\documentclass[conference]{IEEEtran}
\usepackage[hidelinks]{hyperref}
\usepackage{underscore}
\usepackage{xcolor}
\usepackage{enumerate}
\usepackage{wrapfig} 
\usepackage{amssymb}
\makeatletter
\def\@IEEEconsolenoticeconference{}             
\makeatother
\definecolor[named]{Orange}{cmyk}{0,0.82,1,0.01}
\definecolor[named]{Green}{cmyk}{1,0,0.7,0.3}
\definecolor[named]{DarkBlue}{rgb}{0,0,.7}
\definecolor[named]{ACMRed}{cmyk}{0,0.90,0.86,0}
\definecolor[named]{ACMPurple}{cmyk}{0.55,1,0,0.15}
\definecolor[named]{ACMDarkBlue}{cmyk}{1,0.58,0,0.21}
\DeclareSymbolFont{frenchscript}{OMS}{ztmcm}{m}{n}
\makeatletter
\def\comesfrom{\@transition\leftarrowfill}
\def\goesto{\@transition\rightarrowfill}
\def\ngoesto{\@transition\nrightarrowfill}
\def\Goesto{\@transition\Rightarrowfill}
\def\nGoesto{\@transition\nRightarrowfill}
\def\xmapsto{\@transition\mapstofill}
\def\nxmapsto{\@transition\nmapstofill}
\def\@transition#1{\@@transition{#1}}
\newbox\@transbox
\newbox\@arrowbox
\newbox\@downbox
\def\@@transition#1#2%
   {\setbox\@transbox\hbox
      {\vrule height 1.5ex depth .9ex width 0ex\hskip0.25em$\scriptstyle#2$\hskip0.25em}
   \ifdim\wd\@transbox<1.5em
      \setbox\@transbox\hbox to 1.5em{\hfil\box\@transbox\hfil}\fi
   \setbox\@arrowbox\hbox to \wd\@transbox{#1}
   \ht\@arrowbox\z@\dp\@arrowbox\z@
   \setbox\@transbox\hbox{$\mathop{\box\@arrowbox}\limits^{\box\@transbox}$}
   \dp\@transbox\z@\ht\@transbox 10pt
   \mathrel{\box\@transbox}}
\def\nrightarrowfill{$\m@th\mathord-\mkern-6mu%
  \cleaders\hbox{$\mkern-2mu\mathord-\mkern-2mu$}\hfill
  \mkern-6mu\mathord\not\mkern-2mu\mathord\rightarrow$}
\def\Rightarrowfill{$\m@th\mathord=\mkern-6mu%
  \cleaders\hbox{$\mkern-2mu\mathord=\mkern-2mu$}\hfill
  \mkern-6mu\mathord\Rightarrow$}
\def\nRightarrowfill{$\m@th\mathord=\mkern-6mu%
  \cleaders\hbox{$\mkern-2mu\mathord=\mkern-2mu$}\hfill
  \mkern-6mu\mathord\not\mathord\Rightarrow$}
\def\mapstofill{$\m@th\mathord\mapstochar\mathord-\mkern-6mu%
  \cleaders\hbox{$\mkern-2mu\mathord-\mkern-2mu$}\hfill
  \mkern-6mu\mathord\rightarrow$}
\def\nmapstofill{$\m@th\mathord\mapstochar\mathord-\mkern-6mu%
  \cleaders\hbox{$\mkern-2mu\mathord-\mkern-2mu$}\hfill
  \mkern-6mu\mathord\not\mkern-2mu\mathord\rightarrow$}
\makeatother 
\newcommand{\gotoU}[2][ES]{\mathrel{\plat{$\goesto{#2}$}^{\,\aN}_{#1}}}            
\newcommand{\goto}[1]{\mathrel{\goesto{#1}}}            
\newcommand{\gonotto}[1]{\hspace{4pt}\not{\mbox{}}\hspace{-11pt}	
	\stackrel{#1\;}{\longrightarrow}}
\newtheorem{defi}{Definition}
\newtheorem{theo}{Theorem}
\newtheorem{prop}{Proposition}
\newtheorem{lemm}{Lemma}
\newtheorem{corr}{Corollary}
\newtheorem{exam}{Example}
\newtheorem{obse}{Observation}

\newenvironment{definition}{\begin{trivlist} \item[]\begin{defi}}
                           {\end{defi}\end{trivlist}}
\newenvironment{theorem}   {\begin{trivlist} \item[]\begin{theo}}
                           {\end{theo}\end{trivlist}}
\newenvironment{proposition}{\begin{trivlist} \item[]\begin{prop}}
                           {\end{prop}\end{trivlist}}
\newenvironment{lemma}     {\begin{trivlist} \item[]\begin{lemm}}
                           {\end{lemm}\end{trivlist}}
\newenvironment{corollary} {\begin{trivlist} \item[]\begin{corr}}
                           {\end{corr}\end{trivlist}}
\newenvironment{example}   {\begin{trivlist} \item[]\begin{exam}}
                           {\end{exam}\end{trivlist}}
\newenvironment{observation}{\begin{trivlist} \item[]\begin{obse}}
                           {\end{obse}\end{trivlist}}

\newcommand{\df}[1]{Definition~\ref{df:#1}}
\newcommand{\thm}[1]{Theorem~\ref{thm:#1}}

\newcommand{\lem}[1]{Lemma~\ref{lem:#1}}
\newcommand{\ex}[1]{Example~\ref{ex:#1}}
\newcommand{\cor}[1]{Corollary~\ref{cor:#1}}
\newcommand{\obs}[1]{Observation~\ref{obs:#1}}
\newcommand{\sect}[1]{Section~\ref{sec:#1}}
\newcommand{\tab}[1]{Table~\ref{tab:#1}}
\newcommand{\fig}[1]{Figure~\ref{fig:#1}}
\makeatletter     
\newif\if@qeded
\global\@qededfalse
\def\qed{\hfill\rule[0pt]{1.3ex}{1.3ex}\global\@qededtrue}
\def\qedif{\if@qeded\else\qed\fi\global\@qededfalse}
\makeatother
\newenvironment{proof}{\begin{trivlist} \item[\hspace{\labelsep}\it Proof.]}{\qedif\end{trivlist}}
\DeclareMathSymbol{\fT}{\mathord}{frenchscript}{84}  
\DeclareMathSymbol{\A}{\mathord}{frenchscript}{65}   
\DeclareMathSymbol{\Lab}{\mathord}{frenchscript}{76} 
\DeclareMathSymbol{\R}{\mathrel}{frenchscript}{82}   
\DeclareMathSymbol{\C}{\mathord}{frenchscript}{67}   
\DeclareMathSymbol{\Sy}{\mathord}{frenchscript}{83}  
\newcommand{\IN}{\mbox{I\hspace{-1pt}N}}             
\newcommand{\T} {{\rm T}}                    
\newcommand{\E} {P}                          
\newcommand{\F} {Q}                          
\newcommand{\N} {\mathcal{ N}}               
\newcommand{\nN}{\mathcal{ Z}}               
\newcommand{\dN}{\mathcal{ D}}               
\newcommand{\sN}{\mathcal{ B}}               
\newcommand{\pN}{\mathcal{ R}}               
\newcommand{\zN}{\mathcal{ Z}}               
\newcommand{\aN}{\mathcal{ S}}               
\newcommand{\mM}{\mathcal{ H}}               
\newcommand{\M} {\mathcal{ H}}               
\newcommand{\K} {\mathcal{K}}                
\newcommand{\n} {{\rm n}}                    
\newcommand{\Fn}{\textsc{fn}}                
\newcommand{\fn}{{\rm fn}}                   
\newcommand{\dn}{{\rm dn}}                   
\newcommand{\ia}{\iota}                      
\newcommand{\no}{{\rm no}}                   
\newcommand{\bn}{{\rm bn}}                   
\newcommand{\BN}{\mbox{\sc Bn}}              
\newcommand{\RN}{\mbox{\sc Rn}}              
\newcommand{\nil}{\textbf{0}}                
\newcommand{\eqa}{\equiv}                    
\newcommand{\inp}{\mathbin\in}                          
\newcommand{\alt}[2]{#1}                                
\newcommand{\plat}[1]{\raisebox{0pt}[0pt][0pt]{#1}}     
\newcommand{\dom}{{\it dom}}                            
\newcommand{\bis}[2][]{		             		
	\mathrel{\raisebox{1pt}{$\hspace{.1em}\underline{\raisebox{-.5pt}[0em][0pt]{\makebox[.7em]{$\leftrightarrow$}}}$}
                  _{\hspace{.05em}#2}^{\hspace{.1em}#1}}}
\newcommand{\rename}[2]{\mathord{\raisebox{2pt}[0pt]{$#1$}\!/\!#2}} 
\newcommand{\drenb}[2]{\{\mathord{\raisebox{2pt}[0pt]{$#1$}\!/\hspace{-3pt}/\!#2}\}} 
\newcommand{\drensq}[2]{[\mathord{\raisebox{2pt}[0pt]{$#1$}\!/\hspace{-3pt}/\!#2}]} 
\newcommand{\renb}[2]{\{\mathord{\raisebox{2pt}[0pt]{$#1$}\!/\!#2}\}} 
\newcommand{\renbt}[2]{\{\mathord{\raisebox{2pt}[0pt]{$\vec{#1}$}\!/\!\vec{#2}}\}} 
\newcommand{\rensq}[2]{[\mathord{\raisebox{2pt}[0pt]{$#1$}\!/\!#2}]} 
\newcommand{\Ren}[1]{\{\!\{#1\}\!\}} 
\newcommand{\transname}[4][18]{\makebox[#3pt][l]{\raisebox{#1pt}{\textsc{\textbf{\small #2}}:}}
                           #4\hspace{-10pt}\makebox[#3pt]{}}
\newcommand{\debind}[1]{\langle #1 \rangle}
\newcommand{\CCP}{CCS$_\gamma$}                                   
\newcommand{\CCST}{CCS\plat{$_\gamma^{\rm trig}$}}                 
\newcommand{\piIM}{\pi_{\rm IM}}                                  
\newcommand{\piZ}{\pi_{ES}(\nN,\pN)}                              
\newcommand{\piZa}{\pi_{ES}^{\hspace{-.5pt}\not\hspace{.5pt}\alpha}(\nN,\pN)}  
\newcommand{\piWR}{\pi_{ES}^{\dagger}(\nN,\pN)}                    
\newcommand{\piSWR}{\pi_{ES}^{\dagger}(\zN,\pN)}                   
\newcommand{\piRWR}{\pi_{ES}^{\ddagger-}(\nN,\pN)}                 
\newcommand{\piSRWR}{\pi_{ES}^{\ddagger}(\zN,\pN)}                 
\newcommand{\piRR}{\pi_{ES}^{\S}(\zN,\pN)}                        
\newcommand{\sbb}{\stackrel{\scriptscriptstyle\bullet}{\raisebox{0pt}[2pt]{$\sim$}}} 
\newcommand{\as}[1]{\textcolor{ACMPurple}{[}#1\textcolor{ACMPurple}{]}}   
\newcommand{\match}[2]{\textcolor{ACMRed}{[}#1{=}#2\textcolor{ACMRed}{]}} 
\newcommand{\Match}[2]{[#1{=}#2]}                                         
\newcommand{\update}[2]{\textcolor{ACMDarkBlue}{[}#1\mapsto#2\textcolor{ACMDarkBlue}{]}}  
\newcommand{\upd}[2]{\textcolor{ACMDarkBlue}{[}#1\mathbin{\mapsto}#2\textcolor{ACMDarkBlue}{]}}  
\newcommand{\Tcf}{\fT_{\rm cf}^r}                                  
\def\appendixname{Appendix 1}

\begin{document}
\title{Comparing the expressiveness of the $\pi$-calculus and CCS}
\author{\IEEEauthorblockN{Rob van Glabbeek}
\IEEEauthorblockA{Data61, CSIRO, Sydney, Australia\\
School of Computer Science and Engineering, University of New South Wales, Sydney, Australia\\
\texttt{rvg@cs.stanford.edu}}}
\maketitle

\begin{abstract}
This paper shows that the $\pi$-calculus with implicit matching is no more expressive than \CCP, a
variant of CCS in which the result of a synchronisation of two actions is itself an action subject
to relabelling or restriction, rather than the silent action $\tau$.
This is done by exhibiting a compositional translation from the
$\pi$-calculus with implicit matching to {\CCP} that is valid up to strong barbed bisimilarity.

The full $\pi$-calculus can be similarly expressed in {\CCP} enriched with the
triggering operation of {\sc Meije}.

I also show that these results cannot be recreated with CCS in the r\^ole of {\CCP}, not even up to
reduction equivalence, and not even for the asynchronous
$\pi$-calculus without restriction or replication.

Finally I observe that CCS cannot be encoded in the $\pi$-calculus.
\end{abstract}

\section{Introduction}

The $\pi$-calculus \cite{MPWpi1,MPWpi2,Mi99,SW01book} has been advertised as an ``extension to the
process algebra CCS'' \cite{MPWpi1} adding mobility. It is widely believed that the $\pi$-calculus
has features that cannot be expressed in CCS, or other \emph{immobile} process calculi---so called
in \cite{Nestmann06}---such as ACP and CSP\@.\\[1ex]
\mbox{}\hfill\begin{minipage}{3in}{}
``the $\pi$-calculus has a much greater expressiveness than CCS''
\hfill [Sangiorgi \cite{San96}]
\end{minipage}\hfill\mbox{}\\[1ex]
\mbox{}\hfill\begin{minipage}{3in}{}
``Mobility -- of whatever kind -- is important in modern computing. It was not present in CCS or CSP, [...] but
[...] the \emph{$\pi$-calculus} [...] takes mobility of linkage as a primitive notion.''
\hfill [Milner \cite{Mi99}]
\end{minipage}\hfill\mbox{}\\[1ex]
The present paper investigates this belief by formally
comparing the expressive power of the $\pi$-calculus and immobile process calculi.

Following \cite{vG12,vG18} I define one process calculus to be at least as expressive as another
up to a semantic equivalence $\sim$ iff there exists a so-called \emph{valid translation} up to
$\sim$ from the other to the one. Validity entails compositionality, and requires that each
translated expression is $\sim$-equivalent to its original. This concept is parametrised by the
choice of a semantic equivalence that is meaningful for both the source and the target language. Any
language is as expressive as any other up to the universal relation, whereas almost no two
languages are equally expressive up to the identity relation. The equivalence $\sim$ up to which a
translation is valid is a measure for the quality of the translation, and thereby for the degree in
which the source language can be expressed in the target.

Robert de Simone \cite{dS85copy} showed that a wide class of process calculi, including CCS
\cite{Mi90ccs}, CSP \cite{BHR84}, ACP \cite{BK86acp} and SCCS~\cite{Mi83}, are expressible up to
strong bisimilarity in {\sc Meije}~\cite{AB84copy}.
In \cite{vG94a} I sharpened this result by eliminating the crucial r\^ole played by unguarded
recursion in De Simone's translation, now taking aprACP$_R$ as the target language.
Here aprACP$_R$ is a fragment of the language ACP of \cite{BK86acp}, enriched with relational
relabelling, and using action prefixing instead of general sequential composition. It differs
from CCS only in its more versatile communication format, allowing multiway synchronisation instead
of merely handshaking, in the absence of a special action $\tau$, and in the relational nature of
the relabelling operator.
The class of languages that can be translated to {\sc Meije} and aprACP$_R$ are the ones whose
structural operational semantics fits a format due to \cite{dS85copy}, now known as the
\emph{De Simone} format. They can be considered the ``immobile process calculi'' alluded to above.
The $\pi$-calculus does not fit into this class---its operational semantics is not in
De Simone format.

To compare the expressiveness of mobile and immobile process calculi I first of all need to select a
suitable semantic equivalence that is meaningful for both kinds of languages.  A canonical choice is
\emph{strong barbed bisimilarity} \cite{MilS92,SW01book}.
Strong barbed bisimilarity is not a congruence for either CCS or the $\pi$-calculus, but it is used
as a semantic basis for defining suitable congruences on languages \cite{MilS92,SW01book}.  For CCS,
the familiar notion of \emph{strong bisimilarity} \cite{Mi89} arises as the congruence closure of
strong barbed bisimilarity. For the $\pi$-calculus, the congruence closure of strong barbed
bisimilarity yields the notion of \emph{strong early congruence}, called \emph{strong full bisimilarity}
in \cite{SW01book}.  In general, whatever its characterisation in a particular calculus,
\emph{strong barbed congruence} is the name of the congruence closure of strong barbed bisimilarity,
and a default choice for a semantic equivalence \cite{SW01book}.

My first research goal was to find out if there exists a translation from the $\pi$-calculus to CCS that
is valid up to strong barbed bisimilarity. The answer is negative. In fact, no compositional
translation of the $\pi$-calculus to CCS is possible, even when weakening the equivalence
up to which it should be valid from strong barbed bisimilarity to strong reduction equivalence, and even
when restricting the source language to the asynchronous $\pi$-calculus \cite{Bo92} without
restriction and replication. This disproves a result of \cite{BB98}.

\begin{table*}[t]
\caption{Structural operational semantics of CCS}
\label{tab:CCS}
\normalsize
\begin{center}
\framebox{$\begin{array}{cc@{\qquad}c}
\alpha.\E \goto{\alpha} \E &
\displaystyle\frac{\E_j \goto{\alpha} \E_j'}{\sum_{i\in I}\E_i \goto{\alpha} \E_j'}\makebox[0pt][l]{~~($j\in I$)}
 \\[4ex]
\displaystyle\frac{\E\goto{\alpha} \E'}{\E|\F \goto{\alpha} \E'|\F} &
\displaystyle\frac{\E\goto{a} \E' ,~ \F \goto{\bar{a}} \F'}{\E|\F \goto{\tau} \E'| \F'} &
\displaystyle\frac{\F \goto{\alpha} \F'}{\E|\F \goto{\alpha} \E|\F'}\\[4ex]
\displaystyle\frac{\E \goto{\alpha} \E'}{\E\backslash L \goto{\alpha} \E'\backslash L}~~(\alpha\not\in L\cup\bar{L}) &
\displaystyle\frac{\E \goto{\alpha} \E'}{\E[f] \goto{f(\alpha)} \E'[f]} &
\displaystyle\frac{P \goto{\alpha} \E}{A\goto{\alpha}\E}~~(A \stackrel{\rm def}{=} P)
\end{array}$}
\end{center}
\vspace{-6pt}
\end{table*}

My next research goal was to find out if there is a translation from the $\pi$-calculus to any other immobile
process calculus, and if yes, to keep the target language as close as possible to CCS\@.
Here the answer turned out to be positive.
How close the target language can be kept to CCS depends on which version of the $\pi$-calculus I
take as source language. My first choice was the original $\pi$-calculus, as presented in \cite{MPWpi1,MPWpi2},
as it is at least as expressive as its competitors.
It turns out, however, that the matching operator $\Match{x}{y}P$ of \cite{MPWpi1,MPWpi2} is the source
of a complication. The book \cite{SW01book} merely allows matching to occur as part of action prefixing,
as in $\Match{x}{y}u(z).P$ or $\Match{x}{y}\bar uv.P$.
I call this \emph{implicit matching}. Matching was introduced in \cite{MPWpi1,MPWpi2} to facilitate
complete equational axiomatisations of the $\pi$-calculus, and \cite{SW01book} shows that for
that purpose implicit matching is sufficient.

To obtain a valid translation from the $\pi$-calculus with implicit matching (henceforth called $\piIM$) to an upgraded variant of CCS,
the only upgrade needed is to turn the result of a synchronisation of two actions into a visible action, subject
to relabelling or restriction, rather than the silent action $\tau$.
I call this variant {\CCP}, where $\gamma$ is a commutative partial
binary communication function, just like in ACP \cite{BK86acp}.
{\CCP} is a fragment of aprACP$_{\!R}$, which also carries a parameter $\gamma$.
If $\gamma(a,b)=c$, this means that an $a$-action of one component in a parallel composition may
synchronise with a $b$-action of another component, into a $c$-action; if $\gamma(a,b)$ is
undefined, the actions $a$ and $b$ do not synchronise.  CCS can be seen as the instance of {\CCP}
with $\gamma(\bar a,a)=\tau$, and $\gamma$ undefined for other pairs of actions.  But as target
language for my translation I will need another choice of the parameter $\gamma$.

An important feature of ACP, which greatly contributes to its expressiveness, is multiway synchronisation.
This is achieved by allowing an action $\gamma(a,b)$ to synchronise with an action $c$ into
$\gamma(\gamma(a,b),c)$. This feature is not needed for the target language of my translations.
So I require that $\gamma(\gamma(a,b),c)$ is always undefined.

To obtain a valid translation from the full $\pi$-calculus, with an explicit matching operator,
I need to further upgrade {\CCP} with the \emph{triggering} operator of {\sc Meije}, which allows a
relabelling of the first action of its argument only.

By a general result of \cite{vG18}, the validity up to strong barbed bisimilarity of my translation
from $\piIM$ to {\CCP} (and from $\pi$ to {\CCST}) implies that it is even valid up to an equivalence on their
disjoint union that on $\pi$ coincides with strong barbed congruence, or strong early
congruence, and on {\CCST} is the congruence closure of strong barbed bisimilarity under
translated contexts. The latter is strictly coarser than strong bisimilarity, which is the
congruence closure of strong barbed bisimilarity under all {\CCST} contexts.

Having established that $\piIM$ can be expressed in {\CCP}, the possibility remains that
the two languages are equally expressive. This, however, is not the case. There does not exists a
valid translation (up to any reasonable equivalence) from CCS---thus neither from {\CCP}---to the
$\pi$-calculus, even when disallowing the infinite sum of CCS, as well as unguarded recursion.
This is a trivial consequence of the power of the CCS renaming operator, which cannot be mimicked in
the $\pi$-calculus. Using a simple renaming operator that is as finite as the successor function on
the natural numbers, CCS, even without infinite sum and unguarded recursion, allows the
specification of a process with infinitely many weak barbs, whereas this is fundamentally impossible
in the $\pi$-calculus.


\section{CCS}\label{sec:CCS}

CCS \cite{Mi89} is parametrised with a sets $\K$ of {\em agent identifiers} and ${\A}$ of {\em visible actions}.
The set $\overline{\A}$ of {\em co-actions} is $\overline\A:=\{\bar{a}
\mid a\inp {\A}\}$, and $\Lab:=\A \cup \overline\A$ is the set of
\emph{labels}.  The function $\bar{\cdot}$ is extended to $\Lab$ by declaring
$\bar{\bar{\mbox{$a$}}}=a$. Finally, \plat{$Act := \Lab\uplus \{\tau\}$} is the set of
{\em actions}. Below, $a$, $b$, $c$, \ldots range over $\Lab$ and
$\alpha$, $\beta$ over $Act$.
A \emph{relabelling} is a function $f\!:\Lab\mathbin\rightarrow \Lab$ satisfying
\plat{$f(\bar{a})=\overline{f(a)}$}; it extends to $Act$ by $f(\tau):=\tau$.
The class $\T_{\rm CCS}$ of CCS \emph{terms}, \emph{expressions}, \emph{processes} or \emph{agents}
is the smallest class%
\footnote{CCS \cite{Mi89,Mi90ccs} allows arbitrary index sets $I$ in summations $\sum_{i\in I}\!\E_i$.
  As a consequence, $\T_{\rm CCS}$ is a proper class rather than a set. Although this is unproblematic,
  many computer scientists prefer the class of terms to be a set.
  This can be achieved by choosing a cardinal $\kappa$ and requiring the index sets $I$ to satisfy $|I| < \kappa$.
  To enable my translation from the $\pi$-calculus to {\CCST}, $\kappa$ should exceed the size of
  the set of names used in the $\pi$-calculus.}
including:\vspace{-1pt}
\begin{center}
\begin{tabular}{@{}l@{$\;\,$}l@{\quad$\!$}l@{}}
$\alpha.\E$  & \mbox{for $\alpha\inp Act$ and $\E\inp\T_{\rm CCS}$} & \emph{prefixing}\\
$\sum_{i\in I}\!\E_i$ & for $I$ an index set and $\E_i\inp\T_{\rm CCS}$ & \emph{choice} \\
$\E|\F$ & \mbox{for $\E,\F\inp\T_{\rm CCS}$} & \emph{parallel comp.}\\
$\E\backslash L $ & \mbox{for $L\subseteq\Lab$ and $\E\inp\T_{\rm CCS}$} & \emph{restriction} \\
$\E[f]$ & for $f$ a relabelling and $\E\inp\T_{\rm CCS}$ & \emph{relabelling} \\
$A$ & \mbox{for $A\inp\K$} & \emph{recursion}.
\end{tabular}
\end{center}
One writes $\E_1\mathop+\E_2$ for $\sum_{i\in I}\E_i$ when $I\mathop=\{1,2\}$,
and $\nil$ when $I\mathbin=\emptyset$.
Each agent identifier $A \in \K$ comes with a unique \emph{defining equation} of
the form \plat{$A \stackrel{{\rm def}}{=} P$}, with $P \in \T_{\rm CCS}$.
The semantics of CCS is given by the labelled transition relation
$\mathord\rightarrow \subseteq \T_{\rm CCS}\times Act \times\T_{\rm CCS}$.
The transitions \plat{$P\goto{\alpha}Q$} with $P,Q\inp\T_{\rm CCS}$
and $\alpha\inp Act$ are derived from the rules of \tab{CCS}.

Arguably, the most authentic version of CCS \cite{Mi90ccs} features a recursion construct instead of
agent identifiers. Since there exists a straightforward valid transition from the version of CCS presented
here to the one from \cite{Mi90ccs}, the latter is at least as expressive. Therefore, when showing
that a variant of CCS is at least as expressive as the $\pi$-calculus, I obtain a stronger result by
using agent identifiers.

\section[CCS-gamma]{\CCP}\label{sec:gamma}

{\CCP} has four parameters: the same set $\K$ of {\em agent identifiers} as for CCS, an alphabet
$\A$ of \emph{visible actions}, with a subset $\Sy\subseteq \A$ of synchronisations\footnote{These
  have been added solely to prevent multiway synchronisation.}, and a partial
\emph{communication function} $\gamma:(\A{\setminus}\Sy)^2\mathbin\rightharpoonup \Sy\cup\{\tau\}$,
which is commutative, i.e.\ $\gamma(a,b)=\gamma(b,a)$ and each side of this equation is defined
just when the other side is.
Compared to CCS there are no co-actions, so $Act:=\A \uplus\{\tau\}$.

The syntax of {\CCP} is the same as that of CCS, except that parallel composition is denoted $\|$
rather than $|$, following ACP \cite{BK86acp,BW90}. This indicates a semantic difference: the rule
for communication in the middle of \tab{CCS} is for {\CCP} replaced by\vspace{-2ex}
\[\displaystyle\frac{\E\goto{a} \E' ,~ \F \goto{b} \F'}{\E\|\F \goto{c} \E'\| \F'} ~~ (\gamma(a,b)=c).\]
Moreover, relabelling operators $f:\A\rightarrow Act$ are allowed to rename visible actions into
$\tau$, but not vice versa.\footnote{Renaming into $\tau$ could already be done in CCS by means of
  parallel composition. Hence this feature in itself does not add extra expressiveness.}
They are required to satisfy $c\in\Sy \Rightarrow f(c)\in\Sy\cup\{\tau\}$.
These are the only differences between CCS and {\CCP}\@.

\section{Strong barbed bisimilarity}\label{sec:barbed}

The semantics of the $\pi$-calculus and CCS can be expressed by associating a labelled
or a barbed transition system with these languages, with processes as states.
Semantic equivalences are defined on the states of labelled or barbed transition systems, and
thereby on $\pi$- and CCS processes.
\begin{definition}
A \emph{labelled transition system} (LTS) is pair $(S,\rightarrow)$ with $S$ a class (of
\emph{states}) and ${\rightarrow} \subseteq S \times A \times S$ a \emph{transition relation}, for
some suitable set of \emph{actions} $A$.
\end{definition}

\noindent
I write $P \goto{\alpha} Q$ for $(P,\alpha,Q) \inp {\rightarrow}$,
$P {\goto{\alpha}}$ for $\exists Q.~P \goto{\alpha}Q$,\linebreak[4] and \mbox{$P \gonotto{\alpha}$}
for its negation.
The structural operational semantics of CCS presented before creates an LTS with as states all
CCS processes and the transition relation derived from the operational rules, with $A:=Act$.

\begin{definition}\label{df:strong bisimulation}
A \emph{strong bisimulation} is a symmetric relation $\R$ on the states of an LTS
such that
\begin{itemize}
\item if $P\R Q$ and \plat{$P\goto{\alpha} P'$} then $\exists Q'.~Q\goto{\alpha}Q' \wedge P'\R Q'$.
\end{itemize}
Processes $P$ and $Q$ are \emph{strongly bisimilar}---notation
\plat{$P\mathbin{\bis{}} Q$}---if $P \R Q$ for some strong bisimulation $\R$.
\end{definition}

\noindent
As is well-known, $\bis{}$ is an equivalence relation, and a strong bisimulation itself.
Through the operational semantics of {\CCP}, strong bisimilarity is defined on {\CCP} processes.

\begin{definition}
A \emph{barbed transition system} (BTS) is a triple $(S,\mapsto,\downarrow)$ with $S$ a class (of
\emph{states}), ${\mapsto} \subseteq S \times S$ a \emph{reduction relation}, and
${\downarrow} \subseteq S \times B$ an \emph{observability predicate} for some suitable set of \emph{barbs} $B$.
\end{definition}

One writes $P{\downarrow_b}$ for $P \inp S$ and $b\inp B$ when $(P,b) \inp {\downarrow}$.
A BTS can be extracted from an LTS with $\tau\inp A$, by means of a partial observation function
$O:A\rightharpoonup B$. The states remain the same, the reductions are taken to be the transitions
labelled $\tau$ (dropping the label in the BTS), and $P{\downarrow_b}$ holds exactly when there is a
transition \plat{$P \goto{\alpha} Q$} with $O(\alpha)=b$.

In this paper I consider labelled transition systems whose actions $\alpha\inp A$ are of the forms
presented in \tab{actions}. Here $x$ and $y$ are \emph{names}, drawn from the disjoint union of two
sets $\zN$ and $\pN$ of \emph{public} and \emph{private} names, and $M$ is a (possibly empty)
\emph{matching sequence}, a sequence of \emph{matches} $\match{x}{y}$ with $x,y\in\zN\uplus\pN$ and $x\mathbin{\neq}y$.
The set of names occurring in $M$ is denoted $\n(M)$.
In \tab{actions}, also the
\emph{free names} $\fn(\alpha)$ and
\emph{bound names} $\bn(\alpha)$ of an action $\alpha$ are defined.
The set of \emph{names} of $\alpha$ is $\n(\alpha):=\fn(\alpha)\cup\bn(\alpha)$.
Consequently, also the actions $Act$ of my
instantiation of {\CCP} need to have the forms of \tab{actions}.
For the translation into barbed transition systems I take \plat{$B:=\zN\cup\overline\zN$}, where
\plat{$\overline\zN:=\{\bar{a} \mid a\inp {\zN}\}$},
and $O(\alpha)$ as indicated in \tab{actions}, provided $M\mathbin=\varepsilon$ and $O(\alpha)\inp B$.

\begin{table}[t]
\caption{The actions}
\label{tab:actions}
\normalsize
\vspace{-1ex}
\begin{center}
$\begin{array}{@{}|c|l|c|l|l|l|@{}}
\hline
\alpha & \mbox{Kind} & O(\alpha) & \fn(\alpha) & \bn(\alpha) \\
\hline
M\tau     & \mbox{Silent}       &   -    & \emptyset        & \emptyset \\
M\bar xy & \mbox{Free output}  & \bar x & \n(M)\cup\{x,y\} & \emptyset \\
M\bar x(y)&\mbox{Bound output} & \bar x & \n(M)\cup\{x\}   & \{y\}     \\
M xy     & \mbox{Free input}   &      x & \n(M)\cup\{x,y\} & \emptyset \\
M x(y)   & \mbox{Bound input}  &      x & \n(M)\cup\{x\}   & \{y\}     \\
\hline
\end{array}$
\end{center}
\vspace{-1ex}
\end{table}

\begin{definition}\label{df:barbed bisimulation}
A \emph{strong barbed bisimulation} is a symmetric relation $\R$ on the states of a BTS
such that
\begin{itemize}
\item if $P\R Q$ and \plat{$P\longmapsto P'$} then $\exists Q'.~Q\longmapsto Q'\wedge P'\R Q'$
\item and if $P\R Q$ and $P{\downarrow_b}$ then also $Q{\downarrow_b}$.
\end{itemize}
Processes $P$ and $Q$ are \emph{strongly barbed bisimilar}---notation
\plat{$P\stackrel{\scriptscriptstyle\bullet}{\raisebox{0pt}[2pt]{$\sim$}} Q$}---if $P \R Q$ for some strong barbed bisimulation $\R$.
\end{definition}

\noindent
Again, $\sbb$ is an equivalence relation, and a strong barbed bisimulation itself.
Through the above definition, strong barbed bisimilarity is defined on all LTSs occurring in this
paper, as well as on my instantiation of {\CCP}. It can also be used to compare processes from
different LTSs, namely by taking their disjoint union.

\section[The pi-calculus]{The $\pi$-calculus}\label{sec:pi}
\hypertarget{piLate}{}

The $\pi$-calculus \cite{MPWpi1,MPWpi2} is parametrised with an infinite set $\N$ of \emph{names} and,
for each $n\inp\IN$, a set of $\K_n$ of \emph{agent identifiers} of arity $n$.
The set $\T_\pi$ of $\pi$-calculus \emph{terms}, \emph{expressions}, \emph{processes} or
\emph{agents} is the smallest set including:
\begin{center}
\begin{tabular}{@{}lll@{}}
$\nil$ && \emph{inaction}\\
$\tau.P$ & for $P\inp\T_\pi$ & \emph{silent prefix}\\
$\bar xy.P$ & for $x,y\inp \N$ and $P\inp\T_\pi$ & \emph{output prefix}\\
$x(y).P$ & for $x,y\inp \N$ and $P\inp\T_\pi$ & \emph{input prefix}\\
$(\nu y)P$ & for $y\inp\N$ and $P\inp\T_\pi$ & \emph{restriction} \\
$\Match{x}{y}P$ & for $x,y\inp\N$ and $P\inp\T_\pi$ & \emph{match} \\
$P|Q$ & for $P,Q\inp\T_\pi$ & \emph{parallel comp.} \\
$P+Q$ & for $P,Q\inp\T_\pi$ & \emph{choice} \\
$A(y_1,...,y_n)$ & for $A\inp \K_n$ and $y_i\inp\N$ & \emph{defined agent}\\
\end{tabular}
\end{center}
The order of precedence among the operators is the order of the listing above.
A process $\alpha.\nil$ with $\alpha=\tau$ or $\bar x y$ or $x(y)$ is often written $\alpha$.

$\n(P)$ denotes the set of all names occurring in a process $P$.
An occurrence of a name $y$ in a term is \emph{bound} if it occurs in a subterm of the
form $x(y).P$ or $(\nu y)P$; otherwise it is \emph{free}. The set of names occurring free
(resp.\ bound) in a process $P$ is denoted $\fn(P)$ (resp.\ $\bn(P)$).

Each agent identifier $A \inp \K_n$ is assumed to come with a unique \emph{defining equation} of
the form\vspace{-.75ex} $$A(x_1,\ldots,x_n) \stackrel{{\rm def}}{=} P\vspace{-.25ex}$$ where the names $x_i$ are
all distinct and $\fn(P)\subseteq \{x_1,\ldots,x_n\}$.

\textcolor{DarkBlue}{The $\pi$-calculus with implicit matching ($\piIM$) drops the matching operator, instead
allowing prefixes of the form $M\bar xy.P$, $Mx(y).P$ and $M\tau.P$, with $M$ a matching sequence.}

\alt{A \emph{substitution} is a partial function $\sigma\!\!:\!\N\mathbin\rightharpoonup\N$ such
  that $\N{\setminus}(\dom(\sigma)\mathbin\cup{\it range}(\sigma))$ is infinite.}
{A \emph{substitution} is a partial function $\sigma\!\!:\!\N\mathbin\rightharpoonup\N$;
it is \emph{finite} when $\dom(\sigma)$ is finite.}
For $\vec x \mathbin= (x_1,\ldots,x_n)$, $\vec y \mathbin= (y_1,\ldots,y_n)\linebreak[2]\inp\N^n\!$,
$\renbt{y}{x}$ denotes the substitution given by $\sigma(x_i)\mathbin=y_i$
for $1\mathbin\leq i \mathbin\leq n$.
One writes $\{\rename{y}{x}\}$ when $n\mathord=1$.

For $x\inp\N$, $x\as{\sigma}$ denotes $\sigma(x)$ if $x\inp\dom(\sigma)$ and $x$ otherwise;
$M\as{\sigma}$ is the result of changing each occurrence of a name $x$ in $M$ into $x\as{\sigma}$,
while dropping resulting matches $\match{y}{y}$.

For a \alt{substitution}{finite} $\sigma$, the process $P\sigma$ is obtained from $P\inp\T_\pi$
by simultaneous substitution, for all $x\inp\dom(\sigma)$, of $x\as{\sigma}$ for all
free occurrences of $x$ in $P$, with change of bound names to avoid name capture.
\hypertarget{substitution}{A formal inductive definition is:}
\[\begin{array}{r@{~=~}l}
\nil\sigma & \nil\\
(\textcolor{DarkBlue}{M}\tau.P)\sigma & \textcolor{DarkBlue}{M\as{\sigma}}\tau.(P\sigma) \\
(\textcolor{DarkBlue}{M}\bar xy.P)\sigma & \textcolor{DarkBlue}{M\as{\sigma}}\overline{\raisebox{0pt}[5pt]{$x\as{\sigma}$}}y\as{\sigma}.(P\sigma) \\
(\textcolor{DarkBlue}{M}x(y).P)\sigma & \textcolor{DarkBlue}{M\as{\sigma}}x\as{\sigma}(z).(P\renb{z}{y}\sigma) \\
((\nu y)P)\sigma & (\nu z)(P\renb{z}{y}\sigma) \\
(\Match{x}{y}P)\sigma & \Match{x\as{\sigma}}{y\as{\sigma}}(P\sigma) \\
(P|Q)\sigma & (P\sigma)|(Q\sigma) \\
(P+Q)\sigma & (P\sigma) + (Q\sigma) \\
A(\vec{y})\sigma & A(\vec{y}\as{\sigma}) \\
\end{array}\]
where $z$ is chosen outside $\fn((\nu y)P)\cup \dom(\sigma)\cup{\it range}(\sigma)$;
in case $y \notin \dom(\sigma)\cup{\it range}(\sigma)$ one always picks $z:=y$.

A \emph{congruence} is an equivalence relation $\sim$ on $\T_\pi$ such that $P\mathbin\sim Q$ implies
$\tau.P \mathbin\sim \tau.Q$,
$\bar xy.P \mathbin\sim \bar xy.Q$,
$x(y).P \mathbin\sim x(y).Q$,
$(\nu y)P \mathbin\sim (\nu y)Q$,
$\Match{x}{y}P \mathbin\sim \Match{x}{y}Q$,
$P|U \mathbin\sim Q|U$,
$U|P \mathbin\sim U|Q$,
$P+U \mathbin\sim Q+U$ and
$U+P \mathbin\sim U+Q$.
Let $\eqa$ be the smallest congruence on $\T_\pi$ allowing renaming of bound names, i.e.,
that satisfies $x(y).P \eqa x(z).(P\renb{z}{y})$ and $(\nu y)P \eqa (\nu z)(P\renb{z}{y})$
for any $z\notin \fn((\nu y)P)$.
If $P\eqa Q$, then $Q$ is obtained from $P$ by means of \emph{$\alpha$-conversion}.
Due to the choice of $z$ above, substitution is precisely defined only up to $\alpha$-conversion.

Note that $P \eqa Q$ implies that $\fn(P)=\fn(Q)$, and also that $P\sigma \eqa Q\sigma$ for any
\alt{}{finite }substitution $\sigma$.

\section[The semantics of the pi-calculus]{The semantics of the $\pi$-calculus}\label{sec:pi-semantics}

\begin{figure}
  \input{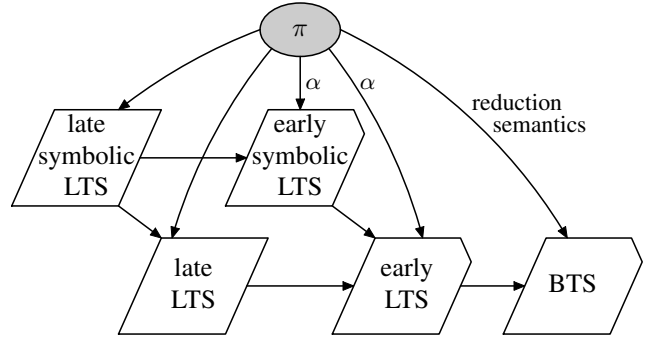}
  \centerline{\raisebox{1ex}{\box\graph}}
  \caption{Semantics of the $\pi$-calculus}\label{fig:pisemantics}
\end{figure}

Whereas CCS has only one operational semantics, the $\pi$-calculus is equipped with at least five,
as indicated in \fig{pisemantics}. The \emph{late} operational semantics stems from \cite{MPWpi2},
the origin of the $\pi$-calculus. It is given by the action rules of \tab{pi}.
These rules generate a labelled transition system in which the states are the
$\pi$-calculus processes and the transitions are labelled with the actions $\tau$, $\bar x y$,
$\bar x(y)$ and $x(y)$ of \tab{actions} (always with $M$ the empty string).
Here I take $\zN:=\N$ and $\pN:=\emptyset$.

\textcolor{DarkBlue}{For $\piIM$, rule \textsc{\textbf{\small match}} is omitted.
A process $\Match{x}{y}\alpha.P$ has no outgoing transitions, similar to $\nil$.}

\begin{table*}[tp]
\caption{Late structural operational semantics of the $\pi$-calculus}
\hypertarget{Late}{}
\label{tab:pi}
\normalsize
\begin{center}
\framebox{$\begin{array}{ccc}
\transname[13]{tau}{25}{         \tau.P \goto{\tau} P }&
\transname[13]{output}{38}{  \bar x y.P \goto{\bar x y} P }&
\hspace{-24pt}
\transname[13]{input}{25}{   x(y).P \goto{x(z)} P\{\rename{z}{y}\}~~(z\mathbin{\not\in} \fn((\nu y)P)) }\\[2ex]
\transname[13]{sum}{20}{
  \displaystyle\frac{P \goto{\alpha} P'}{P+Q \goto{\alpha} P'} }&
\textcolor{purple}{\transname[13]{match}{30}{
  \displaystyle\frac{P \goto{\alpha} P'}{\Match{x}{x}P \goto{\alpha} P'} }}&
\transname[13]{ide}{25}{
  \displaystyle\frac{P\{\rename{\vec{y}}{\vec{x}}\} \goto{\alpha} P'}{A(\vec{y}) \goto{\alpha} P'}
  ~~(A(\vec{x})\stackrel{\rm def}{=} P) }\\[4ex]
\transname{par}{8}{
  \displaystyle\frac{P\goto{\alpha} P'}{P|Q \goto{\alpha} P'|Q}~~(\bn(\alpha)\cap\fn(Q)=\emptyset) }&
\hspace{-3pt}
\transname{com}{26}{
  \displaystyle\frac{P\goto{\bar x y} P' ,~ Q \goto{x (z)} Q'}{P|Q \goto{\tau} P'| Q'\{\rename{y}{z}\}} }&
\transname{close}{32}{
  \displaystyle\frac{P\goto{\bar x(z)} P' ,~ Q \goto{x (z)} Q'}{P|Q \goto{\tau} (\nu z)(P'| Q')} }\\[4ex]
\transname{res}{18}{
  \displaystyle\frac{P \goto{\alpha} P'}{(\nu y)P \goto{\alpha} (\nu y)P'}~~(y\not\in \n(\alpha)) } &
\hspace{-5pt}
\transname{alpha-open}{28}{
  \displaystyle\frac{P \goto{\bar x y} P'}{(\nu y)P \goto{\bar x(z)} P'\{\rename{z}{y}\}}~~
  \left(\begin{array}{@{}l@{}} y \neq x \\ z\mathbin{\not\in} \fn((\nu y)P') \end{array}\right) }\hspace{-75pt}&
\end{array}$}\\[3pt]
The rules \textsc{\textbf{\small sum}}, \textsc{\textbf{\small par}}, \textsc{\textbf{\small com}}
and \textsc{\textbf{\small close}} additionally have symmetric forms, with the r\^oles of $P$ and $Q$ exchanged.
\end{center}
\vspace{2ex}
\caption{Early structural operational semantics of the $\pi$-calculus}
\label{tab:pi-early}
\normalsize
\begin{center}
\framebox{$\begin{array}{c@{\hspace{-3pt}}c@{\hspace{-5pt}}c}
\transname[13]{tau}{25}{         \tau.P \goto{\tau} P }&
\transname[13]{output}{38}{  \bar x y.P \goto{\bar x y} P }&
\hspace{-24pt}
\transname[13]{early-input}{60}{   x(y).P \goto{xz} P\renb{z}{y} }\hspace{-14pt}\\[2ex]
\transname[13]{sum}{20}{
  \displaystyle\frac{P \goto{\alpha} P'}{P+Q \goto{\alpha} P'} }&
\textcolor{purple}{\transname[13]{match}{30}{
  \displaystyle\frac{P \goto{\alpha} P'}{\Match{x}{x}P \goto{\alpha} P'} }}&
\transname[13]{ide}{25}{
  \displaystyle\frac{P\{\rename{\vec{y}}{\vec{x}}\} \goto{\alpha} P'}{A(\vec{y}) \goto{\alpha} P'}
  ~~(A(\vec{x})\stackrel{\rm def}{=} P) }\\[4ex]
\transname{par}{8}{
  \displaystyle\frac{P\goto{\alpha} P'}{P|Q \goto{\alpha} P'|Q}~~(\bn(\alpha)\cap\fn(Q)=\emptyset) }&
\hspace{-13pt}
\transname{early-com}{50}{
  \displaystyle\frac{P\goto{\bar xy} P' ,~ Q \goto{xy} Q'}{P|Q \goto{\tau} P'| Q'} }
  \hspace{-26pt}&
\transname{early-close}{50}{
  \displaystyle\frac{P\goto{\bar x(z)} P' ,~ Q \goto{x z} Q'}{P|Q \goto{\tau} (\nu z)(P'| Q')}~~(z\mathbin{\notin}\fn(Q)) }  \hspace{-40pt}\\[4ex]
\transname{res}{18}{
  \displaystyle\frac{P \goto{\alpha} P'}{(\nu y)P \goto{\alpha} (\nu y)P'}~~(y\not\in \n(\alpha)) } &
\transname{open}{30}{
  \displaystyle\frac{P \goto{\bar x y} P'}{(\nu y)P \goto{\bar x(y)} P'}~~
  \left(\begin{array}{@{}l@{}} y \neq x \end{array}\right) }\hspace{-55pt}&
\transname{alpha}{26}{
  \displaystyle\frac{P \eqa Q, \quad Q \goto{\alpha} Q'}{P \goto{\alpha} Q'}}
\end{array}$}
\end{center}
\end{table*}

\begin{table*}[t]
\caption{Late symbolic structural operational semantics of the $\pi$-calculus}
\label{tab:pi-late-symbolic}
\hypertarget{LS}{}
\normalsize
\begin{center}
\framebox{$\begin{array}{c@{\hspace{-11pt}}c@{\hspace{-11pt}}c@{\hspace{-3pt}}}
\transname[13]{tau}{25}{         \textcolor{DarkBlue}{M}\tau.P \goto{\textcolor{DarkBlue}{M}\tau} P }&
\transname[13]{output}{38}{  \textcolor{DarkBlue}{M}\bar x y.P \goto{\textcolor{DarkBlue}{M}\bar x y} P }&
\hspace{-48pt}
\transname[13]{input}{20}{   \textcolor{DarkBlue}{M}x(y).P \goto{\textcolor{DarkBlue}{M}x(z)} P\{\rename{z}{y}\}~~(z\mathbin{\not\in} \fn((\nu y)P)) }\\[2ex]
\transname[13]{sum}{20}{
  \displaystyle\frac{P \goto{\alpha} P'}{P+Q \goto{\alpha} P'} }&
\textcolor{purple}{\transname[13]{symb-match}{45}{
  \displaystyle\frac{P \goto{\alpha} P'}{\Match{x}{y}P \goto{\match{x}{y}\alpha} P'} }}&
\hspace{-14pt}
\transname[13]{ide}{25}{
  \displaystyle\frac{P\{\rename{\vec{y}}{\vec{x}}\} \goto{\alpha} P'}{A(\vec{y}) \goto{\alpha} P'}
  ~~(A(\vec{x})\stackrel{\rm def}{=} P) }\\[4ex]
\transname{par}{8}{
  \displaystyle\frac{P\goto{\alpha} P'}{P|Q \goto{\alpha} P'|Q}~~(\bn(\alpha)\cap\fn(Q)=\emptyset) }&
\hspace{-20pt}
\transname{symb-com}{43}{
  \displaystyle\frac{P\goto{M \bar x y} P' ,~ Q \goto{N v (z)} Q'}{P|Q \goto{\match{x}{v}MN\tau} P'| Q'\{\rename{y}{z}\}} }&
\hspace{-40pt}
\transname{symb-close}{50}{
  \displaystyle\frac{P\goto{M \bar x(z)} P' ,~ Q \goto{N v (z)} Q'}{P|Q \goto{\match{x}{v}MN\tau} (\nu z)(P'| Q')} }\hspace{-45pt}\\[4ex]
\transname{res}{18}{
  \displaystyle\frac{P \goto{\alpha} P'}{(\nu y)P \goto{\alpha} (\nu y)P'}~~(y\not\in \n(\alpha)) } &
\transname{symb-alpha-open}{58}{
  \displaystyle\frac{P \goto{M \bar x y} P'}{(\nu y)P \goto{M \bar x(z)} P'\{\rename{z}{y}\}}~~
  \left(\begin{array}{@{}l@{}} y \neq x \\ z\mathbin{\not\in} \fn((\nu y)P')  \\
                               y \notin \n(M)\end{array}\right)}\hspace{-85pt}&
\end{array}$}\\[3pt]
For the $\pi$-calculus, the \textcolor{DarkBlue}{blue $M$s} are omitted; for $\piIM$ the
\textcolor{purple}{purple rules}.
\end{center}
\vspace{2ex}
\hypertarget{ES}{}
\caption{Early symbolic structural operational semantics of the $\pi$-calculus}
\label{tab:pi-early-symbolic}
\normalsize
\begin{center}
\framebox{$\begin{array}{c@{\hspace{-1pt}}c@{\hspace{-3pt}}c}
\transname[13]{tau}{25}{         \textcolor{DarkBlue}{M}\tau.P \goto{\textcolor{DarkBlue}{M}\tau} P }&
\transname[13]{output}{38}{  \textcolor{DarkBlue}{M}\bar x y.P \goto{\textcolor{DarkBlue}{M}\bar x y} P }&
\hspace{-30pt}
\transname[13]{early-input}{55}{  \textcolor{DarkBlue}{M} x(y).P \goto{\textcolor{DarkBlue}{M}xz} P\renb{z}{y} }\hspace{-24pt}\\[2ex]
\transname[13]{sum}{20}{
  \displaystyle\frac{P \goto{\alpha} P'}{P+Q \goto{\alpha} P'} }&
\hspace{-10pt}
\textcolor{purple}{\transname[13]{symb-match}{45}{
  \displaystyle\frac{P \goto{\alpha} P'}{\Match{x}{y}P \goto{\match{x}{y}\alpha} P'} }}\hspace{-20pt}&
\hspace{-20pt}
\transname[13]{ide}{25}{
  \displaystyle\frac{P\{\rename{\vec{y}}{\vec{x}}\} \goto{\alpha} P'}{A(\vec{y}) \goto{\alpha} P'}
  ~~(A(\vec{x})\stackrel{\rm def}{=} P) }\hspace{-20pt}\\[4ex]
\transname{par}{8}{
  \displaystyle\frac{P\goto{\alpha} P'}{P|Q \goto{\alpha} P'|Q}~~\textcolor{Orange}{(\bn(\alpha)\cap\fn(Q)=\emptyset) }}&
\hspace{-75pt}
\transname{e-s-com}{35}{
  \displaystyle\frac{P\goto{M\bar xy} P' ,~ Q \goto{Nvy} Q'}{P|Q \goto{\match{x}{v}MN\tau} P'| Q'} }
  \hspace{-75pt}&
\textcolor{Orange}{\transname{e-s-close}{35}{
  \displaystyle\frac{P\goto{M\bar x(z)} P' ,~ Q \goto{Nv z} Q'}{P|Q \goto{\match{x}{v}MN\tau} (\nu z)(P'| Q')}~~(z\mathbin{\notin}\fn(Q)) }}  \hspace{-30pt}\\[4ex]
\textcolor{Orange}{\transname{res}{18}{
  \displaystyle\frac{P \goto{\alpha} P'}{(\nu y)P \goto{\alpha} (\nu y)P'}~~(y\not\in \n(\alpha)) }} &
\hspace{-40pt}
\textcolor{Orange}{\transname{symb-open}{40}{
  \displaystyle\frac{P \goto{M \bar x y} P'}{(\nu y)P \goto{M\bar x(y)} P'}~~
  \left(\begin{array}{@{}l@{}} y \neq x \\  y \notin \n(M)\end{array}\right) }}\hspace{-65pt}&
\hspace{-10pt}
\textcolor{Orange}{\transname{alpha}{26}{
  \displaystyle\frac{P \eqa Q, \quad Q \goto{\alpha} Q'}{P \goto{\alpha} Q'}}}\hspace{-10pt}
\end{array}$}\vspace{-1pt}
\end{center}
\vspace{-1ex}
\end{table*}

In \cite{MPWpi2} the \emph{late} and \emph{early} bisimulation semantics
of the $\pi$-calculus were proposed.

\begin{definition}\label{late bisimulation}
A \emph{late bisimulation} is a symmetric relation $\R$ on $\pi$-processes
such that, whenever $P\R Q$, $\alpha$ is either $\tau$ or $\bar xy$
and $z \not\in\n(P)\cup\n(Q)$, 
\begin{enumerate}
\item if $P\goto{\alpha} P'$ then $\exists Q'$ with $Q\goto{\alpha}Q'$ and $P'\R Q'$,
\item if $P\mathbin{\goto{x(z)}} P'$ then  $\exists Q'\forall y.~Q\mathbin{\goto{x(z)}}Q'
  \wedge P'\renb{y}{z}\R Q'\renb{y}{z}$,
\item if $P\goto{\bar x(z)} P'$ then $\exists Q'$ with $Q\goto{\bar x(z)}Q'$ and $P'\R Q'$.
\end{enumerate}
Processes $P$ and $Q$ are \emph{late bisimilar}---notation
$P\stackrel{.}{\sim}_L Q$---if $P \R Q$ for some late bisimulation $\R$.
They are \emph{late congruent}---notation
$P\mathbin{\sim_L} Q$---if $P\{\rename{\vec{y}}{\vec{x}}\} \mathbin{\stackrel{.}{\sim}_L}
Q\{\rename{\vec{y}}{\vec{x}}\}$ for any substitution $\{\rename{\vec{y}}{\vec{x}}\}$.
\end{definition}
\noindent
Early bisimilarity ($\stackrel{.}{\sim}_E$) and congruence ($\sim_E$) are defined likewise, but
with $\forall y\exists Q'$ instead of $\exists Q'\forall y$.
In \cite{MPWpi2,SW01book} it is shown that $\stackrel{.}{\sim}_L$ and $\stackrel{.}{\sim}_E$ are congruences
for all operators of the $\pi$-calculus, except for the input prefix.
$\sim_E$ and $\sim_L$ are congruence relations for the entire language; in fact they are the
congruence closures of $\stackrel{.}{\sim}_L$ and $\stackrel{.}{\sim}_E$, respectively.
By definition, $\mathord{\stackrel{.}{\sim}_L}\subseteq \mathord{\stackrel{.}{\sim}_E}$, and thus
$\mathord{\sim_L}\subseteq \mathord{\sim_E}$.

\begin{lemma}[\cite{MPWpi2}]\rm\label{lem:alpha}
  Let $P\eqa Q$ and $\bn(\alpha) \cap \n(Q)=\emptyset$.\\
  If $P\goto{\alpha}P'$ then $Q\goto{\alpha}Q'$ for some $Q'$ with $P' \eqa Q'$.
\end{lemma}

\noindent
This implies that $\eqa$ is a late bisimulation, so that ${\eqa} \subset {\sim_L}$.

In \cite{MPWpi3} the \emph{early} operational semantics of the $\pi$-calculus is proposed, presented
in \tab{pi-early}; it uses free input actions $xy$ instead of bound inputs $x(y)$. This is also the
semantics of \cite{SW01book}. The semantics in \cite{MPWpi3,SW01book}
requires us to identify processes modulo $\alpha$-conversion before applying the operational rules.
This is equivalent to adding rule \textsc{\textbf{\small alpha}}
of \tab{pi-early}.

A variant of the late operational semantics incorporating rule \textsc{\textbf{\small alpha}} is also possible.
In this setting rule \textsc{\textbf{\small alpha-open}} can be simplified to \textsc{\textbf{\small open}},
and likewise \textsc{\textbf{\small input}} to \plat{$x(y).P \goto{x(y)} P$}.
By \lem{alpha}, the late operational semantics with \textsc{\textbf{\small alpha}} gives rise
to the same notions of early and late bisimilarity as the late operational semantics without
\textsc{\textbf{\small alpha}}; the addition of this rule is entirely optional.
Interestingly, the rule \textsc{\textbf{\small alpha}} is not optional in the early operational
semantics, not even when reinstating \textsc{\textbf{\small alpha-open}}.
\begin{example}\label{ex:alpha}
Consider the process $P:=\bar xy | (\nu y)(x(z))$. One has \plat{$(\nu y)(x(z)) \goto{x(z)}_L (\nu y)\nil$}
and thus \plat{$P \goto{\tau}_L \nil|(\nu y)\nil$} by \textsc{\textbf{\small com}}.
However, $(\nu y)(x(z)) \goto{xy}_E (\nu y)\nil$ is
forbidden by the side condition of \textsc{\textbf{\small res}}, so in the early semantics without 
\textsc{\textbf{\small alpha}} process $P$ cannot make a $\tau$-step.
Rule \textsc{\textbf{\small alpha}} comes to the rescue here, as it allows
\plat{$P \eqa \bar xy | (\nu w)(x(z)) \goto{\tau}_E \nil|(\nu w)\nil$}.
\end{example}

\noindent
By the following lemma, the early transition relation \plat{$\goto{}_{E}$} is completely determined by
the late transition relation \plat{$\goto{}_{\alpha L}$} with \textsc{\textbf{\small alpha}}:

\begin{lemma}[\cite{MPWpi3}]\rm\label{lem:aL=E}
Let $P\inp\T_\pi$ and $\beta$ be $\tau$, $\bar x y$ or $\bar x(y)$.
  \begin{itemize}
  \item $\!\!P \mathbin{\goto{\beta}_{E}} Q$ iff $P \mathbin{\goto{\beta}_{\alpha L}} Q$.
  \item $\!\!P \mathbin{\goto{xy}_{E}} Q$ iff $P \mathbin{\goto{x(z)}_{\alpha L}} R$
      for some $R$,\,$z$ with $Q\mathbin{\eqa}R\renb{y}{z}$.
  \end{itemize}
\end{lemma}

\noindent
The early transition relations allow a more concise definition of early bisimilarity:
\begin{proposition}[\cite{MPWpi3}]\rm
An \emph{early bisimulation} is a symmetric relation $\R$ on $\T_\pi$
such that, whenever $P\R Q$ and $\alpha$ is an action with
$\bn(\alpha) \cap (\n(P) \cup \n(Q))=\emptyset$,
\begin{itemize}
\item if $P\goto{\alpha}_E P'$ then $\exists Q'$ with $Q\goto{\alpha}_E Q'$ and $P'\R Q'$.
\end{itemize}
Processes $P$ and $Q$ are early bisimilar iff $P \R Q$ for some early bisimulation $\R$.
\end{proposition}

Through the general method of \sect{barbed}, taking $\zN:=\N$ and $\pN:=\emptyset$, a barbed
transition system can be extracted from the late or early labelled transition system of the
$\pi$-calculus; by Lemmas~\ref{lem:alpha} and~\ref{lem:aL=E} the same BTS is obtained
either way. This defines strong barbed bisimilarity $\sbb$ on $\T_\pi$.
The congruence closure of $\sbb$ is early congruence \cite{SW01book}.
In \cite{Mi92} a \emph{reduction semantics} of the $\pi$-calculus is given, that yields a BTS right away.
Up to strong barbed bisimilarity, this BTS is the same as the one extracted from the late or early LTS\@.

In \cite{San96} yet another operational semantics of the $\pi$-calculus was
introduced, in a style called \emph{symbolic} by Hennessy \&
Lin \cite{HL95}, who had proposed it for a version of value-passing CCS\@.
It is presented in \tab{pi-late-symbolic}.
The transitions are labelled with actions $\alpha$ of the form $M\beta$, where $M$ is a
matching sequence and $\beta$ an action as in the late operational semantics.
When $x\mathbin{\neq} y$ the matching sequence $M$ prepended with $\match{x}{y}$ is denoted $\match{x}{y}M$;
however, $\match{x}{x}M$ simply denotes $M$.

In the operational semantics of CCS, $\tau$-actions can be thought of as reactions that
actually take place, whereas a transition labelled $a$ merely represents the potential of a
reaction with the environment, one that can take place only if the environment offers a
complementary transition $\bar a$. In case the environment never does an $\bar a$, this
potential will not be realised. A reduction semantics (as in \cite{Mi99}) yields a BTS that only represents directly the
realised actions---the $\tau$-transitions or \emph{reductions}---and reasons about the
potential reactions by defining the semantics of a system in terms of reductions that can
happen when placing the system in various contexts. An LTS,
on the other hand, directly represents transitions that could happen under some
conditions only, annotated with the conditions that enable them. For CCS, this annotation is
the label $a$, saying that the transition is conditional on an $\bar a$-signal from the
environment. As a result of this, semantic equivalences defined on labelled transitions systems
tend to be congruences for most operators right away, and do not need much closure under contexts.

Seen from this perspective, the operational semantics of the $\pi$-calculus of \tab{pi} or~\ref{tab:pi-early} is a
compromise between a pure reduction semantics and a pure labelled transition system semantics.
Input and output actions are explicitly included to signal potential reactions that are
realised in the presence of a suitable communication partner, but actions whose occurrence is
conditional on two different names $x$ and $y$ denoting the same channel are entirely omitted,
even though any $\pi$-process can be placed in a context in which $x$ and $y$ will be identified.
As a consequence of this, the early and late bisimilarities need to be closed under all
possible substitutions or identifications of names before they turn into early and late
congruences.
The operational semantics of \tab{pi-late-symbolic} adds the conditional transitions that where
missing in \tab{pi}, and hence can be seen as a true labelled transition system semantics.

In this paper I need the early symbolic operational semantics of the $\pi$-calculus, presented in
\tab{pi-early-symbolic}.  Although new, it is the logical combination of the early and the (late)
symbolic semantics.  Its transitions that are labelled with actions having an empty matching
sequence are exactly the transitions of the early semantics, so the BTS extracted from this
semantics is the same.

\textcolor{DarkBlue}{For $\piIM$, rule \textsc{\textbf{\small symb-match}} is omitted, but
\textsc{\textbf{\small tau}}, \textsc{\textbf{\small output}} and \textsc{\textbf{\small input}}
carry the matching sequence $M$ (indicated in blue).}

\section{Valid translations}

A \emph{signature} $\Sigma$ is a set of \emph{operator symbols} $g$, each of which is equipped with
an \emph{arity} $n\in\IN$. The set $\T_\Sigma$ of \emph{closed terms} over $\Sigma$ is the smallest
set such that, for all $g \in \Sigma$,
\[P_1,\dots,P_n \in \T_\Sigma ~~\Rightarrow~~ g(P_1,\dots,P_n) \in \T_\Sigma\;.\]
Call a language \emph{simple} if its expressions are the closed terms $\T_\Sigma$ over some
signature $\Sigma$. The $\pi$-calculus is simple in this sense; its signature consists of
the binary operators $+$ and $|$,
the unary operators $\tau$, $\bar xy.$, $x(y).$, $(\nu y)$ and $\Match{x}{y}$ for $x,y\in\N$,
and 
the nullary operators (or \emph{constants}) $\nil$ and $A(y_1,\dots,y_n)$ for $A \in \K_n$ and $y_i\in \N$.
CCS is not quite simple, since it features the infinite choice operator.

Let $\mathcal{L}$ be a language.
An $n$-ary $\mathcal{L}$-context $C$ is an $\mathcal{L}$-expression that may contain
special \emph{variables} $X_1,\dots,X_n$---its \emph{holes}.
For $C$ an $n$-ary context, $C[P_1,\dots,P_n]$ is the result of substituting $P_i$ for $X_i$,
for each $i=1,\dots,n$.

\begin{definition}
Let $\mathcal{L}'$ and $\mathcal{L}$ languages, generating sets of closed
terms $\T_{\mathcal{L}'}$ and $\T_{\mathcal{L}}$. Let $\mathcal{L}'$ be simple, with signature $\Sigma$.\linebreak[4]
A \emph{translation} from $\mathcal{L}'$ to $\mathcal{L}$
(or an \emph{encoding} from $\mathcal{L}'$ into $\mathcal{L}$)
is a function $\fT:\T_{\mathcal{L}'} \rightarrow \T_{\mathcal{L}}$.
It is \emph{compositional} if for each $n$-ary operator $g\in\Sigma$ there exists an $n$-ary
$\mathcal{L}$-context $C_g$ such that $\fT(g(P_1,\dots,P_n)) = C_g[\fT(P_1),\dots,\fT(P_n)]$.

Let $\sim$ be an equivalence relation on $\T_{\mathcal{L}'} \cup \T_{\mathcal{L}}$.
A translation $\fT$ from $\mathcal{L}'$ to $\mathcal{L}$ is \emph{valid up to $\sim$}
if it is compositional and $\fT(P) \sim P$ for each $P\in \T_{\mathcal{L}'}$.
\end{definition}
The above definition stems in essence from \cite{vG12,vG18}, but could be simplified here
since \cite{vG12,vG18} also covered the case that $\mathcal{L}'$ is not simple.
Moreover, here I restrict attention to what are called \emph{closed term languages} in \cite{vG18}.

\section[The unencodability of pi into CCS]{The unencodability of $\pi$ into CCS}

In this section I show that there exists no translation of the $\pi$-calculus to CCS that is valid
up to $\sbb$. I even show this for the fragment \plat{$\pi^\P_A$} of the
(asynchronous) $\pi$-calculus without choice, recursion, matching and restriction
(thus only featuring inaction, action prefixing and parallel composition).

\begin{definition}\label{df:reduction bis}
  \emph{Strong reduction bisimilarity}, $\bis{r}$, is defined just as strong barbed equivalence in
  \df{barbed bisimulation}, but without the requirement on barbs.
\end{definition}
\noindent
I show that there is no translation of \plat{$\pi^\P_A$} to CCS that is valid up to $\bis{r}$.
As $\bis{r}$ is coarser than $\stackrel{\scriptscriptstyle\bullet}{\raisebox{0pt}[2pt]{$\sim$}}$,
this implies my claim above.
It may be useful to read this section in parallel with the first half of \sect{related}.

\begin{definition}\label{df:unguarded}
Let ${\leftarrowtail}$ be the smallest preorder on CCS contexts such that
$\sum_{i\in I}E_i \leftarrowtail E_j$ for all $j\mathbin\in I$,
$E|F \leftarrowtail E$,
$E|F \leftarrowtail F$,
$E\backslash L \leftarrowtail E$,
$E[f] \leftarrowtail E$ and
$A \leftarrowtail P$ for all $A \in \K$ with $A \stackrel{\rm def}{=} P$.
A variable $X$ occurs \emph{unguarded} in a context $E$ if $E \leftarrowtail X$.
\end{definition}

\noindent
If the hole $X_1$ occurs unguarded in the unary context $E[\;\,]$ and
$U \goto{\tau}$ (resp.\ $U \goto{\tau}\goto{\tau}$) then
$E[U] \goto{\tau}$ (resp.\ $E[U] \goto{\tau}\goto{\tau}$).

\begin{lemma}\label{lem:binary}\rm
Let $E[\;\,]$ be a unary and $C[\;,\;]$ a binary CCS context, and $P,Q,P',Q',U\in\T_{\rm CCS}$.
If $E[C[P,Q]]\goto{\tau}$ and $U\goto{\tau}$ but neither
$E[C[P',Q]]\goto{\tau}$
nor $E[C[P,Q']]\goto{\tau}$
nor $E[U]\goto{\tau}\goto{\tau}$,
then $C[P,Q]\goto{\tau}$.
\end{lemma}

\noindent
\begin{wrapfigure}[6]{r}{0.1\textwidth}
  \vspace{-2.8ex}
  \input{prooftrees}
  \centerline{\box\graph}\vspace{3pt}
  \centerline{\footnotesize CCS proof trees~}
\end{wrapfigure}
{\it Proof.}
Since the only rule in the operational semantics of CCS with multiple premises has a conclusion
labelled $\tau$, it can occur at most once in the derivation of a CCS transition.
Thus, such a derivation is a tree with at most two branches, as illustrated at the right.
Now consider the derivation of \plat{$E[C[P,Q]]\goto{\tau}$}. If none of its branches prods into the
subprocess $P$, the transition would be independent on what is substituted here, thus yielding
\plat{$E[C[P',Q]]\goto{\tau}$}. Thus, by symmetry, both $P$ and $Q$ are visited by branches of this proof.
It suffices to show that these branches come together within the context $C$, as this implies \plat{$C[P,Q]\goto{\tau}$}.
So suppose, towards a contradiction, that the two branches come together in $E$.
Then $E$ must have the form $E_1[E_2[\;\,]|E_3[\;\,]]$, where the hole $X_1$ occurs unguarded in $E_2$, $E_3$ as
well as $E_1$. But in that case \plat{$E[U]\goto{\tau}\goto{\tau}$}, contradicting the assumptions.
\qed

\begin{lemma}\label{lem:ternary}\rm
  If $D[\;\,,\;\,,\;\,]$ is a ternary CCS context, $P_1,P_2,P_3 \inp \T_{\rm CCS}$,
  and \plat{$D[P_1,P_2,P_3] \goto{\tau}$}, then
  there exists an $i\mathbin{\in}\{1,2,3\}$ and a CCS context $E[\;\,]$
  such that $D'[P] \goto{\tau} E[P]$ for any $P\in\T_{\rm CCS}$.
  Here $D'$ is the unary context obtained from $D[\;\,,\;\,,\;\,]$ by
  substituting $P_j$ for the hole $X_j$, for all $j \mathbin{\in}\{1,2,3\}$, $j\neq i$.
\end{lemma}
\begin{proof}
  Since the derivation of \plat{$D[P_1,P_2,P_3] \goto{\tau}$} has at most two branches, one of the
  $P_i$ is not involved in this proof at all. Thus, the derivation remains valid if any other
  process $P$ is substituted in the place of that $P_i$; the target of the transition remains
  the same, except for $P$ taking the place of $P_i$ in it.
  \qed
\end{proof}

\begin{theorem}\label{thm:unencodable}\rm
  There is no translation from \plat{$\pi^\P_A$}
  to CCS that is valid up to $\bis{r}$.
\end{theorem}

\begin{proof}
\newcommand{\op}{|}
Suppose, towards a contradiction, that $\fT$ is a translation from \plat{$\pi^\P_A$} to CCS that is
valid up to $\bis{r}$. By definition, this means that $\fT$ 
is compositional and that $\fT(P) \bis{r} P$ for any \plat{$\pi^\P_A$}-process $P$.

As $\fT$ is compositional, there exists a ternary CCS context $D[\;\,,\;\,,\;\,]$ such that,
for any \plat{$\pi^\P_A$}-processes $R,S,T$,
\[\fT\big(\bar xv \mathbin{\big|} x(y).(R{\op}S{\op}T)\big)=D[\fT(R),\fT(S),\fT(T)].\]
Since $\bar xv \big| x(y).(\nil{\op}\nil{\op}\nil) \!\!\goto{\tau} $ as well as\vspace{-2pt}
$\fT\big(\bar xv \big| x(y).(\nil{\op}\nil{\op}\nil)\big) \bis{r} \linebreak[4] \bar xv \big| x(y).(\nil{\op}\nil{\op}\nil)$,
it follows that $\fT\big(\bar xv \big| x(y).(\nil{\op}\nil{\op}\nil) \big) \goto{\tau} $,
i.e., \plat{$D[\fT(\nil),\fT(\nil),\fT(\nil)] \goto{\tau} $}.
Hence \lem{ternary} can be applied. For simplicity I assume that $i\mathbin=1$; the other two cases
proceed in the same way. So there is a CCS context $E[\;\,]$ such that
$D[P,\fT(\nil),\fT(\nil)] \mathop{\goto{\tau}} E[P]$ for all CCS terms $P$. In particular,
$\fT\big((\bar xv \big| x(y).(R{\op}\nil{\op}\nil)\big) \mathbin= D[\fT\!(R),\fT\!(\nil),\fT\!(\nil)] \goto{\tau} E[\fT\!(R)]$
for all \plat{$\pi^\P_A$}-processes $R$.
\hfill \refstepcounter{equation}(\theequation)\let\four\theequation

I examine the translations of the $\pi$-calculus expressions $\bar xv \big| x(y).(R{\op}\nil{\op}\nil)$, for
$R\in\{\textcolor{ACMPurple}{\bar yz|v(w)}, \textcolor{ACMRed}{0 |v(w)},
\textcolor{ACMDarkBlue}{\bar yz|0},\textcolor{Green}{\tau}\}$.

Since $\bar xv \big| x(y).(\textcolor{ACMPurple}{\bar yz|v(w)}{\op}\nil{\op}\nil)\goto{\tau}\goto{\tau}$
and $\fT$ respects $\bis{r}$,\vspace{-1ex}
\[
\fT\big(\bar xv \big| x(y).(\textcolor{ACMPurple}{\bar yz|v(w)}{\op}\nil{\op}\nil)\big)\goto{\tau}\goto{\tau}
\vspace{-1ex}
\]
In the same way, neither
$\fT\big(\bar xv \big| x(y).(\textcolor{ACMRed}{0 |v(w)}{\op}\nil{\op}\nil)\big)\goto{\tau}\goto{\tau}$
nor $\fT\big(\bar xv \big| x(y).(\textcolor{ACMDarkBlue}{\bar yz|0} {\op}\nil{\op}\nil)\big)\goto{\tau}\goto{\tau}$.
\hfill \refstepcounter{equation}(\theequation)\\
Furthermore, since $\fT$ respects $\bis{r}$ and there is no $S\inp\T_\pi$ such that
$\bar xv | x(y). (\textcolor{ACMPurple}{\bar yz|v(w)}{\op}\nil{\op}\nil) \goto{\tau} S \gonotto{\tau}$,
there is no $S\inp\T_{\rm CCS}$ with
$\fT\big(\bar xv | x(y). (\textcolor{ACMPurple}{\bar yz|v(w)}{\op}\nil{\op}\nil)\big) \goto{\tau} S \gonotto{\tau}$.
\hfill \refstepcounter{equation}(\theequation)

By (1) and (3), $E[\fT(\textcolor{ACMPurple}{\bar yz|v(w)})]\goto{\tau}$.

By (1) and (2), $E[\fT(\textcolor{ACMRed}{0|v(w)})]\gonotto{\tau}$
and $E[\fT(\textcolor{ACMDarkBlue}{\bar yz|0})]\gonotto{\tau}$.

Since $\fT$ is compositional, there is a binary CCS context $C_|[\;,\;]$ such that
$\fT(P|Q)=C_|[\fT(P),\fT(Q)]$ for any $P,Q\in\T_\pi$.
It follows that
\[\begin{array}{l}
E[C_|[\fT(\bar yz),\fT(v(w))]]\goto{\tau}\\
E[C_|[\fT(\nil),\fT(v(w))]]\gonotto{\tau}\\
E[C_|[\fT(\bar yz),\fT(\nil)]]\gonotto{\tau}.
\end{array}\]
Moreover since $\tau\goto{\tau}$, also $U:=\fT(\tau) \goto{\tau}$,
but, since it is not the case that 
$\bar xv \big| x(y).(\textcolor{Green}{\tau}{\op}\nil{\op}\nil)\goto{\tau}\goto{\tau}\goto{\tau}$,
neither holds
$\fT\big(\bar xv \big| x(y).(\textcolor{Green}{\tau}{\op}\nil{\op}\nil)\big)\goto{\tau}\goto{\tau}\goto{\tau}$,
and neither $E[U]\goto{\tau}\goto{\tau}$.
So by \lem{binary},
$\fT(\bar yz|v(w))=C_|[\fT(\bar yz),\fT(v(w))]\goto{\tau}$, yet $\bar yz | v(w) \gonotto{\tau}$.
This contradicts the validity of $\fT$ up to $\bis{r}$.\vspace{-1pt}
\end{proof}

\section[A valid translation of pi into CCS-gamma]{A valid translation of $\piIM$ into {\CCP}}\label{sec:the encoding}

Given a set $\N$ of names, I now define the parameters $\K$, $\A$ and $\gamma$ of the language
{\CCP} that will be the target of my encoding. First of all, $\K$ will be the disjoint union of all
the sets $\K_n$ for $n \in \IN$, of $n$-ary agent identifiers from the chosen instance of the $\pi$-calculus.

Take $p \notin \N$. Let $\pN_0 := \{{}^{\varsigma}\! p \mid \varsigma\inp\{e,\ell,r\}^*\}$.
The set $\pN$ of \emph{private names} is $\{u^\upsilon \mid u \in \pN_0 \wedge \upsilon\inp\{'\}^*\}$.
Let $\aN = \{s_1,s_2,\dots\}$ be an infinite set of \emph{spare names}, disjoint from $\N$ and $\pN$.
Let $\zN:= \N \uplus \aN$ and $\M:= \zN \uplus \pN$.%
\footnote{The names in $\aN$ and in $\pN{\setminus}\pN_0$ exist
  solely to make the substitutions $\renbt{y}{x}^\aN$, $\eta$ and $p_y$ surjective.
  Here $\sigma$ is surjective iff $\dom(\sigma) \subseteq {\it range}(\sigma)$.}

I take $Act$ to be the set of all expressions $\alpha$ from \tab{actions}, as defined in
\sect{barbed} (in terms of $\zN$ and $\pN$), so $\A\mathbin{:=}Act{\setminus}\{\tau\}$.
The communication function $\gamma$ is given by $\gamma(M\bar x y, N v y)=\match{x}{v}MN\tau$,
just as for rule \textsc{\textbf{\small e-s-com}} in \tab{pi-early-symbolic}.

For $\vec{x} = (x_1,\dots,x_n) \inp \N^n$ and $\vec{y} = (y_1,\dots,y_n) \in \M^n$, with the
$x_i$ distinct, let $\renbt{y}{x}^\aN\!:\aN \cup \{x_1,\dots,x_n\} \rightharpoonup \M$
be the substitution $\sigma$ with $\sigma(x_i) \mathbin= y_i$ and
$\sigma(s_i) \mathbin= x_i$ for $i\mathbin=1,...,n$, and $\sigma(s_i) = s_{i-n}$ for $i>n$.
These functions extend homo\-mor\-phically to $\A$ and thereby constitute {\CCP} relabellings.
Abbreviate $[\renbt{y}{x}^\aN]$ by $[\rename{\vec{y}}{\vec{x}}]$ and $[\renb{z}{y}^\aN]$ by $\rensq{z}{y}$.

For $\eta\inp\{\ell,r,e\}$ and $y\inp\zN$, let the surjective substitutions $\eta\!:\pN\rightharpoonup\pN$
and $p_y\!:\!\{y\}\cup \pN\rightarrow\{y\}\cup\pN$ be given by:
\begin{center}
\begin{tabular}{@{\quad}l@{\hspace{4pt}:=\hspace{5pt}}l@{~~}ll@{\hspace{4pt}:=\hspace{5pt}}l@{~~}l@{}}
$\eta ({}^{\varsigma}\! p)$ & ${}^{\eta\varsigma}\!p$ & &
$p_y ( y)$ & $p$ &
$p_y (p') := y$ \\
$\eta ({}^{\varsigma}\! p^{\upsilon\prime})$ & ${}^{\varsigma}\!p^\upsilon$ & if $\varsigma\mathbin{\neq} \eta\zeta$ &
$p_y (u)$ & $e(u)$ & if $u \neq y,p'$
\end{tabular}
\end{center}
These $\sigma\!: \M \rightharpoonup \M$ are injective, i.e., $x\as{\sigma}\mathbin{\neq} y\as{\sigma}$ when $x\mathbin{\neq} y$.\linebreak
Also they yield {\CCP} relabellings.
The following compositional encoding, which will be illustrated with examples in \sect{examples},
defines my translation from $\piIM$ to {\CCP}.
\[\begin{array}{@{}l@{~:=~}ll@{}}
\fT(\nil) & \nil \\
\fT(M\tau.P) & M\tau.\fT(P) \\
\fT(M\bar xy.P) & M\bar xy.\fT(P) \\
\fT(Mx(y).P) & \sum_{z\in\M}Mxz.\big(\fT(P)\rensq{z}{y}\big) \\
\fT((\nu y)P) & \fT(P)[p_y]\\
\fT(P\mid Q) & \fT(P)[\ell]\mathbin{\|}\fT(Q)[r] \\
\fT(P+Q) & \fT(P)+\fT(Q) \\
\fT(A(\vec y)) & A\rensq{\vec y}{\vec x} \qquad \qquad \plat{when $A(\vec{x}) \stackrel{{\rm def}}{=} P$}\\
\end{array}\]
where the {\CCP} agent identifier $A$ has the defining equation $A=\fT(P)$ when
\plat{$A(\vec{x}) \stackrel{{\rm def}}{=} P$} was the defining equation of the agent identifier $A$\
from the $\pi$-calculus.

To explain what this encoding does, inaction, silent prefix, output prefix and choice are translated
homomorphically. The input prefix is translated into an infinite sum over all possible input values
$z$ that could be received, of the received message $Mxz$ followed by the continuation process
$\fT(P)\rensq{z}{y}$. Here $\rensq{z}{y}$ is a CCS relabelling operator that simulates substitution of
$z$ for $y$ in $\fT(P)$. This implements the rule \textsc{\textbf{\small early-input}}
from \tab{pi-early-symbolic}. Agent identifiers are also translated homomorphically, except that
their arguments $\vec{y}$ are replaced by relabelling operators.

Restriction is translated by simply dropping the restriction operator, but renaming the restricted
name $y$ into a private name $p$ that generates no barbs. The operator $[p_y]$ injectively renames
all private names ${}^{\varsigma}\!p$ that occur in the scope of $(\nu y)$ by tagging all of them with a tag $e$.
This ensures that the new private name $p$ is fresh, so that no name clashes can occur that in
$\piIM$ would have been prevented by the restriction operator.

Parallel composition is almost translated homomorphically. However, each private name on the right
is tagged with an $r$, and on the left with an $\ell$. This guarantees that private names introduced
at different sides of a parallel composition cannot interact. Interaction is only possible when
the name is passed on in the appropriate way.

The main result of this paper states the validity of the above translation,
and thus that {\CCP} is at least as expressive as $\piIM$:

\begin{theorem}\label{thm:piCCS}
For $P\in\T_\pi$ one has $\fT(P)\sbb P$.
\end{theorem}
See \hyperlink{appendix}{the appendix} for a proof.\footnote{Appendices~\hyperlink{appendix}{1}
and~\hyperlink{App2}{2} present two different proofs of \thm{piCCS}---the first shorter,
but the second conceptually simpler, avoiding counterintuitive detours.}
\thm{piCCS} says that each $\pi$-calculus process is strongly barbed
bisimilar to its translation as a {\CCP} process. The labelled transition systems of the $\pi$-calculus and {\CCP\hspace{-.7pt}}
are both of the type presented in \sect{barbed}\hspace{-1pt},\linebreak[4] i.e.\ with transition labels taken from \tab{actions}.
There also the associated barbs are defined. By \thm{piCCS} each $\pi$ transition \plat{$P\goto{\tau}P'$}
can be matched by a {\CCP} transition $\fT(P)\goto{\tau}Q$ with $\fT(P')\mathbin{\sbb} Q$. Likewise, each {\CCP}
transition \mbox{$\fT(P)\mathbin{\goto{\tau}}Q$} can be matched by a $\pi$ transition $P\goto{\tau}P'$ with $\fT(P')\sbb Q$.
Moreover, if $P$ has a barb $x$ (or $\bar x$) then so does $\fT(P)$, and vice versa.
Here a $\pi$ or {\CCP} process $P$ has a barb $a\in\zN\cup\overline{\zN}$\linebreak[4] iff $P \goto{ay} P'$ or
$P\goto{a(y)} P'$ for some name $y\in\M$ and process $P'$.
Transitions $P \goto{M\bar xy} P'\!$, $P \goto{M\bar x(y)} P'\!$, $P \goto{Mxy} P'$ or
$P \goto{Mx(y)} P'$ with $M\neq\varepsilon$ or $x\in\pN$ generate no barbs.

\section{The ideas behind this encoding}

The above encoding combines seven ideas, each of which appears to be necessary to achieve the
desired result. Accordingly, the translation could be described as the composition of seven encodings,
leading from $\piIM$ to {\CCP} via six intermediate languages. Here a language comprises syntax as
well as semantics.  Each of the intermediate languages has a labelled transition system semantics
where the labels are as described in \sect{barbed}. Accordingly, at each step it is well-defined
whether strong barbed bisimilarity is preserved, and one can show it is. These proofs go by
induction on the derivation of transitions, where the transitions with visible labels are necessary
steps even when one would only be interested in the transitions with $\tau$-labels.
There are various orders in which the seven steps can be taken. The seven steps are:
\begin{enumerate}
\item Moving from the late operational semantics (\tab{pi}) to the early one (\tab{pi-early}).
  This translation is syntactically the identity function, but still its validity requires proof, as the
  generated LTS changes. The proof amounts to showing that the same barbed transition system is
  obtained before and after the translation---see \sect{pi-semantics}.
\item Moving from a regular operational semantics (\tab{pi-early}) to a symbolic one (\tab{pi-early-symbolic}).
  This step commutes with the previous one.
\item Renaming the bound names of a process in such a way that the result is clash-free \cite{BB98},
  meaning that all bound names are different and no name occurs both free and bound.
  The trick is to do this in a compositional way. The relabelling operators $[\ell]$, $[r]$ and
  $[p_y]$ in the final encoding stem from this step.
\item Eliminating the need for rule \textsc{\textbf{\small alpha}} in the operational semantics.
  This works only for clash-free processes, as generated by the previous step.
\item \hbadness=10000
  Dropping the restriction operators, while preserving strong barbed bisimilarity.
  This eliminates the orange parts of \tab{pi-early-symbolic}.
  For this purpose clash-freedom and the elimination of \textsc{\textbf{\small alpha}} are necessary.
\item Changing all occurrences of substitutions into applications of CCS relabelling operators.
\item The previous six steps generate a language with a semantics in the De Simone format.
  So from here on a translation to {\sc Meije} or aprACP$_R$ is known to be possible.
  The last step, to {\CCP}, involves changing the remaining form of name-binding into an infinite sum.
\end{enumerate}

\begin{figure}[ht]
  \input{translation0}
  \centerline{\raisebox{1ex}{\box\graph}}
  \caption{Translation from the $\pi$-calculus with implicit matching to \CCP}\label{fig:translation}
  \centering{\footnotesize The intermediate languages
    $\pi_{\!\scriptscriptstyle\rm IM}\scriptstyle(\nN,\pN)$ and
    $\pi_{\!\scriptscriptstyle\rm IM}^{\dagger\hspace{-1.6pt}}\scriptstyle(\nN,\pN)$ are not yet defined.}
\end{figure}

As indicated in \fig{translation}, my translation maps the $\pi$-calculus with implicit matching to a subset of {\CCP}.
On that subset, $\pi$-calculus behaviour can be replayed faithfully, at least up to strong early
congruence, the congruence closure of strong barbed bisimilarity (cf.\ \cite{vG18}).
However, the interaction between a translated $\pi$-calculus process and a {\CCP} process outside the
image of the translation may be disturbing, and devoid of good properties.
Also, in case intermediate languages are encountered on the way from $\piIM$ to {\CCP}, which is just
one of the ways to prove my result, no guarantees are given on the sanity of those languages outside
the image of the source language, i.e.\ on their behaviour outside the realm of clash-free processes
after Step 3 has been made.

\section{Triggering}\label{sec:triggering}

To include the general matching operator in the source language I need to extend the target language
with the \emph{triggering} operator $s{\Rightarrow}\E$ of {\sc Meije} \cite{AB84copy,dS85copy}:
\[\frac{\E \goto\alpha \E'}{s{\Rightarrow}\E \goto{s\alpha} \E'}\]
{\sc Meije} features \emph{signals} and \emph{actions}; each signal $s$ can be ``applied''
to an action $\alpha$, and doing so yields an action $s\alpha$.
In this paper the actions are as in \tab{actions}, and a signal is an expression $\Match{x}{y}$ with $x,y\in\N$;
application of a signal to an action was defined in \sect{pi-semantics}.

Triggering cannot be expressed in {\CCP}, as rooted weak bisimilarity \cite{BW90}, the weak
congruence of \cite{Mi89,Mi90ccs}, is a congruence for {\CCP} but not for triggering.
However, rooted branching bisimilarity \cite{GW96} is a congruence for triggering~\cite{vG11}.

My translation from $\piIM$ to {\CCP} can be extended into one from the full $\pi$-calculus to {\CCST}
by adding the clause $$\fT(\Match{x}{y}P) ~:=~ \Match{x}{y}{\Rightarrow}\fT(P).$$
\thm{piCCS} applies to this extended translation as well.

\section{Examples}\label{sec:examples}

\newcommand{\y}{\textcolor{ACMPurple}{y}}
\newcommand{\by}{\textcolor{ACMPurple}{\bar y}}
\newcommand{\s}{\textcolor{ACMPurple}{s}}
\newcommand{\bs}{\textcolor{ACMPurple}{\bar s}}

\begin{example}
The outgoing transitions of $x(y).\bar y w$ are

  \expandafter\ifx\csname graph\endcsname\relax
   \csname newbox\expandafter\endcsname\csname graph\endcsname
\fi
\ifx\graphtemp\undefined
  \csname newdimen\endcsname\graphtemp
\fi
\expandafter\setbox\csname graph\endcsname
 =\vtop{\vskip 0pt\hbox{%
    \graphtemp=.5ex
    \advance\graphtemp by 0.350in
    \rlap{\kern 0.300in\lower\graphtemp\hbox to 0pt{\hss $x(y).\bar y w$\hss}}%
    \graphtemp=.5ex
    \advance\graphtemp by 0.050in
    \rlap{\kern 1.300in\lower\graphtemp\hbox to 0pt{\hss $\bar z_1 w$\hss}}%
\pdfliteral{
q [] 0 d 1 J 1 j
0.576 w
0.072 w
q 0 g
75.384 -7.992 m
82.8 -7.2 l
76.68 -11.376 l
75.384 -7.992 l
B Q
0.576 w
43.2 -21.6 m
76.032 -9.648 l
S
Q
}%
    \graphtemp=\baselineskip
    \multiply\graphtemp by -1
    \divide\graphtemp by 2
    \advance\graphtemp by .5ex
    \advance\graphtemp by 0.200in
    \rlap{\kern 0.875in\lower\graphtemp\hbox to 0pt{\hss $xz_1$\hss}}%
    \graphtemp=.5ex
    \advance\graphtemp by 0.050in
    \rlap{\kern 1.950in\lower\graphtemp\hbox to 0pt{\hss $\nil$\hss}}%
\pdfliteral{
q [] 0 d 1 J 1 j
0.576 w
0.072 w
q 0 g
126 -1.8 m
133.2 -3.6 l
126 -5.4 l
126 -1.8 l
B Q
0.576 w
104.4 -3.6 m
126 -3.6 l
S
Q
}%
    \graphtemp=\baselineskip
    \multiply\graphtemp by -1
    \divide\graphtemp by 2
    \advance\graphtemp by .5ex
    \advance\graphtemp by 0.050in
    \rlap{\kern 1.650in\lower\graphtemp\hbox to 0pt{\hss $\bar z_1 w$\hss}}%
    \graphtemp=.5ex
    \advance\graphtemp by 0.250in
    \rlap{\kern 1.300in\lower\graphtemp\hbox to 0pt{\hss $\bar z_2 w$\hss}}%
\pdfliteral{
q [] 0 d 1 J 1 j
0.576 w
0.072 w
q 0 g
75.456 -17.208 m
82.8 -18 l
75.888 -20.736 l
75.456 -17.208 l
B Q
0.576 w
43.56 -23.4 m
75.672 -19.008 l
S
Q
}%
    \graphtemp=\baselineskip
    \multiply\graphtemp by 1
    \divide\graphtemp by 2
    \advance\graphtemp by .5ex
    \advance\graphtemp by 0.288in
    \rlap{\kern 0.877in\lower\graphtemp\hbox to 0pt{\hss $xz_2$\hss}}%
    \graphtemp=.5ex
    \advance\graphtemp by 0.250in
    \rlap{\kern 1.950in\lower\graphtemp\hbox to 0pt{\hss $\nil$\hss}}%
\pdfliteral{
q [] 0 d 1 J 1 j
0.576 w
0.072 w
q 0 g
126 -16.2 m
133.2 -18 l
126 -19.8 l
126 -16.2 l
B Q
0.576 w
104.4 -18 m
126 -18 l
S
Q
}%
    \graphtemp=\baselineskip
    \multiply\graphtemp by -1
    \divide\graphtemp by 2
    \advance\graphtemp by .5ex
    \advance\graphtemp by 0.250in
    \rlap{\kern 1.650in\lower\graphtemp\hbox to 0pt{\hss $\bar z_2 w$\hss}}%
    \graphtemp=.5ex
    \advance\graphtemp by 0.465in
    \rlap{\kern 1.300in\lower\graphtemp\hbox to 0pt{\hss $\vdots$\hss}}%
    \graphtemp=.5ex
    \advance\graphtemp by 0.650in
    \rlap{\kern 1.300in\lower\graphtemp\hbox to 0pt{\hss $\bar z_n w$\hss}}%
\pdfliteral{
q [] 0 d 1 J 1 j
0.576 w
0.072 w
q 0 g
76.68 -39.024 m
82.8 -43.2 l
75.384 -42.408 l
76.68 -39.024 l
B Q
0.576 w
43.2 -28.8 m
76.032 -40.752 l
S
Q
}%
    \graphtemp=\baselineskip
    \multiply\graphtemp by 1
    \divide\graphtemp by 2
    \advance\graphtemp by .5ex
    \advance\graphtemp by 0.500in
    \rlap{\kern 0.875in\lower\graphtemp\hbox to 0pt{\hss $xz_n$\hss}}%
    \graphtemp=.5ex
    \advance\graphtemp by 0.650in
    \rlap{\kern 1.950in\lower\graphtemp\hbox to 0pt{\hss $\nil$\hss}}%
    \graphtemp=.5ex
    \advance\graphtemp by 0.650in
    \rlap{\kern 2.050in\lower\graphtemp\hbox to 0pt{\hss .\hss}}%
\pdfliteral{
q [] 0 d 1 J 1 j
0.576 w
0.072 w
q 0 g
126 -45 m
133.2 -46.8 l
126 -48.6 l
126 -45 l
B Q
0.576 w
104.4 -46.8 m
126 -46.8 l
S
Q
}%
    \graphtemp=\baselineskip
    \multiply\graphtemp by -1
    \divide\graphtemp by 2
    \advance\graphtemp by .5ex
    \advance\graphtemp by 0.650in
    \rlap{\kern 1.650in\lower\graphtemp\hbox to 0pt{\hss $\bar z_n w$\hss}}%
    \hbox{\vrule depth0.700in width0pt height 0pt}%
    \kern 2.050in
  }%
}%

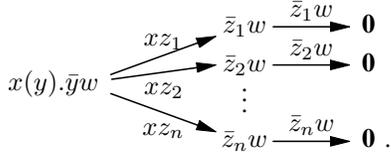
  \centerline{\raisebox{1ex}{\box\graph}}
  \vspace{2ex}

\noindent
The same applies to its translation
$\sum_{z\in \M} xz.\big((\bar y w.\nil)[\rename{z}{y}]\big)$.

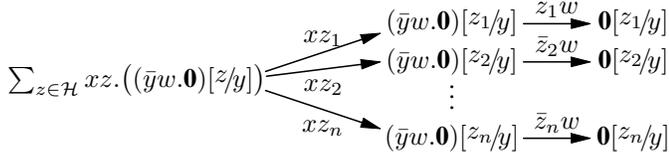
  \expandafter\ifx\csname graph\endcsname\relax
   \csname newbox\expandafter\endcsname\csname graph\endcsname
\fi
\ifx\graphtemp\undefined
  \csname newdimen\endcsname\graphtemp
\fi
\expandafter\setbox\csname graph\endcsname
 =\vtop{\vskip 0pt\hbox{%
    \graphtemp=.5ex
    \advance\graphtemp by 0.350in
    \rlap{\kern 0.675in\lower\graphtemp\hbox to 0pt{\hss $\sum_{z\in \M} xz.\big((\bar y w.\nil)[\rename{z}{y}]\big)$\hss}}%
    \graphtemp=.5ex
    \advance\graphtemp by 0.050in
    \rlap{\kern 2.325in\lower\graphtemp\hbox to 0pt{\hss $(\bar y w.\nil)[\rename{z_1}{y}]$\hss}}%
\pdfliteral{
q [] 0 d 1 J 1 j
0.576 w
0.072 w
q 0 g
132.984 -7.776 m
140.4 -7.2 l
134.136 -11.16 l
132.984 -7.776 l
B Q
0.576 w
97.2 -21.6 m
133.56 -9.504 l
S
Q
}%
    \graphtemp=\baselineskip
    \multiply\graphtemp by -1
    \divide\graphtemp by 2
    \advance\graphtemp by .5ex
    \advance\graphtemp by 0.200in
    \rlap{\kern 1.650in\lower\graphtemp\hbox to 0pt{\hss $xz_1$\hss}}%
    \graphtemp=.5ex
    \advance\graphtemp by 0.050in
    \rlap{\kern 3.275in\lower\graphtemp\hbox to 0pt{\hss $\nil[\rename{z_1}{y}]$\hss}}%
\pdfliteral{
q [] 0 d 1 J 1 j
0.576 w
0.072 w
q 0 g
212.4 -1.8 m
219.6 -3.6 l
212.4 -5.4 l
212.4 -1.8 l
B Q
0.576 w
194.4 -3.6 m
212.4 -3.6 l
S
Q
}%
    \graphtemp=\baselineskip
    \multiply\graphtemp by -1
    \divide\graphtemp by 2
    \advance\graphtemp by .5ex
    \advance\graphtemp by 0.050in
    \rlap{\kern 2.875in\lower\graphtemp\hbox to 0pt{\hss $\bar z_1 w$\hss}}%
    \graphtemp=.5ex
    \advance\graphtemp by 0.250in
    \rlap{\kern 2.325in\lower\graphtemp\hbox to 0pt{\hss $(\bar y w.\nil)[\rename{z_2}{y}]$\hss}}%
\pdfliteral{
q [] 0 d 1 J 1 j
0.576 w
0.072 w
q 0 g
133.056 -17.136 m
140.4 -18 l
133.488 -20.664 l
133.056 -17.136 l
B Q
0.576 w
97.56 -23.4 m
133.272 -18.936 l
S
Q
}%
    \graphtemp=\baselineskip
    \multiply\graphtemp by 1
    \divide\graphtemp by 2
    \advance\graphtemp by .5ex
    \advance\graphtemp by 0.288in
    \rlap{\kern 1.653in\lower\graphtemp\hbox to 0pt{\hss $xz_2$\hss}}%
    \graphtemp=.5ex
    \advance\graphtemp by 0.250in
    \rlap{\kern 3.275in\lower\graphtemp\hbox to 0pt{\hss $\nil[\rename{z_2}{y}]$\hss}}%
\pdfliteral{
q [] 0 d 1 J 1 j
0.576 w
0.072 w
q 0 g
212.4 -16.2 m
219.6 -18 l
212.4 -19.8 l
212.4 -16.2 l
B Q
0.576 w
194.4 -18 m
212.4 -18 l
S
Q
}%
    \graphtemp=\baselineskip
    \multiply\graphtemp by -1
    \divide\graphtemp by 2
    \advance\graphtemp by .5ex
    \advance\graphtemp by 0.250in
    \rlap{\kern 2.875in\lower\graphtemp\hbox to 0pt{\hss $\bar z_2 w$\hss}}%
    \graphtemp=.5ex
    \advance\graphtemp by 0.465in
    \rlap{\kern 2.325in\lower\graphtemp\hbox to 0pt{\hss $\vdots$\hss}}%
    \graphtemp=.5ex
    \advance\graphtemp by 0.650in
    \rlap{\kern 2.325in\lower\graphtemp\hbox to 0pt{\hss $(\bar y w.\nil)[\rename{z_n}{y}]$\hss}}%
\pdfliteral{
q [] 0 d 1 J 1 j
0.576 w
0.072 w
q 0 g
134.136 -39.24 m
140.4 -43.2 l
132.984 -42.624 l
134.136 -39.24 l
B Q
0.576 w
97.2 -28.8 m
133.56 -40.896 l
S
Q
}%
    \graphtemp=\baselineskip
    \multiply\graphtemp by 1
    \divide\graphtemp by 2
    \advance\graphtemp by .5ex
    \advance\graphtemp by 0.500in
    \rlap{\kern 1.650in\lower\graphtemp\hbox to 0pt{\hss $xz_n$\hss}}%
    \graphtemp=.5ex
    \advance\graphtemp by 0.650in
    \rlap{\kern 3.275in\lower\graphtemp\hbox to 0pt{\hss $\nil[\rename{z_n}{y}]$\hss}}%
\pdfliteral{
q [] 0 d 1 J 1 j
0.576 w
0.072 w
q 0 g
212.4 -45 m
219.6 -46.8 l
212.4 -48.6 l
212.4 -45 l
B Q
0.576 w
194.4 -46.8 m
212.4 -46.8 l
S
Q
}%
    \graphtemp=\baselineskip
    \multiply\graphtemp by -1
    \divide\graphtemp by 2
    \advance\graphtemp by .5ex
    \advance\graphtemp by 0.650in
    \rlap{\kern 2.875in\lower\graphtemp\hbox to 0pt{\hss $\bar z_n w$\hss}}%
    \hbox{\vrule depth0.765in width0pt height 0pt}%
    \kern 3.500in
  }%
}%

  \centerline{\box\graph}
  \vspace{2ex}
\end{example}
Here the $z_i$ range over all names in $\N$. Below I flatten such a picture by drawing the arrows only
for one name $z$, which however still ranges over $\N$.

\begin{example}\label{ex:par}
The transitions of $P = x(y).\bar y w \mid \bar x u.u(v)$ are

  \expandafter\ifx\csname graph\endcsname\relax
   \csname newbox\expandafter\endcsname\csname graph\endcsname
\fi
\ifx\graphtemp\undefined
  \csname newdimen\endcsname\graphtemp
\fi
\expandafter\setbox\csname graph\endcsname
 =\vtop{\vskip 0pt\hbox{%
    \graphtemp=.5ex
    \advance\graphtemp by 0.100in
    \rlap{\kern 0.600in\lower\graphtemp\hbox to 0pt{\hss $(x(y).\bar y w)|\bar x u.u(v)$\hss}}%
    \graphtemp=.5ex
    \advance\graphtemp by 0.100in
    \rlap{\kern 1.950in\lower\graphtemp\hbox to 0pt{\hss $\bar z w | \bar x u.u(v)$\hss}}%
\pdfliteral{
q [] 0 d 1 J 1 j
0.576 w
0.072 w
q 0 g
104.4 -5.4 m
111.6 -7.2 l
104.4 -9 l
104.4 -5.4 l
B Q
0.576 w
86.4 -7.2 m
104.4 -7.2 l
S
Q
}%
    \graphtemp=\baselineskip
    \multiply\graphtemp by -1
    \divide\graphtemp by 2
    \advance\graphtemp by .5ex
    \advance\graphtemp by 0.100in
    \rlap{\kern 1.375in\lower\graphtemp\hbox to 0pt{\hss $xz$\hss}}%
    \graphtemp=.5ex
    \advance\graphtemp by 0.100in
    \rlap{\kern 3.050in\lower\graphtemp\hbox to 0pt{\hss $\nil | \bar x u.u(v)$\hss}}%
\pdfliteral{
q [] 0 d 1 J 1 j
0.576 w
0.072 w
q 0 g
187.2 -5.4 m
194.4 -7.2 l
187.2 -9 l
187.2 -5.4 l
B Q
0.576 w
169.2 -7.2 m
187.2 -7.2 l
S
Q
}%
    \graphtemp=\baselineskip
    \multiply\graphtemp by -1
    \divide\graphtemp by 2
    \advance\graphtemp by .5ex
    \advance\graphtemp by 0.100in
    \rlap{\kern 2.525in\lower\graphtemp\hbox to 0pt{\hss $\bar z w$\hss}}%
    \graphtemp=.5ex
    \advance\graphtemp by 0.600in
    \rlap{\kern 0.600in\lower\graphtemp\hbox to 0pt{\hss $(x(y).\bar y w)|u(v)$\hss}}%
\pdfliteral{
q [] 0 d 1 J 1 j
0.576 w
0.072 w
q 0 g
45 -28.8 m
43.2 -36 l
41.4 -28.8 l
45 -28.8 l
B Q
0.576 w
43.2 -14.4 m
43.2 -28.8 l
S
Q
}%
    \graphtemp=.5ex
    \advance\graphtemp by 0.350in
    \rlap{\kern 0.600in\lower\graphtemp\hbox to 0pt{\hss $~~~~~\bar x u$\hss}}%
    \graphtemp=.5ex
    \advance\graphtemp by 0.600in
    \rlap{\kern 1.950in\lower\graphtemp\hbox to 0pt{\hss $\bar z w | u(v)$\hss}}%
\pdfliteral{
q [] 0 d 1 J 1 j
0.576 w
0.072 w
q 0 g
142.2 -28.8 m
140.4 -36 l
138.6 -28.8 l
142.2 -28.8 l
B Q
0.576 w
140.4 -14.4 m
140.4 -28.8 l
S
Q
}%
    \graphtemp=.5ex
    \advance\graphtemp by 0.350in
    \rlap{\kern 1.950in\lower\graphtemp\hbox to 0pt{\hss $~~~~~\bar x u$\hss}}%
\pdfliteral{
q [] 0 d 1 J 1 j
0.576 w
0.072 w
q 0 g
111.6 -41.4 m
118.8 -43.2 l
111.6 -45 l
111.6 -41.4 l
B Q
0.576 w
79.2 -43.2 m
111.6 -43.2 l
S
Q
}%
    \graphtemp=\baselineskip
    \multiply\graphtemp by -1
    \divide\graphtemp by 2
    \advance\graphtemp by .5ex
    \advance\graphtemp by 0.600in
    \rlap{\kern 1.375in\lower\graphtemp\hbox to 0pt{\hss $xz$\hss}}%
    \graphtemp=.5ex
    \advance\graphtemp by 0.600in
    \rlap{\kern 3.050in\lower\graphtemp\hbox to 0pt{\hss $\nil | u(v)$\hss}}%
\pdfliteral{
q [] 0 d 1 J 1 j
0.576 w
0.072 w
q 0 g
221.4 -28.8 m
219.6 -36 l
217.8 -28.8 l
221.4 -28.8 l
B Q
0.576 w
219.6 -14.4 m
219.6 -28.8 l
S
Q
}%
    \graphtemp=.5ex
    \advance\graphtemp by 0.350in
    \rlap{\kern 3.050in\lower\graphtemp\hbox to 0pt{\hss $~~~~~\bar x u$\hss}}%
\pdfliteral{
q [] 0 d 1 J 1 j
0.576 w
0.072 w
q 0 g
194.4 -41.4 m
201.6 -43.2 l
194.4 -45 l
194.4 -41.4 l
B Q
0.576 w
162 -43.2 m
194.4 -43.2 l
S
Q
}%
    \graphtemp=\baselineskip
    \multiply\graphtemp by -1
    \divide\graphtemp by 2
    \advance\graphtemp by .5ex
    \advance\graphtemp by 0.600in
    \rlap{\kern 2.525in\lower\graphtemp\hbox to 0pt{\hss $\bar z w$\hss}}%
    \graphtemp=.5ex
    \advance\graphtemp by 1.100in
    \rlap{\kern 0.600in\lower\graphtemp\hbox to 0pt{\hss $(x(y).\bar y w)|\nil$\hss}}%
\pdfliteral{
q [] 0 d 1 J 1 j
0.576 w
0.072 w
q 0 g
45 -64.8 m
43.2 -72 l
41.4 -64.8 l
45 -64.8 l
B Q
0.576 w
43.2 -50.4 m
43.2 -64.8 l
S
Q
}%
    \graphtemp=.5ex
    \advance\graphtemp by 0.850in
    \rlap{\kern 0.600in\lower\graphtemp\hbox to 0pt{\hss $~~~~~u q$\hss}}%
    \graphtemp=.5ex
    \advance\graphtemp by 1.100in
    \rlap{\kern 1.950in\lower\graphtemp\hbox to 0pt{\hss $\bar z w | \nil$\hss}}%
\pdfliteral{
q [] 0 d 1 J 1 j
0.576 w
0.072 w
q 0 g
142.2 -64.8 m
140.4 -72 l
138.6 -64.8 l
142.2 -64.8 l
B Q
0.576 w
140.4 -50.4 m
140.4 -64.8 l
S
Q
}%
    \graphtemp=\baselineskip
    \multiply\graphtemp by -1
    \divide\graphtemp by 2
    \advance\graphtemp by .5ex
    \advance\graphtemp by 0.850in
    \rlap{\kern 1.950in\lower\graphtemp\hbox to 0pt{\hss $~~~~~u q$\hss}}%
\pdfliteral{
q [] 0 d 1 J 1 j
0.576 w
0.072 w
q 0 g
111.6 -77.4 m
118.8 -79.2 l
111.6 -81 l
111.6 -77.4 l
B Q
0.576 w
79.2 -79.2 m
111.6 -79.2 l
S
Q
}%
    \graphtemp=\baselineskip
    \multiply\graphtemp by -1
    \divide\graphtemp by 2
    \advance\graphtemp by .5ex
    \advance\graphtemp by 1.100in
    \rlap{\kern 1.375in\lower\graphtemp\hbox to 0pt{\hss $xz$\hss}}%
    \graphtemp=.5ex
    \advance\graphtemp by 1.100in
    \rlap{\kern 3.050in\lower\graphtemp\hbox to 0pt{\hss $\nil | \nil$\hss}}%
\pdfliteral{
q [] 0 d 1 J 1 j
0.576 w
0.072 w
q 0 g
221.4 -64.8 m
219.6 -72 l
217.8 -64.8 l
221.4 -64.8 l
B Q
0.576 w
219.6 -50.4 m
219.6 -64.8 l
S
Q
}%
    \graphtemp=.5ex
    \advance\graphtemp by 0.850in
    \rlap{\kern 3.050in\lower\graphtemp\hbox to 0pt{\hss $~~~~~u q$\hss}}%
\pdfliteral{
q [] 0 d 1 J 1 j
0.576 w
0.072 w
q 0 g
198 -77.4 m
205.2 -79.2 l
198 -81 l
198 -77.4 l
B Q
0.576 w
158.4 -79.2 m
198 -79.2 l
S
Q
}%
    \graphtemp=\baselineskip
    \multiply\graphtemp by -1
    \divide\graphtemp by 2
    \advance\graphtemp by .5ex
    \advance\graphtemp by 1.100in
    \rlap{\kern 2.525in\lower\graphtemp\hbox to 0pt{\hss $\bar z w~~~~~~~~~$\hss}}%
    \graphtemp=.5ex
    \advance\graphtemp by 0.800in
    \rlap{\kern 1.600in\lower\graphtemp\hbox to 0pt{\hss \textcolor{ACMRed}{$\bar u w \textcolor{black}| u(v)$}\hss}}%
\pdfliteral{
q [] 0 d 1 J 1 j
0.576 w
0.072 w
q 0 g
93.384 -47.664 m
97.2 -54 l
90.864 -50.184 l
93.384 -47.664 l
B Q
0.576 w
57.6 -14.4 m
92.088 -48.888 l
S
Q
}%
    \graphtemp=\baselineskip
    \multiply\graphtemp by -1
    \divide\graphtemp by 2
    \advance\graphtemp by .5ex
    \advance\graphtemp by 0.475in
    \rlap{\kern 1.075in\lower\graphtemp\hbox to 0pt{\hss \textcolor{ACMRed}{$\tau$}\hss}}%
\pdfliteral{
q [] 0 d 1 J 1 j
0.576 w
0.072 w
q 0 g
198.36 -69.12 m
205.2 -72 l
197.784 -72.72 l
198.36 -69.12 l
B Q
0.576 w
133.2 -61.2 m
198.072 -70.92 l
S
Q
}%
    \graphtemp=\baselineskip
    \multiply\graphtemp by -1
    \divide\graphtemp by 2
    \advance\graphtemp by .5ex
    \advance\graphtemp by 0.925in
    \rlap{\kern 2.350in\lower\graphtemp\hbox to 0pt{\hss \textcolor{ACMRed}{$\tau$}\hss}}%
\pdfliteral{
q [] 0 d 1 J 1 j
0.576 w
0.072 w
q 0 g
201.096 -65.232 m
206.64 -70.2 l
199.44 -68.472 l
201.096 -65.232 l
B Q
0.576 w
162 -46.8 m
200.232 -66.888 l
S
Q
}%
    \graphtemp=\baselineskip
    \multiply\graphtemp by -1
    \divide\graphtemp by 2
    \advance\graphtemp by .5ex
    \advance\graphtemp by 0.812in
    \rlap{\kern 2.560in\lower\graphtemp\hbox to 0pt{\hss \textcolor{ACMRed}{$~~~[z{=}u]\tau$}\hss}}%
    \hbox{\vrule depth1.200in width0pt height 0pt}%
    \kern 3.400in
  }%
}%

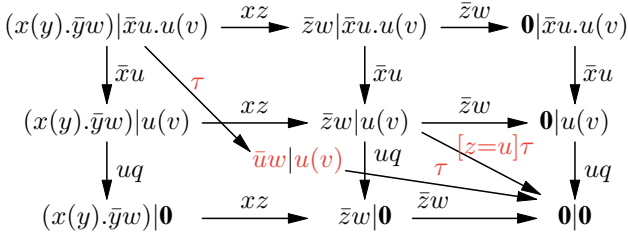
  \centerline{\raisebox{1ex}{\box\graph}}
  \vspace{2ex}

\noindent
Here \textcolor{ACMRed}{$\bar u w \textcolor{black}| u(v)$} is the special case of $\bar z w | u(v)$ obtained by
taking $z:=u$. It thus also has outgoing transitions labelled $\bar uw$ and $uq$, for $q\in \N$.

Up to strong bisimilarity, the same transition system is obtained by the translation $\fT(P)$ of $P$
in {\CCP}.
\[\fT(P) = \left( \sum_{z \in \M} xz.((\bar y w.\nil) [\rename{z}{y}])\!\right)\!\![\ell] \left\| \left(\!\bar x u.\!\sum_{z \in \M}
  u z (\nil[\rename{z}{v}])\!\right)\!\![r]\right.\]
Since there are no restriction operators in this example, the relabelling operators $[\ell]$ and
$[r]$ are of no consequence.
Here $\displaystyle\fT(P) \mathbin{\stackrel{\textcolor{ACMRed}{\tau}}\rightarrow}
\textcolor{ACMRed}{(\bar y w.\nil) [\rename{u}{y}]\textcolor{black}{[\ell]
    \raisebox{-2pt}{$\left\|\rule{0pt}{12pt}\right.$}}\! \!\sum_{z \in \M} \! u z (\nil[\rename{z}{v}])}[r]
\mathbin{\stackrel{\textcolor{ACMRed}{\tau}}\rightarrow} \nil [\rename{u}{y}][\ell] \| \nil[\rename{w}{v}][r]$.\vspace{-3pt}
\end{example}

\begin{example}
Let $Q = (\nu x)\left(x(y).\bar y w \mid (\nu u)\big(\bar x u.u(v)\big)\right)$.
It has no other transitions than
\[Q \goto\tau (\nu x)(\nu u)\big(\bar u w | u(v)\big) \goto\tau (\nu x)(\nu u)(\nil|\nil).\]
Its translation $\fT(Q)$ into {\CCP} is
\[\left(\!\!\left( \sum_{z \in \M} xz.((\bar y w.\nil) [\rename{z}{y}])\!\right)\!\![\ell] \left\| \left(\!\bar x u.\!\sum_{z \in \M}
  u z (\nil[\rename{z}{v}])\!\right)\!\![p_u][r]\right.\!\!\right)\!\![p_x]\]
Up to strong bisimilarity, its transition system is the same as that of $P$ or $\fT(P)$ from \ex{par},
except that in transition labels the name $u$ is renamed into the private name $^{er\!}p$, and $x$
is renamed into the private name $p$. One has $\fT(Q) \stackrel{\scriptscriptstyle\bullet}{\raisebox{0pt}[2pt]{$\sim$}}Q$,
since private names generate no barbs.
\end{example}

\begin{example}
The process $(\nu x)(x(y)) \mid (\nu x)(\bar x u)$ has no outgoing transitions. Accordingly, its translation
\[\left. \left(\sum_{z\in \M}xz.(\nil\renb{z}{y})\right)\![p_x][\ell]\; \right\| (\bar x u) [p_x][r]\]
only has outgoing transitions labelled $^{\ell\!}pz$ for $z \inp\M$ and $\overline{^{r\!}p}u$.\linebreak[4]
Since the names $^{\ell\!}p$ and $^{r\!}p$ are private, these transitions generate no barbs.
In this example, the relabelling operators $[\ell]$ and $[r]$ are essential.
Without them, the mentioned transitions would have complementary names, and communicate into a
$\tau$-transition.
\end{example}

\begin{example}
Let $P \mathbin= (\nu y)\big(\bar xy. \bar y w\big) \mid x(u).u(v)$.
Then \[P \goto\tau (\nu y)\big( \bar y w \mid y(v) \big) \goto\tau (\nu y)(\nil|\nil).\]
Now $\fT\left((\nu y)\big(\bar xy. \bar y w\big)\right) = (\bar xy. \bar y w.\nil)[p_y]$ and
\[\fT(x(u).u(v)) = \sum_{z\in\M} xz.\left(\left(\sum_{z\in\M} uz.(\nil[\rename{z}{v}])\right)[\rename{z}{u}]\right).\]
Hence $\fT\left((\nu y)\big(\bar xy. \bar y w\big)\right)[\ell] \goto{\bar x \mbox{}^\ell\!p} (\bar y w.\nil)[p_y][\ell]$.
Since the substitution $r$ used in the relabelling operator $[r]$ is surjective, there is a
name $s$ that is mapped to $\mbox{}^\ell\!p$, namely $\mbox{}^\ell\!p'$. Considering that $\fT(x(u).u(v)) \goto{xs} \fT(u(v))[\rename{s}{u}]$,
\[\fT(P) \goto\tau
 (\bar y w.\nil)[p_y][\ell] \left\| \left(\sum_{z\in\M} (uz.\nil)[\rename{z}{v}]\right)[\rename{s}{u}][r]\right.\]
These parallel components can perform actions $\overline{\mbox{}^\ell\!p}w$ and $\mbox{}^\ell\!pw$,
synchronising into a $\tau$-transition, and thereby mimicking the behaviour of $P$.
\end{example}

\begin{example}
Let $P \mathbin= (\nu y)\big(\bar xy. (\nu y)(\bar y w)\big) \mid x(u).u(v)$.
Then $P \goto\tau (\nu y)\big( (\nu y)(\bar y w) \mid y(v) \big) \gonotto\tau$.
One obtains
\[\fT(P) \goto\tau
 (\bar y w.\nil)[p_y][p_y][\ell] \left\| \left(\sum_{z\in\M} uz.(\nil[\rename{z}{v}])\right)[\rename{s}{u}][r]\right.\]
for a name $s$ that under $[r]$ maps to $\mbox{}^\ell\!p$.
Now the left component can do an action $\overline{^{\ell e}\!p}w$, whereas the left component can
merely match with $\overline{^{\ell}\!p}w$.
No synchronisation is possible. This shows why it is necessary that the relabelling $[p_y]$ not only
renames $y$ into $p$, but also $p$ into $^e\!p$.
\end{example}

\begin{example}
Let $P=x(y).x(w).\bar w u$. Then
$$P|\bar x v.\bar x \y.\y(v) \goto\tau x(w).\bar w u | \bar x \y.\y(v) \goto\tau \by u | \y(v)\goto\tau \nil|\nil.$$
Therefore, $\fT(P|\bar x v.\bar x \y.\y(v))$ must also be able to start with three consecutive $\tau$-transitions.
Note that
\[\fT(P|\bar x v.\bar x \y.\y(v)) = \fT(P)[\ell] \left\| \left(\bar x v.\bar x \y.\!\sum_{z \in \M}
  \y z (\nil[\rename{z}{\y}])\right)\![r]\right.\]
with
\[\fT(P) = \sum_{z \in \M} xz. \left(\left( \sum_{z \in \M} xz.((\bar w u.\nil) [\rename{z}{w}])\right)[\rename{z}{y}]\right).\]
The only way to obtain $\fT(P|\bar x v.\bar x \y.\y(v)) \goto\tau \goto\tau \goto\tau$
is when $\fT(P) \goto{x v} Q \goto{x \y} \goto{\by u}$.
The {\CCP} process $Q$ must be
\[\left( \sum_{z \in \M} xz.((\bar w u.\nil) [\rename{z}{w}])\right)[\rename{v}{y}].\]
Given the semantics of the CCS relabelling operator, one must have
$\sum_{z \in \M} xz.((\bar w u.\nil) [\rename{z}{w}]) \goto{\alpha}$,
such that applying the relabelling $[\rename{v}{y}]$ to $\alpha$ yields $x\y$.
When simply taking $[\renb{v}{y}]$ for $[\rename{v}{y}]$, that is, the relabelling that changes
all occurrences of the name $y$ in a transition label into $v$, this is not possible.
This shows that a simplification of my translation without use of the spare names $\aN$ would not
be valid.

Crucial for this example is that I only use surjective substitutions.
$[\rename{v}{y}]$ is an abbreviation of $[\renb{v}{y}^\aN]$. Here $\renb{v}{y}^\aN$
is a surjective substitution that not only renames $y$ into $v$, but also sends a spare name $\s$ to $y$.
This allows me to take $\alpha:= x\s$.
Consequently, in deriving the transition $\sum_{z \in \M} xz.((\bar w u.\nil) [\rename{z}{w}]) \goto{\alpha}$,
I choose $z$ to be $\s$, so that
\[\sum_{z \in \M} xz.((\bar w u.\nil) [\rename{z}{w}]) \goto{x\s} (\bar w u.\nil) [\rename{\s}{w}] \goto{\bs u} \nil[\rename{\s}{w}].\]
Putting this in the scope of the relabelling $[\rename{v}{y}]$ yields
\[Q  \goto{x\y} (\bar w u.\nil) [\rename{\s}{w}][\rename{v}{y}] \goto{\by u} \nil[\rename{\s}{w}][\rename{v}{y}]\]
as desired, and the example works out.\footnote{This use of spare names solves the
  problem raised in \cite[Footnote~5]{BB98}.}
\end{example}
This example shows that spare names play a crucial role in intermediate states of
{\CCP}-translations.
In general this leads to stacked relabellings from true names into spare ones and back.
Making sure that in the end one always ends up with the right names calls for particularly careful
proofs that do not cut corners in the bookkeeping of names.

A last example showing a crucial feature of my translation is discussed in \sect{related}.

\section[The unencodability of CCS into pi]{The unencodability of CCS into $\pi$}

Let $f:\A \rightarrow \A$ be a CCS relabelling function satisfying $f(x_iy)=x_{i+1}y$.
Here $(x_i)_{i=0}^\infty$ is an infinite sequence of names, and $\A$ is as in \sect{barbed}.
The CCS process $A$ defined by
$$A := x_0y.\nil + \tau.(A[f])$$ satisfies
$\exists P.~A\goto{\tau}^* P \wedge P{\downarrow_{x_i}}$ for all $i \geq 0$, i.e., it has
infinitely many \emph{weak barbs}. It is easy to check that all weak barbs of a
$\pi$-calculus process $Q$ must be free names of $Q$, of which there are only finitely many.
Consequently, there is no $\pi$-calculus process $Q$ with $A \sbb Q$, and hence no translation of CCS
in the $\pi$-calculus that is valid up to $\sbb$.\footnote{In \cite{Palamidessi03} it was already
  mentioned, by reference to Pugliese [personal communication, 1997] that CCS relabelling operators
  cannot be encoded in the $\pi$-calculus.}

\section{Related work}\label{sec:related}

My translation from $\piIM$ to {\CCP} is inspired by an earlier translation $\mathcal{E}$ from a
version of the $\pi$-calculus to CCS, proposed by Banach \& van Breugel \cite{BB98}.
The paper \cite{BB98} takes $\A := \{\langle x,y\rangle \mid x,y\in\N\}$ for the visible CCS actions;
action $\langle x,y\rangle$ corresponds with my $xy$, and its complement $\overline{\langle x,y\rangle}$
with my $\bar xy$. On the fragment of $\pi$ featuring
inaction, prefixing, choice and parallel composition, the encoding of \cite{BB98} is given by\vspace{-1ex}
\[\begin{array}{@{}l@{~:=~}ll@{}}
\mathcal{E}(\nil) & \nil \\
\mathcal{E}(\tau.P) & \tau.\mathcal{E}(P) \\
\mathcal{E}(\bar xy.P) & \overline{\langle x,y\rangle}.\mathcal{E}(P) \\
\mathcal{E}(x(y).P) & \sum_{z\in\N} \langle x,z\rangle.\big(\mathcal{E}(P)\rensq{z}{y}\big) \\
\mathcal{E}(P\mid Q) & \mathcal{E}(P) \mid \mathcal{E}(Q) \\
\mathcal{E}(P+Q) & \mathcal{E}(P)+\mathcal{E}(Q). \\
\end{array}\]
The main result of \cite{BB98} (Theorem 5.3), stating the correctness of this encoding,
says that $P \bis{r} Q$ iff $\mathcal{E}(P) \bis{r} \mathcal{E}(Q)$, for all $\pi$-processes $P$ and $Q$.
Here $\bis{r}$ is strong reduction bisimilarity---see \df{reduction bis}.
In fact, replacing the call to Lemma 3.5 in the proof of this theorem by a call to Lemma 3.4, they could
equally well have claimed the stronger result that $P \bis{r} \mathcal{E}(P)$ for all $\pi$-processes $P$,
i.e., that $\mathcal{E}$ is valid up to $\bis{r}$.

This result contradicts my \thm{unencodable} and thus must be flawed.
Where it fails can be detected by pushing the counterexample process
$P := \bar x v \mid x(y).R$ with $R := \bar{y}u|v(w)$,
used in the proof of \thm{unencodable}, through the encoding of \cite{BB98}.
I claim that while $P \goto{\tau} \bar{v}u|v(w) \goto{\tau}$, its translation $\mathcal{E}(P)$ cannot do two $\tau$-steps.
Hence $P \not\bis{r} \mathcal{E}(P)$. Using a trivial process $Q$ such that $P \bis{r} Q \bis{r} \mathcal{E}(Q)$,
this also constitutes a counterexample to \cite[Theorem 5.3]{BB98}.

Note that $\mathcal{E}(R) = \overline{\langle y,u\rangle}.\nil \mid \sum_{z\in N}\langle v,z\rangle.(\nil\rensq{z}{w})$.
This process can perform the actions $\overline{\langle y,u\rangle}$ as well
as $\langle v,u\rangle$, but no action $\tau$, since $y\neq v$.
Now \[\mathcal{E}(P) = \overline{\langle x,v\rangle}.\nil \mid
 \sum_{z\in N}\langle x,z\rangle.(\mathcal{E}(R)\rensq{z}{y}).\]
Its only $\tau$-transition goes to $\nil\mid\mathcal{E}(R)\rensq{v}{y}$.
This process can perform the actions $\overline{\langle v,u\rangle}$ as well
as $\langle v,u\rangle$, but still no action $\tau$, since $\rensq{v}{y}$
is a CCS relabelling operator rather than a substitution, and it is applied only after
any synchronisations between 
$\overline{\langle y,u\rangle}.\nil$ and $\sum_{z\in N}\langle v,z\rangle.(\nil\rensq{z}{w})$
are derived.

My own encoding $\fT$ translates the processes $P$ and $R$ essentially in the same way,
but now there is a transition $\fT(R) \goto{\match{y}{v}\tau} (\nil\|\nil\rensq{u}{w})$.
The renaming $\rensq{v}{y}$ turns this synchronisation into a $\tau$:
\[
\fT(P)\goto\tau \fT(R)\rensq{v}{y} \goto\tau (\nil\|\nil\rensq{u}{w})\rensq{v}{y}.
\]
The crucial innovation of my approach over \cite{BB98} in this regard is the switch from the early
to the early symbolic semantics of the $\pi$-calculus, combined with a switch from CCS as target
language to {\CCP}.

\newcommand{\csp}[1][P]{\textsc{csp}[#1]}

In \cite{Ros10}, Roscoe argues that CSP is at least as expressive as the $\pi$-calculus. As evidence
he present a translation from the latter to the former.  Roscoe does not provide a criterion for the
validity of such a translation, nor a result implying that a suitable criterion has been met.  The
following observations show that his transition is not compositional, and that it is debatable
whether it preserves a reasonable semantic equivalence.
\begin{itemize}
\item[(1)]
Roscoe translates $\tau.P$ as $tau\mathop{\rightarrow} \csp[P]$,
where $\rightarrow$ is CSP action prefixing and $\csp[P]$ is the translation of the $\pi$-expression $P$.
Here $tau$ is a visible CSP action, that is renamed into $\tau$ only later in the translation, when
combining prefixes into summations.
Thus, on the level of prefixes, the translation does not preserve (strong) barbed bisimilarity or
any other suitable semantic equivalence.
This problem disappears when we stop seeing prefixing and choice as separate operators in the
$\pi$-calculus, instead using a guarded choice $\sum_{i\in I} \alpha_i.P_i$.
\item[(2)]
Roscoe translates $x(y).P$ into $x?z \rightarrow \csp[P\renb{z}{y}]$.
This is not compositional, since the translation of $x(y).P$ does not merely call the translation of
$P$ as a building block, but the result of applying a substitution to $P$. Substitution is not a CSP
operator; it is applied to the $\pi$-expression $P$ before translating it.
While this mode of translation has some elegance, it is not compositional, and it remains
questionable whether a suitable weaker correctness criterion can be formulated that takes the
place of compositionality here.
\item[(3)]
To deal with restriction,  \cite{Ros10} works with translations $\csp[P]_{\kappa,\sigma}$, where
two parameters $\kappa$ and $\sigma$ are passed along that keep track of sets of fresh names to
translate restricted names into.
The set of fresh names $\sigma$ is partitioned in the translation of $P|Q$ (page 388), such that both
sides get disjoint sets of fresh names to work with. Although the idea is rather similar to the one used here, the
passing of the parameters makes the translation non-compositional. In a compositional translation
$\csp[P|Q]$ the arguments $P$ and $Q$ may appear in the translated CSP process only in the shape
$\csp[P]$ and $\csp[Q]$, not $\csp[P]_{\kappa,\sigma'}$ for new values of $\sigma'$.
\end{itemize}
As pointed out in \cite{GN16,Parrow16},
even the most bizarre translations can be found valid if one only imposes requirements based on
semantic equivalence, and not compositionality. Roscoe's translation is actually rather elegant.
However, we do not have a decent criterion to say to what extent it is a valid translation.
The expressiveness community strongly values compositionality
as a criterion, and this attribute is the novelty brought in by my translation. 

\section{Conclusion}

This paper exhibited a compositional translation from the $\pi$-calculus to {\CCP} extended with
triggering that is valid up to strong barbed bisimilarity, thereby showing that the latter language
is at least as expressive as the former. Triggering is not needed when restricting to the
$\pi$-calculus with implicit matching (as used for instance in \cite{SW01book}).
Conversely, I observed that CCS (and thus certainly {\CCP}) cannot be encoded in the $\pi$-calculus.
I also showed that the upgrade of CCS to {\CCP} is necessary to capture the expressiveness of the
$\pi$-calculus. 

A consequence of this work is that any system specification or verification that is carried out in
the setting of the $\pi$-calculus can be replayed in {\CCP}. The main idea here is to replace the
names that are kept private in the $\pi$-calculus by means of the restriction operator, by names
that are kept private by means of a careful bookkeeping ensuring that the same private name is never
used twice. Of course this in no way suggests that it would be preferable to replay $\pi$-calculus
specifications or verifications in {\CCP}.

My translation encodes the restriction operator $(\nu y)$ from the $\pi$-calculus by renaming $y$
into a ``private name''. Crucial for this approach is that private names generate no barbs,
in contrast with standard approaches where all names generate barbs. This use of private names is
part of the definition of strong barbed bisimilarity $\sbb$ on my chosen instance of {\CCP},
and justified since that definition is custom made in the present paper. The use of private names
can be avoided by placing an outermost CCS restriction operator around any translated $\pi$-process.
This, however, would violate the compositionality of my translation.

The use of infinite summation in my encoding might be considered a serious drawback.
However, when sticking to a countable set of $\pi$-calculus names, only countable summation is
needed, which, as shown in \cite{vG94a}, can be eliminated in
favour of unguarded recursion with infinitely many recursion equations.
As the original presentation of the $\pi$-calculus already allows unguarded recursion with infinitely many recursion equations
\cite{MPWpi2} the latter can not reasonably be forbidden in the target language of the translation. 
Still, it is an interesting question whether infinite sums or infinite sets of recursion equations
can be avoided in the target language if we rule them out in the source language.
My conjecture is that this is possible, but at the expense of further upgrading {\CCP}, say
to aprACP$_R^\tau$. This would however require work that goes well beyond what is presented here.

An alternative approach is to use a version of CCS featuring a \emph{choice quantifier} \cite{Lut03}
instead of infinitary summation, a construct that looks remarkably like an infinite sum, but is as
finite as any quantifier from predicate logic. A choice quantifier binds a data variable $z$ (here
ranging over names) to a single process expression featuring $z$.
The present application would need a function from names to CCS relabelling operators.
When using this approach, the size of translated expressions becomes linear in the size of the originals.

It could be argued that choice quantification is a step towards mobility.
On the other hand, if mobility is associated more with scope extrusion than with name binding
itself, one could classify {\CCP} with choice quantification as an immobile process algebra.
A form of choice quantification is standard in mCRL2 \cite{GM14}, which is often regarded
``immobile''.

My translation from $\pi$ to {\CCP} has a lot in common with the attempted translation of $\pi$ to
CCS in \cite{BB98}. That one is based on the early operational semantics of CCS, rather
than the early symbolic one used here. As a consequence, substitutions there cannot be
eliminated in favour of relabelling operators.

A crucial step in my translation yields an intermediate language with an operational semantics in De
Simone format. In \cite{FMQ96} another representation of the $\pi$-calculus is given through an
operational semantics in the De Simone format. It uses a different way of dealing with substitutions.
This type of semantics could be an alternative stepping stone in an encoding from the
$\pi$-calculus into {\CCP}.

In \cite{Palamidessi03} Palamidessi showed that there exists no uniform encoding of the
$\pi$-calculus into a variant of CCS\@. Here \emph{uniform} means that $\fT(P|Q)=\fT(P)|\fT(Q)$.
This does not contradict my result in any way, as my encoding is not uniform.
Palamidessi \cite{Palamidessi03} finds uniformity a reasonable criterion for encodings, because it
guarantees that the translation maintains the degree of distribution of the system. In \cite{PNG13},
however, it is argued that it is possible to maintain the degree of distribution of a system upon
translation without requiring uniformity. In fact, the translation offered here is a good example of
one that is not uniform, yet maintains the degree of distribution.

Gorla \cite{Gorla:unified} proposes five criteria for valid encodings, and shows that there exists no
valid encoding of the $\pi$-calculus (even its asynchronous fragment) into CCS\@.
Gorla's proof heavily relies on the criterion of \emph{name invariance} imposed on valid encodings.
It requires for  $P\in\T_\pi$ and an injective substitution $\sigma$ that $\fT(P\sigma) = \fT(P)\sigma'$
for some substitution $\sigma'$ that is obtained from $\sigma$ through a \emph{renaming policy}.
Furthermore, the renaming policy is such that if $\dom(\sigma)$ is finite, then also $\dom(\sigma')$
is finite. This latter requirement is not met by the encoding presented here, for a single name
$x\in \N$ corresponds with an infinite set of actions $xy$, the ``names'' of CCS, and a substitution
that merely renames $x$ into $z$ must rename each action $xy$ into $zy$ at the CCS end, thus
violating the finiteness of $\dom(\sigma')$.

My encoding also violates Gorla's compositionality requirement, on grounds that $\fT(P)$ appears
multiple times (actually, infinitely many) in the translation of $Mx(y).P$. It is however compositional by the definition in
\cite{vG12} and elsewhere. My encoding satisfies all other criteria of \cite{Gorla:unified}
(operational correspondence, divergence reflection and success sensitiveness).

\bibliography{references,$HOME/Stanford/lib/abbreviations,$HOME/Stanford/lib/dbase}

\begin{figure*}[t]
  \input{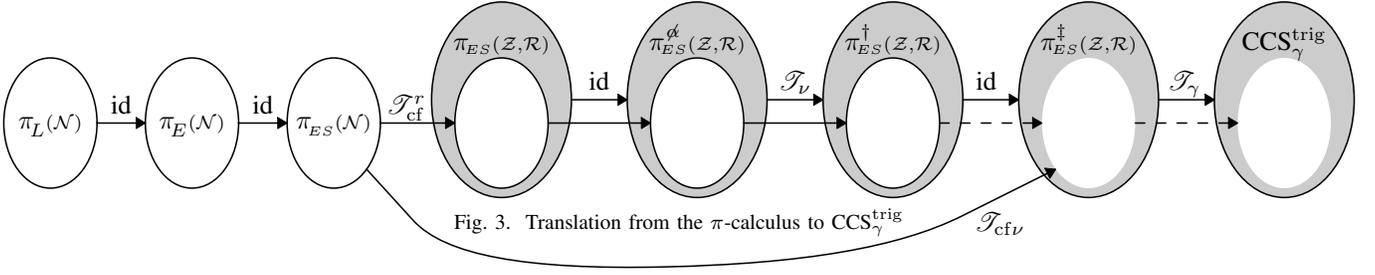}
  \centerline{\raisebox{1ex}{\box\graph}}
  \vspace{-5.5ex}
  \caption{Translation from the $\pi$-calculus to CCS$_\gamma^{\rm trig}$}\label{fig:Translation}
  \vspace{1ex}
\end{figure*}

\newpage
\appendix\hypertarget{appendix}{}

As indicated in \fig{Translation}, my translation from $\pi$ to {\CCST} proceeds in seven steps.
\sect{the encoding} presents a translation in one step: essentially the composition of these constituent translations.
Its decomposition in \fig{Translation} describes both how I found it, and how I prove its validity.

Each of the eight languages in \fig{Translation} comprises syntax, determining what are the valid
expressions or processes, a structural operational semantics generating an LTS, and a BTS extracted
from the LTS in the way described in \sect{barbed}.
The subscript $L$, $E$ or $ES$ in \fig{Translation} tells whether I mean the $\pi$-calculus equipped
with the late, the early, or the early symbolic semantics. The argument $\N$ denotes the set of names
employed by this version of the $\pi$-calculus. In step 3 the set of names is extended from $\N$ to
$\nN\uplus\pN$, where $\N\mathbin\subset\nN$. I write $\nN,\pN$ instead of $\nN\uplus\pN$ to indicate that
only the names in $\nN$---the \emph{public} ones---generate barbs, and to impose a mild restriction on which names to allow
within defining equations of agent identifiers.
The superscript $^{\hspace{-.5pt}\not\hspace{.5pt}\alpha}$ indicates that rule \textsc{\textbf{\small alpha}} is
deleted from the operational semantics, and $^{\dagger}$ that moreover the restriction operator is dropped.
The superscript $\ddagger$ indicates a variant of the calculus to which relabelling operators have
been added, and where the substitutions in rules 
\textsc{\textbf{\small early-input}} and \textsc{\textbf{\small ide}} are replaced by relabelling operators.
The interior white ellipses denote the classes of \emph{clash-free} processes, defined in
Sections~\ref{sec:clash-free} and~\ref{sec:ruthless}.

My translation starts from the $\pi$-calculus $\pi_L(\N)$ with the late operational semantics, as defined in \cite{MPWpi2}.
The first step is the identity mapping to $\pi_E(\N)$, the calculus with the same syntax but the early
operational semantics. The validity of this translation step is the statement that each
$\pi_L(\N)$-expression $P$ is strongly barbed bisimilar with the same expression $P$, but now seen as a
state in the LTS generated by the early operational semantics. As remarked in \sect{pi-semantics},
this is an immediate consequence of Lemmas~\ref{lem:alpha} and~\ref{lem:aL=E}.

The second translation step likewise goes to the $\pi$-calculus with the early symbolic semantics.
Its validity has been concluded at the end of \sect{pi-semantics}.

Skipping step 3 for the moment, \sect{alpha} describes the fourth translation step
by studying the identity translation from $\pi_{ES}(\nN,\pN)$ to $\pi_{ES}^{\hspace{-.5pt}\not\hspace{.5pt}\alpha}(\nN,\pN)$.
As shown by \ex{alpha}, rule \textsc{\textbf{\small alpha}} is not redundant in the early (symbolic)
semantics, and thus this step in not valid in general. As a consequence,
$\pi_{ES}^{\hspace{-.5pt}\not\hspace{.5pt}\alpha}(\nN,\pN)$ is a weird calculus, that no doubt is
unsuitable for many practical purposes. Nevertheless, \sect{alpha} shows that this translation step\pagebreak[2]\vspace*{1.1ex}
is valid on the subclass of clash-free processes, defined in \sect{clash-free}. That is, each clash-free process
$P$ in $\pi_{ES}(\nN,\pN)$ is strongly barbed bisimilar to the same process $P$ seen as a state
in $\pi_{ES}^{\hspace{-.5pt}\not\hspace{.5pt}\alpha}(\nN,\pN)$.

Step 5, recorded in \sect{restriction}, eliminates the restriction operator from the language by translating (sub)expressions
$(\nu z)P$ into $P$. This step does not preserve $\sbb$ for the language as a
whole, but, since there are no $\pN$-barbs, it does so on the sublanguage that arises as
the image of the previous translation steps, namely on the clash-free processes in $\pi_{ES}^{\hspace{-.5pt}\not\hspace{.5pt}\alpha}(\nN,\pN)$.

\sect{relabelling} shows that the substitutions that occur in the operational semantics of
\tab{pi-early-symbolic} may be replaced by relabelling operators, while preserving strong barbed
bisimilarity on clash-free processes.
This proves the validity of Step 6, the identity translation from \plat{$\piWR$} to \plat{$\piSRWR$}
After this step, the resulting language is in De Simone format.
Moreover, as shown in \sect{steps7}, it is easily translated into {\CCST}.

To compose the above fourth step with the first two steps of my translation I need an
intermediate step that maps each process $P$ in $\pi_{ES}(\N)$ to a clash-free process.
A first proposal $\Tcf$ for such a translation appears in \sect{ruthless}.
It replicates agent identifiers and their defining equations---this surely preserves $\sbb$\,---and
renames bound names by giving all binders $(\nu y)$ a fresh name---this preserves $\eqa$ and thus certainly $\sbb$.
Due to the introduction of fresh names, the target of the translation is $\pi_{ES}(\nN\uplus\pN)$ rather than
$\pi_{ES}(\N)$.

Even though $\Tcf$ preserves $\sbb$, I have to reject it as a valid translation.
The main objection is that it employs operations on processes---\emph{ruthless substitutions}---that
are not syntactic operators of the target language. Thereby it fails the criterion of
compositionality, even if the ruthless substitutions are applied in a compositional manner.
In addition, it creates an infinite amount of replicated agent identifiers, whereas I strive not to
increase the number of agent identifiers found in the source language.

To overcome these problems, \sect{ruthless relabelling} studies the composed translation $\fT_\nu\circ\Tcf$
from $\pi_{ES}(\N)$ to \plat{$\piSRWR$}, and shows that in this translation the applications of
ruthless substitutions can be replaced by applications of relabelling operators.
In fact, the latter can be seen as syntactic counterparts of the ruthless substitutions.
This replacement preserves the validity of the translation.
This change also makes the replication of agent identifiers unnecessary, and thus solves both
objections against $\Tcf$.

\sect{simplify} shows that the composed translation so obtained is equivalent to the one presented in
\sect{the encoding}.

A different proof, which first moves from substitutions to relabellings and only then eliminates
restriction, is presented in \hyperlink{App2}{Appendix 2}.

\subsection{Clash-free processes}\label{sec:clash-free}

In this section I consider the $\pi$-calculus $\pi(\nN\uplus\pN)$ where the set of available names is the
disjoint union of sets $\nN$ and $\pN$ of \emph{public} and \emph{private} names.
For the purposes of this section it doesn't matter whether the calculus is equipped with the late, the
early, the late symbolic or the early symbolic operational semantics.

A process $P$ in $\pi(\nN\uplus\pN)$ is \emph{well-typed} w.r.t.\ the partition $(\nN,\pN)$ if for each binder
$x(z)$ occurring in $P$ one has $z\mathbin\in\nN$, and for each binder $(\nu y)$
occurring in $P$ one has $z\mathbin\in\pN$.\linebreak[4] Let $\pi(\nN,\pN)$ be the variant of $\pi(\nN\uplus\pN)$ in
which all defining equations \plat{$A(\vec{x}) \stackrel{{\rm def}}{=} P$} satisfy the restriction
that $P$ is well-typed and $x_i\mathbin\in\nN$ for $i=1,\dots,n$. In addition, when extracting a BTS from
the LTS generated by this version of the $\pi$-calculus, only names in $\nN$ generate barbs.

In the setting of $\pi(\nN,\pN)$ I sharpen the definition of $P\sigma$ from \sect{pi}, the
application of a \alt{}{finite }substitution $\sigma$ to a process $P$. Namely, when choosing a name
$z\notin \fn((\nu y)P)\cup \dom(\sigma)\cup{\it range}(\sigma)$ to replace a bound name $y$, one
always picks $z\in\nN$ if $y\in \nN$ and $z\in\pN$ if $y\in \pN$.

In \cite{BB98} a process $P$ is called \emph{clash-free} if all occurrences of binders $x(y)$ and
$(\nu y)$ use a different name $y$ and no name occurs both free and bound. Below I employ a more
liberal version of this notion by requiring this only for binders $(\nu y)$, and by allowing the arguments of a
$+$-operator to share bound names. At the same time I sharpen the concept by including as binders of
$P$ not only those occurring in $P$, but also those occurring in the body $Q$ of a defining equation
\plat{$A(\vec{x})\stackrel{\rm def}{=}Q$} for which $A$, directly or indirectly, occurs in $P$.  Furthermore,
I require clash-free processes to be well-typed.

\begin{definition}
Let $h(P)$, the \emph{hereditary subprocesses} of $P\mathbin\in\T_\pi$, be the smallest set of
processes containing $P$ such that
\begin{itemize}
\item if $Q\mathbin\in h(P)$ and $R$ is a subterm of $Q$ then $R\mathbin\in h(P)$, and
\item if $A(\vec{y}) \mathbin\in h(P)$ and \plat{$A(\vec{x}) \stackrel{{\rm def}}{=} Q$} then $Q\in h(P)$.
\end{itemize}
Let $\RN(P)$, the \emph{restriction-bound names} of $P$, be the set of all names $y$ such that a
process $(\nu y)Q$ occurs in $h(P)$.
\end{definition}

\begin{observation}
If $P$ is a well-typed $\pi(\nN,\pN)$ process then $\RN(P)\subseteq\pN$.
\end{observation}

\begin{definition}\label{df:clash-free}
  A $\pi(\nN,\pN)$ process $P$ is \emph{clash-free}\footnote{The concept defined here ought to be
    called ``restriction-clash-free'', but is abbreviated ``clash-free'' in
    Sections~\ref{sec:clash-free}--\ref{sec:restriction}. In \sect{ruthless} a more general notion of
    ``full'' clash-freedom will be defined.} if
  \begin{enumerate}
  \item $P$ is well-typed,
  \item for each $(\nu y)Q \mathbin\in h(P)$ one has $y \mathbin{\notin}\RN(Q)$,
  \item for each $Q|R\in h(P)$ one has $\RN(Q)\cap\RN(R)=\emptyset$, and
  \item $\fn(P)\cap\RN(P)=\emptyset$.
  \end{enumerate}
\end{definition}
A substitution $\sigma$ is \emph{clash-free} on a well-typed $\pi(\nN,\pN)$ process $P$ if
$\RN(P) \cap (\dom(\sigma)\cup{\it range}(\sigma)) = \emptyset$.
In that case, $P\sigma$, defined in \sect{pi}, does not involve renaming of bound names from $\pN$.

\begin{observation}\label{obs:clash-free substitution}
If $P$ is clash-free and $\sigma$ is clash-free on $P$, then $P\sigma$ is clash-free and $\RN(P\sigma)=\RN(P)$.
\end{observation}

\begin{lemma}\label{lem:clash-free ide}
  If $A(\vec{y})$ is clash-free and \plat{$A(\vec{x}) \stackrel{{\rm def}}{=} P$}
  then the substitution $\renbt{y}{x}$
  from rule  \textsc{\textbf{\small ide}}
  is clash-free on $P$.
\end{lemma}
\begin{proof}
  Since $x_i\mathbin\in\nN$ for $i=1,\dots,n$,
  one has $x_i\mathbin{\notin}\pN\mathbin\supseteq\RN(P)$.
  Since $y_i\mathbin\in\fn(A(\vec{y}))$ 
  one has $y_i\mathbin{\notin}\RN(A(\vec{y})) \mathbin=\RN(P)$.
\end{proof}

\begin{lemma}\label{lem:clash-free early-input}
  If $Mx(y).P$ is clash-free and $z\notin\RN(Mx(y).P)$ then substitution $\renb{z}{y}$
  from \textsc{\textbf{\small early-input}}  is clash-free on $P$.
\end{lemma}
\begin{proof}
  One has $y\in\nN$, so $y\mathbin{\notin}\RN(P)\subseteq\pN$, and
  $z\notin\RN(P)$ because $\RN(P)= \RN(Mx(y).P)$.
\end{proof}

\begin{definition}\label{df:pi-unguarded}
  Let ${\leftarrowtail}\subseteq \T_{\pi} \times \T_{\pi}$ be the smallest preorder such that
$P+Q \leftarrowtail P$,
$P+Q \leftarrowtail Q$,
$P|Q \leftarrowtail P$,
$P|Q \leftarrowtail Q$,\linebreak[4]
$\Match{x}{y}P \mathbin{\leftarrowtail} P$,
$(\nu y)P \mathbin\leftarrowtail P$ and
$A(\vec{y}) \mathbin\leftarrowtail P\renbt{y}{x}$ when \plat{$A(\vec{x}) \stackrel{{\rm def}}{=} P$}.
\end{definition}

The relation $\leftarrowtail$, like in \df{unguarded}, connects the left-hand sides of conclusions in
operational rules with the left-hand sides of the corresponding premises.

\begin{lemma}\label{lem:clash-free premises}
If $P \leftarrowtail Q$ and $P$ is clash-free, then so is $Q$.
\end{lemma}
\begin{proof}
  Only the last two cases are nontrivial. For $(\nu y)P \mathbin\leftarrowtail P$ use:
  $\fn(P)\mathbin\subseteq\fn((\nu y)P)\mathord\cup\{y\}$ and $\RN(P)\mathbin\subseteq\RN((\nu y)P)\mathord\setminus\{y\}$.\vspace{1pt}

  Suppose $A(\vec{y})$ is clash-free and \plat{$A(\vec{x}) \stackrel{{\rm def}}{=} P$}.
  By definition $P$ is clash-free.
  Apply \lem{clash-free ide} and \obs{clash-free substitution}.
\end{proof}

\subsection{The elimination of \textsc{\textbf{\small alpha}}}\label{sec:alpha}

As a consequence of \lem{alpha}, up to late congruence, rule \textsc{\textbf{\small alpha}} is
redundant in the late operational semantics of the $\pi$-calculus, and hence not included in \tab{pi}.
As shown by \ex{alpha}, \textsc{\textbf{\small alpha}} is not redundant in the early
operational semantics, or in the early symbolic one. However, this section establishes a few basic
properties of the early (symbolic) operational semantics, leading to the conclusion that restricted to the class
of clash-free $\pi(\nN,\pN)$ processes, \textsc{\textbf{\small alpha}} is redundant up to $\sbb$.
These results apply to the early as well as the early symbolic semantics.

Let $P \goto{\alpha}_\bullet Q$ denote that the transition $P \goto{\alpha} Q$ is derivable
from the rules of \tab{pi-early-symbolic} without using rule \textsc{\textbf{\small alpha}}.

\begin{lemma}\label{lem:free names}
Let $P\goto{\alpha}Q$. If $\alpha=M\bar x y$ or $M\tau$ then $\n(\alpha)\subseteq\fn(P)$.
If $\alpha=M x y$ or $\alpha=M \bar x (y)$ then $\n(M)\cup\{x\}\subseteq\fn(P)$.
\end{lemma}
\begin{proof}
  Trivial inductions on the inference of $P\goto{\alpha}Q$.
\end{proof}

\begin{lemma}\label{lem:free names 2}
If $P\goto{\alpha}Q$ then $\fn(Q)\subseteq\fn(P)\cup\n(\alpha)$.
\end{lemma}
\begin{proof}
  A trivial induction on the inference of $P\goto{\alpha}Q$.
\end{proof}

\begin{lemma}\label{lem:bound names}
If $P\goto{M\bar x(y)}_\bullet Q$ then $y\in\RN(P)$.
\end{lemma}
\begin{proof}
  A trivial induction on the inference of $P\goto{M\bar x(y)}Q$.
\end{proof}

\begin{lemma}\label{lem:clash-free transitions}
If $R$ is clash-free, $\alpha$ is not of the form $Mxz$ with $z\in\RN(R)$,
and $R\goto{\alpha}_\bullet R'$, then $R'$ is clash-free and $\RN(R')\subseteq\RN(R)$.
Moreover, $\bn(\alpha)\cap\RN(R')=\emptyset$.
\end{lemma}

\begin{proof}
  With induction on the derivation of $R \goto{\alpha}_\bullet R'$.
  \begin{itemize}
    \item The cases that $R \mathbin{\goto{\alpha}_\bullet} R'$ is derived by rule \textsc{\textbf{\small tau}} or
      \textsc{\textbf{\small early-output}} are trivial, using that $R'\hspace{-1pt}$ is a subexpression of $R$.
    \item  Suppose $R\mathbin{\goto{\alpha}_\bullet} R'$ is derived by \textsc{\textbf{\small early-input}}.
      Then $R\mathbin=Mx(y).P$, $\alpha\mathbin=Mxz$ with $z\notin\RN(R)$ and $R'\mathbin=P\renb{z}{y}$.
      Since $\fn(P)\subseteq \fn(R)\cup\{y\}$, $\RN(P)=\RN(R)$ and $\nN\ni y \notin\RN(P)$,
      the process $P$ is clash-free. By \lem{clash-free early-input}, the
      substitution $\renb{z}{y}$ is clash-free on $P$.
      Thus $R'=P\renb{z}{y}$ is clash-free by \obs{clash-free substitution}, and
      $\RN(R')=\RN(R)$.\vfill
    \item The cases that $R \goto{\alpha}_\bullet R'$ is derived by rule \textsc{\textbf{\small sum}},
      \textsc{\textbf{\small ide}} or \textsc{\textbf{\small symb-match}} are trivial.\vfill
    \item  Suppose \plat{$R\goto{\alpha}_\bullet R'$} is derived by \textsc{\textbf{\small par}}.
      Then $R=P|Q$, $P\goto{\alpha}_\bullet P'$, $R'=P'|Q$, and if $\alpha\mathbin=Mxz$ then $z\notin\RN(P)$.
      By \lem{clash-free premises}, $P$ and $Q$ are clash free.
      So by induction $P'$ is clash-free and $\RN(P')\subseteq\RN(P)$.
      Moreover, $\bn(\alpha)\cap\RN(P')=\emptyset$.
      If $z\in\bn(\alpha)$ then $z\in\RN(P)$ by \lem{bound names} and thus $z\notin\RN(Q)$ by the
      clash-freedom of $R$. Hence $z\notin\RN(R')$. Moreover $\RN(R')\subseteq\RN(R)$.
      I still need to check that $R'$ is clash-free.

      Since $\RN(P)\cap\RN(Q)=\emptyset$, also  $\RN(P')\cap\RN(Q)=\emptyset$.
      It remains to establish that $\fn(P'|Q)\cap\RN(P'|Q)=\emptyset$.
      So suppose $y\in\fn(P'|Q)$. In case $y\in\fn(P|Q)$ then $y\notin\RN(P|Q)\supseteq\RN(P'|Q)$ by
      the clash-freedom of $R$.
      Hence, in view of Lemmas~\ref{lem:free names} and~\ref{lem:free names 2}, the only remaining
      cases are that $\alpha$ has the shape $Mxy$ or $M\bar x(y)$, and $y\in\fn(P')$.
      In the first case the assumption made in \lem{clash-free transitions} yields $y\notin\RN(R)\supseteq\RN(P'|Q)$.
      In the latter case $y\in\RN(P)$ by \lem{bound names}, so $y\notin\RN(Q)$ by the clash-freedom
      of $R$, and $y\notin\RN(P')$ as $P'$ is clash-free.\vfill
    \item  Suppose \plat{$R\goto{\alpha}_\bullet R'$} is derived by \textsc{\textbf{\small e-s-com}}.
      Then $R=P|Q$, $\alpha=\match{x}{v}MN\tau$, $P\goto{M\bar xy}_\bullet P'$, $Q\goto{Nvy}_\bullet Q'$ and $R'\mathbin=P'|Q'$.
      By \lem{free names}, $y\in\fn(P)$, so $y\notin\RN(Q)$ by the clash-freedom of $R$.
      By \lem{clash-free premises}, $P$ and $Q$ are clash free.
      Hence, by induction, $P'$ and $Q'$ are clash-free, $\RN(P')\subseteq\RN(P)$ and $\RN(Q')\subseteq\RN(Q)$.

      Since $\RN(P)\cap\RN(Q)=\emptyset$, also  $\RN(P')\cap\RN(Q')=\emptyset$.
      Since $y\in\fn(P)$, $\fn(P')\cup\fn(Q') \subseteq \fn(P)\cup\fn(Q)$ by Lemmas~\ref{lem:free names} and~\ref{lem:free names 2}.
      Thus $\fn(R')\cap\RN(R')=\emptyset$ by the clash-freedom of $R$, so $R'$ is clash-free.\vfill
    \item  Suppose \plat{$R\goto{\alpha}_\bullet R'$} is derived by \textsc{\textbf{\small e-s-close}}.
      Then $R=P|Q$, $\alpha=\match{x}{v}MN\tau$, $P\goto{M\bar x(z)}_\bullet P'$, $Q\goto{Nvz}_\bullet Q'$,
      $R'\mathbin=(\nu z)(P'|Q')$ and $z\notin\fn(Q)$.
      By \lem{bound names}, $z\in\RN(P)$, so $z\notin\RN(Q)$ by the clash-freedom of $R$.
      By \lem{clash-free premises}, $P$ and $Q$ are clash free.
      Hence, by induction, $P'$ and $Q'$ are clash-free, $\RN(P')\subseteq\RN(P)$ and $\RN(Q')\subseteq\RN(Q)$.
      Moreover, $z\mathbin{\notin}\RN(P')$. As $z\mathbin\in\RN(P)\mathbin\subseteq\RN(R)$, it follows that $\RN(R')\subseteq\RN(R)$.

      Since $\RN(P)\cap\RN(Q)=\emptyset$, also  $\RN(P')\cap\RN(Q')=\emptyset$.
      Since $\fn(R')\mathbin=(\fn(P'){\cup}\fn(Q')){\setminus}\{z\} \mathbin\subseteq \fn(P){\cup}\fn(Q)\linebreak[3] =\fn(R)$
      by Lemmas~\ref{lem:free names} and~\ref{lem:free names 2},
      $\fn(R')\cap\RN(R')=\emptyset$ by the clash-freedom of $R$.
      Finally, as $z\notin\RN(Q)\supseteq\RN(Q')$ and $z\notin\RN(P')$, $z\notin\RN(P'|Q')$.
      Hence $R'$ is clash-free.\vfill
    \item  Suppose $R\goto{\alpha}_\bullet R'$ is derived by \textsc{\textbf{\small res}}.
      Then $R=(\nu y)P$, $R'=(\nu y)P'$, and \plat{$P\goto{\alpha}_\bullet P'$}.
      Moreover, $y\notin\n(\alpha)$.
      By \lem{clash-free premises}, $P$ is clash-free.
      If $\alpha$ is of the form $Mxz$ then $z\neq y$ and thus $z\notin\RN(P)$.
      Hence, by induction, $P'$ is clash-free, $\RN(P')\subseteq\RN(P)$ and $\bn(\alpha)\cap\RN(P')=\emptyset$.
      It follows that $\RN(R')\subseteq\RN(R)$ and $\bn(\alpha)\cap\RN(R')=\emptyset$.

      Since $R$ is clash-free, $y\notin\RN(P)\supseteq\RN(P')$.
      It remains to establish that $\fn(R')\cap\RN(R')=\emptyset$.
      So suppose $w\in\fn(R')\subseteq \fn(P')$. Then $w\neq y$.
      Moreover, $w\notin\RN(P')$ by the clash-freedom of $P'$.
      Hence $w\notin\RN(R')$.
    \item  Finally, suppose \plat{$R\goto{\alpha}_\bullet R'$} is derived by \textsc{\textbf{\small symb-open}}.
      Then $R\mathbin=(\nu y)P$, $\alpha\mathbin=M\bar x (y)$ and $P\mathbin{\goto{M\bar x y}_\bullet} R'$.
      Moreover, $y\mathbin{\neq} x$ and $y\mathbin{\notin}\n(M)$.
      By \lem{clash-free premises}, $P$ is clash-free.
      Hence, by induction, $R'$ is clash-free and $\RN(R')\subseteq\RN(P)\subseteq\RN(R)$.
      Since $R$ is clash-free, $y\notin\RN(P)\supseteq\RN(R')$.
\qed
  \end{itemize}
\end{proof}

\begin{lemma}\label{lem:input names}
If $P\goto{Mxy}_\bullet Q$ with $y\notin\fn(P)$ and $z\notin\RN(P)$ then $P\goto{Mxz}_\bullet P'$
for some $P'$ with $P'\eqa Q\renb{z}{y}$.
\end{lemma}
\begin{proof}
  A trivial induction on the inference of $P\goto{Mxy}Q$.
\end{proof}
Let $\chi$ be an inference of $P\goto{\alpha} Q$. Then $\bn_\chi(P)$ is the union of all $\bn(P')$
for $P'\goto{\beta}Q'$ a transition that appears in $\chi$.

\begin{lemma}\label{lem:input names alpha}
Let $\chi$ be an inference of $P\goto{Mxy} Q$ with $y\mathbin{\notin}\fn(P)$ and $z\mathbin{\notin}\bn_\chi(P)$,
then $P\goto{Mxz} P'$ for some $P'\mathbin\eqa Q\renb{z}{y}$.
Moreover, $P\goto{Mxz} P'$ has an inference no deeper than $\chi$.
\end{lemma}
\begin{proof}
  A trivial induction on $\chi$.
\end{proof}

Given a substitution $\sigma\!\!:\!\nN\uplus\pN\mathbin\rightharpoonup\nN\uplus\pN$,
let $\sigma\update{y}{z}$ denote the substitution with $\dom(\sigma\update{y}{z})=\dom(\sigma)\cup\{y\}$,
defined by $\sigma\update{y}{z}(y):=z$ and $\sigma\update{y}{z}(x):=\sigma(x)$ when $x\neq y$.
\begin{observation}\label{obs:substitution}
  If $w\notin\fn((\nu y)P) \cup\dom(\sigma)\cup{\it range}(\sigma)$ then
  $P\sigma\update{y}{z} \eqa P\renb{w}{y}\sigma\renb{z}{w}$.
\end{observation}

Recall that $((\nu y)P)\sigma:=(\nu w)(P\renb{w}{y}\sigma)$ for some
$w\notin\fn((\nu y)P)\cup\dom(\sigma)\cup{\it range}(\sigma)$.
Thus if $U\eqa ((\nu y)P)\sigma$, then, w.l.o.g.\ $\alpha$-converting the topmost name $w$ first,
$U=(\nu z)Q$ with $Q\eqa P\renb{w}{y}\sigma\renb{z}{w}$.
By \obs{substitution} $Q\eqa P\sigma\update{y}{z}$.
Here $z\notin\fn((\nu w)(P\renb{w}{y}\sigma)) = \fn(((\nu y)P)\sigma)$.
Likewise, if $U\mathbin\eqa (Mx(y).P)\sigma$ then $U\mathbin=M\as{\sigma} x\as{\sigma} (z).Q$ with
$Q\mathbin\eqa P\sigma\upd{y}{z}$.

\begin{corollary}\label{cor:alpha substitution}
If $z\notin\fn(((\nu y)P)\sigma)$\\ then
$(\nu z)(P\sigma\update{y}{z}) \mathbin{\eqa} ((\nu y)P)\sigma$.
\end{corollary}

\newcommand{\fna}{{\rm fn}}

In the next lemma I am mostly interested in the
case $\dom(\sigma)=\emptyset$, but the general case is needed to deal inductively with the cases
\textsc{\textbf{\small res}} and \textsc{\textbf{\small symb-open}}. 
The condition $\fna(\alpha\as{\sigma})\cap \RN(U)=\emptyset$ is needed for case \textsc{\textbf{\small res}} only.

\begin{lemma}\label{lem:alpha early}
  Suppose $R \mathbin{\goto{\alpha}} R'\!$, $U$ is clash-free,
  $\sigma$ is a finite substitution with $\fna(\alpha\as{\sigma})\cap \RN(U)=\emptyset$
  and  $U \eqa R\sigma$.
  \begin{enumerate}[(a)]
  \item
    If $\bn(\alpha)=\emptyset$
    then $\exists U'\!.~ U \goto{\alpha\as{\sigma}}_\bullet U' \eqa R'\sigma$.
  \item
    If $\alpha\mathbin=M\bar x(y)$
    then $\exists z, U'\!.\, U \mathbin{\goto{M\as{\sigma}\bar{x}\as{\sigma}(z)}_\bullet}
    U' \mathbin\eqa R'\sigma\upd{y}{z}$.
  \end{enumerate}
\end{lemma}

\begin{proof}
  By induction on the depth of the inference of $R\mathbin{\goto{\alpha}}R'\!$.
  Ad (a):
  \begin{itemize}
  \item The cases that $R \goto{\alpha} R'$ is derived by rule \textsc{\textbf{\small tau}} or
    \textsc{\textbf{\small output}} are trivial.
  \item Suppose $R\goto{\alpha}R'$ is derived by \textsc{\textbf{\small early-input}}.
    Then $R=Mx(y).P$, $\alpha\mathbin=Mxz$ and $R'\mathbin=P\renb{z}{y}$.
    Since $U\eqa R\sigma$,\linebreak[4]
    $U=M\as{\sigma}x\as{\sigma}(w).Q$ with $w\notin\fn(((\nu y)P)\sigma)$ and $Q\eqa P\sigma\update{y}{w}$.
    By application of rule \textsc{\textbf{\small early-input}}, $U\goto{\alpha\as{\sigma}}_\bullet Q\renb{z\as{\sigma}}{w}$.
    Moreover,
    \[Q\renb{z\as{\sigma}}{w}\eqa P\sigma\update{y}{w}\renb{z\as{\sigma}}{w} \eqa P\renb{z}{y}\sigma = R'\sigma.\]
  \item The cases that $R \goto{\alpha} R'$ is derived by rule \textsc{\textbf{\small sum}},
    \textsc{\textbf{\small symb-match}}, \textsc{\textbf{\small par}} or \textsc{\textbf{\small alpha}} are trivial.
  \item Suppose $R\goto{\alpha}R'$ is derived by \textsc{\textbf{\small ide}}.
    Then $R=A(\vec{y})$, \plat{$A(\vec{x}) \stackrel{{\rm def}}{=} P$} and $P\renbt{y}{x}\goto{\alpha}R'$.
    Furthermore, $U=R\sigma$.
    By \lem{clash-free premises} $P\renb{\vec{y}\as{\sigma}}{\vec{x}}\big(\eqa P\renbt{y}{x}\sigma\big)$ is clash-free.
    By \lem{clash-free ide} and \obs{clash-free substitution},
    $\RN(P\renb{\vec{y}\as{\sigma}}{\vec{x}})\mathbin=\RN(P)\mathbin=\RN(A(\vec{y}))=\RN(U)$,
    so $\fna(\alpha\as{\sigma})\cap \RN(P\renb{\vec{y}\as{\sigma}}{\vec{x}})=\emptyset$.
    Thus, by induction, $P\renb{\vec{y}\as{\sigma}}{\vec{x}}\goto{\alpha}_\bullet U' \eqa R'\sigma$.
    By rule \textsc{\textbf{\small ide}} $U \goto{\alpha}_\bullet U'$.\vfill
  \item Suppose $R\goto{\alpha}R'$ is derived by \textsc{\textbf{\small e-s-com}}.
    Then $R=P|Q$, $\alpha=\match{x}{v}MN\tau$, $P\goto{M\bar xy}P'$, $Q\goto{Nvy}Q'$ and $R'\mathbin=P'|Q'$.
    Moreover, $U\mathbin=V|W$ with $V\mathbin\eqa P\sigma$ and $W\mathbin\eqa Q\sigma$.
    Since $y\in\fn(P)$ by \lem{free names}, $y\as{\sigma}\in\fn(V)\subseteq\fn(U)$.
    As $U$ is clash-free, $y\as{\sigma}\notin\RN(U)\supseteq \RN(V)\cup\RN(W)$.
    Thus $\fna((M\bar xy)\as{\sigma})\cap\RN(V)=\fna((Nvy)\as{\sigma})\cap\RN(W)=\emptyset$.
    So, by induction, $\exists V'.~ V\goto{M\bar xy}_\bullet V' \eqa P'\sigma$ as well as    
    $\exists W'.~ W\goto{Nvy}_\bullet W' \eqa Q'\sigma$.
    Therefore, applying \textsc{\textbf{\small e-s-com}}, $U \goto{\match{x}{v}MN\tau} V'|W' \eqa P'\sigma|Q'\sigma = R'\sigma$.\vfill
  \item Suppose $R\goto{\alpha}R'$ is derived by \textsc{\textbf{\small res}}.
    Then $R=(\nu y)P$, $R'=(\nu y)P'$, and $P\goto{\alpha}P'$ is obtained by a shallower inference.
    Moreover, $y\notin\n(\alpha)$.
    Since $U\eqa R\sigma$, $U=(\nu z)Q$ with $z\notin\fn((\nu y)P)\sigma)$ and $Q\eqa P\sigma\update{y}{z}$.
    So $\fna(\alpha\as{\sigma})\cap (\RN(Q)\cup\{z\})\mathbin=\emptyset$
    and $\alpha\as{\sigma\update{y}{z}}=\alpha\as{\sigma}$.
    By \lem{clash-free premises} process $Q$ is clash-free.
    Thus, by induction, $\exists Q'. Q \goto{\alpha\as{\sigma}}_\bullet Q' \eqa P'\sigma\update{y}{z}$.
    Hence, as $z\notin\n(\alpha\as{\sigma})$,
    $U \goto{\alpha}_\bullet (\nu z)Q'$ by \textsc{\textbf{\small res}}.
    Since $z\mathbin{\notin}\fn(((\nu y)P)\sigma)$ and $z\notin\n(\alpha\as{\sigma})$,
    $z\mathbin{\notin}\fn(((\nu y)P')\sigma)$ by \lem{free names 2}.
    [Namely, if $w\in \fn(((\nu y)P')\sigma)$ then $w=v\as{\sigma}$ with $v\in\fn(P')\setminus\{y\}$.
      Thus $v\in\big(\fn(P)\cup\n(\alpha)\big)\setminus\{y\} \subseteq \fn((\nu y)P)\cup\n(\alpha)$,
      and $w\in\fn(((\nu y)P)\sigma)\cup\n(\alpha\as{\sigma})$.]
    Hence, by \cor{alpha substitution},
    $(\nu z)Q' \mathbin{\eqa} (\nu z)(P'\sigma\update{y}{z}) \mathbin{\eqa} ((\nu y)P')\sigma \mathbin= R'\sigma$.\vfill
  \item $R \mathbin{\goto{\alpha}} R'$ cannot be derived by \textsc{\textbf{\small symb-open}},
    as $\bn(\alpha)\mathbin=\emptyset$.\vfill
  \item Finally, suppose $R\goto{\alpha}R'$ is derived by \textsc{\textbf{\small e-s-close}}.
    Then $R=P|Q$, $\alpha=\match{x}{v}MN\tau$, $P\goto{M\bar x(y)}P'$, $Q\goto{Nvy}Q'$, $y\notin\fn(Q)$ and $R'=(\nu y)(P'|Q')$.
    Moreover, $U\mathbin=V|W$ with $V\mathbin\eqa P\sigma$ and $W\mathbin\eqa Q\sigma$.
    As $\fn(M\bar x(y)\as{\sigma}) \subseteq \fn(\alpha\as{\sigma})$,
    $\fna(M\bar x(y)\as{\sigma})\cap\RN(V)=\emptyset$.
    Consequently, by induction,
    \[\exists z,V'.~ V\goto{M\as{\sigma}\bar x\as{\sigma}(z)}_\bullet V' \eqa P'\sigma\update{y}{z}.\]
    Let $\chi$ be the derivation of $Q\goto{Nvy}Q'$.
    Pick $w\notin \bn_\chi(Q) \cup \RN(W)\cup \fn(W) \cup \fn((\nu y)Q') \cup\dom(\sigma)\cup{\it range}(\sigma)$.
    Applying \lem{input names alpha}, \mbox{$Q\goto{Nvw}Q''$}
    for some $Q''\eqa Q'\renb{w}{y}$.
    Moreover, the inference of this transition is not as deep as the one of $R\goto{\alpha}R'$.
    As $\fna((Nvw)\as{\sigma})\cap\RN(W)=\emptyset$,
    \[\exists W'.~ W\goto{N\as{\sigma}v\as{\sigma}w}_\bullet W' \eqa Q''\sigma \eqa Q'\renb{w}{y}\sigma\]
    by induction.
    Since $z\in\RN(V)$ by \lem{bound names}, and $U$ is clash-free,
    $z\notin\RN(W)$.
    Applying \lem{input names}, 
    \[\exists W''\!.~ W\goto{N\as{\sigma}v\as{\sigma}z}_\bullet W''\eqa W'\renb{z}{w} \eqa Q'\renb{w}{y}\sigma\renb{z}{w}
    \]
    By \obs{substitution}, $Q'\renb{w}{y}\sigma\renb{z}{w}\eqa Q'\sigma\update{y}{z}$.
    Thus, applying \textsc{\textbf{\small e-s-close}}, $U \goto{(\match{x}{v}MN\tau)\as{\sigma}} (\nu z)(V'|W'')$.
    Since $z\mathbin\in\RN(V)\mathbin\subseteq\RN(U)$ and $U$ is clash-free,
    $z\notin\fn(U)=\fn(R\sigma) \supseteq \fn(R'\sigma)$, the last step by 
    Lemmas~\ref{lem:free names} and~\ref{lem:free names 2}.
    Therefore, by  \cor{alpha substitution},
    $(\nu z)(V'|W'') \eqa (\nu z)((P'|Q')\sigma\update{y}{z})  \eqa  ((\nu y)(P'|Q'))\sigma = R'\sigma$.\pagebreak[3]
  \end{itemize}
  Ad (b):
  \begin{itemize}
  \item Suppose $R\goto{M\bar x(y)}R'$ is derived by \textsc{\textbf{\small symb-open}}.
    Then $R=(\nu y)P$ and $P\goto{M\bar xy}R'$ is obtained by a shallower inference.
    Moreover, $y\mathbin{\neq} x$ and $y\mathbin{\notin}\n(M)$.
    Since $U\mathbin\eqa R\sigma$,
    $U=(\nu z)Q$ with $z\notin\fn((\nu y)P)\sigma)$ and $Q\eqa P\sigma\update{y}{z}$.
    By \lem{clash-free premises}, $Q$ is clash-free.
    Since $U$ is clash-free, $z\notin\RN(Q)$.
    So $\fna((M\bar xy)\as{\sigma\update{y}{z}})\cap \RN(Q)=\emptyset$.
    Consequently, by induction,
    \[\exists U'.~ Q \goto{M\as{\sigma}\bar x\as{\sigma}z}_\bullet U' \eqa R'\sigma\update{y}{z}.\]
    By \lem{free names}, $\n(M\bar xy)\subseteq\fn(P)$ and hence
    $\n(M)\cup\{x\}\subseteq\linebreak[2] \fn((\nu y)P)$, so
    $\n(M\as{\sigma})\cup\{x\as{\sigma}\}\subseteq \fn(((\nu y)P)\sigma)\mathrel{\not\hspace{1pt}\ni} z$.
    Therefore, by \textsc{\textbf{\small symb-open}}, $U \goto{M\as{\sigma}\bar x\as{\sigma}(z)}_\bullet U'$.\vfill
  \item The cases that $R \goto{M\bar x(y)} R'$ is derived by rule \textsc{\textbf{\small sum}},
    \textsc{\textbf{\small symb-match}} or \textsc{\textbf{\small alpha}} are again trivial.\vfill
  \item Derivations of $R \goto{M\bar x(y)} R'$ by \textsc{\textbf{\small tau}},
    \textsc{\textbf{\small output}}, \textsc{\textbf{\small early-input}}, \textsc{\textbf{\small e-s-com}} 
    or \textsc{\textbf{\small e-s-close}} cannot occur.\vfill
  \item  The case that $R\goto{M\bar x(y)}R'$ is derived by rule \textsc{\textbf{\small ide}}
    proceeds just as for statement (a) above.\vfill
  \item  Suppose $R\goto{M\bar x(y)}R'$ is derived by \textsc{\textbf{\small par}}.
    Then $R=P|Q$, $P\goto{M\bar x(y)}P'$ is obtained by a shallower inference, and $R'=P'|Q$.
    Moreover, $y\notin \fn(Q)$, and $U=V|W$ with $V\eqa P\sigma$ and $W\eqa Q\sigma$.
    Using that $V$ is clash-free and $\fna((M\bar x(y))\as{\sigma})\cap\RN(V)\mathbin=\emptyset$, by induction
    \[\exists z,V'.~ V\goto{M\as{\sigma}x\as{\sigma}(z)}_\bullet V' \mathbin\eqa P'\sigma\update{y}{z}.\]
    By \lem{bound names} $z\mathbin\in\RN(V)$, so $z\mathbin{\notin}\fn(W)$ by the clash-freedom of $U$.
    By \textsc{\textbf{\small par}} $U \goto{M\as{\sigma}x\as{\sigma}(z)}_\bullet V'|W$.
    Finally, $V'|W\eqa P'\sigma\update{y}{z}\mid Q\sigma\update{y}{z} \eqa  R'\sigma\update{y}{z}$.\vfill
  \item  Suppose $R\goto{M\bar x(v)}R'$ is derived by \textsc{\textbf{\small res}}.
    Then $R=(\nu y)P$, $P\goto{M\bar x(v)}P'$,
    $R'\mathbin=(\nu y)P'$ and $y\notin \n(M\bar x(v))$.
    Since $U\eqa R\sigma$, $U=(\nu z)Q$ with
    $z\notin\fn((\nu y)P)\sigma)$ and $Q\eqa P\sigma\update{y}{z}$.
    So $\fna(M\bar x(v)\as{\sigma})\cap (\RN(Q)\cup\{z\})\mathbin=\emptyset$,\linebreak[4]
    $M\as{\sigma\update{y}{z}}=M\as{\sigma}$ and $\bar x\as{\sigma\update{y}{z}}=\bar x\as{\sigma}$.
    Moreover, $Q$ is clash-free by \lem{clash-free premises}.
    By induction,
    \[\exists w, Q'.~Q \goto{M\as{\sigma}\bar x\as{\sigma}(w)}_\bullet Q' \eqa P'\sigma\update{y}{z}\update{v}{w}.\]
    By \lem{bound names}, $w\in\RN(Q)$, so $w\neq z$ by the clash-freedom of $U$.
    By \lem{free names}, $\n(M)\cup\{x\}\subseteq\fn(P)$,
    so $\n(M\as{\sigma})\cup\{x\as{\sigma}\}\subseteq\fn(((\nu y)P)\sigma)\mathrel{\not\hspace{1pt}\ni} z$.
    Hence, by \textsc{\textbf{\small res}}, \[U \goto{M\as{\sigma}\bar x\as{\sigma}(w)}_\bullet (\nu z)Q'.\]

    Using Lemmas~\ref{lem:free names} and~\ref{lem:free names 2}, $\fn(P')\subseteq \fn(P)\cup\{v\}$
    and thus $\fn(((\nu y)P')\sigma\update{v}{w})\subseteq \fn(((\nu y)P)\sigma\update{v}{w})\cup\{w\}$.
    As $z\mathbin{\notin}\fn((\nu y)P)\sigma)\wedge z\mathbin{\neq} w$, 
    $z\mathbin{\notin}\fn(((\nu y)P')\sigma\update{v}{w})$. It follows that
    \[\begin{array}{@{}c@{~~}c@{~~}c@{}r@{}}
     (\nu z)Q' & \eqa & (\nu z)\big(P'\sigma\update{y}{z}\update{v}{w}\big) \\
               & = &    (\nu z)\big(P'\sigma\update{v}{w}\update{y}{z}\big) & y\neq v \\
               & \eqa & \big((\nu y)P'\big)\sigma\update{v}{w}     & \mbox{\cor{alpha substitution}} \\
               & = & R'\sigma\update{v}{w}.                        & \qed
     \end{array}\]
\end{itemize}
\end{proof}

Now let $\fT_\alpha$ be the identity translation from $\pi_{ES}(\nN,\pN)$ to $\pi_{ES}^{\hspace{-.5pt}\not\hspace{.5pt}\alpha}(\nN,\pN)$,
where $\pi_{ES}^{\hspace{-.5pt}\not\hspace{.5pt}\alpha}(\nN,\pN)$ is the variant of the $\pi_{ES}(\nN,\pN)$
without rule \textsc{\textbf{\small alpha}}.

\begin{theorem}\label{thm:step4}
If $P$ is clash-free then $\fT_\alpha(P)\sbb P$.
\end{theorem}
\begin{proof}
  I have to provide a strong barbed bisimulation on the disjoint union of the LTSs of 
  $\pi_{ES}(\nN,\pN)$ and $\pi_{ES}^{\hspace{-.5pt}\not\hspace{.5pt}\alpha}(\nN,\pN)$.
  Since they have the same states, I make them disjoint by tagging each process in 
  $\pi_{ES}^{\hspace{-.5pt}\not\hspace{.5pt}\alpha}(\nN,\pN)$ with a superscript $^{\hspace{-.5pt}\not\hspace{.5pt}\alpha}$.
  Let \[{\R}:=\{(R,U^{\hspace{-.5pt}\not\hspace{.5pt}\alpha}),
  (U^{\hspace{-.5pt}\not\hspace{.5pt}\alpha},R) \mid U \mbox{~is clash-free and~} R\eqa U\}.\]
  Clearly, $\R$ relates $P$ and $\fT_\alpha(P)$ for clash-free $P$, so
  it suffices to show that $\R$ is a strong barbed bisimulation.
  So suppose $U$ is clash-free and $R \eqa U$, so that $U^{\hspace{-.5pt}\not\hspace{.5pt}\alpha} \R R$
  and $R\R U^{\hspace{-.5pt}\not\hspace{.5pt}\alpha}$.

  Let $U\goto\tau_\bullet Q$. Then $R\goto\tau Q$ by rule \textsc{\textbf{\small alpha}}.
  Moreover, $Q$ is clash-free by \lem{clash-free transitions}, so
  $Q^{\hspace{-.5pt}\not\hspace{.5pt}\alpha} \R Q$.

  Let $R\goto\tau R'$.
  Then $\exists P\!.~ U \goto{\tau}_\bullet P \eqa R'$ by \lem{alpha early}(i).
  Moreover, $P$ is clash-free by \lem{clash-free transitions}, so $R' \R P^{\hspace{-.5pt}\not\hspace{.5pt}\alpha}$.

  Now let $U^{\hspace{-.5pt}\not\hspace{.5pt}\alpha}{\downarrow_b}$ with $b\in \nN\cup\overline\nN$.
  Then $U\goto{by}_\bullet U'$ or $U\goto{b(y)}_\bullet U'$ for some $y$ and $U'$, using the definition of $O$ in \sect{barbed}.
  So $R\goto{by} U'$ or $R\goto{b(y)} U'$ by rule \textsc{\textbf{\small alpha}}.
  Thus $R{\downarrow_b}$.

  Finally, let $R{\downarrow_b}$ with $b\in \nN\cup\overline\nN$.
  Then $R\goto{by}_\bullet R'$ or $R\goto{b(y)}_\bullet R'$ for some $y$ and $R'$.
  Hence $U\goto{by}U'$ or $U\goto{b(z)}U'$ for some $z$ and $U'$ by \lem{alpha early}.
  Thus $U{\downarrow_b}$.
\end{proof}

\subsection{Eliminating restriction operators from the \texorpdfstring{$\pi$}{pi}-calculus}\label{sec:restriction}

The fifth step $\fT_\nu$ of my translation simply drops all restriction operators.
It is defined compositionally by $\fT_\nu((\nu y)P) = \fT_\nu(P)$\vspace{1pt}
and $\fT_\nu(A)=A_\nu$, where $A_\nu$ is a fresh agent identifier with defining
equation \plat{$A_\nu(\vec x) \stackrel{{\rm def}}{=} \fT_\nu(P)$} when
\plat{$A(\vec x) \stackrel{{\rm def}}{=} P$} was the defining equation of $A$;
the translation $\fT_\nu$ acts homomorphically on all other constructs.

The source of this translation step is the well-typed fragment of
$\pi_{E\hspace{-.8pt}S}^{\hspace{-.5pt}\not\hspace{.5pt}\alpha}(\nN,\!\pN\hspace{-.5pt})$.
Its target is \plat{$\piWR$},
a calculus that differs in three ways from $\pi_{E\hspace{-.8pt}S}^{\hspace{-.5pt}\not\hspace{.5pt}\alpha}(\nN,\!\pN\hspace{-.5pt})$. First of all
it only allows well-typed-processes.
Secondly, it allows defining equations\vspace{-.75ex}
$$A(x_1,\ldots,x_n) \stackrel{{\rm def}}{=} P$$
with $\fn(P)\subseteq \{x_1,\ldots,x_n\}\cup\pN$, which is more liberal than the restriction
$\fn(P)\subseteq \{x_1,\ldots,x_n\}$ imposed by $\pi_{E\hspace{-.8pt}S}^{\hspace{-.5pt}\not\hspace{.5pt}\alpha}(\nN,\!\pN\hspace{-.5pt})$.

Intuitively, one may think of such an equation as\vspace{-.75ex} \[A(x_1,\ldots,x_n,\pN) \stackrel{{\rm def}}{=} P,\vspace{-.75ex}\]
in the sense that besides the declared names $x_i$ also all names in $\pN$ are implicitly declared.
A call $A(\vec{y})$ of agent $A$ can likewise be seen as a call $A(\vec{y},\pN)$.  Crucial here is
that although $\vec{y}$ is being substituted for $\vec{x}$, the implicit substitution of $\pN$ for $\pN$ is
the identity. To make that stick, the restriction to well-typed processes has been imposed.
All substitutions $\renb{z}{y}$ and $\renbt{y}{x}$ that are induced by the operational semantics,
namely in rules \textsc{\textbf{\small early-input}} and \textsc{\textbf{\small ide}},
have the property that $y\in\nN$, respectively $x_i\in\nN$. So there never is a need to apply a
substitution renaming an element of $\pN$.

The third and last difference between $\pi_{E\hspace{-.8pt}S}^{\hspace{-.5pt}\not\hspace{.5pt}\alpha}(\nN,\!\pN\hspace{-.5pt})$ and
\plat{$\piWR$}
 is that the latter does not feature restriction operators.  Hence there is no need for rules
\textsc{\textbf{\small res}} and \textsc{\textbf{\small symb-open}}.\linebreak Since the resulting semantics
cannot generate transitions labelled $M\bar x(z)$, rule \textsc{\textbf{\small e-s-close}} can be
dropped as well.
Moreover, $\bn(\alpha)\mathbin=\emptyset$ for all transition labels $\alpha$.
So the side condition of rule \textsc{\textbf{\small par}} can be dropped too.
I also leave out \textsc{\textbf{\small alpha}}.

\begin{lemma}\rm\label{lem:alpha norestriction}
  Let $P$ and $Q$ be \plat{$\piWR$} processes.
  If $P\eqa Q$ and $P\goto{\alpha}P'$ then $Q\goto{\alpha}Q'$ for some $Q'$ with $P' \eqa Q'$.
  Moreover, the size of the derivation of \plat{$Q \goto{\alpha} Q'$} is the same as that of \plat{$P \goto{\alpha} P'$}.
\end{lemma}

\begin{proof}
  A trivial induction on the inference of $P\goto{\alpha}P'$.
\end{proof}
This testifies that \textsc{\textbf{\small alpha}} is redundant.
Thus, in total, all orange parts of \tab{pi-early-symbolic} are dropped.

I proceed to prove the validity of translation $\fT_\nu$.
For $\alpha$ an action in $\pi_{E\hspace{-.8pt}S}^{\hspace{-.5pt}\not\hspace{.5pt}\alpha}(\nN,\!\pN\hspace{-.5pt})$
one defines the \emph{debinding} of $\alpha$ by $\debind{\alpha}\mathbin{:=}\alpha$ if $\alpha$ has the form
$M\tau$, $Mxz$ or $M\bar x y$, and $\debind{M\bar x(y)}:=M\bar xy$.

\begin{lemma}\rm\label{lem:nu-out}
If $R$ is clash-free and $R \goto{\alpha}_\bullet R'$, where
$\alpha$ is not of the form $Mxz$ with $z\mathbin\in\RN(R)$,
then $\fT_\nu(R) \goto{\debind\alpha}\eqa \fT_\nu(R')$.
\end{lemma}
\begin{proof}
By induction on the derivation of $R \goto{\alpha}_\bullet R'$.
\begin{itemize}
\item The cases that  $R \goto{\alpha}_\bullet R'$ is derived by 
\textsc{\textbf{\small tau}} or \textsc{\textbf{\small output}} are trivial.
\item Suppose $R \goto{\alpha}_\bullet R'$ is derived by \textsc{\textbf{\small early-input}}.
  Then $R\mathbin=Mx(y).P$, $\alpha\mathbin=Mxz$, $R'=P\renb{z}{y}$ and $\fT_\nu(R)\mathbin=Mx(y).\fT_\nu(P)$.
  Moreover, $\fT_\nu(R)\goto{\debind\alpha} \fT_\nu(P)\renb{z}{y}$.
  Thus it suffices to show that
  $\fT_\nu(P)\renb{z}{y}\eqa\fT_\nu(P\renb{z}{y})$.
  This holds because the substitution $\renb{z}{y}$ is clash-free on $P$, using \lem{clash-free early-input}.
\item Suppose $R \goto{\alpha}_\bullet R'$ is derived by \textsc{\textbf{\small ide}}.
  Then $R=A(\vec{y})$ with \plat{$A(\vec{x}) \stackrel{{\rm def}}{=} P$}, so $P\renbt{y}{x}\goto{\alpha}_\bullet R'$.
  Moreover $\fT_\nu(R)=A_\nu(\vec{y})$, with $A_\nu$ defined by \plat{$A_\nu(\vec{x}) \stackrel{{\rm def}}{=} \fT_\nu(P)$}.
  By \lem{clash-free premises} $P\renbt{y}{x}$ is clash-free.
  Since $\RN(P\renbt{y}{x})\mathbin=\RN(P)\mathbin=\RN(R)$,
  $\alpha$ is not of the form $Mxz$ with $z\in\linebreak[2]\RN(P\renbt{y}{x})$.
  So by induction $\fT_\nu(P\renbt{y}{x})\mathbin{\goto{\debind\alpha}\eqa} \fT_\nu(R')$.

  By \lem{clash-free ide}, the substitution $\renbt{y}{x}$ is clash-free on $P$.
  Hence $\fT_\nu(P\renbt{y}{x})\eqa\fT_\nu(P)\renbt{y}{x}$.
  \lem{alpha norestriction} yields $\fT_\nu(P)\renbt{y}{x}\mathbin{\goto{\debind\alpha}\eqa} \fT_\nu(R')$.
  Applying rule \textsc{\textbf{\small ide}} yields
  \plat{$\fT_\nu(R) \goto{\debind\alpha}\eqa \fT_\nu(R')$}.
\item The cases that  $R \goto{\alpha}_\bullet R'$ is derived by 
\textsc{\textbf{\small sum}}, \textsc{\textbf{\small symb-match}} or
\textsc{\textbf{\small par}} are trivial.
\item Suppose $R \goto{\alpha}_\bullet R'$ is derived by \textsc{\textbf{\small e-s-com}}.
  Then $R=P|Q$, $\alpha=\match{x}{v}MN\tau$, $P\goto{M\bar xy}_\bullet P'$, $Q\goto{Nvy}_\bullet Q'$, $R'\mathbin=P'|Q'$
  and $\fT_\nu(R)=\fT_\nu(P)|\fT_\nu(Q)$.
  By \lem{clash-free premises}, $P$ and $Q$ are clash-free.
  By \lem{free names}, $y\in\fn(P)\subseteq\fn(R)$, so $y\notin\RN(Q)$ by the clash-freedom of $R$.
  By induction, $\fT_\nu(P)\goto{M\bar xy}\eqa \fT_\nu(P')$ and $\fT_\nu(Q)\goto{Nvy}\eqa \fT_\nu(Q')$.
  Thus  $\fT_\nu(R) \goto{\debind\alpha}\eqa \fT_\nu(R')$ by application of rule \textsc{\textbf{\small e-s-com}}.
\item Suppose $R \goto{\alpha}_\bullet R'$ is derived by \textsc{\textbf{\small e-s-close}}.
  Then $R=P|Q$, $\alpha=\match{x}{v}MN\tau$, $P\goto{M\bar x(z)}_\bullet P'$, $Q\goto{Nvz}_\bullet Q'$
  with $z\notin\fn(Q)$, $R'\mathbin=P'|Q'$, and $\fT_\nu(R)=\fT_\nu(P)|\fT_\nu(Q)$.
  By \lem{clash-free premises}, $P$ and $Q$ are clash-free.
  By \lem{bound names}, $z\in\RN(P)$, so $z\notin\RN(Q)$ by the clash-freedom of $R$.
  Consequently, by induction, $\fT_\nu(P)\goto{M\bar xz}\eqa \fT_\nu(P')$ and $\fT_\nu(Q)\goto{Nvz}\eqa \fT_\nu(Q')$.
  Thus  $\fT_\nu(R) \goto{\debind\alpha}\eqa \fT_\nu(R')$ by application of rule \textsc{\textbf{\small e-s-com}}.
\item Suppose $R \goto{\alpha}_\bullet R'$ is derived by \textsc{\textbf{\small res}}.
  Then $R=(\nu y)P$, $R'=(\nu y)P'$, $P\goto{\alpha}_\bullet P'$ and $y\notin\n(\alpha)$.
  By \lem{clash-free premises}, $P$ is clash-free.
  Moreover, $\alpha$ is not of the form $Mxz$ with $z\mathbin\in\RN(P)\mathbin\subseteq\RN(R)$.
  By induction $\fT_\nu(P)\goto{\debind\alpha}\eqa \fT_\nu(P')$.
  Now $\fT_\nu(R) = \fT_\nu(P) \goto{\debind\alpha}\eqa \fT_\nu(P') = \fT_\nu(R')$.
\item Suppose $R \goto{\alpha}_\bullet R'$ is derived by \textsc{\textbf{\small symb-open}}.
  Then $R=(\nu y)P$, $\alpha=M\bar x (y)$ and $P\goto{M\bar x y}_\bullet R'$.
  By \lem{clash-free premises}, $P$ is clash-free.
  By induction $\fT_\nu(P)\goto{M\bar x y}\eqa \fT_\nu(R')$.
  Now $\fT_\nu(R) = \fT_\nu(P) \goto{\debind\alpha}\eqa \fT_\nu(R')$.
\qed
\end{itemize}
\end{proof}

The next lemma makes use of the set $\no(\beta)$ of \emph{non-output} names of an action $\beta$ in
\plat{$\piWR$}.
Here $\no(\beta):=\n(\beta)$ if $\beta$ has the form $M\tau$ or $Mxz$, whereas
$\no(M\bar x y):=\n(M)\cup\{x\}$.

\begin{lemma}\rm\label{lem:nu-in}
  If $R$ is clash-free and $\fT_\nu(R) \goto{\beta} U$,
  where $\no(\beta)\cap\RN(R)=\emptyset$,
  then \plat{$R \goto{\alpha}_\bullet R'$} for some $\alpha$ and
  $R'$ with $\debind\alpha=\beta$ and $\fT_\nu(R')\eqa U$.
\end{lemma}
\begin{proof}
By induction on the derivation of $\fT_\nu(R) \goto{\beta} U$, and a nested structural
induction on $R$.
\begin{itemize}
\item The cases $R=M\tau.P$ and $R=M\bar x y.P$ are trivial.
\item Let $R=Mx(y).P$.
  Then $\fT_\nu(R)=Mx(y).\fT_\nu(P)$.
  Hence $\beta\mathbin=Mxz$ and $U\mathbin=\fT_\nu(P)\renb{z}{y}$.
  Furthermore, $R\mathbin{\goto{Mxz}_\bullet} P\renb{z}{y}$.
  Finally,
  $\fT_\nu(P\renb{z}{y})\eqa\fT_\nu(P)\renb{z}{y}$,
  since $\renb{z}{y}$ is clash-free on $P$, using \lem{clash-free early-input}.
\item 
Let $R=A(\vec{y})$ with \plat{$A(\vec{x}) \stackrel{{\rm def}}{=} P$}.
Then $\fT_\nu(R)=A_\nu(\vec{y})$ with  \plat{$A_\nu(\vec{x}) \stackrel{{\rm def}}{=} \fT_\nu(P)$}.
So $\fT_\nu(P)\renbt{y}{x}\goto{\beta} U$. 
By \lem{clash-free ide}, the substitution $\renbt{y}{x}$ is clash-free on $P$.
Hence $\fT_\nu(P)\renbt{y}{x}\eqa\fT_\nu(P\renbt{y}{x})$ and
\lem{alpha norestriction} yields $\fT_\nu(P\renbt{y}{x})\goto{\beta}\eqa U$.
By \lem{clash-free premises} $P\renbt{y}{x}$ is clash-free.
Since $\RN(P\renbt{y}{x})\mathbin=\RN(P)\mathbin=\RN(R)$ one has
$\no(\beta)\cap\RN(P\renbt{y}{x})=\emptyset$.
So by induction \plat{$P\renbt{y}{x}\goto{\alpha}_\bullet R'$} for some $\alpha$ and
$R'$ with $\debind\alpha=\beta$ and $\fT_\nu(R')\eqa U$.
Now \plat{$R\goto{\alpha}_\bullet R'$} by \textsc{\textbf{\small ide}}.
\item
  Let $R=P|Q$. Then $\fT_\nu(R)=\fT_\nu(P)|\fT_\nu(Q)$. Suppose $\fT_\nu(R) \goto{\beta} U$
  is derived by \textsc{\textbf{\small par}}.
  Then $\fT_\nu(P) \goto{\beta} V$ and $U=V|\fT_\nu(Q)$. By \lem{clash-free premises} $P$ is clash-free.
Since $\RN(P)\subseteq\RN(P|Q)$, $\no(\beta)\cap\RN(P)=\emptyset$.
So by induction \plat{$P \goto{\alpha}_\bullet P'$} for some $\alpha$ and $P'$ with
$\debind\alpha\mathbin=\beta$ and $\fT_\nu(P')\eqa V$.
By \lem{bound names} $\bn(\alpha)\subseteq\RN(P)\subseteq\RN(R)$, so
$\bn(\alpha)\cap\fn(Q)=\emptyset$ by the clash-freedom of $R$.
Thus $R \goto{\alpha}_\bullet P'|Q$ by \textsc{\textbf{\small par}},
and $\fT_\nu(P'|Q)\eqa V|\fT_\nu(Q)\mathbin=U$.\vspace{3pt}

  Now suppose $\fT_\nu(R) \goto{\beta} U$ is derived by \textsc{\textbf{\small e-s-com}}.
  Then $\beta=\match{x}{v}MN\tau$, $\fT_\nu(P)\goto{M\bar xy} V$, $\fT_\nu(Q)\goto{Nvy} W$ and $U\mathbin=V|W$.
  By \lem{clash-free premises} $P$ and $Q$ are clash-free.
  Since $\RN(P)\subseteq\RN(P|Q)$, $\no(M\bar x y)\cap\RN(P)=\emptyset$.
  So by induction either \plat{$P \goto{M\bar x y}_\bullet P'$}
  or \plat{$P \goto{M\bar x (y)}_\bullet P'$} for some $P'$ with $\fT_\nu(P')\eqa V$.

  In the first case $y\in\fn(P)\subseteq\fn(R)$ by \lem{free names}, so $y\notin\RN(R)\supseteq\RN(Q)$
  by the clash-freedom of $R$. Hence also $\no(Nv y)\cap\RN(Q)=\emptyset$.
  By induction \plat{$Q \goto{Nv y}_\bullet Q'$} for some $Q'$ with $\fT_\nu(Q')=W$.
  So $R \goto{\alpha}_\bullet P'|Q'$ by \textsc{\textbf{\small e-s-com}},
  and $\fT_\nu(P'|Q')\eqa V|W\mathbin=U$.

  In the second case $y\in\RN(P)\subseteq\RN(R)$ by \lem{bound names}, so $y\notin\fn(Q)\cup\RN(Q)$ by the clash-freedom of $R$.
  Hence $\no(Nv y)\cap\RN(Q)=\emptyset$.
  By induction \plat{$Q \goto{Nv y}_\bullet Q'$} for some $Q'$ with $\fT_\nu(Q')=W$.
  So $R \goto{\alpha}_\bullet (\nu y)(P'|Q')$ by \textsc{\textbf{\small e-s-close}},
  and $\fT_\nu((\nu y)(P'|Q'))=\fT_\nu(P'|Q')\eqa V|W\mathbin=U$.
\item Finally, let $R=(\nu y)P$. Then $\fT_\nu(R)=\fT_\nu(P)$.
  Moreover, $\no(\beta)\cap\RN(P)=\emptyset$.
  By induction, \plat{$P \goto{\alpha}_\bullet P'$} for some $\alpha$ and
  $P'$ with $\debind\alpha=\beta$ and $\fT_\nu(P')\eqa U$.

  In case $y\notin\n(\alpha)$, $R \goto{\alpha}_\bullet (\nu y)P'$ by \textsc{\textbf{\small res}}.
  Moreover, $\fT_\nu((\nu y)(P'))=\fT_\nu(P')\eqa U$.

  In case $y\in\n(\alpha)=\n(\beta)$, using that $\RN(R)\mathbin\ni y \mathbin{\notin} \no(\beta)$,\linebreak[3]
  $\beta$ must have the form $M\bar x y$ with $y\neq x$ and $y\notin\n(M)$.
  So $\alpha$ is either $M\bar x (y)$ or $M\bar x y$. If $\alpha=M\bar x (y)$ then $y\in\RN(P)$ by \lem{bound names},
  contradicting the clash-freedom of $R$. So $\alpha=M\bar x y$.
  Now $R\goto{M\bar x (y)}_\bullet P'$ by \textsc{\textbf{\small symb-open}}.
\qed
\end{itemize}
\end{proof}

\begin{lemma}\label{lem:input names general}
Let $P$ be a process in $\pi_{ES}(\nN,\!\pN)$.
If $P\goto{Mxy}_\bullet Q$ then $\exists R.~P\goto{Mxz}_\bullet R$ for any $z\notin\RN(P)$.
\end{lemma}
\begin{proof}
  A trivial induction on the inference of $P\goto{Mxy}Q$.
\end{proof}

\begin{theorem}\label{thm:step5}
If $P$ is clash-free then $\fT_\nu(P)\sbb P$.
\end{theorem}
\begin{proof}
  Let \[{\R}:=\left\{(R,U),  (U,R) \left|\, \begin{array}{@{}l@{}}R \mbox{~in~}
  \pi_{E\hspace{-.8pt}S}^{\hspace{-.5pt}\not\hspace{.5pt}\alpha}(\nN,\!\pN\hspace{-.5pt})\hspace{-1.1pt}\mbox{~is
    clash-free}\\
  \mbox{and~} U\eqa \fT_\nu(R)\end{array}\right\}\right..\]
  It suffices to show that $\R$ is a strong barbed bisimulation.
  So suppose $R$ is clash-free and $U\eqa \fT_\nu(R)$.

  Let $R\goto\tau_\bullet R'$.
  Then $\fT_\nu(R)\goto{\tau}\eqa \fT_\nu(R')$ by \lem{nu-out},
  and $U\goto{\tau} U' \eqa \fT_\nu(R')$ by \lem{alpha norestriction}.
  Moreover, $R'$ is clash-free by \lem{clash-free transitions}, so $R' \R U'$.

  Let $U \goto{\tau} U'$.
  Then $\fT_\nu(R) \goto{\tau}\eqa U'$ by \lem{alpha norestriction}.
  Thus, by \lem{nu-in}, \plat{$R \goto{\tau}_\bullet R'$} for some $R'$ with $\fT_\nu(R')\eqa U'$.
  Moreover, $R'$ is clash-free by \lem{clash-free transitions}, so $U' \R R'$.

  Now let $R{\downarrow_b}$ with $b\in \nN\cup\overline\nN$.
  Then $R\goto{by}_\bullet R'$ or $R\goto{b(y)}_\bullet R'$ for some $y$ and $R'$, using the definition of $O$ in \sect{barbed}.
  By \lem{input names general} I may assume, w.l.o.g., that if $b\mathbin\in\nN$ then $y\mathbin{\notin}\RN(R)$.
  So $\fT_\nu(R)\goto{by}\eqa \fT_\nu(R')$ by \lem{nu-out},
  and $U\goto{by}\eqa \fT_\nu(R')$ by \lem{alpha norestriction}.
  Thus $U{\downarrow_b}$.

  Finally, let $U{\downarrow_b}$ with $b=x\in\nN$ or $b=\bar x$ with $x\in\overline\nN$.
  Then $U\goto{by} U'$ for some $y$ and $U'$.
  So $\fT_\nu(R)\goto{by} U''$ for some $y$ and $U''$ by \lem{alpha norestriction}.
  Since $R$ is well-typed, $\RN(R)\subseteq \pN$, so $x\notin\RN(P)$.
  By \lem{input names general} I may assume, w.l.o.g., that if $b=x$ then $y\mathbin{\notin}\RN(R)$.
  Hence $\no(by)\cap\RN(R)=\emptyset$.
  Thus, by \lem{nu-in}, $R\goto{by}_\bullet R'$ or $R\goto{b(y)}_\bullet R'$ for some $R'$.
  Hence $R{\downarrow_b}$.
\end{proof}

\subsection{Replacing substitution by relabelling}\label{sec:relabelling}

Recall that $\piWR$ is the version of the $\pi$-calculus without restriction, equipped with the early symbolic
operational semantics (\tab{pi-early-symbolic} without the orange rules), using
$\nN\uplus\pN$ as the set of names, subject to the following typing restrictions:
(i) each binder $x(z)$ occurring in a process satisfies $z\mathbin\in\nN$, and
(ii) each defining equation \plat{$A(x_1...,x_n) \stackrel{{\rm def}}{=} P$}
satisfies $x_i\in\nN$ and $\fn(P)\subseteq \{x_1,\ldots,x_n\}\cup\pN$.

$\piRWR$ is the variant of $\piWR$ to which has been added a relabelling operator $\as{\sigma}$ for
each substitution $\sigma$ with $\dom(\sigma)$ finite and $\dom(\sigma)\cap\pN=\emptyset$;
its structural operation semantics is given by\vspace{-2ex}
\[\frac{P\goto{\alpha}P'}{P[\sigma]\goto{\alpha\as{\sigma}}P'[\sigma]}\,.\]
Moreover, the substitutions $\renb{z}{y}$ and $\renbt{y}{x}$ that appear in rules
\textsc{\textbf{\small early-input}} and \textsc{\textbf{\small ide}} are replaced by
applications of the relabelling operators $[\renb{z}{y}]$ and $[\renbt{y}{x}]$, respectively.

This section defines a collection of ``clash-free'' $\piWR$ processes, and shows that on
clash-free processes the identity translation from $\piWR$ to $\piRWR$ preserves $\sbb$.

\begin{lemma}\rm\label{lem:substitution}
  Let $R$ be a \plat{$\piWR$} process
  and $\sigma$ a substitution with $\dom(\sigma)$ finite and $\dom(\sigma)\cap\pN=\emptyset$.
  If \plat{$R \goto{\alpha} R'$}
  then \plat{$R{\sigma} \goto{\alpha\as{\sigma}}\eqa R'{\sigma}$}.
\end{lemma}
\begin{proof}
  With induction on the derivation of $R \goto{\alpha} R'$.
  \begin{itemize}
    \item The cases that $R \goto{\alpha} R'$ is derived by rule \textsc{\textbf{\small tau}} or
      \textsc{\textbf{\small output}} are trivial.
    \item  Suppose $R\mathbin{\goto{\alpha}} R'$ is derived by rule \textsc{\textbf{\small early-input}}.
      Then $R\mathbin=\textcolor{black}{M}x(y).P$, $\alpha\mathbin=\textcolor{black}{M}xz$ and
      $R'\mathbin=P\renb{z}{y}$. So
      $R{\sigma} \mathbin= \textcolor{black}{M\as{\sigma}}x\as{\sigma}(w).(P\renb{w}{y}{\sigma}) \wedge
      \alpha\as{\sigma}\mathbin=\textcolor{black}{M\as{\sigma}}x\as{\sigma}z\as{\sigma}$
      where $w$ is chosen outside $\fn((\nu y)P)\cup \dom(\sigma)\cup{\it range}(\sigma)$.
      By \textsc{\textbf{\small early-input}}
      $R{\sigma} \mathbin{\goto{\alpha\as{\sigma}}}P\renb{w}{y}{\sigma}\renb{z\as{\sigma}}{w}\eqa P\renb{z}{y}{\sigma}$.
    \item Suppose $R\mathbin=A(\vec y)$.
      Let $A(\vec x) \mathbin{\stackrel{{\rm def}}{=}} P$. Say $\vec{x}\mathbin=(x_1,\dots,x_n)$.
      Then $P\renb{\vec y}{\vec x}\mathbin{\goto{\alpha}}R'\!$, so by induction
      $P\renb{\vec y}{\vec x}{\sigma} \mathbin{\goto{\alpha\as{\sigma}}\eqa} R'{\sigma}\!$.

      Moreover, $R{\sigma} \mathbin=A(\vec y\as{\sigma})$ and
      $P\renb{\vec y\as{\sigma}}{\vec x} \eqa P\renb{\vec y}{\vec x}{\sigma}$.
      Here I use that $\fn(P)\cap\dom(\sigma)\subseteq\{x_1,\dots,x_n\}$.
      So $P\renb{\vec y\as{\sigma}}{\vec x} \mathbin{\goto{\alpha\as{\sigma}}\eqa} R'{\sigma}$ by \lem{alpha norestriction}.
      Thus, by rule \textsc{\textbf{\small ide}}, \plat{$R{\sigma} \goto{\alpha\as{\sigma}}\eqa R'{\sigma}$}.
    \item The cases that $R \goto{\alpha} R'$ is derived by rule \textsc{\textbf{\small sum}},
      \textsc{\textbf{\small symb-match}}, \textsc{\textbf{\small par}} or \textsc{\textbf{\small e-s-com}} are trivial.
  \qed
  \end{itemize}
\end{proof}

Define the \emph{input arguments} $\ia(\alpha)$ of an action $\alpha$ by
\begin{center}
$\ia(M\bar x y)=\ia(M\tau) := \emptyset$ \quad and \quad $\ia(Mxz)=\{z\}$.
\end{center}

\begin{lemma}\label{lem:fn}\rm
Let $P$ be a $\piWR$ process.
If $P\goto{\alpha}Q$ then
$\n(\alpha){\setminus}\ia(\alpha)\subseteq\fn(P) \cup \pN$
and
$\fn(Q)\subseteq\fn(P)\cup \ia(\alpha) \cup \pN$.
\end{lemma}
\begin{proof}
  A trivial induction on the inference of $P\goto{\alpha}Q$.
\end{proof}

\noindent
Lemmas~\ref{lem:free names} and~\ref{lem:free names 2} make the same statements for the calculus
$\pi_{ES}(\nN,\pN)$, but without the additions $\cup\,\pN$.  These additions are needed for the case of
recursion, because the bodies of defining equations may introduce names from $\pN$.

\begin{lemma}\rm\label{lem:input universality pre}
  Let $P$ be a \plat{$\piWR$} process.
  If $P\goto{Mxz}P_z$ and $w\notin\fn(P)$ then there is a process $P_w$ such that
  $P\mathbin{\goto{Mxw}}P_w$ and $P_z \eqa P_w\renb{z}{w}$.
  Moreover, the size of the derivation of \plat{$P\mathbin{\goto{Mxw}}P_w$} is the same as that of \plat{$P\goto{Mxz}P_z$}.
\end{lemma}

\begin{proof}
  A trivial induction on the inference of $P\goto{Mxz}P_z$.
\end{proof}

\begin{lemma}\rm\label{lem:input universality}
  Let $P$ be a \plat{$\piWR$} process.
  If $P\goto{Mxw}P_w$ with $w\notin\fn(P)$ then for each name $z$ there is a process $P_z$ such that
  $P\goto{Mxz}P_z$ and $P_z \eqa P_w\renb{z}{w}$.
  Moreover, the size of the derivation of \plat{$P\mathbin{\goto{Mxz}}P_z$} is the same as that of \plat{$P\goto{Mxw}P_w$}.
\end{lemma}

\begin{proof}
  A trivial induction on the inference of $P\goto{Mxw}P_w$.
\end{proof}

\begin{lemma}\rm\label{lem:substitution reverse}
  Let $R$ be a \plat{$\piWR$} process
  and $\sigma$ a substitution with $\dom(\sigma)$ finite and $\dom(\sigma)\cap\pN=\emptyset$.
  If \plat{$R\sigma \goto{\beta} U$}
  with $\ia(\beta)\cap\dom(\sigma)\subseteq {\it range}(\sigma)$
  then \plat{$R \goto{\alpha} R'$} for some $\alpha$ and $R'$ with
  $\alpha\as{\sigma}\mathbin=\beta$ and $R'\sigma \mathbin\eqa U\!$.
  Moreover, the size of the derivation of \plat{$R \goto{\alpha} R'$} is the same as that of \plat{$R\sigma \goto{\beta} U$}.
\end{lemma}
\begin{proof}
  With induction on the size of the derivation of $R\sigma \mathbin{\goto{\beta}} U\!$.
  \begin{itemize}
    \item The cases that $R\sigma \goto{\beta} U$ is derived by rule \textsc{\textbf{\small tau}} or
      \textsc{\textbf{\small output}} are trivial.
    \item Suppose $R\sigma \goto{\beta} U$ is derived by \textsc{\textbf{\small early-input}}.
      Then $R=\textcolor{black}{M}x(y).P$ and
      $R{\sigma} = \textcolor{black}{M\as{\sigma}}x\as{\sigma}(w).(P\renb{w}{y}{\sigma})$ with
      $w\notin \fn((\nu y)P) \cup \dom(\sigma)\cup{\it range}(\sigma)$. So
      $\beta=\textcolor{black}{M\as{\sigma}}x\as{\sigma}v$ and
      $U=P\renb{w}{y}{\sigma}\renb{v}{w}$.
      Since $v\in\ia(\beta)$, there is a $z$ with $z\as{\sigma}=v$.
      By \textsc{\textbf{\small early-input}} $R\goto{\alpha}R'$ with
      $\alpha=\textcolor{black}{M}xz$ and $R'=P\renb{z}{y}$.
      Now $\alpha\as{\sigma}=\beta$ and $R'\sigma\eqa U$.
    \item Suppose $R\sigma \goto{\beta} U$ is derived by \textsc{\textbf{\small ide}}.
      Then $R=A(\vec y)$ and $R\sigma\mathbin=A(\vec{y}\as{\sigma})$.
      Let \plat{$A(\vec x) \mathbin{\stackrel{{\rm def}}{=}} P$}. Say $\vec{x}\mathbin=(x_1,\dots,x_n)$.
      Then $P\renb{\vec{y}\as{\sigma}}{\vec{x}} \goto{\beta} U$.
      Since $P\renb{\vec{y}\as{\sigma}}{\vec{x}} \eqa P\renbt{y}{x}\sigma$,
      using that $\fn(P)\cap\dom(\sigma)\subseteq\{x_1,\dots,x_n\}$, \lem{alpha norestriction} yields
      $P\renbt{y}{x}\sigma\goto{\beta}\eqa U$.
      So by induction \plat{$P\renbt{y}{x} \goto{\alpha} R'$} for some $\alpha$ and $R'$ with
      $\alpha\as{\sigma}=\beta$ and $R'\sigma \eqa U$.
      By rule \textsc{\textbf{\small ide}} $R\goto{\alpha}R'$.
    \item The cases that $R\sigma \goto{\beta} U$ is derived by rule \textsc{\textbf{\small sum}},
      \textsc{\textbf{\small symb-match}} or \textsc{\textbf{\small par}} are trivial.
    \item Suppose $R\sigma \goto{\beta} U$ is derived by \textsc{\textbf{\small e-s-com}}.
      Then $R=P|Q$, $R\sigma=P\sigma|Q\sigma$, $\beta=\match{x}{v}MN\tau$, $P\sigma\goto{M\bar xy} V$,
      $Q\sigma\goto{Nvy} W$ and $U\mathbin=V|W$. By \lem{fn}, $y\mathbin\in\fn(P\sigma)$,
      so $y\mathbin\in{\it range}(\sigma)$ or $y\mathbin{\notin}\dom(\sigma)$. Hence
      $\ia(Nvy)\cap\dom(\sigma)\linebreak[4]\subseteq {\it range}(\sigma)$. By induction, there are matching sequences $K,L$
      with $K\as{\sigma}\mathbin=M$ and $L\as{\sigma}\mathbin=N$,
      names $q,r,z,u$ with $q\as{\sigma}=x$, $r\as{\sigma}=v$,
      $z\as{\sigma}=y$ and $u\as{\sigma}=y$, and processes $P'$ and $Q'$ with $P'\sigma\eqa V$ and
      $Q'\sigma\eqa W$, such that $P\goto{K\bar qz} P'$ and $Q\goto{L ru} Q'$.
      Pick $w\notin\fn(Q)$. By \lem{input universality pre} there is a process $P_w$ such that
      $Q\mathbin{\goto{Lrw}}Q_w$ and $Q' \eqa Q_w\renb{u}{w}$.
      By \lem{input universality} there is a process $P_z$ such that
      $Q\mathbin{\goto{Lrz}}Q_z$ and $Q_z \eqa Q_w\renb{z}{w}$.
      Note that $Q'\sigma \eqa Q_z\sigma$.
      By \textsc{\textbf{\small e-s-com}} $R\goto{\match{q}{r}KL\tau} P'|Q_z$.
      Moreover, $(\match{q}{r}KL\tau)\as{\sigma}=\match{x}{v}MN\tau$ and $(P'|Q_z)\sigma\linebreak[2] \eqa V|W=U$.
  \qed
  \end{itemize}
\end{proof}

For $P$ a $\piRWR$ process, let $\widehat P$ be the $\piWR$ process obtained from $P$ by recursively replacing
each subterm $Q[\sigma]$ by $Q\sigma$, and each agent identifier $A$ by $A\hat{~}$.
Here $A\hat{~}$ is a fresh agent identifier with defining equation
\plat{$A\hat{~}(\vec x) \stackrel{{\rm def}}{=} \widehat P$} when
\plat{$A(\vec x) \stackrel{{\rm def}}{=} P$} was the defining equation of $A$.

\begin{lemma}\label{lem:relabelling forth}
If $R \goto{\alpha} R'$ then \plat{$\widehat R \goto{\alpha}\eqa \widehat{R'}$}.
\end{lemma}
\begin{proof}
By induction of the inference of $R \goto{\alpha} R'$.
\begin{itemize}
\item Suppose $R \goto{\alpha} R'$ is derived by \textsc{\textbf{\small tau}}.
  Then $R\mathbin=\textcolor{black}{M}\tau.P$, $\alpha\mathbin=\textcolor{black}{M}\tau$
  and $R'\mathbin=P$. Moreover, $\widehat R = \textcolor{black}{M}\tau.\widehat P$
  and $\widehat R \goto{\alpha} \widehat {R'}$.
\item The case that $R \goto{\alpha} R'$ is derived by \textsc{\textbf{\small output}} proceeds likewise.
\item Suppose $R \goto{\alpha} R'$ is derived by \textsc{\textbf{\small early input}}.
  Then $R\mathbin=\textcolor{black}{M}x(y).P$, $\alpha\mathbin=\textcolor{black}{M}xz$
  and $R'\mathbin=P\rensq{z}{y}$. Moreover $\widehat R = \textcolor{black}{M}x(y).\widehat P$ and
  $\widehat R \goto{\alpha} \widehat P\renb{z}{y} = \widehat {P[\renb{z}{y}]} = \widehat {R'}$.
\item Suppose $R \goto{\alpha} R'$ is derived by \textsc{\textbf{\small sum}}.
  Then $R=P+Q$ and $P \goto{\alpha} R'$.
  Now $\widehat R = \widehat P + \widehat Q$.
  By induction $\widehat P \goto{\alpha}\eqa \widehat {R'}$.
  Hence, by \textsc{\textbf{\small sum}}, $\widehat R \goto{\alpha}\eqa \widehat {R'}$.
\item Suppose $R \goto{\alpha} R'$ is derived by \textsc{\textbf{\small symb-match}}.
  Then $R=\Match{x}{y}P$, $P \goto{\beta} R'$ and $\alpha\mathbin=\match{x}{y}\beta$.
  Now $\widehat R = \Match{x}{y}\widehat P$. By induction $\widehat P \goto{\beta}\eqa \widehat {R'}$.
  Hence, by \textsc{\textbf{\small symb-match}}, $\widehat R \goto{\alpha}\eqa \widehat {R'}$.
\item Suppose $R \goto{\alpha} R'$ is derived by \textsc{\textbf{\small ide}}.
  Then $R=A(\vec y)$ with \plat{$A(\vec{x})\stackrel{\rm def}{=} P$} and
  $P[\renb{\vec{y}}{\vec{x}}] \goto{\alpha} R'$.
  Now $\widehat{P[\renb{\vec{y}}{\vec{x}}]} = \widehat P \renbt{y}{x}$.
  By induction $\widehat{P[\renb{\vec{y}}{\vec{x}}]} = \widehat P \renbt{y}{x} \goto{\alpha}\eqa \widehat{R'}$.
  Hence, by \textsc{\textbf{\small ide}}, $\widehat R = A\hat{~}(\vec y) \goto{\alpha}\eqa \widehat{R'}$.
\item Suppose $R \goto{\alpha} R'$ is derived by \textsc{\textbf{\small par}}.
  Then $R=P|Q$, $P \goto{\alpha} P'$ and $R'=P'|Q$.
  Now $\widehat R = \widehat P | \widehat Q$
  and $\widehat {R'} = \widehat {P'} | \widehat Q$.
  By induction $\widehat P \goto{\alpha}\eqa \widehat {P'}$.
  Thus $\widehat R \goto{\alpha}\eqa \widehat {R'}$.
\item Suppose $R \goto{\alpha} R'$ is derived by \textsc{\textbf{\small e-s-com}}.
  Then $R\mathbin=P|Q$,
  \plat{$P\goto{M\bar xy} P'$}, {$Q\goto{Nvy} Q'$}, $R'\mathbin=P'|Q'$,
  and $\alpha=\match{x}{v}MN\tau$. Now $\widehat R = \widehat P | \widehat Q$
  and $\widehat {R'} = \widehat {P'} | \widehat {Q'}$.
  By induction
  \plat{$\widehat P\goto{M\bar xy}\eqa \widehat{P'}$} and {$\widehat Q\goto{Nvy}\eqa \widehat {Q'}$}.
  Thus $\widehat R \goto{\alpha}\eqa \widehat {R'}$.
\item Suppose $R \goto{\alpha} R'$ is derived by \textsc{\textbf{\small relabelling}}.
  Then $R\mathbin=P[\sigma]$, $P \goto{\beta} P'$, $R'\mathbin=P'[\sigma]$ and
  $\alpha \mathbin= \beta\as{\sigma}$.
  By induction $\widehat P \goto{\beta}\eqa \widehat{P'}$.
  By \lem{substitution},
  $\widehat R = \widehat P\sigma   \goto{\alpha}\eqa \widehat P'\sigma = \widehat R'$.
\qed
\end{itemize}
\end{proof}

For $A\in\K_n$ an agent identifier with $A(x_1,\ldots,x_n) \stackrel{{\rm def}}{=} P$,
let $\dn(A)=\{x_1,\dots,x_n\}$ be the set of \emph{declared names} of $A$.

\begin{definition}\label{df:BN}
Let $\BN(P)$, the \emph{hereditary bound names} of a $\piRWR$ process $P$, be the set of all names
$y$ such that $h(P)$ contains either a process $Mx(y)Q$, or a process $A(\vec{z})$ with $y\in\dn(A)$,
or a process $Q[\sigma]$ with $y\in\dom(\sigma)$.
\end{definition}

\begin{definition}
The \emph{free names} of a $\piRWR$-process $P$ are defined inductively as follows:
\[\begin{array}{r@{~~=~~}l}
\fn(\nil) & \emptyset \\
\fn(M\tau.P) & \n(M) \cup \fn(P) \\
\fn(M\bar xy.P) & \n(M) \cup \{x,y\} \cup \fn(P) \\
\fn(Mx(y).P) & \n(M) \cup \{x\} \cup \fn(P){\setminus}\{y\} \\
\fn(\Match{x}{y}P) & \{x,y\} \cup \fn(P) \\
\fn(P+Q)=\fn(P\mid Q) & \fn(P)\cup\fn(Q) \\
\fn(A(y_1,\dots,y_n)) & \{y_1,\dots,y_n\}\\
\fn(P[\sigma]) & \{x\as{\sigma} \mid x\in\fn(P)\}
\end{array}\]
\end{definition}

\noindent
In the absence of relabelling operators, this definition agrees with the one from  \sect{pi}.

\begin{definition}\label{df:clash-free RWR}
  A $\piRWR$ process $P$ is \emph{clash-free} if\vspace{-1ex}
  \begin{enumerate}
  \item for each $A(\vec{z})\mathbin\in h(P)$ with $A(\vec{x}) \mathbin{\stackrel{{\rm def}}{=}} Q$ one has $x_i\mathbin{\notin}\BN(Q)$,
  \item for each $Q[\sigma] \in h(P)$ one has $\dom(\sigma)\cap\BN(Q)=\emptyset$,
  \item for each $Mx(y).Q \in h(P)$ one has $y\notin\BN(Q)$, and
  \item $(\fn(P)\cup\pN)\cap\BN(P)=\emptyset$.
  \end{enumerate}
\end{definition}
This definition applies equally well to $\piWR$ processes; here Clause 2 is moot, as 
there are no relabelling operators in $\piWR$.

The relation $\leftarrowtail$ from \df{pi-unguarded} applies to $\piRWR$ as well, except that there
is a clause $P[\sigma] \leftarrowtail P$, and
the last clause is $A(\vec{y}) \leftarrowtail P[\renbt{y}{x}]$ when \plat{$A(\vec{x}) \stackrel{{\rm def}}{=} P$}.

\begin{lemma}\label{lem:clash-free premises RWR}
If $P \leftarrowtail Q$ and $P$ is a clash-free $\piRWR$ process, then so is $Q$.
\end{lemma}
\begin{proof}
  Only the cases of relabelling and recursion are nontrivial. For $P[\sigma] \mathbin\leftarrowtail P$ use:
  $\fn(P)\mathbin\subseteq\fn(P[\sigma])\cup\dom(\sigma)$ and $\BN(P)=\BN(P[\sigma])\mathord\setminus\dom(\sigma)$.\vspace{1pt}

  Suppose $A(\vec{y})$ is clash-free and \plat{$A(\vec{x}) \stackrel{{\rm def}}{=} P$}.
  By definition $P$ satisfies Clauses 1--3. Hence $P[\renbt{y}{x}]$ satisfies Clauses 1 and 3.
  Moreover $P[\renbt{y}{x}]$ satisfies Clause 2 since  $A(\vec{y})$ satisfies Clause 1.
  Finally, $\BN(P[\renbt{y}{x}]) = \BN(P)\cup\{x_1,\dots,x_n\} = \BN(A(\vec{y}))$
  and $\fn(P[\renbt{y}{x})\subseteq \fn(A(\vec{y}))\cup\pN$, so Clause 4 holds too.
\end{proof}

Note that \lem{fn} holds also for $\piRWR$ processes.

\begin{lemma}\label{lem:relabelling back}
If $R$ is clash-free and $\widehat R \goto{\alpha} U$ with $\ia(\alpha)\cap\BN(R)=\emptyset$ then \plat{$R \goto{\alpha} R'$}
for some $R'$ with $\widehat R' \eqa U$.
\end{lemma}
\begin{proof}
  By induction on the size of the derivation of \plat{$\widehat R \goto{\alpha} U$}, with a nested
  induction on the number of topmost renaming operators in $R$.
\begin{itemize}
\item Suppose that $R$ is not of the form $P[\sigma]$.
  The cases that \plat{$\widehat R \goto{\alpha} U$} is derived by
  \textsc{\textbf{\small tau}}, \textsc{\textbf{\small output}}, \textsc{\textbf{\small early-input}},
  \textsc{\textbf{\small sum}}, \textsc{\textbf{\small symb-match}}, \textsc{\textbf{\small ide}} or
  \textsc{\textbf{\small par}} are trivial, similar to the cases
  spelled out in the proof of \lem{relabelling forth}, but using \lem{clash-free premises RWR} to
  establish clash-freedom when applying the induction hypothesis, and also using that
  $R \leftarrowtail P$ implies $\BN(P)\subseteq\BN(R)$.

  Suppose \plat{$\widehat R \goto{\alpha} U$} is derived by rule \textsc{\textbf{\small e-s-com}}.
  Then $R=P|Q$, $\widehat R\mathbin=\widehat P|\widehat Q$, $\alpha=\match{x}{v}MN\tau$, $\widehat P\goto{M\bar xy}V$, $\widehat Q\goto{Nvy}W$ and $\widehat R'\mathbin=V|W$.
  By \lem{clash-free premises RWR} $P$ and $Q$ are clash-free.
  By induction $P \goto{M\bar xy}P'$ for some $P'$ with $\widehat P' \eqa V$.
  By \lem{fn} $y\in\fn(P)\cup\pN\subseteq \fn(R)\cup\pN$, and hence $\ia(Nvy)\cap\BN(Q)=\emptyset$ by the clash-freedom of $R$.
  So by induction $Q \goto{Nvy}Q'$ for some $Q'$ with $\widehat Q' \eqa W$.
  By \textsc{\textbf{\small e-s-com}} $R \goto{\alpha} P'|Q'\eqa U$.
\item Now suppose $R=P[\sigma]$. Then $\ia(\alpha)\cap\dom(\sigma)=\emptyset$ and $\widehat R = \widehat P \sigma$.
  By \lem{substitution reverse} $\widehat P \goto{\beta} V$ for some $\beta$ and $V$ with
  $\beta\as{\sigma}\mathbin=\alpha$ and $V\sigma \mathbin\eqa U$.
  Moreover, the size of the derivation of \plat{$\widehat P \goto{\beta}V$} is the same as that of
  \plat{$\widehat P\sigma \goto{\alpha} U$}.
  Suppose $\ia(\beta) =\{z\}$. Then either $z\in\dom(\sigma)$, so $z\notin\BN(P)$ by the clash-freedom of $R$,
  or $\ia(\beta)=\ia(\alpha)$ and $z\notin\BN(R)\supseteq\BN(P)$.
  So by induction $P\goto{\beta} P'$ for some $P'$ with $\widehat{P'}\eqa V$.
  By rule \textsc{\textbf{\small relabelling}} $R\goto{\alpha}P'[\sigma]$.
  Furthermore one has $\widehat{P'[\sigma]}=\widehat{P'}\sigma \eqa V\sigma \eqa U$.
\qed
\end{itemize}
\end{proof}

\begin{lemma}\label{lem:clash-free transitions 2}
If $R$ in $\piRWR$ is clash-free, $\ia(\alpha)\cap\BN(R)\mathbin=\emptyset$
and $R\goto{\alpha} R'$, then $R'$ is clash-free and $\BN(R')\subseteq\BN(R)$.
\end{lemma}

\begin{proof}
  First I show that $\BN(R')\subseteq\BN(R)$ implies that $R'$ meets Clause 4 of \df{clash-free RWR}.
  That $\pN\cap\BN(R')=\emptyset$ follows since $\pN\cap\BN(R)=\emptyset$ by the clash-freedom of $R$.
  Now suppose $y\in \fn(R')\cap\BN(R')$.
  The case that $y\in\fn(R)$ contradicts the clash-freedom of $R$.
  In view of \lem{fn}, the only remaining case is that $y\in\ia(\alpha)$.
  However, in that case $y\notin\BN(R)\supseteq\BN(R')$ by the assumption of the lemma.

  The rest of the proof proceeds with induction on the derivation of $R \goto{\alpha} R'$.
  \begin{itemize}
    \item The cases that $R \mathbin{\goto{\alpha}} R'$ is derived by rule \textsc{\textbf{\small tau}} or
      \textsc{\textbf{\small early-output}} are trivial, using that $R'\hspace{-1pt}$ \mbox{is a subexpression of $R$.}
    \item  Suppose $R\mathbin{\goto{\alpha}} R'$ is derived by \textsc{\textbf{\small early-input}}.
      Then $R\mathbin=Mx(y).P$, $\alpha\mathbin=Mxz$ with $z\notin\BN(R)$ and $R'\mathbin=P[\renb{z}{y}]$.
      By definition, $P$ satisfies Clauses 1--3 of \df{clash-free RWR}. Hence $R'$ satisfies Clauses 1 and 3.
      Moreover $R'$ satisfies Clause 2 since $R$ satisfies Clause 3.
      Finally, $\BN(R') = \BN(P)\cup\{y\} = \BN(R)$.
    \item The cases that $R \goto{\alpha} R'$ is derived by rule \textsc{\textbf{\small sum}},
      \textsc{\textbf{\small ide}} or \textsc{\textbf{\small symb-match}} are trivial.
    \item  Suppose \plat{$R\goto{\alpha} R'$} is derived by \textsc{\textbf{\small par}}.
      Then $R=P|Q$, $P\goto{\alpha} P'$, $R'=P'|Q$, and if $\alpha\mathbin=Mxz$ then $z\notin\BN(P)$.
      By \lem{clash-free premises RWR}, $P$ and $Q$ are clash free.
      So by induction $P'$ is clash-free and $\BN(P')\subseteq\BN(P)$.
      Hence $\BN(R')\subseteq\BN(R)$ and $R'$ is clash-free.
     \item  Suppose \plat{$R\goto{\alpha} R'$} is derived by \textsc{\textbf{\small e-s-com}}.
      Then $R=P|Q$, $\alpha=\match{x}{v}MN\tau$, $P\goto{M\bar xy} P'$, $Q\goto{Nvy} Q'$ and $R'\mathbin=P'|Q'$.
      By \lem{fn}, $y\in\fn(P)$, so $y\notin\BN(Q)$ by the clash-freedom of $R$.
      By \lem{clash-free premises RWR}, $P$ and $Q$ are clash free.
      Hence, by induction, $P'$ and $Q'$ are clash-free, $\RN(P')\subseteq\RN(P)$ and $\RN(Q')\subseteq\RN(Q)$.
      Thus $\BN(R')\subseteq\BN(R)$ and $R'$ is clash-free.
     \item Suppose \plat{$R\goto{\alpha} R'$} is derived by \textsc{\textbf{\small relabelling}}.
      Then $R=P[\sigma]$, $P\goto\beta P'$, $\beta\as{\sigma}=\alpha$ and $R'=P'[\sigma]$.
      By \lem{clash-free premises RWR}, $P$ is clash free.
  Suppose $\ia(\beta) =\{z\}$. Then either $z\in\dom(\sigma)$, so $z\notin\BN(P)$ by the clash-freedom of $R$,
  or $\ia(\beta)=\ia(\alpha)$ and $z\notin\BN(R)\supseteq\BN(P)$.
      So by induction $P'$ is clash-free and $\BN(P')\subseteq\BN(P)$.
      Hence $\BN(R')\mathbin=\BN(P')\cup\dom(\sigma)\subseteq\BN(P)\cup\dom(\sigma)\mathbin=\BN(R)$.
      It remains to show that $R'$ is clash-free.

      Clauses 1 and 3 of \df{clash-free RWR} hold trivially for $R'$, since they hold for $P'$.
      The clash-freedom of $R$ yields $\dom(\sigma)\cap\BN(P)=\emptyset$, hence
      $\dom(\sigma)\cap\BN(P')=\emptyset$ and Clause 2 holds for $R'$ as well.
\qed
  \end{itemize}
\end{proof}

Let $\fT_\rho$ be the identity translation from $\piWR$ to $\piRWR$.

\begin{theorem}\rm\label{thm:step6}
  $\fT_\rho(P) \mathbin{\sbb} P$ for any clash-free $P$ in $\piWR$.
\end{theorem}

\begin{proof}
  Let \[{\R}:=\left\{(R,U),  (U,R) \left|\, \begin{array}{@{}l@{}}R \mbox{~in~} \piRWR \mbox{~is clash-free}\\
  \mbox{and~} U\eqa \widehat R \mbox{~in~} \piWR \end{array}\right\}\right..\]
Since any clash-free $P$ in $\piWR$ satisfies $\fT_\rho(P)\mathbin=P$ and $\widehat P=P$,
and thus $\fT_\rho(P)\R P$, it suffices to show that $\R$ is a strong barbed bisimulation.

  So suppose $R$ in $\piRWR$ is clash-free and $U\eqa \widehat R$.

  Let $R\goto\tau R'$.
  Then $\widehat R\goto{\tau}\eqa \widehat R'$ by \lem{relabelling forth},
  and $U\goto{\tau} U' \eqa \widehat R'$ by \lem{alpha norestriction}.
  Moreover, $R'$ is clash-free by \lem{clash-free transitions 2}, so $R' \R U'$.

  Let $U \goto{\tau} U'$.
  Then $\widehat R \goto{\tau}\eqa U'$ by \lem{alpha norestriction}.
  Thus, by \lem{relabelling back}, \plat{$R \goto{\tau} R'$} for some $R'$ with $\widehat R'\eqa U'$.
  Moreover, $R'$ is clash-free by \lem{clash-free transitions 2}, so $U' \R R'$.

  Now let $R{\downarrow_b}$ with $b\in \nN\cup\overline\nN$.
  Then $R\goto{by} R'$ for some $y$ and $R'$, using the definition of $O$ in \sect{barbed}.
  So $\widehat R\goto{by}\eqa \widehat R'$ by \lem{relabelling forth},
  and $U\goto{by}\eqa \widehat R'$ by \lem{alpha norestriction}.
  Thus $U{\downarrow_b}$.

  Finally, let $U{\downarrow_b}$ with $b=x\in\nN$ or $b=\bar x$ with $x\in\overline\nN$.
  Then $U\goto{by} U'$ for some $y$ and $U'$.
  So $\widehat R\goto{by} U''$ for some $y$ and $U''$ by \lem{alpha norestriction}.
  By \lem{input universality pre} I may assume, w.l.o.g., that if $b=x$ then $y\mathbin{\notin}\BN(R)$.
  Hence $\ia(by)\cap\BN(R)=\emptyset$.
  Thus, by \lem{relabelling back}, $R\goto{by} R'$ for some $R'$.
  Hence $R{\downarrow_b}$.
\end{proof}

\subsection{Making processes clash-free}\label{sec:ruthless}

Call a $\piZ$ or $\piZa$ process \emph{fully clash-free} if it is restriction-clash-free according to \df{clash-free}
and satisfies Clauses 1, 3 and 4 of \df{clash-free RWR}. Note that any $\piWR$ or $\piRWR$ process
is trivially restriction-clash-free, as these languages only contain well-typed processes $P$ with $\RN(P)\mathbin=\emptyset$.
If $P$ is fully clash-free, then $\fT_\nu(P)$ is a clash-free $\piWR$ process according to \df{clash-free RWR}.
Thus, to connect Steps 1 and 2 of my overall translation from $\pi_L(\N)$ to {\CCST} to steps 4, 5 and 6,
the intermediate Step 3 needs to convert each $\pi_{ES}(\N)$ process into a fully clash-free $\piZ$ process.
This section describes this third step $\Tcf$.

Let $\sN = \{^\varsigma\! s \mid \varsigma \in s^*\}$ and $\dN = \{^\varsigma\! d_i \mid \varsigma \in d^* \wedge i > 0\}$
be the sets of \emph{symbolic} and \emph{declared names}, respectively.
The set $\aN$ of \emph{spare names} and $\pN$ of \emph{private names} were defined in \sect{the encoding}.
The set $\zN$ of \emph{public names} is $\N \uplus \aN \uplus \sN \uplus \dN$,
and the set $\mM$ of all names of the target language will be $\mM = \zN\uplus\pN$.
The string $\varsigma$ in a symbolic, declared or private name is called its \emph{modifier}.

\begin{definition}\label{df:ruthless}
  The \emph{ruthless application} $P\Ren{\sigma}$ of a substitution $\sigma$ to a process $P$ is the result of
  simultaneously replacing each occurrence of an agent identifier $A$ in $P$ by $^{{\sigma}}\!\!A$
  and each occurrence of a name $x$ in $P$ by $x\as{\sigma}$. Here it does not matter if the occurrence
  of $x$ is free or bound. Furthermore, $^{{\sigma}}\!\!A$ is a fresh agent identifier,
  with defining equation \plat{$^{{\sigma}}\!\!A(\vec x\as{\sigma}) \stackrel{{\rm def}}{=} P\Ren{\sigma}$} when
  \plat{$A(\vec x) \stackrel{{\rm def}}{=} P$} was the defining equation of $A$.
\end{definition}
\noindent
A substitution $\sigma\!:\M\rightharpoonup\M$ is called \emph{injective} if $x\as{\sigma} \neq y\as{\sigma}$
for all names $x\mathbin{\neq} y$. It is \emph{surjective} if $\dom(\sigma) \subseteq {\it range}(\sigma)$,
and \emph{bijective} if it is injective as well as surjective.

Ruthless substitution may lead to name capture: a free occurrence of a name becoming bound in $P\Ren{\sigma}$.
When $\sigma$ is injective, which it will be in my applications, this is not possible.

For $y \inp\N$ and $\vec{x} \mathbin= (x_1,\dots,x_n) \inp\N^n$, let the substitutions
$s_y\!:\{y\}\cup \sN \cup \aN \rightarrow \{y\}\cup \sN\cup \aN$
and $d_{\vec{x}}\!:\{x_1,...,x_n\}\cup \dN \cup \aN \rightharpoonup\{x_1,...,x_n\}\cup \dN\cup \aN$ be defined by
\begin{center}
\begin{tabular}{@{~}l@{\hspace{4pt}:=\hspace{5pt}}lll@{\hspace{4pt}:=\hspace{5pt}}l@{~~~~~}l@{}}
$s_y ( y)$ & $s$ & &
$d_{\vec{x}} ( x_i)$ & $d_i$ & for $1\leq i\leq n$
\\
$s_y ({}^\varsigma\! s)$ & $^{s\varsigma}\! s$ & &
$d_{\vec{x}} ({}^{\varsigma}\! d_i)$ & ${}^{d\varsigma}\!d_i$ & for $1\leq i\leq n$
\\
$s_y ( s_1)$ & $y$ & &
$d_{\vec{x}} ( s_i)$ & $x_i$ & for $1\leq i\leq n$
\\
$s_y ( s_{i+1})$ & $s_i$ & &
$d_{\vec{x}} ( s_{i+n})$ & $s_i\;.$ 
\end{tabular}
\end{center}
Also recall the functions $\ell$, $r$ and $p_y$ defined in \sect{the encoding}.
Note that all substitutions $\ell$, $r$, $p_y$, $s_y$ and $d_{\vec{x}}$ are bijective.
Define the translation $\Tcf$ from $\pi(\N)$ to $\pi(\nN,\pN)$ inductively by:\vspace{-2ex}
\[\begin{array}{@{}l@{~:=~}ll@{}}
\Tcf(\nil) & \nil \\
\Tcf(\tau.P) & \tau.\Tcf(P) \\
\Tcf(\bar xy.P) & \bar xy.\Tcf(P) \\
\Tcf(x(y).P) & x(s).\big(\Tcf(P)\Ren{s_y}\big)\\
\Tcf((\nu y)P) & (\nu p)\big(\Tcf(P)\Ren{p_y}\big)\\
\Tcf(\Match{x}{y}P) & \Match{x}{y}\Tcf(P) \\
\Tcf(P\mid Q) & \Tcf(P)\Ren{\ell}\mid\Tcf(Q)\Ren{r} \\
\Tcf(P+Q) & \Tcf(P)+\Tcf(Q) \\
\Tcf(A(\vec y)) & A_{\rm cf}(\vec y) \\
\end{array}\]
where $A_{\rm cf}$ is a fresh agent identifier,\vspace{1pt} with defining
equation \plat{$A_{\rm cf}(\vec x\as{d_{\vec{x}}}) \stackrel{{\rm def}}{=} \Tcf(P)\Ren{d_{\vec{x}}}$} when
\plat{$A(\vec x) \stackrel{{\rm def}}{=} P$} was the defining equation of $A$
(which may now be dropped).
\newcommand{\lift}[1]{\raisebox{2pt}{$^{#1}\hspace{-4pt}$}}

\begin{example}
  Let $P:= (\nu \textcolor{ACMPurple}{y})\bar x\textcolor{ACMPurple}{y}.A(\textcolor{ACMPurple}{y},z)$
  where $A(x_1,x_2)\stackrel{{\rm def}}{=} (\nu \textcolor{ACMRed}{z}).A(\textcolor{ACMRed}{z},x_1)$.
  Then $\Tcf(P)=(\nu \textcolor{ACMPurple}{p})\bar x\textcolor{ACMPurple}{p}.\lift{p_y}A_{\rm cf}(\textcolor{ACMPurple}{p},z)$
  where\vspace{-1ex}
  \[\begin{array}{l@{\quad\!\!\stackrel{{\rm def}}{=}\quad}l}
  A_{\rm cf}(d_1,d_2) & (\nu \textcolor{ACMRed}{p}).\lift{d_{\vec{x}}p_y}A_{\rm cf}(\textcolor{ACMRed}{p},d_1),\\
  \lift{p_y}A_{\rm cf}(d_1,d_2) & (\nu \textcolor{ACMRed}{{}^{e}\!p}).\lift{p_yd_{\vec{x}}p_y}A_{\rm cf}(\textcolor{ACMRed}{{}^{e}\!p},d_1),\\
  \lift{p_yd_{\vec{x}}p_y}A_{\rm cf}({}^{d}\!d_1,{}^{d}\!d_2) & (\nu \textcolor{ACMRed}{{}^{ee}\!p}).\lift{p_yd_{\vec{x}}p_yd_{\vec{x}}p_y}A_{\rm cf}(\textcolor{ACMRed}{{}^{ee}\!p},{}^{d}\!d_1),...\\
  \end{array}\]
\end{example}

\begin{theorem}\label{thm:clash-free}
  Each process $\Tcf(P)$ is fully clash-free.
\end{theorem}

\begin{proof}
In $\Tcf((\nu y)P)$ the operation $\Ren{p_y}$ injectively renames all private names ${}^\varsigma\! p$ in $\Tcf(P)$ by adding a
tag $e$ in front of their modifiers. This frees up the name $p$. The translation of $(\nu y)P$
takes advantage of this by changing the bound name $y$ into $p$.
This ensures that Clause 2 of \df{clash-free} is met.

Likewise, in $\Tcf(x(y).P)$ the $\Ren{s_y}$ injectively renames all symbolic names
${}^\varsigma\! s$ in $\Tcf(P)$ by adding a tag $s$ in front of their modifiers. This
frees up the name $s$. The translation of $Mx(y)P$ finishes by changing the bound name $y$ into $s$.
As a result also Clause 3 of \df{clash-free RWR} is met. 

In $\Tcf(A(\vec y))$ the operation $\Ren{d_{\vec{x}}}$ injectively renames all declared names
${}^\varsigma\! d_i$ for $1\leq i\leq n$ by adding a tag $d$ in front of their
modifiers ${}^\varsigma$. This frees up the names $d_1,\dots,d_n$.
The translation of defining equations takes advantage of that by renaming all declared names $x_i$
into $d_i$. This way Clause 1 of \df{clash-free RWR} is met.

In $\Tcf(P|Q)$ the operations $\Ren{\ell}$ and $\Ren{r}$ injectively rename private names in
$\Tcf(P)$ and $\Tcf(Q)$ by adding a tag $\ell$ or $r$ in front of their modifiers,
depending on whether they occur in the left or the right argument. This validates Clause 3 of \df{clash-free}.

Free names of a $\pi(\N)$ process $P$ are not renamed in $\Tcf(P)$. As a consequence, one has $\fn(\Tcf(P))\subseteq\N$.
Moreover, declared names always end up in $\dN$ and names bound by input prefixes and restriction
always in $\sN$ and $\pN$, respectively. Therefore, translated processes $\Tcf(P)$ are
always well-typed, and Clauses 1 and 4 of \df{clash-free} as well as 4 of \df{clash-free RWR} are met.
\end{proof}
Moreover, after translation, any defining equation \plat{$A(\vec{x}) \mathbin{\stackrel{{\rm def}}{=}} P$} satisfies the condition
that $P$ is well-typed and $x_i\mathbin\in\nN$ for $i=1,\dots,n$.
This ensures that this third step of my translation composes fruitfully with steps four, five and six.

Since all $\Tcf$ does is replication of defining equations and renaming of bound and
declared variables, it preserves strong barbed bisimilarity, regardless whether all barbs are
considered, or only barbs from $\nN$.

\subsection{Replacing ruthless substitution by relabelling}\label{sec:ruthless relabelling}

Even though $\Tcf$ preserves $\sbb$, I have to reject it as a valid translation.
The main objection is that it employs operations on processes---\emph{ruthless substitutions}---that
are not syntactic operators of the target language. Thereby it fails the criterion of
compositionality, even if the ruthless substitutions are applied in a compositional manner.
In addition, it creates an infinite amount of replicated agent identifiers, whereas I strive not to
increase the number of agent identifiers found in the source language.
To overcome these problems, this section shows that the ruthless substitutions can be replaced by
applications of relabelling operators.

Let $\fT_{{\rm cf}\nu}^r := \fT_\nu\circ\Tcf$ be the composed translation
from $\pi_{ES}(\N)$ to \plat{$\piSWR$}. It can be inductively defined just as $\Tcf$
in \sect{ruthless}, except that the clause for restriction reads
$$\fT_{{\rm cf}\nu}^r((\nu y)P) := \Tcf(P)\Ren{p_y}.%
\footnote{Note that $\fT_\nu(\Tcf(P)\Ren{d})=\fT_\nu(\Tcf(P))\Ren{d}$.}$$
By Theorems~\ref{thm:step4} and~\ref{thm:step5} this translation preserves strong barbed bisimilarity:
$\fT_{{\rm cf}\nu}^r(P) \sbb P$ for all $\pi_{ES}(\N)$ processes $P$.

Now let $\piSRWR$ be the variant of $\piRWR$ enriched with relabelling operators $[\sigma]$ for
any bijective substitution $\sigma\!:\M\rightharpoonup\M$ that satisfies that $y\in\pN \Rightarrow \sigma(y)\in\pN$.
Like $\piRWR$, it also contains all relabelling operators $[\sigma]$ with $\dom(\sigma)$ finite,
satisfying $\dom(\sigma)\cap\pN=\emptyset$.

Finally, let $\fT_{{\rm cf}\nu}$ be the translation from $\pi_{ES}(\N)$ to \plat{$\piSRWR$} defined
inductively exactly as $\fT_{{\rm cf}\nu}^{r}$, except that each substitution $\Ren{\sigma}$ is
replaced by the relabelling operator $[\sigma]$.
\thm{relabelling} below will show that
$\fT_{{\rm cf}\nu}^{r}(P) \sbb \fT_{{\rm cf}\nu}(P)$ for any $\pi_{ES}(\N)$ process $P$.
This entails that also the translation $\fT_{{\rm cf}\nu}$ preserves strong barbed bisimilarity:
$\fT_{{\rm cf}\nu}(P) \sbb P$ for all $\pi_{ES}(\N)$ processes $P$.

{\hbadness=10000
The translation $\fT_{{\rm cf}\nu}$ replaces each defining equation 
\plat{$A(\vec x) \stackrel{{\rm def}}{=} P$} that was present in $\pi_{ES}(\N)$ by a defining
equation \plat{$A_{{\rm cf}\nu}(\vec x\as{d_{\vec{x}}}) \stackrel{{\rm def}}{=} \fT_{{\rm cf}\nu}(P)[d_{\vec{x}}]$}.
The total number of agent identifiers and defining equations in the language does not change.
In particular, the translation $\fT_{{\rm cf}\nu}$ does not introduce the infinite set of fresh
agent identifiers $^{{\sigma}}\!\!A$ of \df{ruthless}, and their defining equations
\plat{$^{{\sigma}}\!\!A(\vec x\as{\sigma}) \stackrel{{\rm def}}{=} P\Ren{\sigma}$}.
These where introduced by the translation $\Tcf$, but can be dropped as soon as we have
\thm{relabelling}. They do not form a part of the ultimate translation from $\pi_L(\N)$ to
{\CCST}; they merely played a r\^ole in the validity proof of that translation.

}

\begin{lemma}\rm\label{lem:ruthless substitution}
  Let $R$ be a \plat{$\piSWR$} process
  and $\sigma$ an injective substitution.  If \plat{$R \goto{\alpha} R'$}
  then \plat{$R\Ren{\sigma} \goto{\alpha\as{\sigma}}\eqa R'\Ren{\sigma}$}.
\end{lemma}

\begin{proof}
  With induction on the derivation of $R \goto{\alpha} R'$.
  \begin{itemize}
    \item The cases that $R \goto{\alpha} R'$ is derived by rule \textsc{\textbf{\small tau}} or
      \textsc{\textbf{\small output}} are trivial.
    \item  Suppose $R\mathbin{\goto{\alpha}} R'$ is derived by rule \textsc{\textbf{\small early-input}}.
      Then $R\mathbin=\textcolor{black}{M}x(y).P$, $\alpha\mathbin=\textcolor{black}{M}xz$ and
      $R'\mathbin=P\renb{z}{y}$. So
      $R\Ren{\sigma} \mathbin= \textcolor{black}{M\as{\sigma}}x\as{\sigma}(y\as{\sigma}).(P\Ren{\sigma}) \wedge
      \alpha\as{\sigma}\mathbin=\textcolor{black}{M\as{\sigma}}x\as{\sigma}z\as{\sigma}$.
      By rule \textsc{\textbf{\small early-input}}
      $R\Ren{\sigma} \goto{\alpha\as{\sigma}}P\Ren{\sigma}\renb{z\as{\sigma}}{y\as{\sigma}} \eqa P\renb{z}{y}\Ren{\sigma} = R'\Ren{\sigma}$.
      The last step uses that $\sigma$ is injective.
    \item Suppose $R\mathbin=A(\vec y)$.
      Let $A(\vec x) \mathbin{\stackrel{{\rm def}}{=}} P$. Say $\vec{x}\mathbin=(x_1,\dots,x_n)$.
      Then $P\renb{\vec y}{\vec x}\mathbin{\goto{\alpha}}R'\!$, so by induction
      $$P\renb{\vec y}{\vec x}\Ren{\sigma} \mathbin{\goto{\alpha\as{\sigma}}}\eqa R'\Ren{\sigma}.$$

      Moreover, $R\Ren{\sigma} \mathbin= {}^{{\sigma}}\!\!A(\vec y\as{\sigma})$ and
      $$P\Ren{\sigma}\renb{\vec y\as{\sigma}}{\vec x\as{\sigma}} \eqa P\renb{\vec y}{\vec x}\Ren{\sigma}.$$
      This step uses that $\sigma$ is injective.
      So by \lem{alpha norestriction} $P\Ren{\sigma}\renb{\vec y\as{\sigma}}{\vec x\as{\sigma}} \mathbin{\goto{\alpha\as{\sigma}}\eqa} R'\Ren{\sigma}\!$.
      Thus, by rule \textsc{\textbf{\small ide}}, \plat{$R\Ren{\sigma} \goto{\alpha\as{\sigma}}\eqa R'\Ren{\sigma}$}.
    \item The cases that $R \goto{\alpha} R'$ is derived by rule \textsc{\textbf{\small sum}},
      \textsc{\textbf{\small symb-match}}, \textsc{\textbf{\small par}} or \textsc{\textbf{\small e-s-com}} are trivial.
  \qed
  \end{itemize}
\end{proof}

\begin{lemma}\rm\label{lem:ruthless substitution reverse}
  Let $R$ be a \plat{$\piSWR$} process and $\sigma$ a bijective substitution.
  If \plat{$R\Ren{\sigma} \goto{\beta} U$}
  then \plat{$R \goto{\alpha} R'$} for some $\alpha$ and $R'$ with
  $\alpha\as{\sigma}=\beta$ and $R'\Ren\sigma \eqa U$.
  Moreover, the size of the derivation of \plat{$R \goto{\alpha} R'$} is the same as that of \plat{$R\Ren\sigma \goto{\beta} U$}.
\end{lemma}
\begin{proof}
  With induction on the size of the derivation of the transition $R\Ren{\sigma} \goto{\beta} U$.
  \begin{itemize}
    \item The cases that $R\Ren\sigma \goto{\beta} U$ is derived by rule \textsc{\textbf{\small tau}} or
      \textsc{\textbf{\small output}} are trivial.
    \item Suppose $R\Ren\sigma \goto{\beta} U$ is derived by \textsc{\textbf{\small early-input}}.
      Then $R=\textcolor{black}{M}x(y).P$ and
      $R{\sigma} = \textcolor{black}{M\as{\sigma}}x\as{\sigma}(y\as{\sigma}).(P\Ren{\sigma})$. So
      $\beta=\textcolor{black}{M\as{\sigma}}x\as{\sigma}v$ and
      $U=P\Ren{\sigma}\renb{v}{y\as{\sigma}}$.
      Since $\sigma$ is surjective, there is a $z$ with $z\as{\sigma}=v$.
      By \textsc{\textbf{\small early-input}} $R\goto{\alpha}R'$ with
      $\alpha=\textcolor{black}{M}xz$ and $R'=P\renb{z}{y}$.
      Now $\alpha\as{\sigma}=\beta$ and $R'\Ren\sigma\mathbin=P\renb{z}{y}\Ren\sigma \mathbin\eqa P\Ren{\sigma}\renb{z\as{\sigma}}{y\as{\sigma}}\mathbin= U$.
      The last step uses that $\sigma$ is injective.
    \item Suppose $R\Ren\sigma \goto{\beta} U$ is derived by \textsc{\textbf{\small ide}}.
      Then $R=A(\vec y)$ and $R\Ren\sigma\mathbin={}^{{\sigma}}\!\!A(\vec{y}\as{\sigma})$.
      Let \plat{$A(\vec x) \stackrel{{\rm def}}{=} P$}.
      Then \plat{${}^{{\sigma}}\!\!A(\vec{x}\as{\sigma}) \stackrel{{\rm def}}{=} P\Ren{\sigma}$}.
      So $P\Ren{\sigma}\renb{\vec{y}\as{\sigma}}{\vec{x}\as{\sigma}} \goto{\beta} U$.
      Since $P\Ren{\sigma}\renb{\vec y\as{\sigma}}{\vec x\as{\sigma}} \eqa P\renb{\vec y}{\vec x}\Ren{\sigma}$,
      using that $\sigma$ is injective,
      \lem{alpha norestriction} yields $P\renbt{y}{x}\Ren\sigma\goto{\beta}\eqa U$.
      So by induction \plat{$P\renbt{y}{x} \goto{\alpha} R'$} for some $\alpha$ and $R'$ with
      $\alpha\as{\sigma}=\beta$ and $R'\Ren\sigma \eqa U$.
      By rule \textsc{\textbf{\small ide}} $R\goto{\alpha}R'$.
    \item The cases that $R\Ren\sigma \goto{\beta} U$ is derived by rule \textsc{\textbf{\small sum}},
      \textsc{\textbf{\small symb-match}} or \textsc{\textbf{\small par}} are trivial.
      The case for \textsc{\textbf{\small e-s-com}} is also trivial, when using injectivity of
      $\sigma$ to conclude that there is a unique $z$ with $z\as{\sigma}=y$. This solves the only
      complication in the corresponding case of the proof of \lem{substitution reverse}.
  \qed
  \end{itemize}
\end{proof}

For $P$ a $\piSRWR$ process, let $\widehat P$ be the $\piSWR$ process obtained from $P$ by recursively replacing
each subterm $Q[\sigma]$ by $Q\sigma$ if $\dom(\sigma)$ is finite and $\dom(\sigma)\cap\pN=\emptyset$, and by
$Q\Ren\sigma$ if $\dom(\sigma)$ is infinite and $\sigma$ bijective, and each agent identifier $A$ by $A\hat{~}$.\vspace{1pt}
Here $A\hat{~}$ is a fresh agent identifier with defining equation
\plat{$A\hat{~}(\vec x) \stackrel{{\rm def}}{=} \widehat P$} when
\plat{$A(\vec x) \stackrel{{\rm def}}{=} P$} was the defining equation of $A$.
This definition extends the mapping $\widehat \cdot$ defined in \sect{relabelling} from $\piRWR$ to $\piSRWR$.

\begin{lemma}\label{lem:relabelling forth ruthless}
If $R \goto{\alpha} R'$ then \plat{$\widehat R \goto{\alpha}\eqa \widehat{R'}$}.
\end{lemma}
\begin{proof}
  The proof is the same as the proof of \lem{relabelling forth},
  except that the case of \textsc{\textbf{\small relabelling}} $R=[\sigma]$ involves a further case distinction,
  depending on whether $\dom(\sigma)$ is finite and $\dom(\sigma)\cap\pN=\emptyset$, or $\dom(\sigma)$ is infinite and  $\sigma$
  is bijective. In the first case \lem{substitution} is called, and in the second case \lem{ruthless substitution}.
\end{proof}

\begin{definition}
Let $\BN$ be the smallest function from $\piSRWR$ processes to sets of names, such that
\begin{itemize}
\item $y \in \BN(Mx(y).P)$,
\item $x_1,\dots,x_n\in\BN(A(\vec{y})$ if \plat{$A(x_1,\dots,x_n) \stackrel{{\rm def}}{=} P$},
\item $y \in \BN(P[\sigma])$ if $y\in\dom(\sigma)$ and $\dom(\sigma)$ is finite,
\item $\BN(P) \subseteq \BN(M\tau.P)$,
\item $\BN(P) \subseteq \BN(M\bar x y.P)$ and $\BN(P) \subseteq \BN(Mx(y).P)$,
\item $\BN(P)\cup\BN(Q) \subseteq \BN(P+Q)$,
\item $\BN(P)\cup\BN(Q) \subseteq \BN(P|Q)$,
\item $\BN(P) \subseteq \BN(\Match{x}{y}P)$,
\item $\BN(P) \subseteq \BN(A(\vec{y}))$ if \plat{$A(x_1,\dots,x_n) \stackrel{{\rm def}}{=} P$},
\item $\BN(P) \subseteq \BN(P[\sigma])$ if $\dom(\sigma)$ is finite.
\item $\BN(P[\sigma]) = \{\sigma(y)\mid y\mathbin\in\BN(P)\}$ if $\dom(\sigma)$ is infinite and $\sigma$ bijective.
\end{itemize}
\end{definition}
Note that restricted to $\piRWR$ processes, so that the last clause above does not apply,
this definition agrees with \df{BN}.
Using this definition of $\BN$,
\df{clash-free RWR} of clash-freedom extends to
the general case with the stipulation that Clause 2 on
$Q[\sigma]$ only applies when $\sigma$ is a finite substitution with
$\dom(\sigma)\cap\pN=\emptyset$; the notion does not restrict the use of relabelling operators
$\sigma$ with $\dom(\sigma)$ infinite and $\sigma$ bijective.

I now show that \lem{clash-free premises RWR} (and the definition of $\leftarrowtail$) extends to $\piSRWR$.
\begin{lemma}\label{lem:clash-free premises SRWR}
If $P \leftarrowtail Q$ and $P$ is a clash-free $\piSRWR$ process, then so is $Q$.
\end{lemma}
\begin{proof}
Let $P[\sigma]$ be a clash-free $\piSRWR$ process with $\dom(\sigma)$ infinite and $\sigma$ bijective.
By definition of $\piSRWR$, $\sigma$ satisfies $y\in\pN \Rightarrow \sigma(y)\in\pN$.
This implies that $\pN\cap\BN(P)=\emptyset$. The other clauses of \df{clash-free RWR}
are satisfied trivially, so $P$ is clash-free.

All other cases go as in the proof of \lem{clash-free premises RWR}.
\end{proof}

Note that \lem{fn} holds also for $\piSRWR$ processes.

\begin{lemma}\label{lem:relabelling back ruthless}
If $R$ is clash-free and $\widehat R \goto{\alpha} U$ with $\ia(\alpha)\cap\BN(R)=\emptyset$ then \plat{$R \goto{\alpha} R'$}
for some $R'$ with $\widehat R' \eqa U$.
\end{lemma}
\begin{proof}
  By induction on the size of the derivation of \plat{$\widehat R \goto{\alpha} U$}, with a nested
  induction on the number of topmost renaming operators in $R$.
\begin{itemize}
\item The cases that $R$ is not of the form $P[\sigma]$, or that
  $R=P[\sigma]$ with $\dom(\sigma)$ finite and $\dom(\sigma)\cap\pN=\emptyset$,
  go exactly as in the proof of \lem{relabelling back}.
\item Now suppose $R=P[\sigma]$ with $\dom(\sigma)$ infinite and $\sigma$ bijective. Then $\widehat R = \widehat P \Ren\sigma$.
  By \lem{ruthless substitution reverse} $\widehat P \goto{\beta} V$ for some $\beta$ and $V$ with
  $\beta\as{\sigma}\mathbin=\alpha$ and $V\Ren\sigma \mathbin\eqa U$.
  Moreover, the size of the derivation of \plat{$\widehat P \goto{\beta}V$} is the same as that of
  \plat{$\widehat P\Ren\sigma \goto{\alpha} U$}.
  By \lem{clash-free premises SRWR} $P$ is clash-free and trivially $\ia(\beta)\cap\BN(P)=\emptyset$.
  So by induction $P\goto{\beta} P'$ for some $P'$ with $\widehat{P'}\eqa V$.
  By rule \textsc{\textbf{\small relabelling}} $R\goto{\alpha}P'[\sigma]$.
  Furthermore one has $\widehat{P'[\sigma]}\mathbin=\widehat{P'}\Ren\sigma \mathbin\eqa V\Ren\sigma \mathbin\eqa U$.
\qed
\end{itemize}
\end{proof}

The following result is in analogy with \thm{clash-free}.
\begin{lemma}\label{lem:clash-free}
  Each process $R:=\fT_{{\rm cf}\nu}(P)$ is clash-free and $\BN(R)\subseteq\sN\cup\dN$.
\end{lemma}

\begin{proof}
In $\fT_{{\rm cf}\nu}(x(y).P) \mathbin= x(s).\big(\fT_{{\rm cf}\nu}(P)[s_y]\big)$ the relabelling $[s_y]$
injectively relabels all symbolic names ${}^\varsigma\! s$ in $\BN(\fT_{{\rm cf}\nu}(P))$ by adding a
tag $s$ to their modifiers. This frees up the name $s$. The translation of $Mx(y)P$ takes advantage of this
by changing the bound name $y$ into $s\in\sN$. This ensures that Clause 3 of \df{clash-free RWR} is met.

Likewise, in $\fT_{{\rm cf}\nu}(A(\vec y)) = A_{{\rm cf}\nu}(\vec y)$,\vspace{1pt} where $A_{{\rm cf}\nu}$ is an
agent identifier with defining equation \plat{$A_{{\rm cf}\nu}(\vec x\as{d_{\vec{x}}}) \stackrel{{\rm def}}{=} \fT_{{\rm cf}\nu}(P)[d_{\vec{x}}]$} when
\plat{$A(\vec x) \stackrel{{\rm def}}{=} P$} was the defining equation of $A$,
the $[d_{\vec{x}}]$ injectively renames all declared names ${}^\varsigma\! d_i$ for $1\leq i\leq n$ in $\BN(\fT_{{\rm cf}\nu}(P))$ by adding a tag
$d$ to their modifiers ${}^\varsigma$. This frees up the names $d_1,\dots,d_n$.
The translation of defining equations takes advantage of that by renaming all declared names $x_i$
into $d_i \in\dN$. This way Clause 1 of \df{clash-free RWR} is met.

Relabelling operators $[\sigma]$ with $\dom(\sigma)$ finite do not occur in processes $\fT_{{\rm cf}\nu}(P)$.
Hence Clause 2 of \df{clash-free RWR} is trivially met.
Furthermore, since all relabelling operators employed in the translation never take a name from
$\sN$ or $\dN$ outside of $\sN$ or $\dN$, respectively, one has $\BN(R)\subseteq\sN\cup\dN$.

In $\fT_{{\rm cf}\nu}((\nu y)P) = \fT_{{\rm cf}\nu}(P)[p_y]$ the bound name $y$ is turned into a
free name $p\in\pN$, and all relabelling operators keep such names within $\pN$.
Free names of a $\pi(\N)$ process $P$ are not renamed in $\fT_{{\rm cf}\nu}(P)$. 
As a consequence, one has $\fn(\fT_{{\rm cf}\nu}(P))\subseteq\N\cup\pN$.
Therefore, Clause 4 of \df{clash-free RWR} is met as well.
\end{proof}

\begin{lemma}\label{lem:clash-free transitions 3}
If $R$ in $\piSRWR$ is clash-free, $\ia(\alpha)\cap\BN(R)\linebreak[2]\mathbin=\emptyset$
and $R\goto{\alpha} R'$, then $R'$ is clash-free and $\BN(R')\subseteq\BN(R)$.
\end{lemma}

\begin{proof}
The proof is the same as the one of \lem{clash-free transitions 2}, except that there is now one
extra case to consider: 
Suppose \plat{$R\goto{\alpha} R'$} is derived by \textsc{\textbf{\small relabelling}} and
$R=P[\sigma]$ with $\dom(\sigma)$ infinite and $\sigma$ bijective. Then $P\goto\beta P'$, $\beta\as{\sigma}=\alpha$ and $R'=P'[\sigma]$.
By \lem{clash-free premises SRWR}, $P$ is clash free.
Since  $\ia(\alpha)\cap\BN(R)=\emptyset$ and $\sigma$ is bijective, $\ia(\beta)\cap\BN(P)=\emptyset$.
      So by induction $P'$ is clash-free and $\BN(P')\subseteq\BN(P)$.
      Hence $\BN(R')\subseteq\BN(R)$ and $R'$ is clash-free.
\end{proof}

\begin{theorem}\label{thm:relabelling}
$\fT_{{\rm cf}\nu}^{r}(P) \mathbin{\sbb} \fT_{{\rm cf}\nu}(P)$ for any $\pi_{\scriptstyle \!E\!S}(\N)$ process $P\hspace{-1pt}$.%
\end{theorem}

\begin{proof}
  Let \[{\R}:=\left\{(R,U),  (U,R) \left|\, \begin{array}{@{}l@{}}R \mbox{~in~} \piSRWR \mbox{~is clash-free}\\
  \mbox{and~} U\eqa \widehat R \mbox{~in~} \piSWR \end{array}\right\}\right..\]
\lem{clash-free} shows that $\fT_{{\rm cf}\nu}(P)$ is clash-free for each $\pi_{ES}(\N)$ process $P$.
Since $\widehat{\fT_{{\rm cf}\nu}(P)} = \fT_{{\rm cf}\nu}^{r}(P)$ for any $\pi_{\scriptstyle \!S\!E}(\N)$ process $P$,
it suffices to show that $\R$ is a strong barbed bisimulation.
This follows exactly as in the proof of \thm{step6}, but using Lemmas~\ref{lem:relabelling forth ruthless},%
~\ref{lem:relabelling back ruthless} and~\ref{lem:clash-free transitions 3} instead of~\ref{lem:relabelling forth},%
~\ref{lem:relabelling back} and~\ref{lem:clash-free transitions 2}.
\end{proof}

\subsection{The last steps}\label{sec:steps7}

The translation from $\piSRWR$ to {\CCST}, depicted as the last step in \fig{Translation}, actually
consists of 2 small steps. The first of those consists in moving the relabelling operator
$[\renbt{y}{x}]$ that occurs in rule \textsc{\textbf{\small ide}}, as well as the
relabelling $[d_{\vec{x}}]$ that was added to defining equations of agent identifiers, forwards.
The target language is the variant $\piRR$ of $\piSRWR$ in which agent identifiers may be called
only with their own declared names as parameters, i.e.\ such that $\vec{y}=\vec{x}$ in rule \textsc{\textbf{\small ide}}.
As a result of this, the relabelling or substitution operator in this rule can be dropped.
Let $\fT_{{\rm cf}\nu}'(P)$ be the translation from $\pi(\N)$ to $\piRR$ defined inductively as follows:
\[\begin{array}{@{}l@{~:=~}ll@{}}
\fT_{{\rm cf}\nu}'(\nil) & \nil \\
\fT_{{\rm cf}\nu}'(\tau.P) & \tau.\fT_{{\rm cf}\nu}'(P) \\
\fT_{{\rm cf}\nu}'(\bar xy.P) & \bar xy.\fT_{{\rm cf}\nu}'(P) \\
\fT_{{\rm cf}\nu}'(x(y).P) & x(s).\big(\fT_{{\rm cf}\nu}'(P)[s_y]\big)\\
\fT_{{\rm cf}\nu}'((\nu y)P) & \fT_{{\rm cf}\nu}'(P)[p_y]\\
\fT_{{\rm cf}\nu}'(\Match{x}{y}P) & \Match{x}{y}\fT_{{\rm cf}\nu}'(P) \\
\fT_{{\rm cf}\nu}'(P\mid Q) & \fT_{{\rm cf}\nu}'(P)[\ell]\mid\fT_{{\rm cf}\nu}'(Q)[r] \\
\fT_{{\rm cf}\nu}'(P+Q) & \fT_{{\rm cf}\nu}'(P)+\fT_{{\rm cf}\nu}'(Q) \\
\fT_{{\rm cf}\nu}'(A(\vec y)) & A_{{\rm cf}\nu}'(\vec x)[d_{\vec{x}}][\renb{\vec{y}}{\vec{x}\as{d_{\vec{x}}}}] \\
\end{array}\]
where $A_{{\rm cf}\nu}'$ is a fresh agent identifier,\vspace{1pt} with defining
equation \plat{$A_{{\rm cf}\nu}'(\vec x) \stackrel{{\rm def}}{=} \fT_{{\rm cf}\nu}'(P)$} when
\plat{$A(\vec x) \stackrel{{\rm def}}{=} P$} was the defining equation of $A$.
This translation differs from $\fT_{{\rm cf}\nu}$ only in the case of agent identifiers.
There $\fT_{{\rm cf}\nu}(A(\vec y)) := A_{{\rm cf}\nu}(\vec{y})$ where
$A_{{\rm cf}\nu}$ is a fresh agent identifier,\vspace{1pt} with defining
equation \plat{$A_{{\rm cf}\nu}(\vec x\as{d_{\vec{x}}}) \stackrel{{\rm def}}{=} \fT_{{\rm cf}\nu}(P)[d_{\vec{x}}]$} when
\plat{$A(\vec x) \stackrel{{\rm def}}{=} P$} was the defining equation of $A$.

\begin{theorem}\label{thm:step7a}
$\fT_{{\rm cf}\nu}'(P) \mathbin{\bis{}} \fT_{{\rm cf}\nu}(P)$ for any $\pi_{\scriptstyle \!E\!S}(\N)$ process $P\hspace{-1pt}$.%
\end{theorem}

\begin{proof}
  Let $\R$ be the symmetric closure of the smallest relation between $\piRR$  and $\piSRWR$
  processes such that $\nil \R \nil$,
  $A_{{\rm cf}\nu}'(\vec x)[d_{\vec{x}}][\renb{\vec{y}}{\vec{x}\as{d_{\vec{x}}}} \R A_{{\rm cf}\nu}(\vec y)$,
  and $P \R V \wedge Q \R W$ implies
$\bar xy.P \mathbin\R \bar xy.V$,
$x(y).P \mathbin\R x(y).V$,
$\tau.P \mathbin\R \tau.V$,
$\Match{x}{y}P \mathbin\R \Match{x}{y}V$,
$P|Q \mathbin\R V|W$,
$P{+}Q \mathbin\R V{+}W$ and
$P[\sigma] \mathbin\R V[\sigma]$.
Then $\fT_{{\rm cf}\nu}'(P) \R \fT_{{\rm cf}\nu}(P)$ for any $\pi_{\scriptstyle ES}(\N)$ process $P$,
so it suffices to show that $\R$ is a strong bisimulation, i.e.\ that
\begin{center}  
if $R\R U$ and \plat{$R\goto{\alpha} R'$} then $\exists U'.~U\goto{\alpha}U' \wedge R'\R U'$.
\end{center}
The proof is by induction on the derivation of \plat{$R\goto{\alpha} R'$}, while making a case
distinction based on the construction of $R\R U$.
\begin{itemize}
  \item Suppose $R = A_{{\rm cf}\nu}'(\vec x)[d_{\vec{x}}][\renb{\vec{y}}{\vec{x}\as{d_{\vec{x}}}}$ and $U= A_{{\rm cf}\nu}(\vec y)$.
    Then $A_{{\rm cf}\nu}'(\vec x) \mathbin{\goto{\beta}} Q'$ for some $\beta$ and $Q'$ such that
    $\beta\as{d_{\vec{x}}}[\renb{\vec{y}}{\vec{x}\as{d_{\vec{x}}}}\mathbin=\alpha$
    and $Q'[d_{\vec{x}}][\renb{\vec{y}}{\vec{x}\as{d_{\vec{x}}}}\mathbin=R'$.\vspace{2pt}
    Let \plat{$A(\vec x) \stackrel{{\rm def}}{=} P$} be the defining equation of $A$.
    Then, by \textsc{\textbf{\small ide}}, $\fT_{{\rm cf}\nu}'(P)\goto\beta Q'$ and thus
    $\fT_{{\rm cf}\nu}'(P)[d_{\vec{x}}][\renb{\vec{y}}{\vec{x}\as{d_{\vec{x}}}}\goto\alpha R'$.

    By induction, $\fT_{{\rm cf}\nu}(P)[d_{\vec{x}}][\renb{\vec{y}}{\vec{x}\as{d_{\vec{x}}}}\goto\alpha U'$ for some $U'$ with $R'\R U'$.
    Thus, again by \textsc{\textbf{\small ide}}, $A_{{\rm cf}\nu}(\vec y)\goto\alpha U'$.

    The symmetric case goes likewise.
  \item Suppose $R \R U$ holds because $R=P|Q$ and $U=V|W$ with $P \R V$ and $Q \R W$.
    \begin{itemize}
      \item First suppose that \plat{$R\goto{\alpha} R'$} is derived by \textsc{\textbf{\small par}}.
        Then \plat{$P\goto{\alpha} P'$} has a smaller derivation and $R'=P'|Q$.
        By induction $V \goto{\alpha} V'$ for some $V'$ with $P' \R V'$.
        So by \textsc{\textbf{\small par}} $U \goto{\alpha} V'|W$ and $V'|W \R P'|Q = R'$.
      \item The case that \plat{$R\goto{\alpha} R'$} is derived by \textsc{\textbf{\small e-s-com}}
        (or by the symmetric form of \textsc{\textbf{\small par}}) is equally trivial.    
    \end{itemize}
  \item All other cases are also trivial.
\qed
\end{itemize}
\end{proof}

The language $\piRR$ can almost be recognised as an instance of {\CCP}.
As parameters of {\CCP} I take $\A$ to be the set of all actions $M\tau$, $M\bar xy$ and $Mxy$ with names from $\M$.
The communication function $\gamma\!:\M\rightharpoonup\M$ is given by
$\gamma(M\bar xy,Nvy)=\match{x}{v}MN\tau$, and its commutative variant.\linebreak[4]
Now the parallel composition of \plat{$\piRR$} turns out to be the same as for this instance of {\CCP}.
Likewise, the silent and output prefixes are instances of {\CCP} prefixing.
Moreover, when simply writing $A$ for $A(\vec{x})$, the agent identifiers of $\piRR$ are no
different from {\CCP} agent identifiers.\footnote{Here I assume that all sets $\K_n$ are disjoint,
  i.e., the same $\pi$-calculus agent identifier does occur with multiple arities. When this
  assumption is not met, an arity-index at the {\CCP} identifier is needed.}
However, the input prefix of \plat{$\piRR$} does not occur in {\CCP}.
Yet, one can identify $Mx(y).P$ with $\sum_{z\in\M} Mxz.(P[\renb{z}{y}])$,
for both processes have the very same outgoing transitions.
The $\piRR$ matching operator is
no different from the triggering operator of {\sc Meije} or {\CCST} (see \sect{triggering}): both rename only
the first actions their argument process can perform, namely by adding a single match $\match{x}{y}$
in front of each of them---this match is suppressed when $x{=}y$.

This yields to the following translation from $\piRR$ to {\CCST}:
\[\begin{array}{@{}l@{~:=~}ll@{}}
\fT_\gamma(\nil) & \nil \\
\fT_\gamma(\textcolor{DarkBlue}{M}\tau.P) & \textcolor{DarkBlue}{M}\tau.\fT_\gamma(P) \\
\fT_\gamma(\textcolor{DarkBlue}{M}\bar xy.P) & \textcolor{DarkBlue}{M}\bar xy.\fT_\gamma(P) \\
\fT_\gamma(\textcolor{DarkBlue}{M}x(y).P) & \sum_{z\in\M} \textcolor{DarkBlue}{M}xz.\big(\fT_\gamma(P)[\renb{z}{y}]\big) \\
\color{purple}\fT_\gamma(\Match{x}{y}P) & \color{purple}\match{x}{y}{\Rightarrow}\fT_\gamma(P)\\
\fT_\gamma(P\mid Q) & \fT_\gamma(P)\|\fT_\gamma(Q) \\
\fT_\gamma(P+Q) & \fT_\gamma(P)+\fT_\gamma(Q) \\
\fT_\gamma (A(\vec{x})) & A \\
\end{array}\]
where the {\CCP} agent identifier $A$ has the defining equation $A=\fT_\gamma(P)$ when
\plat{$A(\vec{x}) \stackrel{{\rm def}}{=} P$} was the defining equation of the $\piRR$ agent identifier $A$.

Here the use of the triggering operator can be avoided by restricting attention to the
$\pi$-calculus with implicit matching. For that language the clause for $\fT_\gamma(\Match{x}{y}P)$ can
be dropped, at the expense of the addition of the blue $\textcolor{DarkBlue}{M}$s above, which are absent when dealing with the
full $\pi$-calculus.

\begin{theorem}
$\fT_\gamma(P)\bis{} P$ for each $\piRR$ process $P$.
\end{theorem}
\begin{proof}
Trivial.
\end{proof}

\subsection{A small simplification}\label{sec:simplify}

Putting all steps of my translation from $\pi_L(\N)$ to {\CCST} together, I obtain
\[\begin{array}{@{}l@{~:=~}ll@{}}
\fT_1(\nil) & \nil \\
\fT_1(\textcolor{DarkBlue}{M}\tau.P) & \textcolor{DarkBlue}{M}\tau.\fT_1(P) \\
\fT_1(\textcolor{DarkBlue}{M}\bar xy.P) & \textcolor{DarkBlue}{M}\bar xy.\fT_1(P) \\
\fT_1(\textcolor{DarkBlue}{M}x(y).P) & \sum_{z\in\M}\textcolor{DarkBlue}{M}xz.\big(\fT_1(P)[s_y][\renb{z}{s}]\big) \\
\fT_1((\nu y)P) & \fT_1(P)[p_y]\\
\color{purple}\fT_1(\Match{x}{y}P) & \color{purple}\Match{x}{y}{\Rightarrow}\fT_1(P) \\
\fT_1(P\mid Q) & \fT_1(P)[\ell]\mathbin{\|}\fT_1(Q)[r] \\
\fT_1(P+Q) & \fT_1(P)+\fT_1(Q) \\
\fT_1(A(\vec y)) & A[d_{\vec{x}}][\renb{\vec{y}}{\vec{x}\as{d_{\vec{x}}}}]  \\
\end{array}\]
where the {\CCP} agent identifier $A$ has the defining equation $A=\fT_1(P)$ when
\plat{$A(\vec{x}) \stackrel{{\rm def}}{=} P$} was the defining equation of the $\pi_L(\N)$ agent identifier $A$.

Now I discuss two simplification that can be made to this translation.
The first is obtained by replacing the nested relabellings $[s_y][\renb{z}{s}]$
and $[d_{\vec{x}}][\renb{\vec{y}}{\vec{x}\as{d_{\vec{x}}}}]$ in the clauses for input prefix and agent identifiers
by single relabellings $[\renb{z}{s} \circ s_y]$ and $[\renb{\vec{y}}{\vec{x}\as{d_{\vec{x}}}}\circ d_{\vec{x}}]$.
This surely preserves strong bisimilarity.
After this change, all {\CCST} relabelling operators $[\sigma]$ that are introduced by the
translation have the property that $x\in\sN\cup\dN \Leftrightarrow x\as{\sigma}\in \sN\cup\dN$.

Now let $\fT_{\sN\dN}$ be the translation from 
{\CCST} to {\CCST} that replaces each sum
$\sum_{z\in\M}$ into a sum $\sum_{z\in\N\cup\pN\cup\aN}$ and each relabelling operator $[\sigma]$ by
$[\sigma\upharpoonright\N\cup\pN\cup\aN]$.
It is trivial to show that each process of the form $\fT(P)$ with $P\in\pi_L(\N)$
is strongly barbed bisimilar with $\fT_{\sN\dN}(\fT(P))$. The idea is that names from $\sN\cup\dN$
are never introduced and thus can just as well be dropped from the language.

The resulting simplifications of the renamings $[\renb{z}{s} \circ s_y]$ and
$[\renb{\vec{y}}{\vec{x}\as{d_{\vec{x}}}}\circ d_{\vec{x}}]$ can now be written as $[\renb{z}{y}^\aN]$ and
$[\renbt{y}{x}^\aN]$, or as $\rensq{z}{y}$ and $\rensq{\vec y}{\vec x}$ for short.
Here $[\renb{z}{y}^\aN]$ was defined in \sect{the encoding}. It relabels $y\in\N$ into $z$ and
bijectively maps $\aN$ to $\aN\cup\{y\}$, leaving all other names unchanged. Likewise,
$[\renbt{y}{x}^\aN]$ relabels $x_i$ into $y_i$ for $i=1,\dots,n$ and bijectively maps $\aN$ to $\aN\cup\{x_1,\dots,x_n\}$.
This yields the translation $\fT$ presented in \sect{the encoding}.

\newcommand{\eqaU}{\equiv^\aN}                                         
\newcommand{\piLate}[1][\N]{\hyperlink{Late}{\pi_{L}(#1)}}             
\newcommand{\piES}[1][\N]{\hyperlink{ES}{\pi_{ES}(#1)}}                
\newcommand{\piSur}[1][\N]{\hyperlink{piSur}{\pi_{ES}^{\aN}(#1)}}       
\newcommand{\piRI}[1][\N]{\hyperlink{piRI}{\pi_{ES}^{\,\rho\bullet}(#1)}} 
\newcommand{\piR}[1][\N]{\hyperlink{piR}{\pi_{ES}^{\,\rho}(#1)}}        
\newcommand{\piRp}{\hyperlink{piRp}{\pi_{ES}^{\,\rho}(\N,\pN)}}        
\newcommand{\piRa}{\hyperlink{piRa}{\pi_{ES}^{\,\rho\hspace{-.5pt}\not\hspace{.5pt}\alpha}(\N,\pN)}}  
\newcommand{\pir}{\hyperlink{pir}{\pi_{ES}^{\dagger}(\N,\pN)}}       

\begin{figure*}[t]
  \input{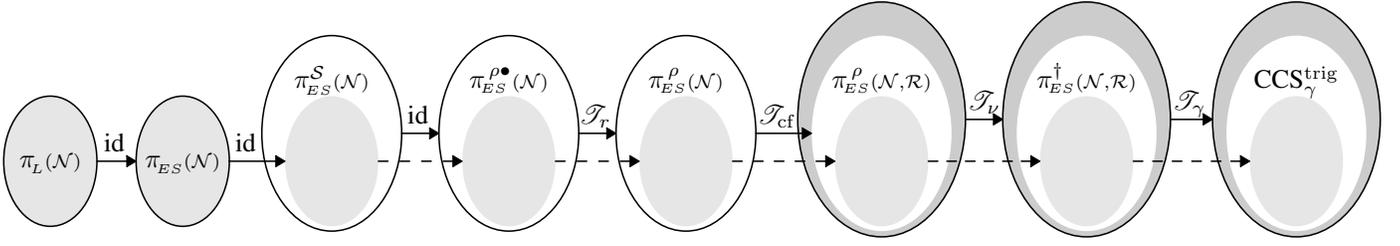}
  \centerline{\raisebox{1ex}{\box\graph}}
  \caption{Translation from the $\pi$-calculus to CCS$_\gamma^{\rm trig}$}\label{fig:Translation2}
\end{figure*}

\newpage
\pdfbookmark[1]{Appendix 2}{contents}
\setcounter{subsection}{0}
\renewcommand*{\theHsubsection}{App2.\the\value{subsection}}
\begin{center}
\hypertarget{App2}{\sc Appendix 2}
\end{center}

As indicated in \fig{Translation2}, my translation from $\pi$ to {\CCST} proceeds in seven steps,
or actually ten when decomposing steps 1, 3 and 6 into two smaller steps each.
\sect{the encoding} presents the translation in one step: the composition of these constituent translations.
Its decomposition in \fig{Translation2} describes how this appendix proves its validity.

Each of the eight languages in \fig{Translation2} comprises syntax, determining what are the valid
expressions or processes, a structural operational semantics generating an LTS, and a BTS extracted
from the LTS in the way described in \sect{barbed}.

My translation starts from the $\pi$-calculus $\piLate$ with the late operational semantics, as defined in \cite{MPWpi2}.
The argument $\N$ denotes the set of names taken as a parameter of the $\pi$-calculus. 
The first step is the identity mapping to $\piES$, the calculus with the same syntax but the early
symbolic operational semantics. The validity of this translation step is the statement that each
$\piLate$-expression $P$ is strongly barbed bisimilar with the same expression $P$, but now seen as a
state in the LTS generated by the early symbolic operational semantics. This has been concluded
already in \sect{pi-semantics}. This first translation step can be decomposed into two smaller steps
by taking either the $\pi$-calculus with the early semantics or the one with the late symbolic
semantics as an halfway point between $\piLate$ and $\piES$.

The calculus $\piSur$ is a variant of $\piES$ in which the substitutions
that occur in the \hyperlink{ES}{early symbolic operational rules \textsc{\textbf{\small early-input}} and
\textsc{\textbf{\small ide}}} are changed into surjective substitutions.  This is achieved by
extending the set of names from $\N$ to $\zN := \N \uplus \aN$, for a countable set of \emph{spare names} $\aN$.
In order to preserve the integrity of the calculus, this change forces new
definitions of $\alpha$-conversion, the free names of a process, and the application of a substitution
to a process. These concepts differ from the old ones only when spare names are involved.
The $\piSur$-processes that employ names from $\N$ only form a subcalculus of $\piSur$
that behaves just like $\piES$. This yields the second translation step.

The calculus $\piRI$ is a variant of $\piSur$ enriched with a CCS-style relabelling
operator $[\sigma]$ for each surjective substitution $\sigma$.
Moreover, all substitutions that are used in the operational semantics are replaced by relabelling operators.
\sect{piRI} documents that the identity mapping is a valid translation from $\piSur$ to $\piRI$.
This works only when using surjective substitutions, and that is the reason $\piSur$
appears in the previous translation step.\pagebreak[2]\vspace*{1.8ex}

In $\piRI$ there is no room for \hyperlink{ES}{rule \textsc{\textbf{\small alpha}}}.
Hence, before translating $\pi_{ES}^\aN(\N)$ into $\piRI$ I show that on $\pi_{ES}^\aN(\N)$
one can equally well use a version of the early symbolic operational semantics in which
all $\alpha$-conversion has been moved into stronger versions of \hyperlink{ES}{rules 
\textsc{\textbf{\small res}} and \textsc{\textbf{\small symb-open}}}.
This can be seen as an intermediate translation step.

The calculus $\piR$ is the variant of $\piRI$ in which agent identifiers may be called
only with their own declared names as parameters, i.e.\ such that $\vec{y}\mathbin=\vec{x}$
in \hyperlink{ES}{rule \textsc{\textbf{\small ide}}}.
In \sect{piR} I establish the validity of a translation from $\piRI$ to $\piR$ that turns each
call $A(\vec{y})$ to an agent identifier with defining equation \plat{$A(\vec{x})\stackrel{\rm def}{=} P$} 
into the expression $A(\vec{x})[\rename{\vec{y}}{\vec{x}}]$.

In $\piRp$ the set of names employed is further extended with a collection $\pN$ of \emph{private names}.
I write $\N,\pN$ instead of $\N\uplus\pN$ to indicate, amongst others, that only the names in
$\zN := \N \uplus \aN$---the \emph{public} ones---generate barbs.
The interior white ellipse denotes the class of \emph{clash-free} processes within $\piRp$, defined
in \sect{piRp}. That section also presents the translation $\fT_{\rm cf}$ that turns each
$\piR$ process into a clash-free $\piRp$ process.

Step 6, recorded in Sections~\ref{sec:piRa}--\ref{sec:pir}, eliminates the restriction operator from the language by
translating (sub)expressions $(\nu z)P$ into $P$. This step does not preserve $\sbb$ for the language as a
whole, but, since there are no $\pN$-barbs, it does so on the sublanguage that arises as
the image of the previous translation steps, namely on the clash-free processes in $\piRp$.
This step is delivered in two substeps, the first of which eliminates $\alpha$-conversion from the
operational semantics, and the second the restriction operator.

After this step, the resulting language $\pir$ is in De Simone format.
Moreover, as shown in \sect{ccstrig}, it is easily translated into {\CCST}.

A different proof of the same result, \thm{piCCS}, saying that $\fT$ is a valid translation from
$\piLate$ to {\CCST}, up to strong barbed bisimilarity, appears in \hyperlink{appendix}{Appendix~1}.
That proof avoids step 2 from the current approach, which is the most laborious one, as well as step 3a.
The remaining steps are delivered in the order $1a{-}1b{-}5{-}6a{-}6b{-}3b{-}4{-}7$.
However, that proof appears less direct: its third step, making processes clash-free, introduces two
auxiliary sets of names $\sN$ and $\dN$, as well as ``ruthless substitutions'', with an infinity of
replicated agent identifiers. The ruthless substitutions are eliminated again after the sixth step,
and the names from $\sN$ and $\dN$ at the very end. The proof in this appendix avoids those detours.

\subsection{Substitution}\label{sec:substitutions}

This section shows to what extend substitution preserves and reflects the behaviour of
$\pi$-calculus processes. As behaviour I will use the labelled translation relation generated by the
early symbolic operational semantics.

In the next section, I will need that the results of this section remain valid for an adapted
version of the $\pi$-calculus, in which the meaning of the construct $x(y).P$ will be changed.
Technically, this will manifest itself by
\begin{itemize}
\item the use of surjective substitutions $\renb{z}{y}^\aN$ and $\renbt{y}{x}^\aN$ instead of $\renb{z}{y}$
  and $\renbt{y}{x}$ in rules \hyperlink{ES}{\textsc{\textbf{\small early-input}} and \textsc{\textbf{\small ide}}},
\item the use of a relation $\eqaU$ instead of $\eqa$ in rule
   \hyperlink{ES}{\textsc{\textbf{\small alpha}}},
\item the use of a function $\Fn$ instead of $\fn$ in rules  \hyperlink{ES}{\textsc{\textbf{\small par}} and
  \textsc{\textbf{\small e-s-close}}},
\item a different definition of \hyperlink{substitution}{$(x(y).P)\sigma$}.
\end{itemize}
The function  $\Fn$ will be defined inductively by
\[\begin{array}{@{}l@{~:=~}l@{}}
\Fn(\nil) & \emptyset \\
\Fn(\textcolor{DarkBlue}{M}\tau.P) & \n(M) \cup \Fn(P) \\
\Fn(\textcolor{DarkBlue}{M}\bar xy.P) & \n(M) \cup \{x,y\} \cup \Fn(P) \\
\Fn(\textcolor{DarkBlue}{M}x(y).P) & \n(M) \cup \{x\} \cup \Fn(P)^-_y \\
\Fn((\nu y)P) & \Fn(P) {\setminus} \{y\} \\
\Fn(\Match{x}{y}P) & \{x,y\} \cup \Fn(P)\\
\Fn(P|Q) & \Fn(P) \cup \Fn(Q) \\
\Fn(P+Q) & \Fn(P) \cup \Fn(Q) \\
\Fn(A(\vec{y})) & \{y_1,\dots,y_n\}. \\
\end{array}\]
where $\Fn(P)^-_y$ is yet to be defined.
So $\Fn$ differs from $\fn$ only when applied to processes $x(y).P$.

I will use the following definition of substitution:
\[\begin{array}{r@{~:=~}l}
\nil\sigma & \nil\\
(\textcolor{DarkBlue}{M}\tau.P)\sigma & \textcolor{DarkBlue}{M\as{\sigma}}\tau.(P\sigma) \\
(\textcolor{DarkBlue}{M}\bar xy.P)\sigma & \textcolor{DarkBlue}{M\as{\sigma}}\overline{\raisebox{0pt}[5pt]{$x\as{\sigma}$}}y\as{\sigma}.(P\sigma) \\
(\textcolor{DarkBlue}{M}x(y).P)\sigma & \textcolor{DarkBlue}{M\as{\sigma}}x\as{\sigma}(z).(P\sigma[y
  \mapsto z]) \\
((\nu y)P)\sigma & (\nu z)(P\renb{z}{y}\sigma) \\
(\Match{x}{y}P)\sigma & \Match{x\as{\sigma}}{y\as{\sigma}}(P\sigma) \\
(P|Q)\sigma & (P\sigma)|(Q\sigma) \\
(P+Q)\sigma & (P\sigma) + (Q\sigma) \\
A(\vec{y})\sigma & A(\vec{y}\as{\sigma}) \\
\end{array}\]
where $P\sigma[y\mathbin{\mapsto} z]$ is yet to be defined.
Here $z$ is chosen outside $\Fn((\nu y)P)\cup \dom(\sigma)\cup{\it range}(\sigma)$;
when $y \mathbin{\notin} \dom(\sigma)\cup{\it range}(\sigma)$ one always picks $z:=y$.
\hyperlink{substitution}{The original definition of substitution} is retrieved by taking
$P\sigma[y\mapsto z] := P\renb{z}{y}\sigma$.

To facilitate reuse, I record below the properties of substitutions, $\Fn$ and $\eqaU$ that will be needed in my proofs.
\begin{eqnarray}
\Fn(P\renbt{y}{x}^\aN) &\subseteq& \Fn(A(\vec{y}))\label{A1} \\
P \eqaU Q &\Rightarrow& \Fn(P) = \Fn(Q) \label{free names alpha} \\
\Fn(P\renb{z}{y}^\aN) &\subseteq& \Fn(x(y).P) \cup \{z\} \label{free names input successors} \\
P{\sigma}[y\mapsto w]\renb{z\as{\sigma}}{w}^\aN &\eqaU& P\renb{z}{y}^\aN{\sigma} \label{A3}\\
\Fn(P\sigma) &=& \{x\as{\sigma} \mid x \inp \Fn(P)\} \label{free names sigma} \\
P\renb{\vec y\as{\sigma}}{\vec x}^\aN &\eqaU& P\renb{\vec y}{\vec x}^\aN{\sigma} \label{recursion sigma} \\
((\nu y)P) &\eqaU& (\nu w)(P\renb{w}{y}) \label{alpha nu} \\
(P\sigma_1)\sigma_2 &\eqaU& P(\sigma_2 \circ \sigma_1)\footnotemark \label{substitution composition} \\
\hspace{-20pt}(\forall x\inp\Fn(P).\,x\as{\sigma} \mathbin= x\as{\sigma'}) &\Rightarrow& P\sigma\eqaU P\sigma'\label{A10}\\
P \eqaU Q &\Rightarrow& P\sigma\eqaU Q\sigma\label{alpha sigma}\\
P\renb{w}{y}^\aN\renb{z}{w} &\eqaU& P\renb{z}{y}^\aN \label{A13} \\
P\epsilon &\eqaU& P \label{empty substitution}\\
U \eqaU (x(y).P)\sigma &\Rightarrow& \begin{array}{@{}l@{}}
  U=x\as{\sigma}(w).V \mbox{\small ~with}\\ V \eqaU P\sigma[y\mapsto w]
  \end{array}\label{A16}\\
U \eqaU (\nu y)P &\Rightarrow& \begin{array}{@{}l@{}}
U=(\nu w)V \mbox{\small ~with}\\ V \eqaU P\renb{y}{w}
  \end{array}\label{A17}
\end{eqnarray}
\footnotetext{For substitutions $\sigma_1$ and
  $\sigma_2$, their composition $\sigma_2\circ\sigma_1$ is defined by
  $\dom(\sigma_2 \circ \sigma_1) = \dom(\sigma_1) \cup \dom(\sigma_2)$
  and $x\as{\sigma_2 \circ \sigma_1} = x\as{\sigma_1}\as{\sigma_2}$.}%
Here $w \inp  \N\mathop{\setminus}\Fn((\nu y)P)$ in (\ref{alpha nu}),
$y \inp \N$ and $w \inp \Fn(x(y)P)$ in (\ref{A13}),
$w \in  \N\mathop{\setminus} (\Fn((\nu y)P) \cup \dom(\sigma) \cup {\it range}(\sigma))$ in (\ref{A3}),
 \plat{$A(\vec x) \mathbin{\stackrel{{\rm def}}{=}} P$} in (\ref{A1}) and (\ref{recursion sigma}),
and $\epsilon$ is the empty substitution.
As is easy to check, these properties hold for the original $\pi$-calculus,
that is, when using $\renb{z}{y}$ and $\renbt{y}{x}$ for $\renb{z}{y}^\aN$ and $\renbt{y}{x}^\aN$,
$\eqa$ for $\eqaU$, $\fn$ for $\Fn$ (i.e.,
$\Fn(P)^-_y := \Fn(P){\setminus}\{y\}$),
and $P\sigma[y\mathbin{\mapsto} z] := P\renb{z}{y}\sigma$.

First I state two standard properties of names in transitions.

Define the \emph{input arguments} $\ia(\alpha)$ of an action $\alpha$ by\vspace{-1ex}
\begin{center}
$\ia(M\bar x (y))=\ia(M\bar x y)=\ia(M\tau) := \emptyset$\quad and\quad $\ia(Mxz)=\{z\}$.
\end{center}

\begin{lemma}\label{lem:free names actions gen}
If $P\goto{\alpha}Q$ then $\n(\alpha){\setminus}(\ia(\alpha)\cup\bn(\alpha))\subseteq\Fn(P)$.
\end{lemma}
\begin{proof}
  A trivial induction on the inference of $P\goto{\alpha}Q$, using (\ref{A1}) and (\ref{free names alpha})
  when $P\goto{\alpha}Q$ is derived by \hyperlink{ES}{\textsc{\textbf{\small ide}} or \textsc{\textbf{\small alpha}}}.
\end{proof}

\begin{lemma}\label{lem:free names successors gen}
If $P\goto{\alpha}Q$ then $\Fn(Q)\subseteq\Fn(P)\cup\ia(\alpha)\cup\bn(\alpha)$.
\end{lemma}
\begin{proof}
  A trivial induction on the inference of $P\goto{\alpha}Q$, using (\ref{free names input successors})
  when $P\goto{\alpha}Q$ is derived by \hyperlink{ES}{\textsc{\textbf{\small early input}}},
  \lem{free names actions gen} when $P\goto{\alpha}Q$ is derived by \hyperlink{ES}{\textsc{\textbf{\small e-s-com}}},
  and (\ref{A1}) and (\ref{free names alpha})
  when $P\goto{\alpha}Q$ is derived by \hyperlink{ES}{\textsc{\textbf{\small ide}} or \textsc{\textbf{\small alpha}}}.
\end{proof}

\begin{definition}
The \emph{depth} of an interference of a transition (from the rules in the
operational semantics) is the length of the longest path in the proof tree of that inference.
The \emph{$\alpha$-depth} is defined likewise, but not counting applications of rule
\hyperlink{ES}{\textsc{\textbf{\small alpha}}}.

Following \cite{MPWpi2}, in the following lemmas the phrase
\begin{quote}
if $P \goto\alpha P'$ then \emph{equally} $Q \goto\alpha Q'$
\end{quote}
means that if $P \goto\alpha P'$ may be inferred then so, by an inference of the same
$\alpha$-depth, may be $Q \goto\alpha Q'$.
\end{definition}

For $\sigma\!: \N\mathbin\rightharpoonup\N$ a substitution and
$\alpha$ an action $M\tau$, $Mxy$, $M\bar xy$ or $M\bar x(z)$,
write $\alpha\as{\sigma}$ for $M\as{\sigma}\tau$, $M\as{\sigma}x\as{\sigma}y\as{\sigma}$,
$M\as{\sigma}\bar x\as{\sigma}y\as{\sigma}$ and $M\as{\sigma}\bar x\as{\sigma}(z\as{\sigma})$, respectively.

\lem{substitution preservation gen} below shows how substitutions preserve behaviour.
On the way to obtain it, I start with the special case of input and (bound) output transitions.

\begin{lemma}\rm\label{lem:substitution preserve gen}
  If \plat{$R \goto{\alpha} R'$}, with $\alpha$ not of the form $M\tau$,
  and $\sigma$ is a substitution \textcolor{Green}{with $\bn(\alpha\as{\sigma}) \cap \Fn(R\sigma) =\emptyset$},
  then equally \plat{$R{\sigma} \goto{\alpha\as{\sigma}}\eqaU R'{\sigma}$}.
\end{lemma}
\begin{proof}
  With induction on the inference of $R \goto{\alpha} R'$.
  \begin{itemize}
    \item  Suppose $R\mathbin{\goto{\alpha}} R'$ is derived by rule \hyperlink{ES}{\textsc{\textbf{\small early-input}}}.
      Then $R\mathbin=\textcolor{DarkBlue}{M}x(y).P$, $\alpha\mathbin=\textcolor{DarkBlue}{M}xz$ and
      $R'\mathbin=P\renb{z}{y}^\aN$. So
      $R{\sigma} \mathbin= \textcolor{DarkBlue}{M\as{\sigma}}x\as{\sigma}(w).(P\sigma[y\mapsto w]) \wedge
      \alpha\as{\sigma}\mathbin=\textcolor{DarkBlue}{M\as{\sigma}}x\as{\sigma}z\as{\sigma}$
      where $w$ is chosen outside $\Fn((\nu y)P)\cup \dom(\sigma)\cup{\it range}(\sigma)$.
      By rule \hyperlink{ES}{\textsc{\textbf{\small early-input}}} and (\ref{A3})
     \[R{\sigma} \mathbin{\goto{\alpha\as{\sigma}}} P{\sigma}[y\mapsto w]\renb{z\as{\sigma}}{w}^\aN
       \eqaU P\renb{z}{y}^\aN{\sigma} = R'\sigma.\]
    \item The case that $R \goto{\alpha} R'$ is derived by \hyperlink{ES}{\textsc{\textbf{\small output}}} is trivial.
    \item Suppose $R\mathbin{\goto{\alpha}} R'$ is derived by rule \hyperlink{ES}{\textsc{\textbf{\small ide}}}.
      Then $R\mathbin=A(\vec y)$.
      Let \plat{$A(\vec x) \stackrel{{\rm def}}{=} P$}. 
      By rule \hyperlink{ES}{\textsc{\textbf{\small ide}}}, $P\renb{\vec y}{\vec x}^\aN\mathbin{\goto{\alpha}}R'\!$.
      \textcolor{Green}{%
      Since $\Fn(P\renb{\vec y}{\vec x}^\aN\sigma) \subseteq \Fn(R\sigma)$ by (\ref{A1}) and (\ref{free names sigma}),
      $\bn(\alpha\as{\sigma})\cap \Fn(P\renb{\vec y}{\vec x}^\aN\sigma) =\emptyset$.}
      Hence $P\renb{\vec y}{\vec x}^\aN{\sigma} \mathbin{\goto{\alpha\as{\sigma}}\eqaU} R'{\sigma}\!$, by induction.
      By Property (\ref{recursion sigma}), $P\renb{\vec y\as{\sigma}}{\vec x}^\aN \eqaU P\renb{\vec y}{\vec x}^\aN{\sigma}$.
      Therefore $P\renb{\vec y\as{\sigma}}{\vec x}^\aN \mathbin{\goto{\alpha\as{\sigma}}\eqaU} R'{\sigma}$,
      by rule \hyperlink{ES}{\textsc{\textbf{\small alpha}}}.
      Thus \plat{$R{\sigma} = A(\vec y\as{\sigma}) \goto{\alpha\as{\sigma}}\eqaU R'{\sigma}$}
      by \hyperlink{ES}{\textsc{\textbf{\small ide}}}.
    \item The cases that $R \goto{\alpha} R'$ is derived by rule \hyperlink{ES}{\textsc{\textbf{\small sum}}},
      \hyperlink{ES}{\textsc{\textbf{\small symb-match}}} or \hyperlink{ES}{\textsc{\textbf{\small alpha}}} (using (\ref{free names alpha})) are trivial. 
    \item Suppose $R \goto{\alpha} R'$ is derived by \hyperlink{ES}{\textsc{\textbf{\small par}}}.
      Then $R = P | Q$, $P \goto{\alpha} P'$ and $R'=P'|Q$.
      \textcolor{Green}{Since $\bn(\alpha\as{\sigma}) \cap \Fn(R\sigma) =\emptyset$,
      $\bn(\alpha\as{\sigma}) \cap \Fn(P\sigma) \mathbin=\emptyset$.}
      So by induction $P\sigma\mathbin{\goto{\alpha\as{\sigma}}\eqaU} P'\sigma$.
      Moreover $R\sigma \mathbin= P\sigma|Q\sigma$ \textcolor{Green}{and $\bn(\alpha\as\sigma)\cap\Fn(Q\sigma)=\emptyset$}.
      So $R\sigma\goto{\alpha\as{\sigma}}\eqaU P'\sigma|Q\sigma = R'\sigma$ by rule \hyperlink{ES}{\textsc{\textbf{\small par}}}.
     \item Let $R \goto{\alpha} R'$ be derived by rule \hyperlink{ES}{\textsc{\textbf{\small res}}}.
      Then $R \mathbin= (\nu y)P$, $P \mathbin{\goto\alpha} P'$, $y \mathbin{\not\in} \n(\alpha)$
      and $R' \mathbin= (\nu y)P'$.
      Pick a $w$ outside \[\Fn((\nu y)P') \cup \Fn((\nu y)P)
      \cup \dom(\sigma) \cup \textit{range}(\sigma) \cup \n(\alpha\as{\sigma}).\]
      Then $R\sigma \eqaU ((\nu w)(P\renb{w}{y}))\sigma = (\nu w)(P\renb{w}{y}\sigma)$ by
      (\ref{alpha nu}) and (\ref{alpha sigma}).
      \textcolor{Green}{Since $w\notin \bn(\alpha\as{\sigma})$ and $\bn(\alpha\as{\sigma}) \cap \Fn(R\sigma) =\emptyset$,
      by (\ref{free names alpha}) also $\bn(\alpha\as{\sigma}) \cap \Fn(P\renb{y}{w}\sigma) =\emptyset$.}
      So by induction, using (\ref{substitution composition}) and the substitution $\sigma \circ \renb{w}{y}$,
      $$P\renb{y}{w}\sigma \eqaU\goto{\alpha\as{\sigma}}\eqaU P'\renb{y}{w}\sigma.$$
      Thus $R\sigma \goto{\alpha\as{\sigma}}\eqaU (\nu w) (P'\renb{y}{w}\sigma)$ by
      \hyperlink{ES}{\textsc{\textbf{\small res}}} and \hyperlink{ES}{\textsc{\textbf{\small alpha}}}.
      Furthermore, $(\nu w) (P'\renb{y}{w}\sigma) \eqaU R'\sigma$ by (\ref{alpha nu}) and (\ref{alpha sigma}).
    \item Let $R \goto{M\bar x(y)} R'$ be derived by rule \hyperlink{ES}{\textsc{\textbf{\small symb-open}}}.
      Then $R \mathbin= (\nu y)P$, $P \goto{M\bar x y} R'$, $y \neq x$ and $y \mathbin{\not\in} \n(M)$.
      Now $R\sigma = (\nu z)(P\renb{z}{y}\sigma)$ for some
      $z \notin \Fn((\nu y)P) \cup \dom(\sigma) \cup \textit{range}(\sigma)$.
      \textcolor{Green}{Since $w := y\as{\sigma}\notin \Fn(R\sigma)$ (the side-condition of this lemma),}
      $R\sigma \eqaU (\nu w)(P\renb{z}{y}\sigma\renb{w}{z})$ by (\ref{alpha nu}).

      Note that $u \as{\renb{w}{z} \circ \sigma\circ \renb{z}{y}} = u\as\sigma$ for $u \in \Fn(P)$.
      Hence $P\renb{z}{y}\sigma\renb{w}{z} \eqaU P\sigma$ by (\ref{substitution composition}) and (\ref{A10}).
      By induction,
      $$P\sigma \goto{M\as{\sigma}\bar x\as{\sigma} w}\eqaU R'\sigma.$$
      By \lem{free names actions gen}, $\n(M) \cup \{x,y\} \subseteq \Fn(P)$.
      Hence $\n(M) \cup \{x\} \subseteq \Fn(R) = \Fn(P){\setminus}\{y\}$, using that $y \notin \n(M) \cup \{x\}$.
      Thus $\n(M\as{\sigma}) \cup \{x\as{\sigma}\} \subseteq \Fn(R\sigma)$ by (\ref{free names sigma}),
      and therefore $w \notin\n(M\as{\sigma}) \cup \{x\as{\sigma}\}$.
      Thus, by \hyperlink{ES}{\textsc{\textbf{\small symb-open}}} and \hyperlink{ES}{\textsc{\textbf{\small alpha}}},\\[1ex]
      \mbox{}\hfill $R\sigma \eqaU (\nu w)(P\sigma) \goto{M\as{\sigma}\bar x\as{\sigma} (w)}\eqaU R'\sigma$.
   \qed
  \end{itemize}
\end{proof}
Before generalising the above result to include the case $\alpha=M\tau$,
I cover three more helpful standard properties.

\begin{lemma}\rm\label{lem:input universality pre 2 gen}
  If $R\goto{Mxz}R_z$ and $w\notin\Fn(R)$ then equally
  $R\goto{Mxw}R_w$ for some process $R_w$ with $R_z \eqaU R_w\renb{z}{w}$.
\end{lemma}

\begin{proof}
  With induction on the inference of $P\goto{Mxz}P_z$.
  \begin{itemize}
    \item  Suppose $R\goto{\textcolor{DarkBlue}{M}xz} R_z$ is derived by rule \hyperlink{ES}{\textsc{\textbf{\small early-input}}}.
      Then $R\mathbin=\textcolor{DarkBlue}{M}x(y).P$ and $R_z \mathbin=P\renb{z}{y}^\aN$. Moreover,
      $R\goto{\textcolor{DarkBlue}{M}xw} R_w := P\renb{w}{y}^\aN$. By (\ref{A13}), $R_z \eqaU R_w\renb{z}{w}$.
    \item  Suppose $R\goto{Mxz} R_z$ is derived by \hyperlink{ES}{\textsc{\textbf{\small par}}}.
      Then $R = P | Q$, $P \goto{Mxz} P_z$ and $R_z=P_z|Q$.
      So by induction there is a process $P_w$ such that $P\goto{Mxw}P_w$ and $P_z \eqaU P_w\renb{z}{w}$.
      Thus $R\goto{Mxw}P_w|Q$ by \hyperlink{ES}{\textsc{\textbf{\small par}}}.
      Moreover, $(P_w|Q)\renb{z}{w} = P_w\renb{z}{w}|Q\renb{z}{w} \eqaU P_z|Q = R_z$ by (\ref{A10}) and (\ref{empty substitution}),
      using that $w\notin\Fn(Q)$.
    \item Let $R \goto{Mxz} R_z$ be derived by rule \hyperlink{ES}{\textsc{\textbf{\small res}}}.
      Then $R \mathbin= (\nu y)P$, $P \goto{Mxz} P_z$, $y \mathbin{\not\in} \n(M) \cup \{x,z\}$ and $R_z = (\nu y)P_z$.
      Pick a $v \notin \Fn(P) \cup \n(M) \cup \{x,y\}$.
      By induction there is a process $P_v$ such that $P\goto{Mxv}P_v$ and $P_z \eqaU P_v\renb{z}{v}$.
      So $R \mathbin{\goto{Mxv}} (\nu y)P_v$ by \hyperlink{ES}{\textsc{\textbf{\small res}}}.
      By (\ref{A10}), (\ref{empty substitution}) and \lem{substitution preserve gen},
      $$R \eqaU R\renb{w}{v} \goto{M x w} R_w \eqaU ((\nu y)P_v)\renb{w}{v}$$
      for some $R_w$. Hence $R \goto{M x w} R_w$ by rule \hyperlink{ES}{\textsc{\textbf{\small alpha}}}.

      By \lem{free names successors gen}
      $\Fn(P_v){\setminus}\{y,v\} \mathbin\subseteq \Fn(P){\setminus}\{y\} \mathbin= \Fn(R) \mathbin{\not\ni} w$.
     Pick $u \notin \Fn(P_v) \cup \Fn(P_v\renb{z}{v}) \cup \{v,w\}$. Then
     \[\begin{array}{@{}r@{\,\eqaU\,}l@{}}
     R_w\renb{z}{w} & ((\nu y)P_v)\renb{w}{v}\renb{z}{w} \eqaU\\
     ((\nu u)P_v\hspace{-1pt}\renb{u}{v})\renb{w}{v}\renb{z}{w} &
     (\nu u)(\hspace{-1pt}P_v\hspace{-1pt}\renb{u}{y}\renb{w}{v}\renb{z}{w}\hspace{-1pt})\\
      \eqaU (\nu u)(P_v\renb{z}{v}\renb{u}{y}) & (\nu y)(P_v\renb{z}{v}) \eqaU R_z
      \end{array}\]
     by  (\ref{alpha nu})--(\ref{alpha sigma}), using that $v \neq y \neq z$ and $w \neq u \neq v$.
   \item The other four cases are trivial, using (\ref{A1}) and (\ref{free names alpha})
      when $P\goto{\alpha}Q$ is derived by \hyperlink{ES}{\textsc{\textbf{\small ide}}} or \hyperlink{ES}{\textsc{\textbf{\small alpha}}}.
    \qed
  \end{itemize}
\end{proof}

\begin{lemma}\rm\label{lem:input universality 2 gen}
  If $R\goto{Mxw}R_w$ with $w\notin\Fn(R)$, and $z$ a name, then equally
  $R\mathbin{\goto{Mxz}}R_z$ for some $R_z$ with $R_z \mathbin\eqaU R_w\renb{z}{w}$.
\end{lemma}

\begin{proof}
  With induction on the inference of $R\goto{Mxw}R_w$.
  \begin{itemize}
    \item Let $R \goto{Mxw} R_w$ be derived by rule \hyperlink{ES}{\textsc{\textbf{\small res}}}.
      Then $R \mathbin= (\nu y)P$, $P \goto{Mxw} P_w$, $y \mathbin{\not\in} \n(M) \cup \{x,w\}$ and $R_w = (\nu y)P_w$.
      Pick a $v \mathbin{\notin} \Fn(P) \cup \Fn(P_w) \cup \n(M) \cup \{x,y\}$.
      By induction there is a process $P_v$ with $P\goto{Mxv}P_v$ and $P_v \eqaU P_w\renb{z}{w}$.
      So $R \mathbin{\goto{Mxv}} (\nu y)P_v$ by \hyperlink{ES}{\textsc{\textbf{\small res}}}.
      By (\ref{A10}), (\ref{empty substitution}) and \lem{substitution preserve gen},
      $$R \eqaU R\renb{z}{v} \goto{M x z} R_z \eqaU ((\nu y)P_v)\renb{z}{v}$$
      for some $R_z$. Hence $R \goto{M x z} R_z$ by rule \hyperlink{ES}{\textsc{\textbf{\small alpha}}}.
     Pick $u \notin \Fn(P_w) \cup  \Fn(P_v) \cup \{z,v,w\}$. Then
     \[\begin{array}{@{}r@{\,\eqaU\,}l@{}}
     R_w\renb{z}{w} & ((\nu y)P_w)\renb{z}{w} \eqaU\\
     ((\nu u)P_w\hspace{-1pt}\renb{u}{y})\renb{z}{w} &
     ((\nu u)P_w\hspace{-1pt}\renb{u}{y}\renb{z}{w}) \eqaU\\
     (\nu u)(\hspace{-1pt}P_w\hspace{-1pt}\renb{z}{w}\renb{u}{y}\renb{z}{v}\hspace{-1pt}) &(\nu u)(P_v\renb{u}{y}\renb{z}{v}) \\
      \eqaU (\nu u)(P_v\renb{u}{y})\renb{z}{v} & ((\nu y)P_v)\renb{z}{v} \eqaU R_z
      \end{array}\]
     by (\ref{alpha nu})--(\ref{alpha sigma}), using that $v \neq z \neq y \neq w \neq u \neq v$.
    \item All other cases proceed as in the proof of \lem{input universality pre 2 gen}.
    \qed
  \end{itemize}
\end{proof}
The above proofs caters both to the full $\pi$-calculus and to $\textcolor{DarkBlue}{\piIM}$, namely by considering
explicitly the input prefix $\textcolor{DarkBlue}{M}x(y)$. Henceforth I focus on the full
$\pi$-calculus, and drop the $\textcolor{DarkBlue}{M}$ from the input prefix. It can always be
retrieved via an application of \hyperlink{ES}{\textsc{\textbf{\small symb-match}}}.

\begin{lemma}\rm\label{lem:bound output universality gen}
  If $R\goto{Mx(z)}R'$ and $w\notin\Fn(R)$ then equally
  $R\goto{Mx(w)}\eqaU R'\renb{w}{z}$.
\end{lemma}
\begin{proof}
  With induction on the inference of $R \goto{Mx(z)} R'$.
  \begin{itemize}
   \item Let $R \goto{M\bar x(z)} R'$ be derived by rule \hyperlink{ES}{\textsc{\textbf{\small symb-open}}}.
     Then $R \mathbin= (\nu z)P$, $P \goto{M\bar x z} R'$, $z \neq x$ and $z \mathbin{\not\in} \n(M)$.
     By \lem{substitution preserve gen}, $P\renb{w}{z} \goto{M\bar x w}\eqaU R'\renb{w}{z}$.
     By \lem{free names actions gen}, $\n(M)\cup\{x\} \subseteq \Fn(R)$.
     Hence $w \notin \n(M)\cup\{x\}$.
     By (\ref{alpha nu}), $R \eqaU (\nu w)(P \renb{w}{z})$, so rules \hyperlink{ES}{\textsc{\textbf{\small symb-open}}} and
     \hyperlink{ES}{\textsc{\textbf{\small alpha}}} yield $R \goto{M\bar x (w)}\eqaU R'\renb{w}{z}$.
   \item Let $R \mathbin{\goto{M\bar x(z)}} R'$ be derived by rule \hyperlink{ES}{\textsc{\textbf{\small res}}}.
     Then $R \mathbin= (\nu y)P$, $P \goto{M\bar x(z)} P'$, $y \not\in \n(M)\cup\{x,z\}$
     and $R' = (\nu y)P'$. Pick a $v \notin \{y\} \cup\Fn(P)\cup \Fn(P') \cup \n(M)\cup\{x\}$.
     By induction $P \goto{M\bar x(v)}\eqaU P'\renb{v}{z}$.\vspace{1pt}
     So $R \goto{M\bar x(v)}\eqaU (\nu y)(P'\renb{v}{z})$ by \hyperlink{ES}{\textsc{\textbf{\small res}}}.
     By (\ref{A10}), (\ref{empty substitution}) and \lem{substitution preserve gen},
     $$R \mathbin\eqaU R\renb{w}{v} \goto{M\bar x (w)}\eqaU ((\nu y)(P'\renb{v}{z}))\renb{w}{v}.$$
     Pick a $u \notin \Fn(P') \cup \Fn(P'\renb{v}{z}) \cup \{z,v,w\}$. Then
     \[\begin{array}{@{}r@{\;\eqaU\;}l@{}}
     ((\nu y)(P'\hspace{-1pt}\renb{v}{z}))\renb{w}{v} & (\nu u)(P'\hspace{-1pt}\renb{v}{z}\renb{u}{y}\renb{w}{v}) \eqaU\\
     (\nu u)(P'\renb{u}{y}\renb{w}{z}) & ((\nu y)P')\renb{w}{z} = R'\renb{w}{z},
      \end{array}\]
     by (\ref{alpha nu})--(\ref{alpha sigma}), using that $\Fn(P') \mathbin{\not\ni} v \neq y \neq z \neq u \neq v$.
     By rule \hyperlink{ES}{\textsc{\textbf{\small alpha}}}, $R \goto{M\bar x (w)}\eqaU R'\renb{w}{z}$.
   \item The other five cases are trivial.
 \qed
   \end{itemize}
\end{proof}

\begin{lemma}\rm\label{lem:substitution preservation gen}
  If \plat{$R \goto{\alpha} R'$} and $\sigma$ is a substitution with 
  $\bn(\alpha\as{\sigma}) \cap \Fn(R\sigma) =\emptyset$
  then equally \plat{$R{\sigma} \goto{\alpha\as{\sigma}}\eqaU R'{\sigma}$}.
\end{lemma}
\begin{proof}
  With induction on the inference of $R \goto{\alpha} R'$.
  \begin{itemize}
    \item The case that $R \goto{\alpha} R'$ is derived by rule \hyperlink{ES}{\textsc{\textbf{\small tau}}}
      is trivial.
    \item The cases that $R \goto{\alpha} R'$ is derived by rule
      \hyperlink{ES}{\textsc{\textbf{\small output}}},
      \hyperlink{ES}{\textsc{\textbf{\small early-input}}},
      \hyperlink{ES}{\textsc{\textbf{\small sum}}},
      \hyperlink{ES}{\textsc{\textbf{\small symb-match}}},
      \hyperlink{ES}{\textsc{\textbf{\small ide}}},
      \hyperlink{ES}{\textsc{\textbf{\small par}}},
      \hyperlink{ES}{\textsc{\textbf{\small res}}},
      \hyperlink{ES}{\textsc{\textbf{\small symb-open}}} or
      \hyperlink{ES}{\textsc{\textbf{\small alpha}}}
      are already covered by \lem{substitution preserve gen}.
    \item The case that $R \goto{\alpha} R'$ is derived by \hyperlink{ES}{\textsc{\textbf{\small e-s-com}}} is trivial.
    \item Suppose $R \goto{\alpha} R'$ is derived by \hyperlink{ES}{\textsc{\textbf{\small e-s-close}}}.
      Then $R = P | Q$, $P \goto{M\bar x (z)} P'$, $Q \goto{N\bar v z} Q'$, $z\notin\Fn(Q)$,
      $\alpha = \Match{x}{v}MN\tau$ and $R'=(\nu z)(P'|Q')$.
      Pick $w \notin \Fn(R) \cup \Fn(R') \cup \dom(\sigma)\cup {\it range}(\sigma)$.
      Now $P \mathbin{\goto{M\bar x (w)}\eqaU} P'\renb{w}{z}$ and
      $Q \goto{N\bar v w}\eqaU Q'\renb{w}{z}$ by Lemmas~\ref{lem:bound output universality gen}
      and~\ref{lem:input universality 2 gen}. Thus
      $P\sigma \mathbin{\goto{M\as{\sigma}\bar x\as{\sigma} (w)}\eqaU} P'\!\renb{w}{z}\sigma$,
      $Q\sigma \mathbin{\goto{N\as{\sigma}\bar v\as{\sigma} w}\eqaU} Q'\renb{w}{z}\sigma$\linebreak[4]
      by \lem{substitution preserve gen}, using that $w\as{\sigma} =w \notin \Fn(P\sigma)$,
      which follows by (\ref{free names sigma}) from $w \notin \Fn(P)\cup {\it range}(\sigma)$.
      Hence $$R\sigma \goto{\match{x\as{\sigma}}{v\as{\sigma}}M\as{\sigma}N\as{\sigma}\tau}
      (\nu w) ((P'|Q')\renb{w}{z}\sigma)$$
      by rule \hyperlink{ES}{\textsc{\textbf{\small e-s-close}}}.
      As $w \notin \Fn((\nu z)(P'|Q'))$, by (\ref{alpha nu}),\\[2ex]
      \mbox{}\hfill $R'\sigma \eqaU (\nu w)((P'|Q')\renb{w}{z}\sigma)$.
  \qed
  \end{itemize}
\end{proof}

Next I would have liked to show that substitutions also reflect behaviour, in the sense that,
possibly under the same side condition as in \lem{substitution preservation gen},
\plat{$R\sigma \goto{\beta} U$} implies \plat{$R \goto{\alpha} R'$} for some $\alpha$ and $R'$ with
$\alpha\as{\sigma}\mathbin=\beta$ and $R'\sigma \eqaU U$.
Unfortunately, this property does not hold.

\begin{example}
  $(x(y).P)\renb{w}{z} \goto{xz} P\renb{w}{z}\renb{z}{y}$, yet there is no
  $\alpha$ such that $\alpha\as{\renb{w}{z}}=xz$. 
\end{example}
In order to rescue the reflection property, I restrict attention to surjective substitutions.

\begin{definition}
A substitution $\sigma\!: \N\mathbin\rightharpoonup\N$ is called \emph{surjective} if
$\dom(\sigma) \subseteq {\it range}(\sigma)$.
\end{definition}
Using rule \hyperlink{ES}{\textsc{\textbf{\small alpha}}}, the below reflection lemma can equally well
be formulated with \plat{$R\sigma \goto{\beta} U'$} instead of \plat{$R\sigma \eqaU U \goto{\beta} U'$}.
However, its inductive proof requires the form stated.\pagebreak[4]

\begin{lemma}\rm\label{lem:substitution reflection gen}
  If \plat{$R\sigma \eqaU U \goto{\beta} U'$}
  and $\sigma$ is a surjective substitution,
  then equally \plat{$R \goto{\alpha} R'$} for some $\alpha$ and $R'$ with
  $\alpha\as{\sigma}\mathbin=\beta$ and $R'\sigma \eqaU U'$.
\end{lemma}
\begin{proof}
  With induction on the $\alpha$-depth of the inference of $U \goto{\beta}U'$, with a nested
  induction on the number of applications of rule \hyperlink{ES}{\textsc{\textbf{\small alpha}}} at
  the end of that inference.
  \begin{itemize}
    \item The cases that $U \goto{\beta} U'$ is derived by rule \hyperlink{ES}{\textsc{\textbf{\small tau}}} or
      \hyperlink{ES}{\textsc{\textbf{\small output}}} are trivial.
    \item Suppose $U \goto{\beta} U'$ is derived by \hyperlink{ES}{\textsc{\textbf{\small early-input}}}.
      Then $R=x(y).P$ and
      $U = x\as{\sigma}(w).V$ with $V \eqaU (P\sigma[y\mapsto w])$ by (\ref{A16}).
      So $\beta=x\as{\sigma}v$ and $U'=V\renb{v}{w}^\aN$.
      Since $\sigma$ is surjective, there is a $z$ with $z\as{\sigma}=v$.
      By \hyperlink{ES}{\textsc{\textbf{\small early-input}}} $R\goto{\alpha}R'$ with
      $\alpha\mathbin=\textcolor{black}{M}xz$ and $R'\mathbin=P\renb{z}{y}^\aN$.
      Now $\alpha\as{\sigma}\mathbin=\beta$ and
      $R'\sigma = P\renb{z}{y}^\aN \sigma \eqaU P{\sigma}[y\mapsto w]\renb{z\as{\sigma}}{w}^\aN \eqaU V\renb{v}{w}^\aN = U'$,
      using (\ref{A3}) and (\ref{A10}).
    \item Suppose $U \goto{\beta} U'$ is derived by \hyperlink{ES}{\textsc{\textbf{\small ide}}}.
      Then $R=A(\vec y)$ and $U \mathbin= R\sigma\mathbin=A(\vec{y}\as{\sigma})$.
      Let \plat{$A(\vec x) \mathord{\stackrel{{\rm def}}{=}} P$}.
      Then $P\renb{\vec{y}\as{\sigma}}{\vec{x}} \mathbin{\goto{\beta}} U'\!$.
      Since $P\renb{\vec{y}\as{\sigma}}{\vec{x}} \eqaU P\renbt{y}{x}\sigma$
      by (\ref{recursion sigma}), $P\renbt{y}{x}\sigma\goto{\beta} U'$ by \hyperlink{ES}{\textsc{\textbf{\small alpha}}}.
      Moreover, the inference of $P\renbt{y}{x}\sigma\goto{\beta} U'$ has the same $\alpha$-depth as that of
      $P\renb{\vec{y}\as{\sigma}}{\vec{x}} \mathbin{\goto{\beta}} U'\!$, which is smaller than the
      $\alpha$-depth of the inference of $U \goto{\beta} U'$.
      Hence, by induction, \plat{$P\renbt{y}{x} \goto{\alpha} R'$} for some $\alpha$ and $R'$ with
      $\alpha\as{\sigma}=\beta$ and $R'\sigma \eqaU U'$.
      By rule \hyperlink{ES}{\textsc{\textbf{\small ide}}} $R\goto{\alpha}R'$.
    \item The cases that $U \goto{\beta} U'$ is derived by rule \hyperlink{ES}{\textsc{\textbf{\small sum}}}
      or \hyperlink{ES}{\textsc{\textbf{\small symb-match}}} are trivial.
    \item Suppose $U \goto{\beta} U'$ is derived by \hyperlink{ES}{\textsc{\textbf{\small par}}}.
      Then $R = P | Q$, $R\sigma \mathbin= P\sigma | Q\sigma$, $U \mathbin= V|W$,
      $P\sigma \mathbin\eqaU V \goto{\beta} V'$, $Q\sigma \mathbin\eqaU W$,
      $\bn(\beta)\cap\Fn(W)\mathbin=\emptyset$ and $U'\mathbin=V'|W$.
      By induction $P\mathbin{\goto{\alpha}} P'$ for some $\alpha$ and $P'$ with
      $\alpha\as{\sigma}\mathbin=\beta$ and $P'\sigma \eqaU V'\!$.
      By~(\ref{free names alpha}) $\bn(\alpha\as\sigma)\cap\Fn(Q\sigma)=\emptyset$, so by (\ref{free names sigma})
      $\bn(\alpha)\cap\Fn(Q)=\emptyset$.
      Consequently, $R = P|Q \goto{\alpha} P'|Q$ by rule \hyperlink{ES}{\textsc{\textbf{\small par}}}, and
      $(P'|Q)\sigma \eqaU V'|W = U'$.
    \item Suppose $U \goto{\beta} U'$ is derived by \hyperlink{ES}{\textsc{\textbf{\small e-s-com}}}.
      In that case $R = P | Q$, $R\sigma = P\sigma | Q\sigma$, $U = V|W$,
      $\beta=\match{x}{v}MN\tau$, $P\sigma\mathbin\eqaU V \goto{M\bar xy} V'$,
      $Q\sigma \mathbin\eqaU W \goto{Nvy} W'$ and $U'\mathbin=V'|W'$.
      By induction, there are matching sequences $K,L$
      with $K\as{\sigma}\mathbin=M$ and $L\as{\sigma}\mathbin=N$,
      names $q,r,z,u$ with $q\as{\sigma}=x$, $r\as{\sigma}=v$,
      $z\as{\sigma}=y$ and $u\as{\sigma}=y$, and processes $P'$ and $Q'$ with $P'\sigma\eqaU V'$ and
      $Q'\sigma\eqaU W'$, such that $P\goto{K\bar qz} P'$ and $Q\goto{L ru} Q'$.

      Pick $w\notin\Fn(Q)$. By \lem{input universality pre 2 gen} there is a process $P_w$ such that
      $Q\mathbin{\goto{Lrw}}Q_w$ and $Q' \eqaU Q_w\renb{u}{w}$.
      By \lem{input universality 2 gen} there is a $P_z$ such that
      $Q\mathbin{\goto{Lrz}}Q_z$ and $Q_z \eqaU Q_w\renb{z}{w}$.
      By rule \hyperlink{ES}{\textsc{\textbf{\small e-s-com}}} $R\goto{\match{q}{r}KL\tau} R' := P'|Q_z$.
      By (\ref{substitution composition})--(\ref{alpha sigma}) $Q'\sigma \eqaU Q_z\sigma$.
      So $(\match{q}{r}KL\tau)\as{\sigma}=\match{x}{v}MN\tau$ and $R'\sigma = (P'|Q_z)\sigma \eqaU V'|W'=U'$.
   \item 
      Suppose $U \goto{\beta} U'$ is derived by rule \hyperlink{ES}{\textsc{\textbf{\small e-s-close}}}.
      Then $R = P | Q$, $R\sigma = P\sigma | Q\sigma$, $U = V|W$,
      $\beta=\match{x}{v}MN\tau$, $P\sigma\eqaU V \goto{M\bar x(z)} V'$,
      $Q\sigma \eqaU W \goto{Nvz} W'$ and $U'=(\nu z)(V'|W')$ for some $z\notin\Fn(W)$.
      Pick $w \notin \Fn(U) \cup \dom(\sigma) \cup {\it range}(\sigma)$.
      By (\ref{free names alpha}) $w \notin \Fn(W) = \Fn(Q\sigma)$ so $w \notin \Fn(Q)$ by (\ref{free names sigma}).
      Now $V\goto{M\bar x(w)}\eqaU V'\renb{w}{z}$ and
      $W\goto{Nvw}\eqaU W'\renb{w}{z}$ by Lemmas~\ref{lem:bound output universality gen}
      and~\ref{lem:input universality 2 gen}.
      Moreover, the inferences of these transitions have a smaller $\alpha$-depth than that of $U \goto{\beta} U'$.
      By induction, using that $w \notin \dom(\sigma) \cup {\it range}(\sigma)$, there are matching sequences $K,L$
      with $K\as{\sigma}\mathbin=M$ and $L\as{\sigma}\mathbin=N$,
      names $q,r$ with $q\as{\sigma}=x$ and $r\as{\sigma}=v$,
      and processes $P'$ and $Q'$ with $P'\sigma\eqaU V'\renb{w}{z}$ and
      $Q'\sigma\eqaU W'\renb{w}{z}$, such that $P\goto{K\bar q(w)} P'$ and $Q\goto{L r w} Q'$.
      Therefore $R\goto{\match{q}{r}KL\tau} (\nu w)(P'|Q')$ by \hyperlink{ES}{\textsc{\textbf{\small e-s-com}}}.
      Moreover, $(\match{q}{r}KL\tau)\as{\sigma} = \beta$.
      By \lem{free names successors gen}, $\Fn(U') \subseteq \Fn(U)\not\ni w$. So by (\ref{alpha nu})
      $((\nu w)(P'|Q'))\sigma = (\nu w)(P'\sigma|Q'\sigma) \eqaU (\nu w)(V'\renb{w}{z}|W'\renb{w}{z}) \eqaU U'$.
    \item Suppose $U \goto{\beta} U'$ is derived by \hyperlink{ES}{\textsc{\textbf{\small res}}}.
      Then $R = (\nu y)P$ and $R\sigma \mathbin= (\nu z)(P\renb{z}{y}\sigma)$
      for some $z \mathbin{\notin} \Fn((\nu y)P) \cup \dom(\sigma) \cup {\it range}(\sigma)$.
      By (\ref{A17}), $U = (\nu w)V$ with $V \eqaU P\renb{z}{y}\sigma\renb{w}{z}$.
      So $w \mathbin{\notin} \n(\beta)$, $V \goto\beta V'$ and $U'\mathbin=(\nu w) V'\!$.
      Since $\sigma$ is surjective, there is a $u$ with $z\as{\sigma}=w$.
      Now $V \eqaU P\renb{z}{y}\sigma\renb{w}{z} \eqaU P\renb{u}{y}\sigma$ by
      (\ref{substitution composition}) and (\ref{A10}).
      By rule \hyperlink{ES}{\textsc{\textbf{\small alpha}}} $P\renb{u}{y}\sigma \goto{\beta} V'$.
      Moreover, the inference of $P\renb{u}{y}\sigma\goto{\beta} V'$ has the same $\alpha$-depth as that of
      $V \goto\beta V'$, which is smaller than the one of $U \goto{\beta} U'$.
      So by induction \plat{$P\renb{u}{y} \goto{\alpha} P'$} for some $\alpha$ and $P'$ with
      $\alpha\as{\sigma}\mathbin=\beta$ and $P'\sigma \eqaU V'$.
      Since $w \notin \n(\beta)$ one has $u \notin \n(\alpha)$.
      Hence $(\nu u)(P\renb{u}{y}) \goto{\alpha} (\nu u)P'$ by \hyperlink{ES}{\textsc{\textbf{\small res}}}.
      Moreover, since $w \notin \Fn((\nu w)V) = \Fn(U)=\Fn(R\sigma)$, using (\ref{free names alpha}),
      $u \notin \Fn(R$) by (\ref{free names sigma}).
      Hence $R \eqaU (\nu u)(P\renb{u}{y})$ by (\ref{alpha nu}) and
      $R \goto{\alpha} (\nu u)P'$ by \hyperlink{ES}{\textsc{\textbf{\small alpha}}}.

      Suppose that there is an $x \in \Fn((\nu u)P')$ with $x\as{\sigma}=w$.
      Then, by \lem{free names successors gen}, $x \in \n(\alpha)$ or $x \in \Fn(R)$.
      In the first case $w \in \n(\alpha\as{\sigma})=\n(\beta)$, which has been ruled out.
      In the second case, by (\ref{free names sigma}), $w \mathbin\in \Fn(R\sigma)\mathbin=\Fn(U)$, contradicting
      $U\mathbin=(\nu w)V$. Hence there is no such $x$.
      Now pick\linebreak[4] $v \notin \Fn(P') \cup \dom(\sigma) \cup {\it range}(\sigma) \cup \Fn(P'\sigma)$ and obtain\vspace{-3pt}
      \[((\nu u)P')\sigma \begin{array}[t]{l@{~}l@{\quad}l@{}}
        \eqaU & ((\nu v)(P'\renb{v}{u}))\sigma & \mbox{by (\ref{alpha nu}), (\ref{alpha sigma})}\\
        =     & (\nu v)(P'\renb{v}{u}\sigma) & \mbox{by definition}\\
        \eqaU & (\nu v)(P'\sigma\renb{v}{w}) & \mbox{by (\ref{substitution composition}), (\ref{A10})}\\
        \eqaU & (\nu w)(P'\sigma)           & \mbox{by (\ref{alpha nu})}\\
        \eqaU & (\nu w)V' = U'.\\[-1ex]
        \end{array}\]
    \item Suppose $U \goto{\beta} U'$ is derived by \hyperlink{ES}{\textsc{\textbf{\small symb-open}}}.
      Then $R = (\nu y)P$ and $R\sigma \mathbin= (\nu z)(P\renb{z}{y}\sigma)$
      for some $z \mathbin{\notin} \Fn((\nu y)P) \cup \dom(\sigma) \cup {\it range}(\sigma)$.
      By (\ref{A17}), $U = (\nu w)V$ with $V \eqaU P\renb{z}{y}\sigma\renb{w}{z}$.
      So $\beta=M\bar x (w)$,  $V \goto{M \bar x w} U'$ and $w \notin \n(M) \cup \{x\}$.
      Since $\sigma$ is surjective, there is a $u$ with $u\as{\sigma}=w$.
      Now $V \eqaU P\renb{z}{y}\sigma\renb{w}{z} \eqaU P\renb{u}{y}\sigma$ by
      (\ref{substitution composition}) and (\ref{A10}).
      By rule \hyperlink{ES}{\textsc{\textbf{\small alpha}}} $P\renb{u}{y}\sigma \goto{M \bar x w} U'$.
      Moreover, the inference of $P\renb{u}{y}\sigma\goto{M \bar x w} U'$ has the same $\alpha$-depth as that of
      $V \goto{M \bar x w} U'$, which is smaller than the one of $U \goto{\beta} U'$.
      By induction \plat{$P\renb{u}{y} \goto{K \bar q r} P'$} for some $K$, $q$, $r$ and $P'$ with
      $K\as{\sigma}\mathbin=M$, $q\as{\sigma}\mathbin=x$, $r\as{\sigma}\mathbin=w$ and $P'\sigma \eqaU U'$.

      Since $u\as{\sigma} = w \notin \Fn((\nu w)V) = \Fn(U) = \Fn(R\sigma)$, using (\ref{free names alpha}),
      $u \notin \Fn(R)$ by (\ref{free names sigma}). Hence $R \eqaU (\nu u)(P\renb{u}{y})$ by (\ref{alpha nu}).
      By \lem{free names actions gen} $r \in \Fn(P\renb{u}{y})$. So if $r\neq u$ then
      $r \mathbin\in \Fn((\nu u)(P\renb{u}{y}) \mathbin= \Fn(R)$, again using (\ref{free names alpha}), and
      $w=r\as\sigma \in \Fn(R\sigma)$ by (\ref{free names sigma}), yielding a contradiction. Thus $r=u$.
      Since $w  = u\as{\sigma} \notin \n(M) \cup \{x\}$ one has $y \notin \n(K) \cup \{q\}$.
      Hence $(\nu u)(P\renb{u}{y}) \goto{K \bar q (u)} P'$
      by \hyperlink{ES}{\textsc{\textbf{\small symb-open}}}
      and $R \goto{K \bar q (u)} P'$ by \hyperlink{ES}{\textsc{\textbf{\small alpha}}}.
    \item Suppose $U \goto{\beta} U'$ is derived by \hyperlink{ES}{\textsc{\textbf{\small alpha}}}.
      Then there is a $V \eqaU U$ such that $V \goto{\alpha} U'$ is derived by a simpler proof.
      So $R\sigma \eqaU V$ and by induction \plat{$R \goto{\alpha} R'$} for some $\alpha$ and $R'$ with
      $\alpha\as{\sigma}\mathbin=\beta$ and $R'\sigma \eqaU U'$.
  \qed
  \end{itemize}
\end{proof}

\subsection{A \texorpdfstring{$\pi$}{pi}-calculus with surjective substitutions}\label{sec:surjective}
\hypertarget{piSur}{}

When needed I will address the $\pi$-calculus of \sect{pi} as $\pi(\N)$, to make the choice of the
set $\N$ of names explicit.\linebreak[4] Here I introduce a variant $\pi^\aN(\N)$ of the $\pi$-calculus
that additionally employs a countably infinite collection of \emph{spare names}
$\aN = \{s_1,s_2,\dots\}$, disjoint with $\N$.
Its syntax is the same as the one of $\pi(\N\uplus\aN)$, except that
in expressions $x(y).P$ one must take $y\in\N$.
Moreover, in defining equations
\plat{$A(\vec{x})\stackrel{\rm def}{=} P$} I require that $x_1,\dots,x_n \in \N$.
The definition of substitution on $\pi^\aN(\N)$ is fine-tuned by requiring that when a bound name
$y\in\N$ is changed into a bound name $z$ to avoid name capture, one always chooses $z \in \N$.

As semantics of $\pi^\aN(\N)$ I use the variant $\gotoU\alpha$ 
of the early symbolic transition relation \plat{$\goto\alpha_{ES}$} 
where the substitutions $\renb{z}{y}$ and $\renbt{y}{x}$ introduced by the operational rules
\hyperlink{ES}{\textsc{\textbf{\small early-input}} and \textsc{\textbf{\small ide}}}
are changed into surjective substitutions $\renb{z}{y}^\aN$ and $\renbt{y}{x}^\aN$.
The motivation for this is to make \lem{substitution reflection gen} applicable to these substitutions.
In order to preserve the validity of (\ref{A1})--(\ref{A17}), this change forces new
definitions of $\eqaU$, $\Fn$ and $\sigma[y\mapsto w]$.
\begin{definition}
Let $\M$ be the set of all names currently in use, for now that is $\N \cup \aN$.
For $\vec{x} = (x_1,\dots,x_n) \inp \N^n$
and $\vec{y} = (y_1,\dots,y_n) \in \M^n$, where the
names $x_i$ are all distinct, define the surjective substitution
$$\renbt{x}{y}^\aN\!:\aN \cup \{x_1,\dots,x_n\} \rightarrow \aN \cup \{x_1,\dots,x_n,y_1,\dots,y_n\}$$
by $\renbt{x}{y}^\aN(x_i) = y_i$ and
$\renbt{x}{y}^\aN(s_i) = x_i$ for $i=1,\dots,n$, and $\renbt{x}{y}^\aN(s_i) = s_{i-n}$ for $i>n$.
\end{definition}
This is in fact the simplest possible adaptation of the substitution $\renbt{y}{x}$ that makes it
surjective. It is essential for my purposes that the axioms
(\ref{A1})--(\ref{A17}) of \sect{substitutions}
be satisfied, so that Lemmas~\ref{lem:substitution preservation gen} and ~\ref{lem:substitution reflection gen}
hold for $\pi^\aN(\N)$. In view of Property (\ref{substitution composition}), to obtain (\ref{A3}) I need to ensure that
\begin{equation}
  \renb{z\as{\sigma}}{w}^\aN \circ \sigma[y\mapsto w] = \sigma\circ\renb{z}{y}^\aN. \label{A11}
\end{equation}
The definition of $\sigma[y\mapsto w]$ can be found by solving this equation.
In \tab{solving}, $\sigma\!:\M \rightharpoonup \M$ is an arbitrary substitution,
$y,w\in\N$, and $w \notin \dom(\sigma) \cup {\it range}(\sigma)$.
\begin{table}
\normalsize
\[\begin{array}{@{}l|c|c|c@{}}
\mbox{\small Name~}n & n\as{\renb{z}{y}^\aN} & n\as{\renb{z}{y}^\aN}\as{\sigma} = \mbox{}\\
&& n\as{\sigma[y\mapsto w]}\as{\renb{z\as{\sigma}}{w}^\aN} & n\as{\sigma[y\mapsto w]}\\
\hline
 y & z & z\as{\sigma} & w \\
 w \scriptstyle\hfill({\rm if}\neq y) & w & w         & s_1 \\
\begin{array}{@{}c@{}}s_1\\[-4pt]\scriptstyle{\rm if~}y= w \end{array} & w & w
   & s_1\\
\begin{array}{@{}c@{}}s_1\\[-4pt]\scriptstyle{\rm if~}y\neq w \end{array} & y & y\as\sigma~ \scriptstyle\neq w
   \hfill\begin{array}{@{\mbox{if~}}c@{}}y\as{\sigma} \in\N \\[-3pt] y\as{\sigma} = s_k\end{array}
   & \begin{array}{@{}c@{}}y\as{\sigma} \\[-3pt] s_{k+1}\end{array}\\
 s_{i+1} & s_i & s_i\as\sigma~\scriptstyle\neq w
   \hfill\begin{array}{@{\mbox{if~}}c@{}}s_i\as{\sigma} \in\N \\[-3pt] s_i\as{\sigma} = s_k\end{array}
   & \begin{array}{@{}c@{}}s_i\as{\sigma} \\[-3pt] s_{k+1}\end{array}\\
\begin{array}{@{}c@{}}x \in\N \\[-4pt] \scriptstyle x\neq y,w\end{array}
   & x & x\as{\sigma}~\scriptstyle\neq w
   \hfill\begin{array}{@{\mbox{if~}}c@{}}x\as{\sigma} \in\N \\[-3pt] x\as{\sigma} = s_k\end{array}
   & \begin{array}{@{}c@{}}x\as{\sigma} \\[-3pt] s_{k+1}\end{array}\\
\hline
\end{array}\]
\caption{Solving (\ref{A11})}
\label{tab:solving}
\vspace{-4ex}
\end{table}
Thus, $\sigma[y\mathbin{\mapsto} w]$ is defined as the substitution that sends a name from the first column to
the corresponding name from the last column. The middle columns show that now (\ref{A11}) is satisfied.

\begin{definition}\label{df:update}
For $x \in \M$ and $y \in \N$, let\\[1ex]
$x_y^- := \left\{\begin{array}{@{}ll@{}} y & \mbox{if}~x=s_1 \\
s_k & \mbox{if}~x=s_{k+1} \\
x   & \mbox{otherwise}\end{array}\right.$
\hfill and\hfill
$x_y^+ := \left\{\begin{array}{@{}ll@{}} s_1 & \mbox{if}~x=y \\
s_{k+1} & \mbox{if}~x=s_{k} \\
x   & \mbox{otherwise.}\end{array}\right.$\\[1ex]
Now $\sigma[y\mathbin{\mapsto} w]$, with $y,w\in\N$, is the substitution with
$\dom(\sigma[y\mathbin{\mapsto} w]) = \dom(\sigma) \cup \{y,w\} \cup \aN$ given by
\[x\as{\sigma[y\mathbin{\mapsto} w]} \mathbin{:=} \left\{\begin{array}{@{}ll@{}}
w & \mbox{if}~x=y \\
(x_y^-\as{\sigma})^+_w & \mbox{otherwise.}
\end{array}\right.\]
\end{definition}

As far as the computational interpretation of the $\pi$-calculus concerns,
the replacement of $P\renbt{y}{x}$ by $P\renbt{y}{x}^\aN$ in rule \textsc{\textbf{\small ide}} is of
no consequence, as the names from $\aN$ do not occur free in $P$ anyway.
However, the replacement of $P\renb{z}{y}$ by $P\renb{z}{y}^\aN$ in rule
\textsc{\textbf{\small early-input}} changes the meaning of the construct $x(y).P$.
In the original $\pi$-calculus there are no free occurrences of $y$ in $x(y).P$, and thus, after
receiving a name $z\neq y$, the resulting process $P\renb{z}{y}$ cannot do an action~$\goto{\bar yw}$.
Here, however, upon receiving a name $z\mathbin{\neq} y$, in the resulting process $P\renb{z}{y}^\aN$ the
spare name $s_1$ is elevated to $y$, so that $P\renb{z}{y}^\aN\goto{\bar yw}$ is possible.
For this reason, $y$ can be considered a free name of $x(y).P$ when $s_1$ is a free name of $P$.
More in general, a spare name $s_k$ can be seen as a potential name $y \inp \N\!$. This potential is
realised when $s_k$ occurs in the scope of $k$ input prefixes, and the choice of $y$ is made by the
outermost of these. Hence the definition of $\Fn$, partly given in \sect{substitutions}, is completed
by defining $\Fn(P)^-_y := \{x^-_y \mid x \in \Fn(P){\setminus}\{y\}\}$.

\df{update} matches this intuition. To apply $\sigma[y\mapsto w]$ to a name $x \neq y$,
first elevate $x$ one level, then apply $\sigma$, and finally undo the elevation.

Note that $\T_{\pi(\N)} \subseteq \T_{\pi^\aN(\N)}$, that is, the ordinary $\pi$-calculus processes
form a subset of the processes in the $\pi$-calculus with surjective substitutions.
A process $P\in \T_{\pi(\N)}$ by definition satisfies 
$\fn(P) \subseteq \N$. The following lemma implies that on $\T_{\pi(\N)}$ there is no difference
between $\Fn$ and $\fn$.

\begin{lemma}\label{lem:free names same}
  If $\fn(P) \subseteq \N$ then $\Fn(P)=\fn(P)$.
\end{lemma}
\begin{proof}
A trivial structural induction on $P$.
\end{proof}
I now proceed to show that the axioms (\ref{A1}), (\ref{free names input successors}) and
(\ref{free names sigma}) from \sect{substitutions}---the ones not mentioning $\eqaU$---are satisfied.

\begin{lemma}\label{lem:free names sigma}
$\Fn(P\sigma) = \{x\as{\sigma} \mid x \inp \Fn(P)\}$.
\end{lemma}
\begin{proof}
With induction on the size of the parse tree of $P$.
Note that this size is preserved under substitution.
All cases are trivial, except for the two detailed below.

Let $P \mathbin= x(y).Q$. Then
$P\sigma \mathbin= x\as{\sigma}(z).(Q\sigma[y \mapsto z])$,
so $\Fn(P\sigma) = \{x\as{\sigma}\} \cup \Fn(Q\sigma[y \mapsto z])^-_z$.
By induction, 
\[\begin{array}{l}
\Fn(Q\sigma[y \mapsto z]) = \{v\as{\sigma[y \mapsto z]} \mid v \inp \Fn(Q)\} =\\
\{z\mid y\in \Fn(Q)\} \cup \{(v^-_y\as{\sigma})^+_z \mid v \inp \Fn(Q){\setminus}\{y\}\},
\end{array}\]
so $\Fn(Q\sigma[y \mapsto z])^-_z =\{x^-_z \mid x \in \Fn(Q\sigma[y \mapsto z]){\setminus}\{z\}\} =
\{v^-_y\as{\sigma} \mid v \inp \Fn(Q){\setminus}\{y\}\}$.
Here I use that $u^+_z \neq z$ and $(u^+_z)^-_z = u$ for all names $u$.

Since $\Fn(P) = \{x\} \cup \{v^-_y \mid v \in \Fn(Q){\setminus}\{y\}\}$, it follows that
$\Fn(P\sigma) = \{x\as{\sigma} \mid x \inp \Fn(P)\}$.

Now let $P = (\nu y)Q$. Then $P\sigma = (\nu z)(Q\renb{z}{y}\sigma)$, so
$\Fn(P\sigma) = \Fn(Q\renb{z}{y}\sigma) {\setminus} \{z\}$.
By induction (applied twice) $\Fn(Q\renb{z}{y}\sigma) = \{v\as{\renb{z}{y}}\as{\sigma} \mid v \inp \Fn(Q)\}$.
Considering that $z\notin \Fn(Q) \cup \dom(\sigma) \cup {\it range}(\sigma)$,
$$\Fn(Q\renb{z}{y}\sigma) {\setminus} \{z\} = \{v\as{\sigma} \mid v \inp \Fn(Q){\setminus}\{y\}\}.$$
Since $\Fn(P) = \Fn(Q){\setminus}\{y\}$, again the claim follows.
\end{proof}
For recursion equations \plat{$A(\vec{x})\stackrel{\rm def}{=} P$} the $\pi$-calculus requires that
$\fn(P) \subseteq \{x_1,\dots,x_n\}$, and I didn't bother to change that for $\pi^\aN$.
Since I require that $x_1,\dots,x_n \in \N$, \lem{free names same} implies that also 
$\Fn(P) \subseteq \{x_1,\dots,x_n\}$. 
Using this, (\ref{A1}) immediately follows from (\ref{free names sigma}), which is \lem{free names sigma}.

\begin{lemma}
$\Fn(P\renb{z}{y}^\aN) \subseteq \Fn(x(y).P) \cup \{z\}$.
\end{lemma}
\begin{proof}
If $v \in \Fn(P\renb{z}{y}^\aN)$ then $v=x\as{\renb{z}{y}^\aN}$ for some $x\in\Fn(P)$,
using \lem{free names sigma}. If moreover $v\neq z$ then $v=x^-_y$.
Hence $v \in \Fn(P)^-_y \subseteq \Fn(x(y).P)$.
\end{proof}
Now I define $\eqaU$ in such a way that (\ref{free names alpha}) holds.

\begin{definition}\label{df:eqaU}
For $y,z\in\N$ with $y\neq z$, let $\renb{z}{y}^3$ be the substitution $\sigma$ with $\dom(\sigma)\mathbin=\{y,z,s_1\}$,
$\sigma(y)\mathbin=z$, $\sigma(z)\mathbin=s_1$ and $\sigma(s_1)\mathbin=y$.
Moreover, let $\renb{y}{y}^3 := \epsilon$, the substitution with $\dom(\epsilon)=\emptyset$.
Let $\eqaU$ be the smallest congruence satisfying
\begin{enumerate}[(i)]
\item $(\nu y)P \eqaU (\nu z)(P\renb{z}{y})$ for any $z\notin \Fn((\nu y)P)$,
\item and $x(y).P \eqaU x(z).(P\renb{z}{y}^3)$ for any $z \in \N$.
\end{enumerate}
\end{definition}
Note that in $x(y).P$ the spare name $s_1$ occurring in $P$ will be elevated to $y$,
but in $x(z).Q$ it will be elevated to $z$. Therefore, when renaming the
bound name $y$ into $z$, the spare name $s_1$ should be renamed into $y$ right away.
Moreover, there is no reason to require that $z$ not occur free in $P$; its free occurrences
can simply be renamed into $s_1$.\footnote{An alternative form of \df{eqaU} requires (ii)
only for $z \mathbin{\notin} \Fn((\nu y)P)$. This yields the same equivalence $\eqaU$ as \df{eqaU},
since even when  $z \in \Fn((\nu y)P)$ one derives
$x(y).P \eqaU x(w).(P\renb{w}{y}^3) \eqaU x(z).(P\renb{z}{y}^3)$ for some $w \notin \Fn((\nu y)P)$,
using (\ref{3-composition}), that $P\renb{z}{y}^3\renb{w}{z}^3 \eqaU P\renb{w}{y}^3$.}

\begin{lemma}\label{lem:free names alpha}
$P \eqaU Q \Rightarrow \Fn(P) = \Fn(Q)$.
\end{lemma}
\begin{proof}
Using transitivity of $\eqaU$, it suffices to prove this for the special case that $P$ and $Q$
differ by only one application of the generating equations from \df{eqaU}.
Below I deal with the special case that this single application does not occur within a proper
subterm of $P$; the general case then follows by a straightforward structural induction on $P$.
There are two possibilities to consider.

Let $P = x(y).R$ and $Q = x(z).(R\renb{z}{y}^3)$.
In this case $\Fn(P) = \{x\} \cup \Fn(R)^-_y$ and $\Fn(Q) = \{x\} \cup \Fn(R\renb{z}{y}^3)^-_z$.
Moreover, $\Fn(R)^-_y \mathbin= \{v^-_y\!\mid v \inp \Fn(R){\setminus}\{y\}\}$ and, by \lem{free names sigma},
\[\Fn(R\renb{z}{y}^3)^-_z = \{v\as{\renb{z}{y}^3}^-_z \mid v \inp \Fn(R) \wedge v\as{\renb{z}{y}^3} \neq z\}.\]
Note that $v=y$ iff $v\as{\renb{z}{y}^3}=z$.
So take $v\neq y$.
A simple case distinction shows that $v\as{\renb{z}{y}^3}^-_z = v^-_y$.
Consequently, $\Fn(R)^-_y = \Fn(R\renb{z}{y}^3)^-_z$ and thus $\Fn(P) =\Fn(Q)$.

Let $P \mathbin= (\nu y).R$ and $Q \mathbin= (\nu z).(R\renb{z}{y})$ with $z\mathbin{\notin} \Fn((\nu y)P)$.
Then $\Fn(P) = \Fn(R){\setminus}\{y\}$ and $\Fn(Q) = \Fn(R\renb{z}{y}){\setminus}\{z\}$.
By \lem{free names sigma}, $\Fn(R\renb{z}{y}) = \{v\as{\renb{z}{y}} \mid v \inp \Fn(R)\}$.
When $v\neq y,z$ one has $v = v\as{\renb{z}{y}} \neq z$, and when $v\as{\renb{z}{y}} \neq z$ then
$v\neq y,z$. Consequently, $\Fn(P) =\Fn(Q)$.
\end{proof}
The next two lemmas pave the way for \lem{substitution composition}, which establishes the validity
of Properties (\ref{substitution composition}) and (\ref{A10}).

\begin{lemma}\label{lem:update composition}
  Let $y,w,z\in \N$,
  Then, for all $x \in \M$,
$$x\as{\sigma_1[y \mathbin{\mapsto} w]} \as{\sigma_2[w \mathbin{\mapsto} z]} \mathbin= x\as{(\sigma_2 \circ \sigma_1)[y \mathbin{\mapsto} z]}.$$
\end{lemma}
\begin{proof}
By \df{update}, $x \as{\sigma_1[y \mathbin{\mapsto} w]}=w$ only if $x=y$. So
\[x \as{\sigma_1[y \mathbin{\mapsto} w]} \as{\sigma_2[w \mathbin{\mapsto} z]} = \left\{\begin{array}{@{}ll@{}}
z & \mbox{if}~x=y \\
((x_y^-\as{\sigma_1})^+_w)^-_w\as{\sigma_2}^+_z & \mbox{otherwise.}
\end{array}\right.\]
Moreover,
\[x\as{(\sigma_2 \circ \sigma_1)[y \mathbin{\mapsto} z]} = \left\{\begin{array}{@{}ll@{}}
z & \mbox{if}~x=y \\
(x_y^-\as{\sigma_1}\as{\sigma_2})^+_z & \mbox{otherwise.}
\end{array}\right.\]
Since $(v^+_w)^-_w = v$ for all $v \in \M$, the lemma follows.
\end{proof}

\begin{lemma}\label{lem:subs}
  Let $y,w,z\inp \N$. 
  Then, for all $x \inp \M$, $$x\as{\sigma[y \mapsto w]}\as{\renb{z}{w}^3} = x\as{\sigma[y \mapsto z]}$$
  $$x\as{\renb{w}{y}^3}\as{\sigma[w \mapsto z]} = x\as{\sigma[y \mapsto z]}.$$
\end{lemma}
\begin{proof}
If $x\mathbin=y$ then $x\as{\sigma[y \mapsto w]}\as{\renb{z}{w}^3} = z = x\as{\sigma[y \mapsto z]}$.
Otherwise, considering that $v^+_w\as{\renb{z}{w}^3} = v^+_z$ for any $v\in \M$,
$x\as{\sigma[y \mapsto w]}\as{\renb{z}{w}^3} = (x_y^-\as{\sigma})^+_w \as{\renb{z}{w}^3}
= (x_y^-\as{\sigma})^+_z =x\as{\sigma[y \mapsto z]}$.
The second statement follows likewise, using that $x\as{\renb{w}{y}^3}^-_w = x^-_y$ when $x \neq y$,
\end{proof}

\begin{lemma}\label{lem:substitution composition}
\begin{enumerate}
\item  $(P\sigma_1)\sigma_2 \eqaU P(\sigma_2 \circ \sigma_1)$.
\item  $(\forall x\inp\Fn(P).\,x\as{\sigma} \mathbin= x\as{\sigma'}) \Rightarrow P\sigma\eqaU P\sigma'$.
\end{enumerate}
\end{lemma}
\begin{proof}
I prove both statements by simultaneous structural induction on $P$.
\begin{enumerate}
\item
All cases are trivial, except for the ones where $P$ has the form $x(y).Q$ or $(\nu y)Q$.

Let $P = x(y).Q$. By the definition of substitution, there is a $u$
such that the last step in the below derivation holds.
Likewise, there are $w$ and $z$ such that the first two steps hold.
\[(P\sigma_1)\sigma_2 \begin{array}[t]{@{~}l@{~}l@{}}
  = & \big(x\as{\sigma_1}(w).(Q\sigma_1[y\mapsto w])\big)\sigma_2\\
  = & x\as{\sigma_1}\as{\sigma_2}(z).\big((Q\sigma_1[y\mathbin{\mapsto} w])\sigma_2[w\mathbin{\mapsto} z]\big)\\
  \eqaU & x\as{\sigma_2\mathop{\circ}\sigma_1}(z).
     \big(Q(\sigma_2[y\mathbin{\mapsto} w]\mathop\circ\sigma_1[w\mathbin{\mapsto} z])\big)\\
  \eqaU & x\as{\sigma_2\mathop{\circ}\sigma_1}(z).\big(Q((\sigma_2\mathop\circ\sigma_1)[y\mathbin{\mapsto} z])\big)\\
  \eqaU & x\as{\sigma_2\mathop{\circ}\sigma_1}(u).\big(Q((\sigma_2\mathop\circ\sigma_1)[y\mathbin{\mapsto} z]\renb{u}{z}^3)\big)\\
  \eqaU & x\as{\sigma_2\mathop{\circ}\sigma_1}(u).\big(Q((\sigma_2\mathop\circ\sigma_1)[y\mathbin{\mapsto} u])\big)\\
  = & P(\sigma_2\circ\sigma_1).
\end{array}\]
Here the third step is an application of the induction hypotheses regarding statement 1) and
the fourth step applies the induction hypotheses regarding statement 2),
also using \lem{update composition}.
The fifth step applies \df{eqaU} of $\eqaU$.
The sixth step applies the induction hypotheses regarding both statements, also using \lem{subs}.

Let $P = (\nu y)Q$.  By the definition of substitution, there is an $u$ 
such that the last step in the below derivation holds.
Likewise, there are $w$ and $z$ such that the first two steps hold.
Now choose a name $v$ outside\\ $\Fn\big(Q\renb{w}{y}\sigma_1\renb{z}{w}\sigma_2\big) \cup
\Fn\big(Q\renb{u}{y}(\sigma_2\circ\sigma_1)\big) \cup\{z,u\}$.
\[(P\sigma_1)\sigma_2 \begin{array}[t]{@{~}l@{~}l@{}}
  = & ((\nu w)(Q\renb{w}{y}\sigma_1))\sigma_2 \\
  = & (\nu z)(Q\renb{w}{y}\sigma_1\renb{z}{w}\sigma_2) \\
  \eqaU & (\nu v)(Q\renb{w}{y}\sigma_1\renb{z}{w}\sigma_2\renb{v}{z}) \\
  \eqaU & (\nu v)(Q\renb{u}{y}(\sigma_2\circ\sigma_1)\renb{v}{u}) \\
  \eqaU & (\nu u)(Q\renb{u}{y}(\sigma_2\circ\sigma_1)) \\
  = & P(\sigma_2\circ\sigma_1).
\end{array}\]
The third and fifth step apply \df{eqaU} of $\eqaU$.
The fourth applies the induction hypotheses regarding both statements, also using that, for all $x\in\Fn(Q)$,
\[x \as{\renb{w}{y}}\as{\sigma_1}\as{\renb{z}{w}}\as{\sigma_2}\as{\renb{v}{z}} \mathbin= 
  x\as{\renb{u}{y}}\as{\sigma_2\circ\sigma_1}\as{\renb{v}{u}}.\]
The latter follows by a straightforward case distinction, considering that
$w \notin \Fn(Q){\setminus}\{y\} \cup \dom(\sigma_1) \cup {\it range}(\sigma_1)$,
that $z \notin \Fn(Q\renb{w}{y}\sigma_1){\setminus}\{w\} \cup \dom(\sigma_2) \cup {\it range}(\sigma_2)$
and that $u \notin \Fn((\nu y)Q) \cup \dom(\sigma_2\circ\sigma_1) \cup {\it range}(\sigma_2\circ\sigma_1)$.
\item
Again all cases are trivial, except for the ones that $P$ has the form $x(y).Q$ or $(\nu y)Q$.

Let $P \mathbin= x(y).Q$ and $u\as{\sigma} \mathbin= u\as{\sigma'}$ for all $u\inp\Fn(P)$.
Then, for certain $z,w\in\N$, $P\sigma = x\as{\sigma}(z).(Q\sigma[y\mathbin{\mapsto} z])$
and $P\sigma' = x\as{\sigma'}(w).(Q\sigma'[y\mathbin{\mapsto} w])$.
By \df{eqaU}, $P\sigma' \eqaU x\as{\sigma'}(z).(Q\sigma'[y\mathbin{\mapsto} w]\renb{z}{w}^3)$.
By the induction hypotheses and \lem{subs},
$P\sigma' \mathbin{\eqaU} x\as{\sigma'}(z).(Q\sigma'[y\mathbin{\mapsto} z])$.

By \lem{free names actions gen}, $x \in \Fn(P)$, and thus $x\as{\sigma} = x\as{\sigma'}$.
To round of the proof by another appeal to the induction hypotheses regarding statement 2),
it suffices to show that $u\as{\sigma[y\mathbin{\mapsto} z]} = u\as{\sigma'[y\mathbin{\mapsto} z]}$
for all $u \inp \Fn(Q)$. This amounts to showing that
$(u_y^-\as{\sigma})^+_z = (u_y^-\as{\sigma'})^+_z$ for all $u \inp \Fn(Q){\setminus}\{y\}$.
Now for each $u \inp \Fn(Q){\setminus}\{y\}$ one has $u^-_y \in \Fn(P)$, by the definition of 
$\Fn(Q)^-_y$. Hence $u_y^-\as{\sigma} = u_y^-\as{\sigma'}$, and the proof of this case is done.

Let $P = (\nu y)Q$ and let $u\as{\sigma} \mathbin= u\as{\sigma'}$ for all $u\inp\Fn(P)$.
Then $P\sigma \mathbin= (\nu z)(Q\renb{z}{y}\sigma)$
and $P\sigma' \mathbin= (\nu w)(Q\renb{w}{y}\sigma')$ for certain names $z,w$.
Pick $v \notin \Fn(P\sigma)\cup\{z,w\}\cup\dom(\sigma)\cup\dom(\sigma')$.
Then, by \df{eqaU}, $P\sigma \eqaU (\nu v)(Q\renb{z}{y}\sigma\renb{v}{z})$
and $P\sigma' \mathbin{\eqaU} (\nu v)(Q\renb{w}{y}\sigma'\renb{v}{w})$.
Now $u\as{\renb{z}{y}}\as{\sigma}\as{\renb{v}{z}} = u\as{\renb{v}{y}}\as{\sigma}$ for all $u \in \Fn(Q)$,
using that $z \mathbin{\notin}\Fn(P) \cup \dom(\sigma) \cup {\it range}(\sigma)$ and $v \mathbin{\notin}\dom(\sigma)$.
Hence, by induction, $Q\renb{z}{y}\sigma\renb{v}{z} \eqaU Q\renb{v}{y}\sigma$.
Likewise, $Q\renb{w}{y}\sigma'\renb{v}{w} \eqaU Q\renb{v}{y}\sigma'$.
By \lem{free names sigma} $\Fn(Q\renb{v}{y}){\setminus}\{v\} \subseteq \Fn(P)$, so, again by
induction, $Q\renb{v}{y}\sigma \eqaU Q\renb{v}{y}\sigma'$, and hence $P\sigma\eqaU P\sigma'$.
\qed
\end{enumerate}
\end{proof}
Property (\ref{alpha nu}) follows immediately from \df{eqaU}, and 
(\ref{A3}) follows from (\ref{substitution composition}) (i.e.\ \lem{substitution composition}.1) and (\ref{A11}).
Also (\ref{recursion sigma}) follows from \lem{substitution composition}, since $\Fn(P)\subseteq\{x_1,\dots,x_n\}$.
Hence it remains to verify Properties (\ref{alpha sigma})--(\ref{A17}).

\begin{lemma}
$P \eqaU Q \Rightarrow P\sigma\eqaU Q\sigma$.
\end{lemma}
\begin{proof}
Using transitivity of $\eqaU$, it suffices to prove this for the special case that $P$ and $Q$
differ by only one application of the generating equations from \df{eqaU}.
Below I deal with the special case that this single application does not occur within a proper
subterm of $P$; the general case then follows by a structural induction on the size of the parse
tree of $P$, recalling that this size is preserved under substitution.

There are two possibilities to consider.

Let $P = x(y).R$ and $Q = x(z).(R\renb{z}{y}^3)$. In this case
$P\sigma \mathbin= x\as{\sigma}(w).(R\sigma[y \mathbin{\mapsto} w])$,
$Q\sigma \mathbin= x\as{\sigma}(u).((R\renb{z}{y}^3)\sigma[z \mathbin{\mapsto} u])$.
By \df{eqaU} $Q\sigma \eqaU x\as{\sigma}(w).(R\renb{z}{y}^3\sigma[z \mathbin{\mapsto} u]\renb{w}{u}^3)$,
so it suffices to show that
$$R\sigma[y \mathbin{\mapsto} w] \eqaU R\renb{z}{y}^3\sigma[z \mathbin{\mapsto} u]\renb{w}{u}^3,$$
which follows from Lemmas~\ref{lem:subs} and~\ref{lem:substitution composition}.

Let $P = (\nu y)R$ and $Q = (\nu z)(R\renb{z}{y})$ with $z \notin \Fn(R)$.
Then $P\sigma = (\nu w)(R\renb{w}{y}\sigma)$, $Q\sigma = (\nu u)(R\renb{z}{y}\renb{u}{z}\sigma)$
for names
$w,u$ with $w,u\notin\Fn(P) \cup \dom(\sigma) \cup {\it range}(\sigma)$ (using that
$\Fn(P)=\Fn(Q)$, by \lem{free names alpha}).
Choose a fresh name $v \mathbin{\notin} \Fn(R\renb{w}{y}\sigma) \cup \Fn(R\renb{z}{y}\renb{u}{z}\sigma) \cup \{w,u\}$.
Then $P\sigma \eqaU (\nu v)(R\renb{w}{y}\sigma\renb{v}{w})$ and $Q\sigma \mathbin{\eqaU} (\nu v)(R\renb{z}{y}\renb{u}{z}\sigma\renb{v}{u})$.
It suffices to show that $R\renb{w}{y}\sigma\renb{v}{w} \mathbin\eqaU R\renb{z}{y}\renb{u}{z}\sigma\renb{v}{u}$.
This follows from \lem{substitution composition}, provided that, for all $x \inp\Fn(R)$,
$$x\as{\renb{w}{y}}\as{\sigma}\as{\renb{v}{w}} = x\as{\renb{z}{y}}\as{\renb{u}{z}}\as{\sigma}\as{\renb{v}{u}}.$$
The latter is established by a simple case distinction.
\end{proof}

\begin{lemma}
  Let $y \in \N$ and $w \notin \Fn(x(y).P)$.\\
  Then $P\renb{w}{y}^\aN\renb{z}{w} \eqaU P\renb{z}{y}^\aN$.
\end{lemma}
\begin{proof}
  Since $w \notin \Fn(P)^-_y$, one has $w^+_y \notin \Fn(P)$.
  Now the statement follows from \lem{substitution composition}, considering that, for all $x \inp\Fn(P)$,
  $x\as{\renb{w}{y}^\aN}\as{\renb{z}{w}} = x\as{\renb{z}{y}^\aN}$.
\end{proof}

Property (\ref{empty substitution}) follows by a straightforward structural induction on $P$, using
\lem{substitution composition} for the cases that $P=x(y).Q$ or $P=(\nu y)Q$.

Using \lem{substitution composition} it is trivial to check, for $y,w,z\in\N$, that
\begin{equation}\label{3-composition}
  P\renb{w}{y}^3\renb{z}{w}^3 \eqaU P\renb{z}{y}^3
\end{equation}
\begin{equation}\label{3S-composition}
  P\renb{w}{y}^3\renb{u}{w}^\aN \eqaU P\renb{u}{y}^\aN.
\end{equation}

\begin{lemma}\rm\label{lem:outside first}
If $U \eqaU x(y).P$ then $U = x(w).V$ for some $V\eqaU P\renb{w}{y}^3$.
Likewise, if $U \eqaU (\nu y)P$ then $U = (\nu w)V$ for some $w \notin \Fn(U)$ and $V\eqaU P\renb{w}{y}$.
\end{lemma}
\begin{proof}
  Since $U \eqaU x(y).P$, one has $U = U_0 \eqaU U_1 \eqaU U_2 \eqaU \cdots \eqaU U_n = x(y).P$,
  and each step $U_i\eqaU U_{i+1}$ involves exactly one application of the generating equations from \df{eqaU}.
  I prove the statement by induction on $n$.

  The case that $n=0$ is trivial; take $w:=y$.

  Let $n\mathbin>0$. By induction $U_1\mathbin= x(w).V$ with $V\eqaU P\renb{w}{y}^3$.

  If for $U = U_0 \eqaU U_1 = x(w).V$ the application of the generating equation from \df{eqaU} occurs entirely
  within $V$, one has $U = x(w).W$ with $W \eqaU V \eqaU P\renb{w}{y}^3$.

  If the application is at top level, then $U = x(z).W$ and either $W \mathbin= V\renb{z}{w}^3$ or $V = W\renb{w}{z}^3$.
  By (\ref{3-composition}), (\ref{empty substitution}), in either case $W \mathbin\eqaU V\renb{z}{w}^3$.
  Thus $W \eqaU P\renb{z}{y}^3$ by (\ref{3-composition}).

  The second statement is obtained in the same way.
\end{proof}
The last property to be checked, (\ref{A16}), follows from \lem{outside first}.
\begin{corollary}
If $U \eqaU (x(y).P)\sigma$ then $U=x\as{\sigma}(w).V$ with $V \eqaU P\sigma[y\mapsto w]$.
\end{corollary}
\begin{proof}
By definition $(x(y).P)\sigma = x\as{\sigma}(z).(P\sigma[y \mapsto z])$ for some $z\in\N$.
By \lem{outside first} $U = x(w).V$ for some $V\eqaU P\sigma[y \mathbin{\mapsto} z]\renb{\hspace{-.5pt}w}{z\hspace{-.5pt}}^3\!$.
By Lemmas~\ref{lem:subs} and~\ref{lem:substitution composition} $V\mathbin\eqaU P\sigma[y\mathbin{\mapsto} w]$.
\end{proof}

Since the axioms (\ref{A1})--(\ref{A17}) of \sect{substitutions} hold for $\pi^\aN(\N)$,
so do Lemmas~\ref{lem:free names actions gen}--\ref{lem:substitution reflection gen} from \sect{substitutions}.

As mentioned earlier, the classical $\pi$-calculus $\pi(\N)$ can be seen as a subcalculus of the
$\pi$-calculus with surjective substitutions $\pi^\aN(\N)$; its processes $P$ satisfy $\fn(P)\subseteq\N$.
\lem{free names same} shows that on this subcalculus one has $\fn(P)=\Fn(P)$, so that $\Fn$ can be
seen as an extension of $\fn$ from $\pi(\N)$ to all of $\pi^\aN(\N)$.

Likewise, the application $P\sigma$ of a substitution $\sigma$ to a process $P$ is unchanged up to
$\eqa$ when $\fn(P)\subseteq\N$ and $\sigma\!:\M \rightharpoonup \N$.
Namely, if $\fn(Q)\subseteq \N$, $w \notin \Fn((\nu y)Q) \cup \dom(\sigma) \cup {\it range}(\sigma)$
and $\sigma\!:\M \rightharpoonup \N$, then $Q\sigma[y\mapsto w] \eqa Q\renb{w}{y}\sigma$.

The next lemma shows that also $\eqaU$ and $\eqa$ coincide on the subcalculus $\pi(\N)$ of $\pi^\aN(\N)$.

\begin{lemma}\label{lem:alpha same}
  Let $\fn(P|Q)\subseteq \N$. 
  Then $P \eqaU Q \Leftrightarrow P \eqa Q$.
\end{lemma}
\begin{proof}
  ``$\Rightarrow$'': It suffices to consider the case that $P$ and $Q$
  differ by only one application of the generating equations from \df{eqaU}.
Below I deal with the special case that this single application does not occur within a proper
subterm of $P$; the general case then follows by structural induction on $P$.

The case that $P = (\nu y)R$ and $Q = (\nu z)(R\renb{z}{y})$ with $z \notin \Fn((\nu y)R)$ is trivial.

So let $P \mathbin= x(y).R$ and $Q \mathbin= x(z).(R\renb{z}{y}^3)$.
As $\fn(P)\mathbin\subseteq\N$ and $\fn(Q)\mathbin\subseteq\N$,
also $\fn(R)\subseteq \N$ and $\fn(R\renb{z}{y}^3)\subseteq \N$, so $z \notin \fn((\nu y)R)$.
Considering that $v\as{\renb{z}{y}^3} = v\as{\renb{z}{y}}$ for all $v \in \Fn(R)$,
and that the classical notion of substitution may be applied here,
$R\renb{z}{y}^3 \eqa R\renb{z}{y}$. Hence $P \eqa Q$.

``$\Leftarrow$'': It suffices to consider the case that $P$ and $Q$
  differ by only one application of the generating equations from \sect{pi}.
Below I deal with the special case that this single application does not occur within a proper
subterm of $P$; the general case then follows by structural induction on $P$.

The case that $P = (\nu y)R$ and $Q = (\nu z)(R\renb{z}{y})$ with $z \notin \fn((\nu y)R)$ is trivial.

So let $P = x(y).R$ and $Q = x(z).(R\renb{z}{y})$ with $z \notin \fn((\nu y)R)$.
In that case $\Fn(R)=\fn(R)\subseteq \N$.
Considering that $v\as{\renb{z}{y}} = v\as{\renb{z}{y}^3}$ for all $v \in \Fn(R)$,
\lem{substitution composition}.2 yields
$R\renb{z}{y} \eqaU R\renb{z}{y}^3$. Hence $P \eqaU Q$.
\end{proof}

Now I will show that the identity $\mbox{id}\!: \T_{\piES} \rightarrow \T_{\piSur}$ is a valid
translation from  $\piES$ to $\piSur$.
My definition of $R \goto\alpha_{ES} R'$ implies that $R,R' \inp \T_{\piES}$ and $\n(\alpha)\subseteq \N$.

\begin{lemma}\label{lem:ES2Sur}
  If $R \goto\alpha_{ES} R'$ then equally
  $R \gotoU\alpha U'$ for some $\piSur$ process $U'$ with $U' \eqaU R'$.
\end{lemma}
\begin{proof}
  With induction on the inference of $R \mathbin{\goto\alpha_{ES}} R'$.
  \begin{itemize}
  \item  Suppose $R \goto{\alpha}_{ES} R'$ is derived by rule \hyperlink{ES}{\textsc{\textbf{\small early-input}}}.
      Then $R= x(y).P$, $\alpha = xz$ and $R'=P\renb{z}{y}$.
      Using \textsc{\textbf{\small early-input}} from $\piSur$,
      $R \mathbin{\gotoU\alpha} U':= P\renb{z}{y}^\aN$.
      By \lem{free names same} $\Fn(P)\subseteq \N$.
      So by (\ref{A10}) $U' \eqaU R'$. 
   \item  Suppose $R \goto{\alpha}_{ES} R'$ is derived by \hyperlink{ES}{\textsc{\textbf{\small ide}}}.
      Then $R\mathbin=A(\vec y)$.
      Let \plat{$A(\vec x) \mathbin{\stackrel{{\rm def}}{=}} P$}. 
      Then $P\renb{\vec y}{\vec x}\goto{\alpha}_{ES} R'\!$.
      So by induction $P\renb{\vec y}{\vec x} \mathbin{\gotoU\alpha} U'$
      for some $U' \eqaU R'$.
      As $\fn(P)\subseteq \N$,
      $P\renb{\vec y}{\vec x} \eqaU P\renb{\vec y}{\vec x}^\aN$ by (\ref{A10}) and \lem{free names same}.
      Hence $P\renb{\vec y}{\vec x}^\aN\gotoU{\alpha} U'$ by \textsc{\textbf{\small alpha}}.
      Thus $R \gotoU{\alpha} U'$ by rule \textsc{\textbf{\small ide}} from $\piSur$.
  \item All other cases are trivial.
      \qed
  \end{itemize}
  \end{proof}

\begin{lemma}\label{lem:Sur2ES}
  If $R \in\T_{\piES}$ and $R\eqaU U \gotoU\alpha U'$ with $\ia(\alpha)\cup\bn(\alpha)\mathbin\subseteq\N$,
  then equally $R \goto\alpha_{ES} R'$ for some $R' \mathbin{\eqaU} U'\!$.
\end{lemma}
\begin{proof}
  With induction on the $\alpha$-depth of the inference of $U \gotoU{\alpha}U'$, with a nested
  induction on the number of applications of rule \hyperlink{ES}{\textsc{\textbf{\small alpha}}} at
  the end of that inference.
  \begin{itemize}
    \item Suppose $U \gotoU{\alpha} U'$ is derived by rule \textsc{\textbf{\small early-input}}.
      Then $R=x(y).P$ and $U = x(w).V$ with $V \eqaU P\renb{w}{y}^3$ by \lem{outside first}.
      Moreover, $\alpha = xz$ and $y,w,z \mathbin\in\N$. So
      $U'\!\mathbin=V \renb{z}{w}^\aN \eqaU P\renb{w}{y}^3\renb{z}{w}^\aN \mathbin{\eqaU} P\renb{z}{y}^\aN$
      by (\ref{3S-composition}).

      By \hyperlink{ES}{\textsc{\textbf{\small early-input}}} from $\piES$,
      $R \mathbin{\goto\alpha_{ES}} R':= P\renb{z}{y}$.
      By \lem{free names same} $\Fn(P)\subseteq \N$.
      By (\ref{A10}) $R'\mathbin{\eqaU} P\renb{z}{y}^\aN \mathbin{\eqaU} U'\!$.
  \item  Suppose $U \gotoU{\alpha} U'$ is derived by rule \textsc{\textbf{\small ide}}.
      Then $R=U\mathbin=A(\vec y)$.
      Let \plat{$A(\vec x) \mathbin{\stackrel{{\rm def}}{=}} P$}.
      Then $P\renb{\vec y}{\vec x}^\aN \gotoU{\alpha} U'$.
      As $P\in T_{\piES}$, and $\{y_1,\dots,y_n\} = \fn(R) \subseteq \N$, also $P\renb{\vec y}{\vec x}\mathbin\in T_{\piES}$.
      Using (\ref{A10}), $P\renb{\vec y}{\vec x} \eqaU P\renb{\vec y}{\vec x}^\aN$.
      So by induction $P\renb{\vec y}{\vec x} \goto\alpha_{ES} R'$ for some $R' \eqaU U'$.
      Consequently, $R \goto{\alpha}_{ES} R'$ by rule \hyperlink{ES}{\textsc{\textbf{\small ide}}} from $\piES$.
  \item  Suppose $U \gotoU{\alpha} U'$ is derived by \textsc{\textbf{\small e-s-close}}.
      Then $R = P | Q$, $U = V | W$, $P \eqaU V \gotoU{M\bar x(z)} V'$,
      $Q \eqaU W \gotoU{Nvz} W'$, $\alpha=\match{x}{v}MN\tau$ and $U'=(\nu z)(V'|W')$ for some $z\notin\Fn(Q)$.
      Pick $w \in \N{\setminus}\Fn(R)$.  Now $V\gotoU{M\bar x(w)}\eqaU V'\renb{w}{z}$ and
      $W\gotoU{Nvw}\eqaU W'\renb{w}{z}$ by Lemmas~\ref{lem:bound output universality gen}
      and~\ref{lem:input universality 2 gen}.
      Moreover, the inferences of these transitions have a smaller $\alpha$-depth than that of $R \gotoU\alpha U'$.
      So $P\goto{M\bar x(w)}_{ES} P'\eqaU V'\renb{w}{z}$ and
      $Q\goto{Nvw}_{ES} Q'\eqaU W'\renb{w}{z}$ by induction.
      Applying rule \hyperlink{ES}{\textsc{\textbf{\small e-s-close}}},
      $R \mathbin{\goto{\alpha}_{ES}} R' := (\nu w)(P'|Q')$.
      \lem{free names successors gen} yields $\Fn(R')\subseteq\Fn(R)\not\ni w$.
      Using this, one obtains
      $R' \eqaU (\nu w)(V'\renb{w}{z}|W'\renb{w}{z}) \eqaU (\nu z)(V'|W') \mathbin= U'\!$.
  \item All other cases are trivial. In the case of \hyperlink{ES}{\textsc{\textbf{\small e-s-com}}}
      apply \lem{free names actions gen} to obtain $y\in\N$.
  \qed
  \end{itemize}
\end{proof}

\begin{table*}
\caption{Early symbolic structural operational semantics of the $\pi$-calculus without rule \textsc{\textbf{\small alpha}}}
\label{tab:pi-early-symbolic-noalpha-surjective}
\hypertarget{alpha}{}
\normalsize
\begin{center}
\framebox{$\begin{array}{@{}c@{\hspace{-1pt}}c@{\hspace{-3pt}}c}
\transname[13]{tau}{25}{         \textcolor{DarkBlue}{M}\tau.P \goto{\textcolor{DarkBlue}{M}\tau} P }&
\transname[13]{output}{38}{  \textcolor{DarkBlue}{M}\bar x y.P \goto{\textcolor{DarkBlue}{M}\bar x y} P }&
\hspace{-30pt}
\transname[13]{early-input}{55}{  \textcolor{DarkBlue}{M} x(y).P \goto{\textcolor{DarkBlue}{M}xz} P\renb{z}{y}^\aN }\hspace{-24pt}\\[2ex]
\transname[13]{sum}{20}{
  \displaystyle\frac{P \goto{\alpha} P'}{P+Q \goto{\alpha} P'} }&
\hspace{-10pt}
\textcolor{purple}{\transname[13]{symb-match}{45}{
  \displaystyle\frac{P \goto{\alpha} P'}{\Match{x}{y}P \goto{\match{x}{y}\alpha} P'} }}\hspace{-20pt}&
\hspace{-20pt}
\transname[13]{ide}{25}{
  \displaystyle\frac{P\{\rename{\vec{y}}{\vec{x}}\}^\aN \goto{\alpha} P'}{A(\vec{y}) \goto{\alpha} P'}
  ~~(A(\vec{x})\stackrel{\rm def}{=} P) }\hspace{-20pt}\\[4ex]
\transname{par}{8}{
  \displaystyle\frac{P\goto{\alpha} P'}{P|Q \goto{\alpha} P'|Q}~~\textcolor{black}{(\bn(\alpha)\cap\Fn(Q)=\emptyset) }}&
\hspace{-75pt}
\transname{e-s-com}{35}{
  \displaystyle\frac{P\goto{M\bar xy} P' ,~ Q \goto{Nvy} Q'}{P|Q \goto{\match{x}{v}MN\tau} P'| Q'} }
  \hspace{-75pt}&
\textcolor{black}{\transname{e-s-close}{35}{
  \displaystyle\frac{P\goto{M\bar x(z)} P' ,~ Q \goto{Nv z} Q'}{P|Q \goto{\match{x}{v}MN\tau} (\nu z)(P'| Q')}~~(z\mathbin{\notin}\Fn(Q)) }}  \hspace{-30pt}\\[4ex]
\textcolor{black}{\transname{res-alpha}{18}{
  \displaystyle\frac{P\drenb{z}{y} \goto{\alpha} P'}{(\nu y)P \goto{\alpha} (\nu z)P'}~~
  \left(\begin{array}{@{}l@{}} 
                               z\mathbin{\not\in} \Fn((\nu y)P)  \\
                               z\not\in \n(\alpha)\end{array}\right)}}\hspace{-65pt}&
\hspace{40pt}
\transname{symb-open-alpha}{58}{
  \displaystyle\frac{P \goto{M \bar x y} P'}{(\nu y)P \goto{M \bar x(z)} P'\drenb{z}{y}}~~
  \left(\begin{array}{@{}l@{}} y \neq x \\ z\mathbin{\notin} \Fn((\nu y)P) \\ 
                               y \notin \n(M)\end{array}\right)}\hspace{-245pt}&
\end{array}$}\vspace{-1pt}
\end{center}
\vspace{-1ex}
\end{table*}

\begin{theorem}\label{step2}
$P \sbb \mbox{id}(P)$ for any $P \in \T_{\piES}$.
\end{theorem}
\begin{proof}
  It suffices to show that the symmetric closure of 
  \[{\R}:=\{(R,U) \mid R \in \T_{\piES}, U \in \T_{\piSur} \wedge R \eqaU U\}\]
  is a strong barbed bisimulation. Let $(R,U)\in{\R}$.
 
  Let $R\goto\tau_{ES} R'$.
  By \lem{ES2Sur}, $R\gotoU\tau U'$ for some $U' \eqaU R'$.
  So $U\gotoU\tau U'$ by rule \textsc{\textbf{\small alpha}} of $\piSur$.

  Let $U\gotoU\tau U'$.
  By \lem{Sur2ES}, $\exists R'.~R\goto\tau_{ES} R' \eqaU U'$.

  Let $U{\downarrow_b}$ with $b\inp \nN\cup\overline\nN$.
  Then $U\gotoU{by} U'$ or $U\gotoU{b(y)} U'$ for some $y$ and $U'$, using the definition of $O$ in \sect{barbed}.
  By Lemmas~\ref{lem:input universality pre 2 gen} and~\ref{lem:bound output universality gen} I may assume,
  without loss of generality, that $y \in \N$ in case of input or bound output actions.
  So $R\goto{by}_{ES} R'$ or $R\goto{b(y)}_{ES} R'$ by \lem{Sur2ES}.
  Thus $R{\downarrow_b}$.

  The implication $R{\downarrow_b} \Rightarrow U{\downarrow_b}$ proceeds likewise.
\end{proof}

\subsection{The elimination of \textsc{\textbf{\small alpha}}}\label{sec:alpha3}
\newcommand{\gotoNU}[2]{\mathrel{\plat{$\goesto{#1}$}^{\aN}_{#2}}}
\newcommand{\gotoN}[1]{\gotoNU{#1}{\bullet}}
\newcommand{\gotoRN}[2]{\mathrel{{-}#1{\gotoN{#2}}}}

Let $\gotoN{\alpha}$ be the transition relation on $\T_\pi$ generated by the rules of
\tab{pi-early-symbolic-noalpha-surjective}.
Here $\drenb{z}{y}$ denotes the substitution $\sigma$ with $\dom(\sigma)=\{y,z\}$, $\sigma(y)=z$ and $\sigma(z)=y$.%
\footnote{In rule \hyperlink{alpha}{\textsc{\textbf{\scriptsize res-alpha}}},
  thanks to the side condition $z\not\in \Fn((\nu y)P)$, using (\ref{A10}), $P\drenb{z}{y} \mathbin{\eqaU} P\renb{z}{y}$.
  For this reason, the versions with $\drenb{z}{y}$ and $\renb{z}{y}$ are equivalent, up to
  strong bisimilarity.  The use of $\drenb{z}{y}$ instead of $\renb{z}{y}$ makes the substitution
  surjective, which will be needed in \sect{piRI}.

  The same can be said about \hyperlink{alpha}{\textsc{\textbf{\scriptsize symb-open-alpha}}},
  since $\Fn((\nu y)P') \subseteq \Fn((\nu y)P)$ by \lem{free names successors gen}.}
Compared to the operational semantics of $\piSur$, rule \hyperlink{ES}{\textsc{\textbf{\small alpha}}} is omitted, but
applications of this rule are incorporated in rules \hyperlink{alpha}{\textsc{\textbf{\small res-alpha}} and
\textsc{\textbf{\small symb-open-alpha}}}.\footnote{Besides the insignificant difference of $\drenb{z}{y}$ versus $\renb{z}{y}$,
  rule \hyperlink{alpha}{\textsc{\textbf{\scriptsize symb-open-alpha}}} differs from
  \hyperlink{LS}{\textsc{\textbf{\scriptsize symb-alpha-open}}}
  in the more restrictive side condition $z \notin \Fn((\nu y)P)$ versus $z \notin \Fn((\nu y)P')$.
  Whereas the two rules are interchangeable up to strong barbed bisimilarity, they are not up to
  strong bisimilarity, for the transition \plat{$(\nu y)(\bar x y.{\bf 0} + \bar z x.{\bf 0}) \goto{\bar x(z)}{\bf 0}$}
  is not derivable from \tab{pi-early-symbolic-noalpha-surjective}. I need \hyperlink{alpha}{\textsc{\textbf{\scriptsize symb-open-alpha}}}
  to obtain a transition relation that is strongly bisimilar with the one of $\piSur$.}
Lemmas~\ref{lem:alpha introduction gen} and~\ref{lem:alpha elimination gen} below say that up to strong
bisimilarity the transitions relations $\gotoU{\alpha}$ and $\gotoN{\alpha}$ of $\piSur$ are equivalent.

\begin{lemma}\rm\label{lem:alpha introduction gen}
If $R \gotoN{\alpha} R'$ then equally $R \gotoU{\alpha}\eqaU R'$.
\end{lemma}
\begin{proof}
With induction on the inference of $R \gotoN{\alpha} R'$. The only nontrivial cases are when 
$R \gotoN{\alpha} R'$ is derived by rule \hyperlink{alpha}{\textsc{\textbf{\small res-alpha}} or \textsc{\textbf{\small symb-open-alpha}}}.

By (\ref{A10}), $P\drenb{z}{y} \eqaU P\renb{z}{y}$ when $z \notin \Fn((\nu y)P)$.
Hence $(\nu y)P \eqaU (\nu z)(P\renb{z}{y}) \eqaU (\nu z)(P\drenb{z}{y})$.
Using this,
each application of rule \hyperlink{alpha}{\textsc{\textbf{\small res-alpha}}} can be mimicked by an application of
\hyperlink{ES}{\textsc{\textbf{\small res}}} followed by one of \hyperlink{ES}{\textsc{\textbf{\small alpha}}}.

Now suppose $R \gotoN{\alpha} R'$ is derived by \hyperlink{alpha}{\textsc{\textbf{\small symb-open-alpha}}}.
Then $R= (\nu y)P$, $\alpha = M\bar x (z)$, $y \neq x$, $z \notin\Fn(R)$, $y \notin \n(M)$,
$P \mathbin{\gotoN{M \bar x y}} P'$ and $R'\mathbin= P'\drenb{z}{y}$.
By induction $P \gotoU{M \bar x y}\eqaU P'\!$.
By \lem{free names actions gen} $\n(M)\cup\{x\} \subseteq \Fn(P){\setminus}\{y\}=\Fn(R)\not\ni z$.
So by \lem{substitution preserve gen} $P\drenb{z}{y} \gotoU{M \bar x z}\eqaU P'\drenb{z}{y}$.\vspace{1pt}
By \hyperlink{ES}{\textsc{\textbf{\small symb-open}}}
$(\nu z)(P\drenb{z}{y}) \gotoU{M \bar x (z)}\eqaU P'\drenb{z}{y} =R'$.
So by rule \hyperlink{ES}{\textsc{\textbf{\small alpha}}} $R \gotoU{\alpha}\eqaU R'$.
\end{proof}

\begin{lemma}\rm\label{lem:alpha elimination gen}
  If $R\eqaU U$ and $R\gotoU{\alpha}R'$ then equally $U\gotoN{\alpha}U'$ for some $U'$ with $R' \eqaU U'$.
\end{lemma}
\begin{proof}
  With induction on the $\alpha$-depth of the inference of $R \gotoU{\alpha}R'$, with a nested
  induction on the number of applications of rule \hyperlink{ES}{\textsc{\textbf{\small alpha}}} at
  the end of the inference.
  \begin{itemize}
    \item The cases that $R \gotoU{\alpha} R'$ is derived by rule \hyperlink{ES}{\textsc{\textbf{\small tau}}} or
      \hyperlink{ES}{\textsc{\textbf{\small output}}} are trivial.
    \item Suppose $R \gotoU{\alpha} R'$ is derived by \hyperlink{ES}{\textsc{\textbf{\small early-input}}}.
      Then $R= x(y).P$, $\alpha = xz$ and $R'=P\renb{z}{y}^\aN$.
      By \lem{outside first} $U=x(w).V$ for some $V\eqaU P\renb{w}{y}^3$.
      Consequently, $U \gotoN{\alpha} V\renb{z}{w}^\aN$.
      Using (\ref{alpha sigma}) and (\ref{3S-composition})      
      $V\renb{z}{w}^\aN \eqaU P\renb{w}{y}^3\renb{z}{w}^\aN \eqaU P\renb{z}{y}^\aN = R'$.\pagebreak
    \item Suppose $R \gotoU{\alpha} R'$ is derived by \hyperlink{ES}{\textsc{\textbf{\small ide}}}.
      Then $R = A(\vec{y})$, so $U=R$. 
      Let \plat{$A(\vec x) \mathbin{\stackrel{{\rm def}}{=}} P$}.
      Then $P\renb{\vec y}{\vec x}^\aN\gotoU{\alpha} R'\!$.
      By induction $P\renb{\vec y}{\vec x}^\aN\gotoN{\alpha}U'$ for some $U'$ with $R' \eqaU U'$.
      By \textsc{\textbf{\small ide}} $U\gotoN{\alpha}U'$.
    \item Suppose $R \gotoU{\alpha} R'$ is derived by \hyperlink{ES}{\textsc{\textbf{\small par}}}.
      Then $R = P | Q$, $P \gotoU{\alpha} P'$, $R'=P'|Q$ and
      $\bn(\alpha)\cap\Fn(Q)=\emptyset$.
      Thus $U \mathbin= V|W$ with $P\mathbin\eqaU V$, $Q \mathbin\eqaU W$ and $\bn(\alpha)\cap\Fn(W)\mathbin=\emptyset$.
      By induction $V\gotoN{\alpha}V'$ for some $V'$ with $P' \eqaU V'$.
      So $U =  V|W \gotoN{\alpha} V'|W$ by rule \textsc{\textbf{\small par}}, and $R'\eqaU V'|W$.
    \item The cases that $R \gotoN{\alpha} R'$ is derived by rule \hyperlink{ES}{\textsc{\textbf{\small e-s-com}}},
      \hyperlink{ES}{\textsc{\textbf{\small e-s-close}}}, \hyperlink{ES}{\textsc{\textbf{\small sum}}} or
      \hyperlink{ES}{\textsc{\textbf{\small symb-match}}} are (also) trivial.
    \item Suppose $R \gotoU{\alpha} R'$ is derived by rule \hyperlink{ES}{\textsc{\textbf{\small res}}}.
      Then $R\mathbin= (\nu z)P$, $P \gotoU{\alpha} P'$,
      $z\not\in \n(\alpha)$ and $R'\mathbin=(\nu z)P'$.
      So $U\mathbin=(\nu y)V$ for some $y \mathbin{\notin} \Fn((\nu z)P)$ and $V\mathbin{\eqaU} P\renb{y}{z}$,
      using \lem{outside first}.  Employing (\ref{alpha sigma}), (\ref{empty substitution}) and \lem{substitution composition},
      $V\drenb{z}{y} \eqaU P\renb{y}{z}\drenb{z}{y} \eqaU P\drenb{z}{z} \eqaU P$. 
      So by induction $V\drenb{z}{y} \gotoN{\alpha} V'$ for some $V' \eqaU P'$.
      By~(\ref{free names alpha})
      $z\not\in \Fn((\nu z)P) = \Fn((\nu y)(P\renb{y}{z})) = \Fn((\nu y)V)$.
      Hence $U = (\nu y)V \gotoN{\alpha} (\nu z)V'$ by rule \hyperlink{alpha}{\textsc{\textbf{\small res-alpha}}}.
      Moreover, $R' = (\nu z)P' \eqaU (\nu z)V'$.
    \item Suppose $R \gotoU{M \bar x (z)} R'$ is derived by \hyperlink{ES}{\textsc{\textbf{\small symb-open}}}.\vspace{1pt}
      Then $R\mathbin= (\nu z)P$, $P \gotoU{M \bar x z} R'$,
      $z \neq x$ and $z\not\in \n(M)$.
      Thus $U\mathbin=(\nu y)V$ for $y \notin \Fn((\nu z)P)$ and $V\mathbin{\eqaU} P\renb{y}{z}$, using \lem{outside first}.
      Now $P\renb{y}{z} \gotoU{M \bar x y} R'\renb{y}{z}$ by \lem{substitution preserve gen}.
      Moreover, the $\alpha$-depth of the inference of $P\renb{y}{z} \gotoU{M \bar x y} R'\renb{y}{z}$
      equals that of $P \mathbin{\gotoU{M \bar x z}} R'$, which is smaller than that of $R \gotoU{M \bar x (z)} R'$.\vspace{1pt}
      By induction $V \gotoN{M \bar x y} V'$ for some $V' \eqaU R'\renb{y}{z}$.
      By (\ref{free names alpha}), $z\not\in \Fn(R) = \Fn(U)$.
      Employing \lem{free names actions gen}, $\n(M)\cup\{x\} \subseteq \Fn(P){\setminus}\{z\}=\Fn(R)\not\ni y$.\vspace{1pt}
      Consequently, $U = (\nu y)V \gotoN{M\bar x (z)} V'\drenb{z}{y}$ by \hyperlink{alpha}{\textsc{\textbf{\small symb-open-alpha}}}.
      Moreover, $V'\drenb{z}{y} \eqaU R'\renb{y}{z}\drenb{z}{y} \eqaU  R'$
      by (\ref{substitution composition})--(\ref{alpha sigma}) and (\ref{empty substitution}), using that $y \notin \Fn(R) \supseteq \Fn(R'){\setminus}\{z\}$.
    \item Suppose $R \gotoU{\alpha} R'$ is derived by \hyperlink{ES}{\textsc{\textbf{\small alpha}}}.
      Then there is a $Q \eqaU R$ such that $Q \gotoU{\alpha} R'$ is derived by a simpler proof.
      So $U \eqaU Q$ and by induction $U \gotoN{\alpha}\eqaU R'$.
   \qed
   \end{itemize}
\end{proof}\vfill
Combining Lemmas~\ref{lem:alpha introduction gen} and~\ref{lem:alpha elimination gen}
yields the following result,
stating that rule \hyperlink{ES}{\textsc{\textbf{\small alpha}}} is not needed on top of the rules of \tab{pi-early-symbolic-noalpha-surjective}.
\begin{corollary}\rm\label{cor:alpha open gen}
  If $R\eqaU U$ and $R\gotoN{\alpha}R'$ then equally $U\gotoN{\alpha}U'$ for some $U'$ with $R' \eqaU U'$.
\end{corollary}
Together, Lemmas~\ref{lem:alpha introduction gen} and~\ref{lem:alpha elimination gen}
imply that the relation $\eqaU$ is a strong (barbed) bisimulation between
$\piSur$ expressions equipped with the semantics of \tab{pi-early-symbolic}, amended as described in
Sections~\ref{sec:substitutions} and~\ref{sec:surjective},
and $\piSur$ expressions equipped with the semantics of \tab{pi-early-symbolic-noalpha-surjective}.
Hence, up to $\sbb$ it doesn't matter whether which of these semantics one employs.

Although I will not need this in this paper, a trivial adaptation of the above proofs shows that
also the early and early symbolic semantics of the $\pi$-calculus, displayed in Tables~\ref{tab:pi-early}
and~\ref{tab:pi-early-symbolic} can be equivalently adapted by dropping rule \hyperlink{Late}{\textsc{\textbf{\small alpha}}}
at the cost of strengthening \hyperlink{Late}{\textsc{\textbf{\small res}}} and 
\hyperlink{Late}{\textsc{\textbf{\small open}}} (or \hyperlink{ES}{\textsc{\textbf{\small symb-open}}})
into \hyperlink{alpha}{\textsc{\textbf{\small res-alpha}}} and \textsc{\textbf{\small open-alpha}}
(or \hyperlink{alpha}{\textsc{\textbf{\small symb-open-alpha}}}).

\subsection{Replacing substitution by relabelling}
\label{sec:piRI}\hypertarget{piRI}{}

Let $\piRI$ be the variant of the $\pi$-calculus with surjective substitutions $\pi_{ES}^\aN(\N)$,
enriched with a postfix-written relabelling operator $[\sigma]$ for each surjective
substitution $\sigma$.
Its operational semantics is given by the rules of \tab{pi-early-symbolic-noalpha-surjective},
together with the rule \hypertarget{relab}{\textsc{\textbf{\small relabelling}}}
\[\frac{P\goto{\alpha}P'}{P[\sigma]\goto{\alpha\as{\sigma}}P'[\sigma]}
  ~~(\bn(\alpha\as\sigma)\cap\Fn(P[\sigma])=\emptyset)\]
for the relabelling operators,
except that the substitutions $\renb{z}{y}^\aN$ and $\renbt{y}{x}^\aN$ that appear in rules
\hyperlink{alpha}{\textsc{\textbf{\small early-input}} and \textsc{\textbf{\small ide}}} are replaced by
applications of the relabelling operators $[\renb{z}{y}^\aN]$ and $[\renbt{y}{x}^\aN]$, respectively,
and the substitution $\drenb{z}{y}$ that appears in rules 
\hyperlink{alpha}{\textsc{\textbf{\small res-alpha}} and \textsc{\textbf{\small symb-open-alpha}}} is replaced by
the relabelling operator $[\drenb{z}{y}]$.
Moreover, rules \hyperlink{alpha}{\textsc{\textbf{\small par}} and \textsc{\textbf{\small e-s-close}}}
gain side conditions $\bn(\alpha)\mathop\cap\Fn(P)\mathbin=\emptyset$ and $z \notin\Fn(P)$,
respectively.\footnote{These side conditions are not essential---they could be left out here, or
  added to the original $\pi$-calculus, without ill effects---but are included to obtain a stronger
  and simpler version of Lemmas~\ref{lem:parameter elimination} and~\ref{lem:clash-free forth} in
  Sections~\ref{sec:piR} and~\ref{sec:piRp}.}
Except for \textsc{\textbf{\small ide}}, $\piRI$ has the semantics of \tab{pi-relabelling}.
The function $\Fn$, used in several rules,
is extended to $\piRI$ by $\Fn(P[\sigma]) := \{x\as{\sigma} \mid x \mathbin\in \Fn(P)\}$.\linebreak
Recall that names are chosen from $\N \uplus \aN$, except that in $x(y).P$ one always has $y\in \N$.
Moreover, in defining equations \plat{$A(\vec{x})\stackrel{\rm def}{=} P$} I require that $x_1,\dots,x_n \in \N$,
and $\fn(P) \subseteq \{x_1,\dots,x_n\}$.\footnote{Optionally, one could furthermore require that no relabelling
  operators may occur in defining equations; this would simplify the definition of $\widehat P$.} 

I will show that the identity $\mbox{id}\!: \T_{\piSur} \rightarrow \T_{\piRI}$ is a valid
translation from $\piSur$ to $\piRI$. As justified in \sect{alpha3}, at the side of $\piSur$ I
will use the transition relation $\gotoN{}$ generated by \tab{pi-early-symbolic-noalpha-surjective}.
Lemmas~\ref{lem:alpha introduction gen} and~\ref{lem:alpha elimination gen} imply that
Lemmas~\ref{lem:input universality 2 gen}--\ref{lem:substitution reflection gen} also hold for $\gotoN{}$.
I will denote the transition relation of $\piRI$ simply by $\goto\alpha$.

For $P$ a $\piRI$ process, let $\widehat P$ be the $\piSur$ process obtained from $P$ by recursively replacing
each subterm $Q[\sigma]$ by $Q\sigma$, and each agent identifier $A$ by $A\hat{~}$.
Here $A\hat{~}$ is a fresh agent identifier with defining equation
\plat{$A\hat{~}(\vec x) \stackrel{{\rm def}}{=} \widehat P$} when
\plat{$A(\vec x) \stackrel{{\rm def}}{=} P$} was the defining equation of $A$.
Note that $\Fn(\widehat P)=\Fn(P)$ for each $P \in \T_{\piRI}$.

\begin{lemma}\label{lem:relabelling forth gen}
If $R \goto{\beta} R'$ then equally \plat{$\widehat R \gotoN{\beta}\eqaU \widehat{R'}$}.
\end{lemma}
\begin{proof}
By induction of the inference of $R \goto{\beta} R'$.
\begin{itemize}
\item Suppose $R \goto{\beta} R'$ is derived by \hyperlink{pi-relabelling}{\textsc{\textbf{\small tau}}}.
  Then $R\mathbin=\tau.P$, $\beta\mathbin=\tau$
  and $R'\mathbin=P$. Moreover, $\widehat R = \tau.\widehat P$
  and $\widehat R \gotoN{\beta} \widehat {R'}$.
\item The case that $R \mathbin{\goto{\beta}} R'$ stems from \hyperlink{pi-relabelling}{\textsc{\textbf{\small output}}} goes likewise.
\item Suppose $R \goto{\beta} R'$ is derived by \hyperlink{pi-relabelling}{\textsc{\textbf{\small early input}}}.
  Then $R=x(y).P$, $\beta=xz$
  and $R'=P[\renb{z}{y}^\aN]$. Moreover, $\widehat R = x(y).\widehat P$ and
  $\widehat R \gotoN{\beta} \widehat P\renb{z}{y}^\aN = \widehat {P[\renb{z}{y}^\aN]} = \widehat {R'}$.
\item Suppose $R \goto{\beta} R'$ is derived by \hyperlink{pi-relabelling}{\textsc{\textbf{\small sum}}}.
  Then $R=P+Q$ and $P \goto{\beta} R'$.
  Now $\widehat R = \widehat P + \widehat Q$.
  So $\widehat P \mathbin{\gotoN{\beta}\eqaU} \widehat {R'}$ by induction.
  Hence, by \hyperlink{alpha}{\textsc{\textbf{\small sum}}}, $\widehat R \gotoN{\beta}\eqaU \widehat {R'}$.
\item Suppose $R \goto{\beta} R'$ is derived by \hyperlink{pi-relabelling}{\textsc{\textbf{\small symb-match}}}.
  Then $R=\Match{x}{y}P$, $P \goto{\alpha} R'$ and $\beta\mathbin=\match{x}{y}\alpha$.
  Now $\widehat R = \Match{x}{y}\widehat P$.
  By induction $\widehat P \gotoN{\alpha}\eqaU \widehat {R'}$.
  By \hyperlink{alpha}{\textsc{\textbf{\small symb-match}}}, $\widehat R \gotoN{\beta}\eqaU \widehat {R'}$.
\item Suppose $R \goto{\beta} R'$ is derived by \textsc{\textbf{\small ide}}.
  Then $R=A(\vec y)$ with \plat{$A(\vec{x})\stackrel{\rm def}{=} P$} and
  $P[\renb{\vec{y}}{\vec{x}}^\aN] \goto{\beta} R'$.
  Now $\widehat R = A\hat{~}(\vec y)$ with \plat{$A\hat{~}\hspace{-1pt}(\hspace{-1pt}\vec{x})\mathbin{\stackrel{\rm def}{=}} \widehat P\hspace{-1pt}$}.
  Moreover, $\widehat{P[\renb{\vec{y}}{\vec{x}}^\aN]} \mathbin= \widehat P \renbt{y}{x}^\aN\!$.
  By induction $\widehat P \renbt{y}{x}^\aN \gotoN{\beta}\eqaU \widehat{R'}$.
  By \hyperlink{alpha}{\textsc{\textbf{\small ide}}}, $\widehat R\gotoN{\beta}\eqaU \widehat{R'}$.
\item Suppose $R \goto{\beta} R'$ is derived by \hyperlink{pi-relabelling}{\textsc{\textbf{\small par}}}.
  Then $R\mathbin=P|Q$, $P \goto{\beta} P'$, $\bn(\beta)\cap\Fn(Q)=\emptyset$ and $R'\mathbin=P'|Q$.
  Now $\widehat R \mathbin= \widehat P | \widehat Q$, $\bn(\beta)\cap\Fn(\widehat Q)=\emptyset$
  and $\widehat {R'} \mathbin= \widehat {P'} | \widehat Q$.
  By induction $\widehat P \gotoN{\beta}\eqaU \widehat {P'}$.
  Thus $\widehat R \gotoN{\beta}\eqaU \widehat {R'}$, by \hyperlink{alpha}{\textsc{\textbf{\small par}}}.
\item Suppose $R \goto{\beta} R'$ is derived by \hyperlink{pi-relabelling}{\textsc{\textbf{\small e-s-com}}}.
  Then $R\mathbin=P|Q$,
  \plat{$P\goto{M\bar xy} P'$}, {$Q\goto{Nvy} Q'$}, $R'\mathbin=P'|Q'$,
  and $\beta=\match{x}{v}MN\tau$. Now $\widehat R = \widehat P | \widehat Q$
  and $\widehat {R'} = \widehat {P'} | \widehat {Q'}$.  By induction
  \plat{$\widehat P\gotoN{M\bar xy}\eqaU \widehat{P'}$} and {$\widehat Q\gotoN{Nvy}\eqaU \widehat {Q'}$}.
  Thus $\widehat R \gotoN{\beta}\eqaU \widehat {R'}$, by \hyperlink{alpha}{\textsc{\textbf{\small e-s-com}}}.
\item Suppose $R \goto{\beta} R'$ is derived by \hyperlink{pi-relabelling}{\textsc{\textbf{\small e-s-close}}}.
  Then $R\mathbin=P|Q$,
  \plat{$P\goto{M\bar x(z)} P'$}, {$Q\goto{Nvz} Q'$}, $R'\mathbin=(\nu z)(P'|Q')$, $z \notin\Fn(Q)$
  and $\beta=\match{x}{v}MN\tau$. Now $\widehat R = \widehat P | \widehat Q$, $z \notin\Fn(\widehat Q)$
  and $\widehat {R'} = (\nu z)(\widehat {P'} | \widehat {Q'})$.  By induction
  \plat{$\widehat P\gotoN{M\bar x(z)}\eqaU \widehat{P'}$} and {$\widehat Q\gotoN{Nvz}\eqaU \widehat {Q'}$}.
  Thus $\widehat R \gotoN{\beta}\eqaU \widehat {R'}$, by \hyperlink{alpha}{\textsc{\textbf{\small e-s-close}}}.
\item Suppose $R \goto{\beta} R'$ is derived by \hyperlink{pi-relabelling}{\textsc{\textbf{\small res-alpha}}}.
  Then $R=(\nu y)P$, $P[\drenb{z}{y}] \goto{\beta} P'$, $R'=(\nu z)P'$, $z \notin\Fn(R)$ and $z \mathbin{\notin} \n(\beta)$.
  Now $\widehat R \mathbin= (\nu y)\widehat P$, $z \mathbin{\notin}\Fn(\widehat R)$ and $\widehat R' \mathbin= (\nu z)\widehat P'$.
  By induction $\widehat P\drenb{z}{y} \mathbin{\gotoN{\beta}\eqaU} \widehat {P'}$.
  Hence, $\widehat R \gotoN{\beta}\eqaU \widehat {R'}$.
\item Suppose $R \goto{\beta} R'$ is derived by \hyperlink{pi-relabelling}{\textsc{\textbf{\small symb-open-alpha}}}.
  Then $R\mathbin=(\nu y)P$, $\beta\mathbin=M\bar x (z)$, $P \goto{M \bar x y} P'$,
  $R'\mathbin=P'[\drenb{z}{y}]$, $y \notin x$, $z \notin \Fn(R)$ and $y \notin \n(M)$.
  Now $\widehat R \mathbin= (\nu y)\widehat P$, $z \mathbin{\notin}\Fn(\widehat R)$ and $\widehat R' \mathbin= \widehat P'\drenb{z}{y}$.
  By induction $\widehat P \mathbin{\gotoN{M \bar x y}\eqaU} \widehat {P'}$.
  Hence, $\widehat R \gotoN{\beta}\eqaU \widehat {R'}$ by \hyperlink{alpha}{\textsc{\textbf{\small symb-open-alpha}}}.
\item Suppose $R \goto{\beta} R'$ is derived by \hyperlink{pi-relabelling}{\textsc{\textbf{\small relabelling}}}.
  Then $R\mathbin=P[\sigma]$, $P \goto{\alpha} P'$, $R'\mathbin=P'[\sigma]$, $\beta \mathbin= \alpha\as{\sigma}$
  and $\bn(\beta) \cap \Fn(R) \mathbin= \emptyset$.
  By induction $\widehat P \mathbin{\gotoN{\alpha}\eqaU} \widehat{P'}$.
  As $\Fn(R)\mathbin=\Fn(\widehat R)\linebreak[4] = \Fn(\widehat P \sigma)$,
  by \lem{substitution preservation gen}
  $\widehat R \mathbin= \widehat P\sigma \gotoN{\beta}\eqaU \widehat P'\sigma \mathbin= \widehat R'$.
\qed
\end{itemize}
\end{proof}

\begin{lemma}\label{lem:relabelling back gen}
If $\widehat R \gotoN{\beta} U$ with $\bn(\beta)\cap\Fn(R)=\emptyset$ then equally \plat{$R \goto{\beta} R'$}
for some $R'$ with $\widehat R' \eqaU U$.
\end{lemma}
\begin{proof}
  By induction on the depth of the inference of $\widehat R \mathbin{\gotoN{\beta}} U$, with a nested
  induction on the number of topmost relabelling operators in $R$.
\begin{itemize}
\item First suppose $R=P[\sigma]$. Then $\widehat R = \widehat P \sigma$.
  By \lem{substitution reflection gen} $\widehat P \gotoN{\alpha} V$ for some $\alpha$ and $V$ with
  $\alpha\as{\sigma}\mathbin=\beta$ and $V\sigma \mathbin\eqaU U$.
  Moreover, the depth of the inference of \plat{$\widehat P \gotoN{\alpha}V$} is the same as that of
  \plat{$\widehat P\sigma \gotoN{\beta} U$}.
  As $\bn(\alpha\as\sigma)\cap\Fn(P[\sigma])=\emptyset$, also $\bn(\alpha)\cap\Fn(P)=\emptyset$.
  So by induction $P\goto{\alpha} P'$ for some $P'$ with $\widehat{P'}\eqaU V$.
  By \hyperlink{pi-relabelling}{\textsc{\textbf{\small relabelling}}} $R\goto{\beta}P'[\sigma]$.
  Furthermore one has $\widehat{P'[\sigma]}=\widehat{P'}\sigma \eqaU V\sigma \eqaU U$.

  Henceforth I suppose that $R$ is not of the form $P[\sigma]$.
\item
  Suppose \plat{$\widehat R \gotoN{\beta} U$} is derived by rule \hyperlink{alpha}{\textsc{\textbf{\small e-s-com}}}.
  Then $R=P|Q$, $\widehat R\mathbin=\widehat P|\widehat Q$, $\beta=\match{x}{v}MN\tau$, $\widehat P\gotoN{M\bar xy}V$,
  $\widehat Q\gotoN{Nvy}W$ and $U \mathbin=V|W$.\vspace{1pt}
  By induction $P \goto{M\bar xy}P'$ and $Q \goto{Nvy}Q'$ for some $P'$ and $Q'$ with $\widehat P' \eqaU V$
  and $\widehat Q' \eqaU W$.
  By \hyperlink{pi-relabelling}{\textsc{\textbf{\small e-s-com}}} $R \goto{\beta} P'|Q'\eqaU V|W = U$.
\item
  Suppose \plat{$\widehat R \gotoN{\beta} U$} is derived by rule \hyperlink{alpha}{\textsc{\textbf{\small res-alpha}}}.
  Then $R=(\nu y)P$, $\widehat R\mathbin=(\nu y )\widehat P$, $\widehat P\drenb{z}{y} \gotoN\beta V$,
  $z \notin \Fn(\widehat R)$, $z \notin \n(\beta)$ and $U = (\nu z)V$.
  As $\bn(\beta)\cap\Fn((\nu y)P)=\emptyset$ and $z \mathbin{\notin} \n(\beta) \cup \Fn((\nu y)P)$, also
  $\bn(\beta)\cap\Fn(P[\drenb{z}{y}])\mathbin=\emptyset$.
  So by induction $P[\drenb{z}{y}] \goto\beta P'$ for some $P'$ such that $\widehat P' \eqaU V$.
  Now $R \goto\beta (\nu z)P'$ by \hyperlink{pi-relabelling}{\textsc{\textbf{\small res-alpha}}}.
  Moreover, $\widehat{(\nu z)P'} = (\nu z)\widehat P' \eqaU (\nu z) V = U$.
\item
  The cases that \plat{$\widehat R \gotoN{\beta} U$} is derived by
  \hyperlink{alpha}{\textsc{\textbf{\small tau}}, \textsc{\textbf{\small output}}, \textsc{\textbf{\small early-input}},
  \textsc{\textbf{\small sum}}, \textsc{\textbf{\small symb-match}}, \textsc{\textbf{\small ide}},
  \textsc{\textbf{\small par}}} or \textsc{\textbf{\small symb-open-alpha}} are also trivial.
\item
  Suppose \plat{$\widehat R \gotoN{\beta} U$} is derived by rule \hyperlink{alpha}{\textsc{\textbf{\small e-s-close}}}.
  Then $R=P|Q$, $\widehat R\mathbin=\widehat P|\widehat Q$, $\beta=\match{x}{v}MN\tau$, $\widehat P\gotoN{M\bar x(z)}V$,
  $\widehat Q\gotoN{Nvz}W$, $z\notin \Fn(\widehat Q)$ and $U = (\nu z)(V|W)$.
  Pick $w \notin \Fn(\widehat R)\cup \Fn(U)$. Then
  $\widehat P \gotoN{M\bar x(w)}V' \eqaU V\renb{w}{z}$ and $\widehat Q \gotoN{Nv w}W' \eqaU W\renb{w}{z}$
  by Lemmas~\ref{lem:bound output universality gen} and~\ref{lem:input universality 2 gen}.
  Moreover, the inferences of these transitions have a smaller depth than that of $\widehat R \gotoN{\beta} U$.
  So by induction $P \goto{M\bar x(w)}P'$ and $Q \goto{Nvw}Q'$ for some $P'$ and $Q'$ with $\widehat P' \eqaU V'$
  and $\widehat Q' \mathbin{\eqaU} W'\!$.
  Hence $R \mathbin{\goto{\beta}} (\nu w)(P'|Q')$ by \hyperlink{pi-relabelling}{\textsc{\textbf{\small e-s-close}}}.
  Moreover,\\[1ex]\mbox{}\hfill\(\begin{array}[b]{@{}l@{}}\widehat{(\nu w)(P'|Q')} = (\nu w)(\widehat P'|\widehat Q')\eqaU (\nu w)(V'|W') \eqaU \mbox{}\\
  (\nu w)((V|W)\renb{w}{z}) \eqaU (\nu z)(V|W) = U. \end{array}\)
  \qed
\end{itemize}
\end{proof}

\begin{table*}[t]
\caption{Structural operational semantics of the $\pi$-calculus with relabelling}
\hypertarget{pi-relabelling}{}
\label{tab:pi-relabelling}
\normalsize
\begin{center}
\framebox{$\begin{array}{@{}c@{\hspace{-1pt}}c@{\hspace{-3pt}}c}
\transname[13]{tau}{25}{         \textcolor{DarkBlue}{M}\tau.P \goto{\textcolor{DarkBlue}{M}\tau} P }&
\transname[13]{output}{38}{  \textcolor{DarkBlue}{M}\bar x y.P \goto{\textcolor{DarkBlue}{M}\bar x y} P }&
\hspace{-30pt}
\transname[13]{early-input}{55}{  \textcolor{DarkBlue}{M} x(y).P \goto{\textcolor{DarkBlue}{M}xz} P[\renb{z}{y}^\aN] }\hspace{-24pt}\\[2ex]
\transname[13]{sum}{20}{
  \displaystyle\frac{P \goto{\alpha} P'}{P+Q \goto{\alpha} P'} }&
\hspace{-10pt}
\textcolor{purple}{\transname[13]{symb-match}{45}{
  \displaystyle\frac{P \goto{\alpha} P'}{\Match{x}{y}P \goto{\match{x}{y}\alpha} P'} }}\hspace{-20pt}&
\hspace{-20pt}
\transname[13]{ide}{25}{
 \displaystyle\frac{P \goto{\alpha} P'}{A(\vec{x}) \goto{\alpha} P'}
  ~~(A(\vec{x})\stackrel{\rm def}{=} P) }\hspace{-20pt}\\[4ex]
\transname{par}{8}{
  \displaystyle\frac{P\goto{\alpha} P'}{P|Q \goto{\alpha} P'|Q}~
  \textcolor{Orange}{\left(\begin{array}{@{}l@{}} \bn(\alpha)\mathop\cap\Fn(P)\mathbin=\emptyset  \\
                             \bn(\alpha)\mathop\cap\Fn(Q)\mathbin=\emptyset\end{array}\right)}}\hspace{-10pt}&
\hspace{-75pt}
\transname{e-s-com}{35}{
  \displaystyle\frac{P\goto{M\bar xy} P' ,~ Q \goto{Nvy} Q'}{P|Q \goto{\match{x}{v}MN\tau} P'| Q'} }
  \hspace{-75pt}&
\textcolor{Orange}{\transname{e-s-close}{35}{
  \displaystyle\frac{P\goto{M\bar x(z)} P' ,~ Q \goto{Nv z} Q'}{P|Q \goto{\match{x}{v}MN\tau} (\nu z)(P'| Q')}~
\left(\!\begin{array}{@{}l@{}} z\mathop{\notin}\Fn(P)  \\
                             z\mathop{\notin}\Fn(Q)\end{array}\!\right)}}\hspace{-30pt}\\[4.5ex]

\textcolor{Orange}{\transname[20]{res-alpha}{0}{
  \displaystyle\frac{P[\drenb{z}{y}] \goto{\alpha} P'}{(\nu y)P \goto{\alpha} (\nu z)P'}~
  \left(\begin{array}{@{}l@{}} 
                               z\mathbin{\not\in} \Fn((\nu y)P)  \\
                               z\not\in \n(\alpha)\end{array}\right)}}\hspace{10pt}&
\hspace{-40pt}
\textcolor{Orange}{\transname[20]{symb-open-alpha}{0}{
  \displaystyle\frac{P \goto{M \bar x y} P'}{(\nu y)P \goto{M\bar x(z)} P'[\drenb{z}{y}]}
  \left(\begin{array}{@{}l@{}} y \neq x \\ z\mathbin{\not\in} \Fn((\nu y)P)  \\
                               y \notin \n(M)\end{array}\!\!\right)}}\hspace{-95pt}&
\hspace{30pt}
\transname[20]{relabelling}{0}{
  \displaystyle\frac{P \goto{\alpha} P'}{P[\sigma] \goto{\alpha\as{\sigma}} P'[\sigma]}~
  \textcolor{Orange}{\left(\begin{array}{@{}l@{}}\bn(\alpha\as\sigma)\cap \mbox{}\\
                              \Fn(P[\sigma])=\emptyset\end{array}\right) }}\hspace{-30pt}
\end{array}$}\vspace{-1pt}
\end{center}
\vspace{-1ex}
\end{table*}

Let $\fT_\rho$ be the identity translation from $\piSur$ to $\piRI$.

\begin{theorem}\label{thm:step 4: relabelling}
  $\fT_\rho(P) \mathbin{\sbb} P$ for any $P \in \T_{\piSur}$.
\end{theorem}
\begin{proof}
  It suffices to show that the symmetric closure of 
  \[{\R}:=\{(R,U) \mid R \in \T_{\piRI}, U \in \T_{\piSur} \wedge \widehat R \eqaU U\}\]
  is a strong barbed bisimulation, as it contains $(\fT_\rho(P),P)$.

  Let $(R,U)\in{\R}$.
  Let $R\goto\tau R'$.
  Then $\widehat R\gotoN\tau \eqaU \widehat R'$ by \lem{relabelling forth gen}
  and $U\gotoN\tau \eqaU \widehat R'$ by \cor{alpha open gen}.
  So $U\gotoN\tau U'$ for some $U'$ with $(R',U')\in{\R}$.

  Let $U\gotoN\tau U'$.
  Then $\widehat R\gotoN\tau \eqaU U'$ by \cor{alpha open gen}.
  By \lem{relabelling back gen} $R\goto\tau R'$ for some $R'$ with $\widehat R' \eqaU U'$.

  Let $U{\downarrow_b}$ with $b\inp \nN\cup\overline\nN$.
  Then $U\gotoN{by} U'$ or $U\gotoN{b(y)} U'$ for some $y$ and $U'$, using the definition of $O$ in \sect{barbed}.
  Hence $\widehat R\gotoN{by}$ or $\widehat R\gotoN{b(y)}$ by \cor{alpha open gen}.
  In the second case, by \lem{bound output universality gen} I may assume, without loss of generality, that $y \not\in \Fn(\widehat R)$.
  So $R\gotoU{by}$ or $R\gotoU{b(y)}$ by \lem{relabelling back gen}.
  Thus $R{\downarrow_b}$.

  The implication $R{\downarrow_b} \Rightarrow U{\downarrow_b}$ proceeds likewise.
\end{proof}

The following example shows why the side condition $\bn(\alpha\as\sigma)\cap\Fn(P[\sigma])=\emptyset$
in rule \hyperlink{relab}{\textsc{\textbf{\small relabelling}}} is necessary.

\begin{example}
  Let $P = \bar x z. x(w). \bar w q \mid x(y).(\nu v) \bar x v.y(r)$.
  In the $\pi$-calculus, or in $\piSur$, this process can do two successive $\tau$-steps in a row, but not three:
  $$P \goto\tau x(w). \bar w q \mid (\nu v) \bar x v.z(r) \goto\tau (\nu u) (\bar u q \mid z(r)) \gonotto\tau.$$
  In the above expression the bound name $u$ may be chosen freely from $\N {\setminus} \{x,z,q\}$.
  The side condition of rule \hyperlink{alpha}{\textsc{\textbf{\small symb-open-alpha}}} prevents
  the choices $u=x,z$, since $x$ and $z$ are free in $(\nu v) \bar x(v).z(r)$. The side condition of
  (the symmetric counterpart of) rule \hyperlink{alpha}{\textsc{\textbf{\small e-s-close}}} prevents the choices $u=x,q$.

  In a version of $\piRI$ in which rule \hyperlink{relab}{\textsc{\textbf{\small relabelling}}}
  does not have the side condition $\bn(\alpha\as\sigma)\cap\Fn(P[\sigma])=\emptyset$, and
  \hyperlink{relab}{\textsc{\textbf{\small e-s-close}}} lacks the side condition $z \in \Fn(P)$, one obtains
  $(\nu v) \bar x v.y(r) \goto{\bar x (z)} y(r)\drensq{z}{v}$
  and thus
  \[\big((\nu v) \bar x v.y(r)\big)\rensq{z}{y} \goto{\bar x (z)} y(r)\drensq{z}{v}\rensq{z}{y}
    \goto{zq} {\bf 0}\drensq{z}{v}\rensq{z}{y}\rensq{q}{r} \]
  which yields
  \[P \begin{array}[t]{cl} \goto\tau &
    x(w). \bar w q \mid \big((\nu v) \bar x v.y(r)\big)\rensq{z}{y} \\
    \goto\tau &
    (\nu z) (\bar w q\rensq{z}{w} \mid y(r)\drensq{z}{v}\rensq{z}{y}) \\
    \goto\tau &
    (\nu z) ({\bf 0}\rensq{z}{w} \mid {\bf 0}\drensq{z}{v}\rensq{z}{y}\rensq{q}{r}) ,
  \end{array}\]
  in violation of \thm{step 4: relabelling}.

  Although $(\nu v) \bar x v.y(r) \gotoN{\bar x (z)} y(r)\drenb{z}{v}$,
  one does not have
  $\big((\nu v) \bar x v.y(r)\big)\renb{z}{y} \mathbin{\,\,\not\!\!\gotoN{\bar x (z)}} y(r)\drenb{z}{v}\renb{z}{y}$.
  This shows that the necessity of the side condition
  $\bn(\alpha\as{\sigma}) \cap \Fn(R\sigma) =\emptyset$
  in \lem{substitution preserve gen}/\ref{lem:substitution preservation gen}.
  Hence $\big((\nu v) \bar x v.y(r)\big)\rensq{z}{y} \goto{\bar x (z)}  y(r)\drensq{z}{v}\rensq{z}{y}$
  violates \lem{relabelling forth gen}.
  Since rule \hyperlink{pi-relabelling}{\textsc{\textbf{\small symb-open-alpha}}} is invoked
  \emph{before} the relabelling $\rensq{z}{y}$ is applied---in contrast with $\piSur$, where it is
  invoked after applying the substitution $\renb{z}{y}$---the choice $u=z$ cannot be avoided.
  Here the side condition of \hyperlink{relab}{\textsc{\textbf{\small relabelling}}} comes to the rescue;
  it rules out the offending transition
  because $z\in \bn(\bar x (z)) \cap \Fn(\big((\nu v) \bar x v.y(r)\big)\rensq{z}{y})$.
\end{example}

\subsection{\texorpdfstring{$\alpha$}{Alpha}-conversion for \texorpdfstring{$\piRI$}{the pi-calculus with relabelling}}
\label{sec:eqaR}

For $\piRI$-processes $P$ and $Q$ write $P \eqa Q$ iff $\widehat P \eqaU \widehat Q$.
If $P \eqa Q$ then $\Fn(P)=\Fn(\widehat P)=\Fn(\widehat Q)=\Fn(Q)$.

\begin{lemma}\label{lem:eqa matches}
  If $P\mathbin{\eqa} Q$ and $P\goto\alpha P'$ with $\bn(\alpha)\cap\Fn(Q)\mathbin=\emptyset$ then
  equally $Q\goto{\alpha}Q'$ for some $Q'$ with $P'\eqa Q'$.
\end{lemma}
\begin{proof}
  Suppose $P \eqa Q$ and $P\goto\alpha P'$ with $\bn(\alpha)\cap\Fn(P)\mathbin=\emptyset$.
  Then $\widehat P \eqaU \widehat Q$, and $\widehat P\gotoN\alpha\eqaU \widehat P'$ by \lem{relabelling forth gen}.
  Therefore, $\widehat Q\gotoN\alpha U$ for some $U \eqaU \widehat P'$ by \cor{alpha open gen}.
  By \lem{relabelling back gen} $Q \goto\alpha Q'$ for some $Q'$ with $\widehat Q' \eqaU U$. Now $P'\eqa Q'$.
\end{proof}

\begin{lemma}\label{lem:eqaR}
\begin{enumerate}[(a)]
\item $(\forall x\inp\Fn(P).\,x\as{\sigma} \mathbin= x\as{\sigma'}) \Rightarrow P[\sigma]\eqa P[\sigma']$,
      \label{free same}
\item $P[\sigma_1][\sigma_2] \eqa P[\sigma_2 \circ \sigma_1]$,\label{relabelling composition}
\item $P[\epsilon] \eqa P$, where $\epsilon$ is the empty substitution,\label{empty relabelling}
\item $(P|Q)[\sigma] \eqa P[\sigma] \mid Q[\sigma]$,\label{par rel}
\item $((\nu y)P)[\sigma] \eqa (\nu y)(P[\sigma])$, when $y\notin\dom(\sigma)\cup{\it range}(\sigma)$,\label{nu rel}
\item $(\nu y)P \eqa (\nu z)(P[\drenb{z}{y}])$, when $z \notin \Fn((\nu y)P)$,\label{relabelling conversion}
\item if $P_1\eqa Q_1$ and $P_2\eqa Q_2$ then $P_1|P_2 \eqa Q_1|Q_2$, and\label{eqa par}
\item if $P \eqa Q$ then $(\nu y)P \eqa (\nu y)Q$ and $P[\sigma]\eqa Q[\sigma]$.\label{eqa nu}
\end{enumerate}
\end{lemma}
\begin{proof}
(\ref{free same}), (\ref{relabelling composition}) and (\ref{empty relabelling})
follow immediately from (\ref{A10}), (\ref{substitution composition}) and (\ref{empty substitution}),
(\ref{par rel}) and (\ref{nu rel}) from the definition of substitution,
(\ref{eqa par}) and (\ref{eqa nu}) from the congruence property of $\eqaU$ and (\ref{alpha sigma}),
and (\ref{relabelling conversion}) follows from (\ref{alpha nu}), using (\ref{free same}) and (\ref{eqa nu})
to replace $\renb{z}{y}$ by $\drenb{z}{y}$.
\end{proof}
The following five lemmas are counterparts for $\piRI$ of Lemmas
\ref{lem:free names actions gen},~\ref{lem:free names successors gen},
\ref{lem:bound output universality gen},~\ref{lem:input universality 2 gen}
and \ref{lem:input universality pre 2 gen} for $\piES$, adapted to avoid non-surjective
substitutions $\renb{w}{z}$.

\begin{lemma}\label{lem:free names actions relabelling}
If $P\goto{\alpha}Q$ then $\n(\alpha){\setminus}(\ia(\alpha)\cup\bn(\alpha))\subseteq\Fn(P)$.
\end{lemma}
\begin{proof}
Immediately from Lemmas~\ref{lem:relabelling forth gen},~\ref{lem:alpha introduction gen}
and \ref{lem:free names actions gen}.
\end{proof}

\begin{lemma}\label{lem:free names successors relabelling}
If $P\goto{\alpha}Q$ then $\Fn(Q)\subseteq\Fn(P)\cup\ia(\alpha)\cup\bn(\alpha)$.
\end{lemma}
\begin{proof}
Immediately from Lemmas~\ref{lem:relabelling forth gen},~\ref{lem:alpha introduction gen}
and \ref{lem:free names successors gen}.
\end{proof}

\begin{lemma}\rm\label{lem:bound output universality relabelling}
  If $R\goto{M\bar x(z)}R_z$ and $w\notin\Fn(R)$ then equally
  $R\goto{M\bar x(w)}\eqa R_z[\drenb{w}{z}]$.
\end{lemma}

\begin{proof}
  Let $R\mathbin{\goto{M\bar x(z)}}R_z$.
  Then $\widehat R\mathbin{\gotoN{M\bar x(z)}\eqaU} \widehat R_z$ by \lem{relabelling forth gen}.
  So $\widehat R\mathbin{\gotoU{M\bar x(z)}} U \mathbin{\eqaU} \widehat R_z$ by \lem{alpha introduction gen}.
  Using that $w\mathbin{\notin}\Fn(R)\linebreak\mathbin=\Fn(\widehat R)$, by \lem{bound output universality gen}
  $\widehat R\gotoU{M\bar x(w)}V$ for a $V \mathbin{\eqaU} U\renb{w}{z}$.
  By \lem{free names successors gen}, $\Fn(U){\setminus}\{z\} \subseteq \Fn(\widehat R)\not\ni w$,
  and therefore $U\drenb{w}{z} \eqaU U\renb{w}{z}$ by (\ref{A10}).
  By \lem{alpha elimination gen} $\widehat R\gotoN{M\bar x(w)}W$ for a $W \mathbin{\eqaU} V$.
  By \lem{relabelling back gen}, using that $w\mathbin{\notin}\Fn(R)$,
  $R \goto{M\bar x(w)} R'$ for an $R'$ with $\widehat R' \eqaU W$. Using (\ref{alpha sigma})
  $\widehat R' \mathbin{\eqaU} \widehat R_z\drenb{w}{z} \mathbin= \widehat{R_z[\drenb{w}{z}]}$, and thus $R' \mathbin\eqa R_z[\drenb{w}{z}]$.
\end{proof}

\begin{lemma}\rm\label{lem:input universality relabelling}
  If $R\goto{Mxz}R_z$ and $z,w\notin\Fn(R)$, then equally
  $R\mathbin{\goto{Mxw}}R_w$ for some $R_w$ with $R_w \mathbin\eqa R_z[\drenb{w}{z}]$.
\end{lemma}

\begin{proof}
  The proof is the same as the one of \lem{bound output universality relabelling},
  but using \lem{input universality 2 gen} instead of \lem{bound output universality gen},
  thereby also using the precondition $z\notin\Fn(R)$.
\end{proof}

\begin{lemma}\rm\label{lem:input universality confluence}
  If $R\goto{Mxz}R_z$ and $w$ is a name, then equally $R\mathbin{\goto{Mxw}}R_w$ for some $R_w$ such
  that  $R_w[\sigma] \mathbin\eqa R_z[\sigma]$ for each surjective substitution $\sigma$ with $z\as\sigma=w\as\sigma$.
\end{lemma}

\begin{proof}
  Let $R\goto{Mxz}R_z$. 
  Then $\widehat R\gotoN{Mxz}\eqaU \widehat R_z$ by \lem{relabelling forth gen}.
  So $\widehat R\gotoU{Mxz} U \eqaU \widehat R_z$ by \lem{alpha introduction gen}.
  Pick $v \notin\Fn(\widehat R)$.
  By \lem{input universality pre 2 gen}
  $\widehat R\gotoU{Mxv}V$ for some $V$ with $U \mathbin{\eqaU} V\renb{z}{v}$,
  and by \lem{input universality 2 gen}
  $\widehat R\gotoU{Mxw}W$ for some $W$ with $W \mathbin{\eqaU} V\renb{w}{v}$.\linebreak
  By \lem{alpha elimination gen} $\widehat R\gotoN{Mxw}W^\dagger$ for some $W^\dagger \mathbin{\eqaU} W$.
  By \lem{relabelling back gen} $R \goto{Mxw} R_w$ for some $R_w$ with $\widehat R_w \eqaU W^\dagger$.
  
  Now let $\sigma$ be a surjective substitution with $z\as\sigma=w\as\sigma$.
  Then $\widehat {R_w[\sigma]} = \widehat R_w \sigma \eqaU V\renb{w}{v}\sigma =
  V\renb{z}{v}\sigma \eqaU \widehat R_z \sigma \eqaU \widehat {R_z[\sigma]}$
  by(\ref{substitution composition}), (\ref{A10}) and (\ref{alpha sigma}),
  so $R_w[\sigma] \eqa R_z[\sigma]$.
\end{proof}

\subsection{Agent identifiers without parameters}
\label{sec:piR}\hypertarget{piR}{}

Let $\piR$ be the variant of $\piRI$ in which agent identifiers may be called
only with their own declared names as parameters, i.e.\ such that $\vec{y}=\vec{x}$ in rule \textsc{\textbf{\small ide}}.
As a result, the relabelling operator $[\renbt{y}{x}^\aN]$ in this rule can be dropped.
The operational semantics of $\piR$ is displayed in \tab{pi-relabelling}.

Let $\fT_r:\T_{\piRI} \rightarrow \T_{\piR}$\vspace{-1pt} be the compositional translation satisfying
$\fT_r(A(\vec{y})) := A_r(\vec x)[\renbt{y}{x}^\aN]$, and acting homomorphically on all other operators.
Here $A_r$ is a fresh agent identifier with defining equation
\plat{$A_r(\vec x) \stackrel{{\rm def}}{=} \fT_r(P)$} when
\plat{$A(\vec x) \stackrel{{\rm def}}{=} P$} was the defining equation of $A$.
Clearly $\Fn(\fT_r(P)) =\Fn(P)$.
I will show that $\fT_r$ is a valid translation from $\piRI$ to $\piR$, up to strong
barbed bisimilarity.

\begin{lemma}\label{lem:parameter elimination}
  If $R \goto\beta R'$ and $\bn(\beta)\cap\Fn(R)=\emptyset$ then $\fT_r(R) \goto\beta \fT_r(R')$.
\end{lemma}
\begin{proof}
By induction of the inference of $R \goto{\beta} R'$.
\begin{itemize}
\item Suppose $R \goto{\beta} R'$ is derived by \hyperlink{pi-relabelling}{\textsc{\textbf{\small early input}}}.
  Then $R=x(y).P$, $\beta=xz$ and $R'=P[\renb{z}{y}^\aN]$.
  Now $\fT_r(R) = x(y).\fT_r(P) \goto\beta \fT_r(P)[\renb{z}{y}^\aN] =\fT_r(R')$.
\item Suppose $R \goto{\beta} R'$ is derived by \textsc{\textbf{\small ide}}.
  Then $R=A(\vec{y})$, $\plat{$A(\vec{x})\stackrel{\rm def}{=} P$}$
  and $P[\renbt{y}{x}^\aN] \goto{\beta} R'$. As $\bn(\beta)\cap\Fn(R)\mathbin=\emptyset$
  and $\Fn(P[\renbt{y}{x}^\aN]) \mathbin\subseteq \Fn(R)$, $\bn(\beta)\cap\Fn(P[\renbt{y}{x}^\aN])\mathbin=\emptyset$.
  By induction
  $$\fT_r(P)[\renbt{y}{x}^\aN] = \fT_r(P[\renbt{y}{x}^\aN]) \goto{\beta} \fT_r(R').$$
  so by \hyperlink{pi-relabelling}{\textsc{\textbf{\small relabelling}}}
  $\fT_r(P) \goto{\alpha} P'$ for some $\alpha$ and $P'$ with $\alpha\as{\renbt{y}{x}^\aN}=\beta$ and $P'[\renbt{y}{x}^\aN]=\fT_r(R')$.
  By rule \hyperlink{pi-relabelling}{\textsc{\textbf{\small ide}}} from $\piR$,
  $A_r(\vec{x})  \goto{\alpha} P'$. So by \hyperlink{pi-relabelling}{\textsc{\textbf{\small relabelling}}}
  $\fT_r(R)= A_r(\vec{x}) [\renbt{y}{x}^\aN] \goto\beta P'[\renbt{y}{x}^\aN]=\fT_r(R')$,
  here using the side condition that $\bn(\beta)\cap\Fn(R)=\emptyset$.
\item Using that $\Fn(\fT_r(Q)) = \Fn(Q)$ for all $Q$, all remaining cases are trivial.
\qed
\end{itemize}
\end{proof}

\begin{lemma}\label{lem:parameter introduction}
If $\fT_r(R) \goto\beta U$ then $R \goto\beta R'$ for some $R'$ with $\fT_r(R')=U$.
\end{lemma}
\begin{proof}
By induction of the inference of $\fT_r(R) \goto{\beta} U$.
\begin{itemize}
\item Suppose $R=A(\vec{y})$. Then $\fT_r(R)= A_r(\vec{x}) [\renbt{y}{x}^\aN]$.
  By \hyperlink{pi-relabelling}{\textsc{\textbf{\small relabelling}}} $\bn(\beta)\cap\Fn(\fT_r(R))=\emptyset$ and
  $A_r(\vec{x}) \goto\alpha V$ for some $\alpha$ and $V$ with $\alpha\as{\renbt{y}{x}^\aN}\mathbin=\beta$
  and $V[\renbt{y}{x}^\aN]\mathbin=U$.
  Let $\plat{$A(\vec{x})\stackrel{\rm def}{=} P$}$, so that $\plat{$A_r(\vec{x})\stackrel{\rm def}{=} \fT_r(P)$}$.
  Via rule \hyperlink{pi-relabelling}{\textsc{\textbf{\small ide}}}
  $\fT_r(P) \goto\alpha V$. By induction $P \goto\alpha P'$ for some $P'$ with $\fT_r(P')=V$.
  As $\Fn(P[\renbt{y}{x}^\aN]) \mathbin\subseteq \Fn(R) = \Fn(\fT_r(R))$,
  $\bn(\beta)\cap\Fn(P[\renbt{y}{x}^\aN])=\emptyset$.
  So by \hyperlink{relab}{\textsc{\textbf{\small relabelling}}} $P[\renbt{y}{x}^\aN] \goto\beta P'[\renbt{y}{x}^\aN]$ .
  By rule \textsc{\textbf{\small ide}} from $\piRI$, $R \goto\beta P'[\renbt{y}{x}^\aN]$.
  Moreover, $$\fT_r(P'[\renbt{y}{x}^\aN])= \fT_r(P')[\renbt{y}{x}^\aN] = V[\renbt{y}{x}^\aN] = U.$$
\item Using that $\Fn(\fT_r(Q)) = \Fn(Q)$ for all $Q$, all other cases are trivial.
\qed
\end{itemize}
\end{proof}

\begin{theorem}\label{thm:step 5}
  $\fT_r(P) \sbb P$ for any $P \in \T_{\piRI}$.
\end{theorem}
\begin{proof}
  It suffices to show that the symmetric closure of 
  \[{\R}:=\{(R,\fT_r(R)) \mid R \in \T_{\piRI}\}\]
  is a strong barbed bisimulation. This follows from Lemmas~\ref{lem:parameter elimination}
  and~\ref{lem:parameter introduction}, also using \lem{bound output universality relabelling} to
  handle the barbs from $R$.
\end{proof}

\subsection{Clash-free processes}
\label{sec:piRp}\hypertarget{piRp}{}

For $\pN$ a collection of \emph{private names}, disjoint with $\N\uplus\aN =: \zN$,
I define $\piRp$ as the variant of $\piR[\N\uplus\pN]$ where in all expressions $x(y).Q$ the name
$y$ must be chosen from $\N$, and in defining equations
\plat{$A(\vec{x})\stackrel{\rm def}{=} P$} one has $x_1,\dots,x_n \in \N$.
The semantics of $\piRp$ equals that of $\piR[\N\uplus\pN]$.

Clearly, the relation $\eqa$ and all results from \sect{eqaR} are inherited by $\piRp$.

\begin{definition}\label{df:RN}
  The set $\RN(P)$ of \emph{restriction-bound} names of a $\piRp$ process $P$
  is inductively defined by
  \begin{itemize}
    \item if $R = (\nu y)P$ then $\RN(R) \supseteq \RN(P)\cup\{y\}$,
    \item if $R = \tau.P$ or $\bar x y.P$ or $x(y).P$  then $\RN(R) \supseteq \RN(P)$,
    \item if $R = [x{=}y]P$ then $\RN(R) \supseteq \RN(P)$,
    \item if $R = P[\sigma]$ then $\RN(R) \supseteq \{x\as\sigma \mid x \in \RN(P)\}$,\vspace{1pt}
    \item if $R = P+Q$ or $P|Q$ then $\RN(R) \supseteq \RN(P) \cup \RN(Q)$,
    \item if $R = A(\vec{x})$ and \plat{$A(\vec{x})\stackrel{\rm def}{=} P$} then $\RN(R) \supseteq \RN(P)$.
  \end{itemize}
\end{definition}
Each statement $y \mathbin\in \RN(R)$ can be proven using $y \mathbin\in \RN((\nu y)P)$ as an axiom, and each of the six clauses
above as proof rules. It follows that in fact one has $=$ wherever $\supseteq$ is written in \df{RN}.

\begin{definition}\label{df:clash}
  The collection $\C$ of \emph{clashing} $\piRp$ processes is the smallest such that
  \begin{enumerate}[(i)]
  \item if $\Fn(R) \cap \RN(R) \neq\emptyset$ then $R \in \C$,\label{i}
  \item if $R=P|Q$ and $\RN(P) \cap \RN(Q) \neq\emptyset$ then $R \in \C$,\label{ii}
  \item if $R=(\nu y)P$ and $y \in \RN(P)$ then $R \in \C$,\label{iii}
  \item if $R=P[\sigma]$ and $y\as\sigma=z\as\sigma$ for different $y,z\in\RN(P)$ then $R \in \C$,\label{iv}
  \item if $P \in \C$ then also $\tau.P \in \C$, $\bar xy.P \in \C$, $x(y).P \in \C$, $[x{=}y]P \in \C$,\label{v}
    $(\nu y)P \in \C$ and $P[\sigma]\in\C$,
  \item if $P \in \C$ or $Q \in \C$ then also $P+Q \in \C$ and $P|Q \in \C$,\label{vi}
  \item if \plat{$A(\vec{x})\stackrel{\rm def}{=} P$} and $P \in \C$ then also $A(\vec{x}) \in \C$,\label{vii}
  \item and if $\RN(R)\cap \zN \neq \emptyset$ then $R \in \C$.\label{viii}
  \end{enumerate}
  A process $P$ is \emph{clash-free} if $P \notin \C$.
\end{definition}
Now take $\pN$ as given in \sect{the encoding}, where also the relabelling operators $[\ell]$,
$[r]$, $[e]$ and $[p_y]$ are defined, with $p_y = \drenb{p}{y}\circ e$.
Let the translation $\fT_{\rm cf}:\T_{\piR} \rightarrow \T_{\piRp}$ satisfy
\[\begin{array}{@{}l@{~:=~}ll@{}}
\fT_{\rm cf}((\nu y)P) & (\nu p)(\fT_{\rm cf}(P)[p_y])\\
\fT_{\rm cf}(P\mid Q) & \fT_{\rm cf}(P)[\ell]\mathbin{\|}\fT_{\rm cf}(Q)[r] \\
\fT_{\rm cf}(A(\vec x)) & A_{\rm cf}(\vec x)
\end{array}\]
and act homomorphically on all other operators.
Here $A_{\rm cf}$ is a fresh agent identifier with defining
equation \plat{$A_{\rm cf}(\vec x) \stackrel{{\rm def}}{=} \fT_{\rm cf}(P)$} when
\plat{$A(\vec x) \stackrel{{\rm def}}{=} P$} was the defining equation of $A$.

\begin{lemma}\label{lem:clash-free gen}
  Let $P\in \T_{\piR}$ be a $\piR$ process. Then $\Fn(\fT_{\rm cf}(P)) = \Fn(P) \subseteq \zN$ and
  $\fT_{\rm cf}(P)$ is clash-free.
\end{lemma}
\begin{proof}
  The first statement follows by a straightforward structural induction on $P$.

  Let $\pN_0 := \{{}^{\varsigma}\! p \mid \varsigma\inp\{e,\ell,r\}^*\}$; these are the names
  ${}^{\varsigma}\! p^\upsilon \in\pN$ with $\upsilon=\varepsilon$.
  A straightforward induction on the derivation of $y \in \RN(\fT_{\rm cf}(P))$ from the rules of \df{RN}
  shows that $\RN(\fT_{\rm cf}(P))\subseteq \pN_0$.

  Using these facts to handle Requirements (\ref{i}) and (\ref{viii}), 
  a straightforward induction on the derivation of $\fT_{\rm cf}(P)\in\C$ from the rules of
  \df{clash} leads to a contradiction.
  For (\ref{ii}) use that $\forall y,z\in\pN_0.~y\as\ell \neq z\as{r}$.
  For (\ref{iii}) use that $z\as{p_y} \neq p$ for all $z \in \pN_0$.
  For (\ref{iv}) use that $p_y$, $\ell$ and $r$ are injective substitutions, whereas any relabelling
  operator $[\sigma]$ inherited from $\piR$ satisfies $\dom(\sigma) \cap \pN = \emptyset$.
\end{proof}

I will show that $\fT_{\rm cf}$ is a valid translation, up to strong barbed bisimilarity.

\begin{lemma}\label{lem:restriction translation}
  $\fT_{\rm cf}((\nu y)P) \eqa (\nu y) \fT_{\rm cf}(P)$ for $P\in\T_{\piR}$.
\end{lemma}
\begin{proof}
By definition $\fT_{\rm cf}((\nu y)P) = (\nu p)(\fT_{\rm cf}(P)[p_y])$.
Since $\Fn(\fT_{\rm cf}(P)) \mathbin= \Fn(P) \mathbin\subseteq \zN$ by \lem{clash-free gen},
$y \notin \Fn(\fT_{\rm cf}(P)[p_y])$. 
$(\nu p)(\fT_{\rm cf}(P)[p_y]) \eqa (\nu y)(\fT_{\rm cf}(P)[p_y][\drenb{y}{p}])\eqa (\nu y) \fT_{\rm cf}(P)$ by
\lem{eqaR}(\ref{relabelling conversion},\ref{relabelling composition},\ref{free same},\ref{empty relabelling},\ref{eqa nu}).
\end{proof}

\begin{lemma}\label{lem:clash-free forth}
  If $R \goto\beta R'$ then $\fT_{\rm cf}(R) \goto\beta \eqa \fT_{\rm cf}(R')$.
\end{lemma}
\begin{proof}
By induction of the inference of $R \goto{\beta} R'$.
Since $R$ is a ${\piR}$ process, $\n(\beta)\cup \n(R)\subseteq\zN$.
\begin{itemize}
\item Suppose $R \goto{\beta} R'$ is derived by \hyperlink{pi-relabelling}{\textsc{\textbf{\small early input}}}.
  Then $R=x(y).P$, $\beta=xz$ and $R'=P[\renb{z}{y}^\aN]$. Now
  $\fT_{\rm cf}(R)=x(y).\fT_{\rm cf}(P) \goto\beta \fT_{\rm cf}(P)[\renb{z}{y}^\aN] = \fT_{\rm cf}(R')$.
\item Suppose $R \goto{\beta} R'$ is derived by \hyperlink{pi-relabelling}{\textsc{\textbf{\small ide}}}.
  Then $R=A(\vec x)$. Let \plat{$A(\vec x) \stackrel{{\rm def}}{=} P$},
  so \plat{$A_{\rm cf}(\vec x) \stackrel{{\rm def}}{=} \fT_{\rm cf}(P)$}.
  Then $P \goto{\beta} R'$.
  By induction $\fT_{\rm cf}(P) \goto\beta\eqa \fT_{\rm cf}(R')$.
  Applying rule \hyperlink{pi-relabelling}{\textsc{\textbf{\small ide}}},
  $\fT_{\rm cf}(R) = A_{\rm cf}(\vec x)  \goto\beta\eqa \fT_{\rm cf}(R')$.
\item
  The cases that $R \goto\beta R'$ is derived by \hyperlink{pi-relabelling}{\textsc{\textbf{\small tau}},
  \textsc{\textbf{\small output}}, \textsc{\textbf{\small sum}} or \textsc{\textbf{\small symb-match}}} are (also) trivial.
\item Suppose $R \goto{\beta} R'$ is derived by \hyperlink{pi-relabelling}{\textsc{\textbf{\small par}}}.
  Then $R=P|Q$, $P \goto\beta P'$, $\bn(\beta)\cap\Fn(R)=\emptyset$ and $R'=P'|Q$.
  By induction $\fT_{\rm cf}(P) \goto\beta V \eqa \fT_{\rm cf}(P')$.
  As $\n(\beta)\subseteq\zN$, $\beta[\ell]\mathbin=\beta$.
  Moreover, $\zN \mathbin\supseteq \Fn(P) \mathbin= \Fn(\fT_{\rm cf}(P)) = \Fn(\fT_{\rm cf}(P)[\ell])$
  by \lem{clash-free gen}, so $\bn(\beta)\cap\Fn(\fT_{\rm cf}(P)[\ell])=\emptyset$.
  Thus, $\fT_{\rm cf}(P)[\ell] \goto\beta V[\ell] \eqa \fT_{\rm cf}(P')[\ell]$
  by rule \hyperlink{pi-relabelling}{\textsc{\textbf{\small relabelling}}} and \lem{eqaR}(\ref{eqa nu}).
  Since $\zN \mathbin\supseteq \Fn(Q) \mathbin= \Fn(\fT_{\rm cf}(Q)) = \Fn(\fT_{\rm cf}(Q)[r])$
  by \lem{clash-free gen}, $\bn(\beta)\cap\Fn(\fT_{\rm cf}(Q)[r])\mathbin=\emptyset$.

  Hence, by rule \hyperlink{pi-relabelling}{\textsc{\textbf{\small par}}},
  \(\fT_{\rm cf}(R) = \fT_{\rm cf}(P)[\ell]|\fT_{\rm cf}(Q)[r] \goto\beta V[\ell]|\fT_{\rm cf}(Q)[r] \eqa
  \fT_{\rm cf}(P')[\ell]|\fT_{\rm cf}(Q)[r] = \fT_{\rm cf}(R')\), in the $\eqa$-step using \lem{eqaR}(\ref{eqa par}).
\item Suppose $R \goto{\beta} R'$ is derived by \hyperlink{pi-relabelling}{\textsc{\textbf{\small e-s-com}}}.
  Then $R\mathbin=P|Q$, {$P\goto{M\bar xy} P'$}, {$Q\goto{Nvy} Q'$}, $R'\mathbin=P'|Q'$,
  and $\beta=\match{x}{v}MN\tau$. By induction $\fT_{\rm cf}(P) \goto{M\bar xy} V \eqa \fT_{\rm cf}(P')$
  and $\fT_{\rm cf}(Q) \goto{Nvy} W \eqa \fT_{\rm cf}(Q')$.
  By rule \hyperlink{pi-relabelling}{\textsc{\textbf{\small relabelling}}} and \lem{eqaR}(\ref{eqa nu})
  $\fT_{\rm cf}(P)[\ell] \goto{M\bar xy} V[\ell] \eqa \fT_{\rm cf}(P')[\ell]$
  and $\fT_{\rm cf}(Q)[r] \goto{Nvy} W[r] \eqa \fT_{\rm cf}(Q')[r]$.
  Thus, by \hyperlink{pi-relabelling}{\textsc{\textbf{\small e-s-com}}} and \lem{eqaR}(\ref{eqa par})
  \(\fT_{\rm cf}(R) = \fT_{\rm cf}(P)[\ell]|\fT_{\rm cf}(Q)[r]\linebreak[3] \goto\beta V[\ell]|W[r] \eqa
  \fT_{\rm cf}(P')[\ell]|\fT_{\rm cf}(Q')[r] = \fT_{\rm cf}(R')\).
\item Suppose $R \goto{\beta} R'$ is derived by \hyperlink{pi-relabelling}{\textsc{\textbf{\small e-s-close}}}.
  Then $R\mathbin=P|Q$, {$P\goto{M\bar x(z)} P'$}, {$Q\goto{Nvz} Q'$}, $z \notin\Fn(R)$, $\beta=\match{x}{v}MN\tau$
  and $R'\mathbin=(\nu z)(P'|Q')$. By induction $\fT_{\rm cf}(P) \goto{M\bar x(z)} V \eqa \fT_{\rm cf}(P')$
  and $\fT_{\rm cf}(Q) \goto{Nvz} W \eqa \fT_{\rm cf}(Q')$.
  Since $z \mathbin{\notin} \Fn(P)\mathbin=\Fn(\fT_{\rm cf}(P))\mathbin=\Fn(\fT_{\rm cf}(P)[\ell])$,
  by rule \hyperlink{pi-relabelling}{\textsc{\textbf{\small relabelling}}} and \lem{eqaR}(\ref{eqa nu})
  one obtains
  $\fT_{\rm cf}(P)[\ell] \goto{M\bar x(z)} V[\ell] \eqa \fT_{\rm cf}(P')[\ell]$
  and likewise $\fT_{\rm cf}(Q)[r] \goto{Nvz} W[r] \eqa \fT_{\rm cf}(Q')[r]$.
  Thus, by \hyperlink{pi-relabelling}{\textsc{\textbf{\small e-s-close}}} and \lem{eqaR}(\ref{eqa par},\ref{eqa nu})
  \(\fT_{\rm cf}(R) \mathbin= \fT_{\rm cf}(P)[\ell]|\fT_{\rm cf}(Q)[r]\linebreak \goto\beta (\nu z)(V[\ell]|W[r]) \mathbin{\eqa}
  (\nu z)(\fT_{\rm cf}(P')[\ell]|\fT_{\rm cf}(Q')[r]) \eqa \fT_{\rm cf}(R')\).
  The last step follows by \lem{restriction translation}.
\item Suppose $R \goto{\beta} R'$ is derived by \hyperlink{pi-relabelling}{\textsc{\textbf{\small res-alpha}}}.
  Then $R=(\nu y)P$, $P[\drenb{z}{y}] \goto{\beta} P'$, $R'=(\nu z)P'$, $z \notin\Fn(R)$ and $z \mathbin{\notin} \n(\beta)$.
  By induction $\fT_{\rm cf}(P[\drenb{z}{y}]) \goto{\beta} V \eqa \fT_{\rm cf}(P')$.
  By \lem{clash-free gen} $z\not\in\Fn(\fT_{\rm cf}(R))$.
  Using this, \lem{restriction translation}, \hyperlink{pi-relabelling}{\textsc{\textbf{\small res-alpha}}},
  \lem{eqaR}(\ref{eqa nu}), and the observation that
  $\fT_{\rm cf}(P[\drenb{z}{y}]) \mathbin= \fT_{\rm cf}(P)[\drenb{z}{y}]$,
  one obtains
  \(\fT_{\rm cf}(R) \eqa\linebreak (\nu y)\fT_{\rm cf}(P) \goto{\beta} (\nu z)V \eqa (\nu z)\fT_{\rm cf}(P')
  \eqa  \fT_{\rm cf}(R')\).

  As $P[\drenb{z}{y}] \goto{\beta} P'$ must be derived by rule
  \hyperlink{pi-relabelling}{\textsc{\textbf{\small relabelling}}},
  $\bn(\beta) \cap \Fn(P[\drenb{z}{y}]) = \emptyset$.
  Since $z \notin\bn(\beta)$ this implies $\bn(\beta) \cap \Fn(R) = \emptyset$.
  Using this, \lem{eqa matches} yields $\fT_{\rm cf}(R)\goto{\beta}\eqa \fT_{\rm cf}(R')$.
\item Suppose $R \goto{\beta} R'$ is derived by \hyperlink{pi-relabelling}{\textsc{\textbf{\small symb-open-alpha}}}.
  Then $R \mathbin= (\nu y)P$, $\beta=M\bar x(z)$, $P \goto{M\bar x y} P'$, $y \mathbin{\neq} x$,
  $z \mathbin{\notin} \Fn(R)$, $y \mathbin{\not\in} \n(M)$ and $R' \mathbin= P'[\drenb{z}{y}]$.
  By induction $\fT_{\rm cf}(P) \goto{M\bar x y} V \eqa \fT_{\rm cf}(P')$.
  By \lem{restriction translation} $\fT_{\rm cf}(R) \eqa (\nu y)\fT_{\rm cf}(P)$.
  By Lemma~\ref{lem:clash-free gen} $z \mathbin{\notin} \Fn(R) = \Fn(\fT_{\rm cf}(R)) = \Fn( (\nu y)\fT_{\rm cf}(P))$.
  So by rule \hyperlink{pi-relabelling}{\textsc{\textbf{\small symb-open-alpha}}}
  $(\nu y)\fT_{\rm cf}(P) \goto{M\bar x (z)} V[\drenb{z}{y}]$.\vspace{1pt}
  By \lem{eqa matches} $\fT_{\rm cf}(R) \goto{M\bar x (z)}\eqa V[\drenb{z}{y}]$.
  Moreover, by  \lem{eqaR}(\ref{eqa nu}), $V[\drenb{z}{y}] \eqa \fT_{\rm cf}(P')[\drenb{z}{y}] = \fT_{\rm cf}(R')$.
\item Suppose $R \goto{\beta} R'$ is derived by \hyperlink{pi-relabelling}{\textsc{\textbf{\small relabelling}}}.
  Then $R=P[\sigma]$, $P \goto\alpha P'$, $\alpha\as\sigma=\beta$, $\bn(\beta) \cap \Fn(R)=\emptyset$ and $R'=P'[\sigma]$.
  By induction $\fT_{\rm cf}(P) \goto\alpha V \eqa \fT_{\rm cf}(P')$.
  As $\bn(\beta) \cap \Fn(\fT_{\rm cf}(R))=\emptyset$, by
  \hyperlink{pi-relabelling}{\textsc{\textbf{\small relabelling}}} and  \lem{eqaR}(\ref{eqa nu})
  $\fT_{\rm cf}(R) \mathbin= \fT_{\rm cf}(P)[\sigma] \goto\beta  V[\sigma] \mathbin\eqa \fT_{\rm cf}(P')[\sigma]\linebreak=\fT_{\rm cf}(R')$.
\qed
\end{itemize}
\end{proof}
\vfill

\begin{lemma}\label{lem:clash-free back}
  If $\fT_{\rm cf}(R) \goto\beta U$ with $\ia(\beta) \cup \bn(\beta)\subseteq\zN$ then $R \goto\beta R'$ for some $R'$ with $\fT_{\rm cf}(R') \eqa U$.
\end{lemma}
\begin{proof}
  By induction of the depth of the inference of the transition $\fT_{\rm cf}(R) \goto{\beta} U$.
  Since $R$ is a ${\piR}$ process, using \lem{clash-free gen}, $\Fn(\fT_{\rm cf}(R)) =\Fn(R)\subseteq\zN$.
\begin{itemize}
\item The case that $\fT_{\rm cf}(R) \goto\beta U$ is derived by \hyperlink{pi-relabelling}{\textsc{\textbf{\small early-input}}}
  proceeds as in the proof of \lem{clash-free forth}, but using $\ia(\beta)\mathbin\subseteq\N\!$.
\item Suppose $\fT_{\rm cf}(R) \goto\beta U$ is derived by \hyperlink{pi-relabelling}{\textsc{\textbf{\small ide}}}.
  Then $R\mathbin=A(\vec x)$. Let \plat{$A(\vec x) \mathbin{\stackrel{{\rm def}}{=}} P$},
  so \plat{$A_{\rm cf}(\vec x) \mathbin{\stackrel{{\rm def}}{=}} \fT_{\rm cf}(P)$}.
  Then $\fT_{\rm cf}(P) \mathbin{\goto{\beta}} U\!$.

  By induction $P \goto\beta R'$ for some $R'$ with $\fT_{\rm cf}(R') \eqa U$.
  Applying rule \hyperlink{pi-relabelling}{\textsc{\textbf{\small ide}}},
  $R \goto\beta R'$.
\item The cases that $\fT_{\rm cf}(R) \goto\beta U$ is derived by \hyperlink{pi-relabelling}{\textsc{\textbf{\small tau}},
  \textsc{\textbf{\small output}}, \textsc{\textbf{\small sum}} or \textsc{\textbf{\small symb-match}}} are (also) trivial.
\item Suppose $\fT_{\rm cf}(R) \goto\beta U$ is derived by \hyperlink{pi-relabelling}{\textsc{\textbf{\small par}}}.
  Then $R\mathbin=P|Q$, $\fT_{\rm cf}(R) = \fT_{\rm cf}(P)[\ell] | \fT_{\rm cf}(Q)[r]$, 
  $\fT_{\rm cf}(P) \goto\alpha V$, $\alpha\as{\ell}\mathbin=\beta$, $\bn(\beta)\cap\Fn(\fT_{\rm cf}(R))\mathbin=\emptyset$
  and $U\mathbin=V[\ell]|\fT_{\rm cf}(Q)[r]$.
  Since $\ia(\beta) \cup \bn(\beta)\subseteq\zN$, also $\ia(\alpha) \cup \bn(\alpha)\subseteq\zN$.
  Therefore, by induction, $P \goto\alpha P'$ for some $P'$ with $\fT_{\rm cf}(P') \eqa V$.
  So $\n(\alpha)\subseteq\zN$ and $\alpha = \alpha[\ell]\mathbin=\beta$.
  By rule \hyperlink{pi-relabelling}{\textsc{\textbf{\small par}}}, using that $\bn(\beta)\cap\Fn(R)\mathbin=\emptyset$
  via \lem{clash-free gen}, $R = P|Q \goto\beta P'|Q$.
  By  \lem{eqaR}(\ref{eqa par},\ref{eqa nu}), $\fT_{\rm cf}(P'|Q) = \fT_{\rm cf}(P')[\ell] | \fT_{\rm cf}(Q)[r] \eqa V[\ell]|\fT_{\rm cf}(Q)[r] = U$.
\item Suppose $\fT_{\rm cf}(R) \goto\beta U$ is derived by \hyperlink{pi-relabelling}{\textsc{\textbf{\small e-s-com}}}.
  Then $R\mathbin=P|Q$, $\fT_{\rm cf}(R) \mathbin= \fT_{\rm cf}(P)[\ell] | \fT_{\rm cf}(Q)[r]$, $\beta\mathbin=\match{x}{v}MN\tau$,
  $\fT_{\rm cf}(P)\goto{M\bar xy} V$, $\fT_{\rm cf}(Q)\goto{Nvz} W$, $y\as\ell\mathbin=z\as{r}$ and $U\mathbin=V[\ell]|W[r]$.
  Here I use \lem{free names actions relabelling} to infer
  that $\n(M) \cup \n(N) \cup \{x,y,v\} \subseteq \N$, so that these names are not affected by $\as\ell$ or $\as{r}$.
  It follows that $z=y$.
  By induction $P \goto{M\bar xy} P$ and $Q \goto{Nvy} Q'$ for processes $P',Q'$ with
  $\fT_{\rm cf}(P') \eqa V$ and $\fT_{\rm cf}(Q') \eqa W$.
  Applying \hyperlink{pi-relabelling}{\textsc{\textbf{\small e-s-com}}},
  $R \goto\beta P'|Q'$. Moreover, employing \lem{eqaR}(\ref{eqa par},\ref{eqa nu}),
  $\fT_{\rm cf}(P'|Q') \mathbin= \fT_{\rm cf}(P')[\ell] \mid \fT_{\rm cf}(Q')[r] \eqa V[\ell]|W[r] = U$.

\item Suppose  $\fT_{\rm cf}(R) \goto\beta U$ is derived by \hyperlink{pi-relabelling}{\textsc{\textbf{\small e-s-close}}}.
  Then $R\mathbin=P|Q$, $\fT_{\rm cf}(R) = \fT_{\rm cf}(P)[\ell] | \fT_{\rm cf}(Q)[r]$, 
  $\fT_{\rm cf}(P) \goto\alpha V$,\linebreak $\fT_{\rm cf}(Q) \goto{\gamma} W$, $\beta=\match{x}{v}MN\tau$,
  $\alpha\as{\ell}=M\bar x (z)$, $\gamma\as{r}=Nvz$, $z \notin\Fn(\fT_{\rm cf}(R))$
  and $U=(\nu z)((V[\ell]|W[r])$. By \lem{free names actions relabelling}
  $\{x,v\} \cup \n(M) \cup \n(N)\subseteq \Fn(\fT_{\rm cf}(R)) \subseteq \zN$.
  So $\alpha=M\bar x (w)$ and $\gamma=Nvu$ with $w[\ell]=u[r]=z$.
  Note that $w,u \notin \Fn(\fT_{\rm cf}(R)) \subseteq \zN$. 
  Pick $q \in \zN{\setminus}\Fn(R)$.
  By taking $q\mathbin{\neq} w,u$ one also has $q \mathbin{\notin} \Fn(V | W | V[\ell] | W[r])$.

  By Lemmas~\ref{lem:bound output universality relabelling} and~\ref{lem:input universality relabelling}
  $\fT_{\rm cf}(P) \goto{M\bar x (q)} V' \eqa V\drenb{q}{w}$ and
  $\fT_{\rm cf}(Q) \goto{Nxq} W' \eqa W\drenb{q}{w}$, and the inferences of these
  transitions have a smaller depth than the one of $\fT_{\rm cf}(R) \goto{\beta} U$.
  Therefore, by induction, $P \goto{M\bar x (q)} P'$ and $Q \goto{Nxq} Q'$ for some $P'$ and $Q'$ with $\fT_{\rm cf}(P') \eqa V'$
  and $\fT_{\rm cf}(Q') \eqa W'$.
  By rule \hyperlink{pi-relabelling}{\textsc{\textbf{\small e-s-close}}} $R \mathbin{\goto\beta} (\nu q)(P'|Q')$.

  By \lem{restriction translation} and \lem{eqaR}(\ref{eqa par},\ref{eqa nu},\ref{relabelling composition},%
  \ref{free same},\ref{par rel},\ref{relabelling conversion})
  \[\fT_{\rm cf}((\nu q)(P'|Q'))
    \begin{array}[t]{@{~}c@{~}l@{}}
    \eqa &  (\nu q)\big((\fT_{\rm cf}(P')[\ell] \mid \fT_{\rm cf}(Q')[r])\big) \\
    \eqa &  (\nu q)\big((V'[\ell] \mid W'[r])\big) \\
    \eqa &  (\nu q)\big((V[\drenb{q}{w}][\ell] \mid W[\drenb{q}{u}][r])\big) \\
    \eqa &  (\nu q)\big((V[\ell][\drenb{q}{z}] \mid W[r][\drenb{q}{z}])]\big) \\
    \eqa &  (\nu q)\big((V[\ell] \mid W[r]) [\drenb{q}{z}]\big) \\
    \eqa &  (\nu z)(V[\ell] \mid W[r]) \eqa U.\\
    \end{array}\]
\item Suppose $R$ has the form $(\nu y)P$, i.e., $\fT_{\rm cf}(R) \goto\beta U$ is derived by
  \hyperlink{pi-relabelling}{\textsc{\textbf{\small res-alpha}}} or \hyperlink{pi-relabelling}{\textsc{\textbf{\small symb-open-alpha}}}.
  By \lem{restriction translation} $\fT_{\rm cf}(R) \eqa (\nu y)\fT_{\rm cf}(P)$,
  so by \lem{eqa matches} $(\nu y)\fT_{\rm cf}(P) \goto\beta V \eqa U$, and the inference of this
  transition has the same depth as the one of $\fT_{\rm cf}(R) \goto{\beta} U$.\vspace{1ex}

  First suppose $(\nu y)\fT_{\rm cf}(P) \goto\beta V$ is derived by \hyperlink{pi-relabelling}{\textsc{\textbf{\small res-alpha}}}.
  Then $\fT_{\rm cf}(P[\drenb{z}{y}]) = \fT_{\rm cf}(P)[\drenb{z}{y}] \goto\beta W$ and $V=(\nu z)W$ for
  some $W$ and $z \notin \Fn(R) \cup \n(\beta)$.
  So by induction $P[\drenb{z}{y}]\goto\beta P'$ for some $P'$ with $\fT_{\rm cf}(P') \eqa W$.
  By \hyperlink{pi-relabelling}{\textsc{\textbf{\small res-alpha}}}
  $R = (\nu y)P \goto\beta (\nu z)P'$. Moreover,
  by \lem{restriction translation} and  \lem{eqaR}(\ref{eqa nu}),
  $$\fT_{\rm cf}((\nu z)P') \eqa (\nu z)\fT_{\rm cf}(P') \eqa (\nu z)W = V \eqa U.$$
  Next suppose  $(\nu y)\fT_{\rm cf}(P) \goto\beta V$ is derived by \hyperlink{pi-relabelling}{\textsc{\textbf{\small symb-open-alpha}}}.
  Then $\beta\mathbin=M\bar x (z)$, $\fT_{\rm cf}(P) \goto{M\bar x y} W$,
  $y\mathbin{\neq} x$, $z \notin\Fn(R)$, $y\notin\n(M)$ and $V=W[\drenb{z}{y}]$.
  By induction $P \goto{M\bar x y} P'$ for some $P'$ with $\fT_{\rm cf}(P') \eqa W$.
  By \hyperlink{pi-relabelling}{\textsc{\textbf{\small res-alpha}}} $R \goto\beta P'[\drenb{z}{y}]$.
  By \lem{eqaR}(\ref{eqa nu}),
  \[\fT_{\rm cf}(P'[\drenb{z}{y}]) \mathbin= \fT_{\rm cf}(P')[\drenb{z}{y}] \mathbin\eqa W[\drenb{z}{y}] \mathbin= V \mathbin\eqa U.\]
\item Suppose  $\fT_{\rm cf}(R) \goto\beta U$ is derived by \hyperlink{pi-relabelling}{\textsc{\textbf{\small relabelling}}}.
  Then $R=P[\sigma]$, $\fT_{\rm cf}(R)=\fT_{\rm cf}(P)[\sigma]$, $\fT_{\rm cf}(P) \goto\alpha V$,
  $\alpha\as\sigma=\beta$, $\bn(\beta) \cap \Fn(R)=\emptyset$ and $U=V[\sigma]$.
  Since the relabelling operator $[\sigma]$ stems from $\piR$, $\dom(\sigma) \subseteq \zN$.  
  Hence $\ia(\beta) \cup \bn(\beta)\subseteq\zN$ implies $\ia(\alpha) \cup \bn(\alpha)\subseteq\zN$.
  By induction $P \goto\alpha P'$ for some $P'$ with $\fT_{\rm cf}(P') \eqa V$.
  By \hyperlink{pi-relabelling}{\textsc{\textbf{\small relabelling}}}
  $R = P[\sigma] \goto\beta  P'[\sigma]$.
  Moreover, by \lem{eqaR}(\ref{eqa nu}), $\fT_{\rm cf}(P'[\sigma]) = \fT_{\rm cf}(P')[\sigma] \eqa V[\sigma]=U$.
\qed
\end{itemize}
\end{proof}

\begin{theorem}
  $\fT_{\rm cf}(P) \sbb P$ for any $P \in \T_{\piR}$.
\end{theorem}
\begin{proof}
  It suffices to show that the symmetric closure of 
  \[{\R}:=\{(R,U) \mid R \in \T_{\piR} \wedge U \eqa \fT_{\rm cf}(R)\}\]
  is a strong barbed bisimulation, as it contains $(P,\fT_{\rm cf}(P))$.

  Let $(R,U)\mathbin\in{\R}$.
  Let $R\goto\tau R'$.
  Then $\fT_{\rm cf}(R)\goto\tau \eqa \fT_{\rm cf}(R')$ by \lem{clash-free forth}
  and $U\goto\tau \eqa \fT_{\rm cf}(R')$ by \lem{eqa matches}.
  Thus $U\goto\tau U'$ for some $U'$ with $(R',U')\in{\R}$.

  Let $U\goto\tau U'$.
  Then $\fT_{\rm cf}(R)\goto\tau \eqa U'$ by \lem{eqa matches}.
  By \lem{clash-free back} $R\goto\tau R'$ for some $R'$ with $\fT_{\rm cf}(R') \eqa U'$.

  Let $U{\downarrow_b}$ with $b\inp \nN\cup\overline\nN$.
  Then $U\goto{by} U'$ or $U\goto{b(y)} U'$ for some $y$ and $U'$, using the definition of $O$ in \sect{barbed}.
  By Lemmas~\ref{lem:input universality confluence} and~\ref{lem:bound output universality relabelling} I may assume,
  without loss of generality, that $y \in\zN{\setminus}\Fn(U)$ in case of input or bound output actions.
  Hence $\fT_{\rm cf}(R) \goto{by}$ or $\fT_{\rm cf}(R) \goto{b(y)}$ by \lem{eqa matches}.
  So $R\goto{by}$ or $R\goto{b(y)}$ by \lem{clash-free back}.
  Thus $R{\downarrow_b}$.

  The implication $R{\downarrow_b} \Rightarrow U{\downarrow_b}$ proceeds likewise.
\end{proof}

\subsection{Eliminating \texorpdfstring{$\alpha$}{alpha}-conversion for clash-free processes}
\label{sec:piRa}\hypertarget{piRa}{}

Let $\goto{\alpha}_\bullet$ be the transition relation on \plat{$\T_{\piRp}$} obtained by
from the transition relation $\goto\alpha$ by replacing the rules 
\hyperlink{pi-relabelling}{\textsc{\textbf{\small res-alpha}} and \textsc{\textbf{\small symb-open-alpha}}}
from \tab{pi-relabelling} by the rules
rules \hyperlink{ES}{\textsc{\textbf{\small res}} and \textsc{\textbf{\small symb-open}}} from \tab{pi-early-symbolic}.
Thus, the $\alpha$-conversion implicit in these rules has been removed.

As witnessed by \ex{alpha}, on $\piRp$ the transition relation $\goto{\alpha}_\bullet$ differs
essentially from $\goto{\alpha}$, and fails to derive some crucial transitions.
However, I will show that restricted to the clash-free processes within $\piRp$
this transition relation is just as suitable as $\goto{\alpha}$; up to strong barbed bisimilarity
there is no difference. Since the only the clash-free processes in $\piRp$ occur as the target of
the previous translation step, the relation $\goto{\alpha}_\bullet$ may be used in the source of the
next.

\begin{lemma}\label{lem:restriction bound names}
If $P\goto{M\bar x(y)}_\bullet Q$ then $y\in\RN(P)$.
\end{lemma}
\begin{proof}
  A trivial induction on the inference of $P\mathbin{\goto{M\bar x(y)}_\bullet} Q$.
\end{proof}

\begin{lemma}\label{lem:alpha introduce}
If $R$ is clash-free and $R \mathbin{\goto{\alpha}_\bullet} R'$ then $R \mathbin{\goto{\alpha}\eqa} R'\!$.
\end{lemma}

\begin{proof}
By induction on the derivation of $R \goto{\alpha}_\bullet R'$.
\begin{itemize}
\item Suppose $R \goto{M\bar x (y)}_\bullet R'$ is derived by \hyperlink{ES}{\textsc{\textbf{\small symb-open}}}.\vspace{1pt}
  Then $R=(\nu y)P$, $P \goto{M\bar x y}_\bullet R'$, $y \neq x$ and $y \notin \n(M)$.
  By induction $P \goto{M\bar x y} U$ for a process $U\eqa R'$.
  So by \hyperlink{pi-relabelling}{\textsc{\textbf{\small symb-open-alpha}}} $R \goto{M\bar x (y)} U[\drenb{y}{y}]$.
  Moreover, $U[\drenb{y}{y}] \eqa U \eqa R'$ by \lem{eqaR}(\ref{empty relabelling},\ref{free same}).
\item Suppose $R \goto{\alpha}_\bullet R'$ is derived by \hyperlink{ES}{\textsc{\textbf{\small res}}}.
  Then $R=(\nu y)P$, $P \goto\alpha_\bullet P'$, $R'=(\nu y)P'$ and $y\notin \n(\alpha)$.
  For $u \in\bn(\alpha)$, by \lem{restriction bound names} $u\mathbin\in\RN(P)$, so
  by the clash-freedom of $P$, $u\notin \Fn(P)$.
  By induction $P \goto\alpha V$ for some $V \eqa P'$.
  So by \hyperlink{pi-relabelling}{\textsc{\textbf{\small relabelling}}}
  $P[\drenb{y}{y}] \goto\alpha V[\drenb{y}{y}]$ and
  by rule \hyperlink{pi-relabelling}{\textsc{\textbf{\small res-alpha}}}
  $R \goto\alpha (\nu y)(V[\drenb{y}{y}])$.
  By \lem{eqaR}(\ref{empty relabelling},\ref{free same},\ref{eqa nu})
  $(\nu y)(V[\drenb{y}{y}]) \eqa (\nu y)V \eqa (\nu y)P' = R'$
\item All other cases are trivial with \lem{eqaR}(\ref{eqa par},\ref{eqa nu}).
\qed
\end{itemize}
\end{proof}

\begin{lemma}\label{lem:alpha eliminate}
Let $R$ be clash-free.
If $R \goto{\alpha} R'$, $\bn(\alpha)=\emptyset$ and $\ia(\alpha) \cap \RN(R)\mathbin=\emptyset$, then $R \goto{\alpha}_\bullet U$
for some $U$ with $R' \mathbin\eqa U$.

If $R \mathbin{\goto{M\bar x(z)}} R'$ then $R \mathbin{\goto{M \bar x(y)}_\bullet} U$
for some $y$ and $U$ with $R' \mathbin{\eqa} U[\drenb{z}{y}]$ and $z \notin \Fn(U){\setminus}\{y\}$.
\end{lemma}

\begin{proof}
By induction on the depth of the derivation of $R \mathbin{\goto{\alpha}} R'\!$.
\begin{itemize}
\item The cases that  $R \goto{\alpha}_\bullet R'$ is derived by 
  \hyperlink{pi-relabelling}{\textsc{\textbf{\small tau}}, \textsc{\textbf{\small output}},
  \textsc{\textbf{\small early-input}}, \textsc{\textbf{\small sum}},
  \textsc{\textbf{\small symb-match}} or \textsc{\textbf{\small ide}}} are trivial.
\item Suppose $R \goto{\alpha} R'$ is derived by \hyperlink{pi-relabelling}{\textsc{\textbf{\small par}}}.
  Then $R=P|Q$, $P \goto{\alpha} P'$, $\bn(\alpha) \cap \Fn(Q)= \emptyset$ and $R'=P'|Q$.

  First assume $\bn(\alpha)=\emptyset$.
  Since $\ia(\alpha) \cap \RN(R)\mathbin=\emptyset$, also $\ia(\alpha) \cap \RN(P)\mathbin=\emptyset$.
  By induction $P \goto{\alpha}_\bullet V$ for some
  $V$ with $P'\eqa V$. By rule \hyperlink{pi-relabelling}{\textsc{\textbf{\small par}}}
  $R \goto{\alpha}_\bullet V|Q$. Moreover, $R' = P'|Q  \eqa V|Q$ by \lem{eqaR}(\ref{eqa par}).

  Next assume $\alpha \mathbin= M\bar x(z)$.  So $z \notin \Fn(R)$. By induction $P \goto{M\bar x (y)}_\bullet V$ for some
  $y,V$ with $P'\eqa V[\drenb{z}{y}]$. By \lem{restriction bound names}, $y \in \RN(P)$, so $y \notin \Fn(R)$ by the
  clash-freedom of $R$. Hence $R \goto{M\bar x (y)}_\bullet V|Q$ by rule 
  \hyperlink{pi-relabelling}{\textsc{\textbf{\small par}}}.
  Moreover, $R' = P'|Q  \eqa V[\drenb{z}{y}]|Q[\drenb{z}{y}] \eqa (V|Q)[\drenb{z}{y}]$
  by \lem{eqaR}(\ref{empty relabelling},\ref{free same},\ref{eqa par},\ref{par rel}),
  using that $y,z \notin \Fn(Q)$.
  Finally, by \lem{free names successors relabelling}, $\Fn(V|Q){\setminus}\{y\} \subseteq \Fn(R) \not\ni z$.
\item Suppose $R \goto{\alpha} R'$ is derived by rule \hyperlink{pi-relabelling}{\textsc{\textbf{\small e-s-com}}}.
  Then $R=P|Q$, $\alpha=\match{x}{v}MN\tau$, $P \goto{M\bar x y} P'$, $Q \goto{M v y} Q'$ and $R'=P'|Q'$.
  By \lem{free names actions relabelling} $y \in \Fn(P)$. So by the clash-freedom of $R$, $y \mathbin{\notin} \RN(Q)$.
  By induction $P \goto{M\bar x y}_\bullet V$ and $Q \goto{M v y}_\bullet W$ with $P \eqa V$ and $Q' \eqa W$.
  Hence by \hyperlink{pi-relabelling}{\textsc{\textbf{\small e-s-com}}} $R \goto{\alpha}_\bullet V|W$, and
  $R' = P'|Q'  \eqa V|W$.
\item Suppose $R \goto{\alpha} R'$ is derived by rule \hyperlink{pi-relabelling}{\textsc{\textbf{\small e-s-close}}}.
  Then $R=P|Q$, $\alpha=\match{x}{v}MN\tau$, $P \goto{M\bar x (z)} P'$, $Q \goto{M v z} Q'$, $z \notin\Fn(R)$
  and $R'\mathbin=(\nu z)(P'|Q')$.
  By induction, one has $P \goto{M\bar x (y)}_\bullet V$ for some $y$ and $V$ with $P' \mathbin\eqa V[\drenb{z}{y}]$
  and $z \mathbin{\notin} \Fn(V){\setminus}\{y\}$. By \lem{restriction bound names} $y\mathbin\in\RN(P)$, so
  by the clash-freedom of $R$, $y\notin \RN(Q)$ and $y\notin \Fn(R)$.
  By \lem{input universality relabelling} $Q \goto{M v y} W_y$ for some $W_y \mathbin\eqa Q'[\drenb{y}{z}]$.
  So by induction $Q \goto{Mv y}_\bullet W$ for some $W \eqa W_y$.
  By \hyperlink{pi-relabelling}{\textsc{\textbf{\small e-s-close}}} $R \goto\alpha_\bullet (\nu y)(V|W)$.
  Moreover,
  \[\begin{array}{@{}c@{~}c@{~}l@{}}
  R' = (\nu z)(P'|Q') & \eqa &  (\nu z)(V[\drenb{z}{y}] \mid Q'[\drenb{z}{y}][\drenb{z}{y}]) \\
  & \eqa & (\nu z)(V[\drenb{z}{y}] \mid W_y[\drenb{z}{y}]) \\
  & \eqa & (\nu z)((V|W )[\drenb{z}{y}]) \\
  & \eqa & (\nu y)(V|W)
  \end{array}\]
  by \lem{eqaR}(\ref{eqa nu},\ref{free same},\ref{relabelling composition},\ref{empty relabelling},\ref{eqa nu},%
  \ref{par rel},\ref{relabelling conversion}),
  in the last step using that $z \notin\Fn(R) \supseteq \Fn((\nu y)(V|W))$ by \lem{free names successors relabelling}.
\item Suppose $R \goto{\alpha} R'$ is derived by \hyperlink{pi-relabelling}{\textsc{\textbf{\small relabelling}}}.
  Then $R=P[\sigma]$, $P\goto{\beta} P'$, $\beta\as\sigma = \alpha$ and $R'=P'[\sigma]$.

  First assume $\bn(\alpha)=\emptyset$. 
  Since $\ia(\alpha) \cap \RN(R)\mathbin=\emptyset$, also $\ia(\beta) \cap \RN(P)\mathbin=\emptyset$.
  By induction $P \goto{\beta}_\bullet V$ for some
  $V$ with $P'\eqa V$. Now $R \goto{\alpha}_\bullet V[\sigma]$
  by \hyperlink{pi-relabelling}{\textsc{\textbf{\small relabelling}}}.
  Moreover, $R' \mathbin= P'[\sigma] \mathbin\eqa V[\sigma]$ by \lem{eqaR}(\ref{eqa nu}).

  Next assume $\alpha \mathbin= M\bar x(z)$ and $\beta \mathbin= K\bar q (w)$. So $w\as{\sigma}\mathbin=z$.
  Then $z \notin \Fn(R)$ by the side-condition of \hyperlink{pi-relabelling}{\textsc{\textbf{\small relabelling}}}.
  By induction $P \goto{K\bar q (v)}_\bullet V$ for some
  $v$ and $V$ with $P'\eqa V[\drenb{w}{v}]$.
  Let $y:=v\as\sigma$.
  By \lem{restriction bound names}, $v \in \RN(P)$, so $y \in \Fn(R)$.
  Hence $y \notin \Fn(R)$ by the clash-freedom of $R$.
  So there is no $n \in\Fn(P)$ with $n\as\sigma=y$ or $n\as\sigma=z$.
  Since $\Fn(V) \subseteq \Fn(P) \cup \{v\}$ by \lem{free names successors relabelling},
  the name $v$ is the only possible $n \in\Fn(V)$ with $n\as\sigma=y$ or $n\as\sigma=z$.
  By rule \hyperlink{pi-relabelling}{\textsc{\textbf{\small relabelling}}}
  $R \goto{M\bar x (y)}_\bullet V[\sigma]$.
  \lem{eqaR}(\ref{eqa nu},\ref{relabelling composition},\ref{free same}) yields
  $R' = P'[\sigma] \eqa V[\drenb{w}{v}][\sigma] \eqa V[\sigma][\drenb{z}{y}]$.
  Finally, by \lem{free names successors relabelling}, $\Fn(V[\sigma]){\setminus}\{y\} \subseteq \Fn(R) \not\ni z$.

\item Suppose $R \goto{M\bar x (z)} R'$ is derived by \hyperlink{pi-relabelling}{\textsc{\textbf{\small symb-open-alpha}}}.
  Then $R=(\nu y)P$, $P \goto{M\bar x y} P'$, $R'=P'[\drenb{z}{y}]$,\vspace{1pt} $y \neq x$, $z \notin \Fn(R)$ and $y \notin \n(M)$.
  By induction $P \goto{M\bar x y}_\bullet U$ for a process $U$ with $P' \mathbin\eqa U$.
  So by \hyperlink{ES}{\textsc{\textbf{\small symb-open}}} $R \mathbin{\goto{M\bar x (y)}_\bullet} U\!$.
  Moreover, $R' = P'[\drenb{z}{y}] \eqa U[\drenb{z}{y}]$ by \lem{eqaR}(\ref{eqa nu}).
  Finally, by \lem{free names successors relabelling}, $\Fn(U){\setminus}\{y\} \subseteq \Fn(R) \not\ni z$.

\item Suppose $R \goto{\alpha} R'$ is derived by rule \hyperlink{pi-relabelling}{\textsc{\textbf{\small res-alpha}}}.
  Then $R=(\nu y)P$, $P[\drenb{z}{y}] \goto\alpha P'$, $R'=(\nu z)P'$ and $z\notin \Fn(R)\cup\n(\alpha)$.

  First assume $\bn(\alpha)=\emptyset$.
  As $\ia(\alpha) \cap \RN(R)\mathbin=\emptyset$, $y \notin \ia(\alpha)$ and $\ia(\alpha) \cap \RN(P)\mathbin=\emptyset$.
  Hence $y \notin \n(\alpha)$ by \lem{free names actions relabelling}.
  By \hyperlink{pi-relabelling}{\textsc{\textbf{\small relabelling}}} $P\goto\alpha V$
  for a $V$ with $V[\drenb{z}{y}]=P'$.
  So by induction $P \goto{\alpha}_\bullet W$ for some $W\eqa V$. 
  So by \hyperlink{ES}{\textsc{\textbf{\small res}}} $R \goto{\alpha}_\bullet (\nu y)W$.
  By \lem{free names successors relabelling}
  $\Fn(V){\setminus}\{y\} \subseteq \n(\alpha)\cup\Fn(P){\setminus}\{y\}=\n(\alpha)\cup\Fn(R) \not\ni z$.
  Hence $R' = (\nu z)P' = (\nu z)(V[\drenb{z}{y}]) \eqa (\nu y)V$ by \lem{eqaR}(\ref{relabelling conversion}).

  Next assume $\alpha = M\bar x(w)$. By \lem{free names actions relabelling}
  $\n(M) \cup \{x\} \subseteq \Fn(R) \not\ni y$.
  So by \hyperlink{pi-relabelling}{\textsc{\textbf{\small relabelling}}} $P\goto{M\bar x (u)} W$
  for a $W$ with $W[\drenb{z}{y}]=P'$. Here $u:=w\as{\renb{z}{y}}$.
  By induction $P\goto{M\bar x (v)}_\bullet V$ for some $v$ and $V$ with $W \eqa V[\drenb{u}{v}]$
  and $w \notin \Fn(V){\setminus}\{v\}$.
  By \lem{restriction bound names} $v \in\RN(P)$. So $v\neq y$ by the clash-freedom of $R$.
  Hence $R\goto{M\bar x (v)}_\bullet (\nu y)V$ by \hyperlink{ES}{\textsc{\textbf{\small res}}}.
  Pick a name $q \notin \Fn(R') \cup \Fn((\nu y)V) \cup \{w,v\}$.
  Now
  \[\begin{array}{@{}c@{~}c@{~}l@{}}
  R' = (\nu z)P' & = & (\nu z)(W[\drenb{z}{y}])  \\
  & \eqa & (\nu q)(W[\drenb{z}{y}][\drenb{q}{z}])  \\
  & \eqa & (\nu q)(V[\drenb{u}{v}][\drenb{z}{y}][\drenb{q}{z}]) \\
  & \eqa & (\nu q)(V[\drenb{q}{y}][\drenb{w}{v}])  \\
  & \eqa & ((\nu q)(V[\drenb{q}{y}]))[\drenb{w}{v}]  \\
  & \eqa & ((\nu y)V)[\drenb{w}{v}]
  \end{array}\]
  by \lem{eqaR}(\ref{relabelling conversion},\ref{eqa nu},\ref{relabelling composition},\ref{free same},%
  \ref{nu rel},\ref{relabelling conversion},\ref{eqa nu}), provided that
  \begin{equation}\label{ab}
  n\as{\drenb{u}{v}}\as{\drenb{z}{y}}\as{\drenb{q}{z}} = n\as{\drenb{q}{y}}\as{\drenb{w}{v}}
  \end{equation}
  for all $n\mathbin\in\Fn(V)$. This has to be checked only for $n\in\{y,v\}$, because $q \mathbin{\notin} \Fn(V){\setminus}\{y\}$,
  $w \mathbin{\notin} \Fn(V){\setminus}\{v\}$,
  $u$ is either $w$ or $z$, and $z \notin \Fn(R) \supseteq \Fn(V){\setminus}\{y,v\}$ by \lem{free names successors relabelling}.
  To check (\ref{ab}) I recall that $y \mathbin{\neq} v \mathbin{\neq} q \mathbin{\neq} w \mathbin{\neq} z$ and consider two cases.
  \begin{itemize}
  \item Let $w \neq y$. Then $u=w$, and (\ref{ab}) holds for $n=y,v$.
  \item Let $w = y$. Then $u=z\neq y \neq q$ and again (\ref{ab}) holds for $n=y,v$.
  \end{itemize}
  Finally, $w \notin \Fn((\nu y)V){\setminus}\{v\}$.
\qed
\end{itemize}
\end{proof}

Below a substitution $\sigma$ is called \emph{clash-free} on a $\piRp$ process $P$ iff
$\RN(P) \cap (\dom(\sigma) \cup {\it range}(\sigma)) = \emptyset$.

\begin{observation}\label{obs:clash-free substitution relabelling}
If $P$ is clash-free and $\sigma$ is clash-free on $P$, then $P[\sigma]$ is clash-free and $\RN(P[\sigma])=\RN(P)$.
\end{observation}

\begin{lemma}\label{lem:clash-free early-input relabelling}
  If $x(y).P$ is clash-free and $z\notin\RN(x(y).P)$ then substitution $\renb{z}{y}^\aN\!$
  from \hyperlink{pi-relabelling}{\textsc{\textbf{\small early-input}}} is clash-free on $P\!$.
\end{lemma}
\begin{proof}
  By definition $\dom(\renb{z}{y}^\aN) \subseteq \zN$ and ${\it range}(\renb{z}{y}^\aN) \subseteq \zN \cup \{z\}$.
  Since $x(y).P$ is clash-free, so is process $P$ by \df{clash}(\ref{v}), and $\RN(P)=\RN(x(y).P)$ by \df{RN}.
  Hence $\RN(P)\subseteq \pN$ and $z \notin \RN(x(y).P)=\RN(P)$.
\end{proof}

\begin{lemma}\label{lem:preservation of clash-freedom}
If $R \mathbin{\goto{\alpha}_\bullet} R'$ with $\ia(\alpha)\cap\RN(R)=\emptyset$ and $R$ is clash-free, then so is $R'$.
Moreover, $\RN(R') \subseteq \RN(R)$ and $\bn(\alpha) \cap \RN(R')=\emptyset$.
\end{lemma}
\begin{proof}
By induction on the derivation of $R \mathbin{\goto{\alpha}_\bullet} R'$.
Since $R$ is clash-free, $\RN(R) \subseteq \pN$.
\begin{itemize}
\item By \df{clash}(iv-v) a subexpression $P$ of a clash-free\linebreak process $R$ is also clash-free,
  and by \df{RN} $\RN(P)\subseteq\RN(R)$.
  Using this, the cases that  $R \goto{\alpha}_\bullet R'$ is derived by 
  \hyperlink{pi-relabelling}{\textsc{\textbf{\small tau}}, \textsc{\textbf{\small output}},
  \textsc{\textbf{\small sum}} or \textsc{\textbf{\small symb-match}}} are trivial.
\item Suppose $R \goto{\alpha}_\bullet R'$ is derived by \hyperlink{pi-relabelling}{\textsc{\textbf{\small early-input}}}.
  Then $R\mathbin=x(y).P$, $\alpha\mathbin=xz$ with $z \mathbin{\notin} \RN(R)$, and $R'\mathbin=P[\renb{z}{y}^\aN]$.
  By \df{clash}(\ref{v}) $P$ is clash-free and by \df{RN} $\RN(P)=\RN(R)$.
  By \lem{clash-free early-input relabelling} and \obs{clash-free substitution relabelling}
  $R'$ is clash-free and $\RN(R') = \RN(R)$. Moreover, $\bn(\alpha)=\emptyset$.
\item Suppose $R \goto{\alpha}_\bullet R'$ is derived by \hyperlink{pi-relabelling}{\textsc{\textbf{\small ide}}}.
  Then $R=A(\vec x)$. Let \plat{$A(\vec x) \stackrel{{\rm def}}{=} P$}.
  Then $P \goto\alpha_\bullet R'$.
  Process $P$ is clash-free by \df{clash}(\ref{vii}), and $\RN(P)\mathbin=\RN(R)$ by \df{RN}.
  By induction $\RN(R') \subseteq \RN(P)= \RN(R)$, process $R'$ is clash-free,
  and $\bn(\alpha) \cap \RN(R')=\emptyset$.
\item Suppose $R \goto{\alpha}_\bullet R'$ is derived by \hyperlink{pi-relabelling}{\textsc{\textbf{\small par}}}.
  Then $R=P|Q$, $P \goto{\alpha}_\bullet P'$ and $R'=P'|Q$.
  Since $\ia(\alpha)\cap\RN(R)=\emptyset$ and $\RN(P)\subseteq\RN(R)$, also $\ia(\alpha)\cap\RN(P)=\emptyset$.
  Using that $P$ is clash-free, by induction $P'$ is clash-free and $\RN(P')\subseteq \RN(P)$.
  Moreover, $\bn(\alpha) \cap \RN(P')=\emptyset$.
  Now $\RN(R') \mathbin= \RN(P') \cup \RN(Q) \mathbin\subseteq \RN(P) \cup \RN(Q) \mathbin= \RN(R)$.\vspace{1ex}

  Let $z \in \bn(\alpha)$. Then $z \in \RN(P)$ by \lem{restriction bound names},
  so by the clash-freedom of $R$ one has $z \notin\RN(Q)$.
  Moreover, $z \notin \RN(P')$ by the above, so $z \notin \RN(R')$.\vspace{1ex}

  To show that $R'$ is clash-free, using that $P'$ and $Q$ are clash-free, it suffices to show
  (\ref{i}) $\Fn(R')\cap \RN(R')=\emptyset$ and (\ref{ii}) $\RN(P')\cap \RN(Q) = \emptyset$.
  The latter follows since $\RN(P)\cap \RN(Q) = \emptyset$ by the clash-freedom of $R$.
  For the former, \lem{free names successors relabelling} yields
  $\Fn(R') \subseteq \Fn(R) \cup \ia(\alpha) \cup \bn(\alpha)$.
  Moreover  $\Fn(R)\cap \RN(R)=\emptyset$ by the clash-freedom of $R$,
  $\ia(\alpha)\cap\RN(R)\mathbin=\emptyset$ by assumption, and $\bn(\alpha)\cap\RN(R')\mathbin=\emptyset$ as derived above.
  This entails (\ref{i}).
\item Suppose $R \goto{\alpha}_\bullet R'$ is derived by rule \hyperlink{pi-relabelling}{\textsc{\textbf{\small e-s-com}}}.
  Then $R=P|Q$, $\alpha=\match{x}{v}MN\tau$, $P \goto{M\bar x y}_\bullet P'$, $Q \goto{M v y}_\bullet Q'$ and $R'=P'|Q'$.
  By \lem{free names actions relabelling} $y \in \Fn(P)$, so $y \notin \RN(Q)$ by the clash-freedom of $R$.
  Using that $P$ and $Q$ are clash-free, by induction $P'$  and $Q'$ are clash-free, $\RN(P')\subseteq \RN(P)$
  and $\RN(Q')\subseteq \RN(Q)$. It follows that
  $\RN(R') \mathbin= \RN(P') \cup \RN(Q') \mathbin\subseteq \RN(P) \cup \RN(Q) \mathbin= \RN(R)$.\vspace{1ex}

  To show that $R'$ is clash-free, using that $P'$ and $Q'$ are clash-free, it suffices to show
  (\ref{i}) $\Fn(R')\cap \RN(R')=\emptyset$ and (\ref{ii}) $\RN(P')\cap \RN(Q') = \emptyset$.
  The latter follows since $\RN(P)\cap \RN(Q) = \emptyset$ by the clash-freedom of $R$.
  For the former, \lem{free names successors relabelling} yields $\Fn(R') \subseteq \Fn(R)$.
  As $\Fn(R)\cap \RN(R)=\emptyset$ by the clash-freedom of $R$, this entails (\ref{i}).
\item Suppose $R \goto{\alpha}_\bullet R'$ is derived by rule \hyperlink{pi-relabelling}{\textsc{\textbf{\small e-s-close}}}.
  Then $R=P|Q$, $\alpha=\match{x}{v}MN\tau$, $P \goto{M\bar x (z)}_\bullet P'$, $Q \goto{M v z}_\bullet Q'$
  and $R'=(\nu z)(P'|Q')$.
  By \lem{restriction bound names} $z \in \RN(P)$, so $y \notin \RN(Q)$ by the clash-freedom of $R$.
  Using that $P$ and $Q$ are clash-free, by induction $P'$  and $Q'$ are clash-free, $\RN(P')\subseteq \RN(P)$,
  $\RN(Q')\subseteq \RN(Q)$ and $z \notin\RN(P')$. It follows that
  $\RN(R') = \RN(P') \cup \RN(Q') \cup \{z\} \subseteq \RN(P) \cup \RN(Q) = \RN(R)$.\vspace{1ex}

  To show that $R'$ is clash-free, using that $P'$ and $Q'$ are clash-free, it suffices to show
  (\ref{i}) $\Fn(R')\cap \RN(R')=\emptyset$, (\ref{ii}) $\RN(P')\cap \RN(Q') = \emptyset$ and
  (\ref{iii}) $z \notin \RN(P'|Q')$. That (\ref{i}) and (\ref{ii}) hold follows exactly as in the case of
  \hyperlink{pi-relabelling}{\textsc{\textbf{\small e-s-com}}} above.
  For (\ref{iii}), using that $z \mathbin\in \RN(P)$ and $z \mathbin{\notin}\RN(P')$, in case $z\in \RN(P'|Q')$ then
  $z \in \RN(Q') \subseteq \RN(Q)$, contradicting the clash-freedom of $R$.

\item Suppose $R \goto{\alpha}_\bullet R'$ is derived by \hyperlink{ES}{\textsc{\textbf{\small res}}}.
  Then $R=(\nu y)P$, $P \goto{\alpha}_\bullet P'$, $y \notin\n(\alpha)$ and $R'=(\nu y)P'$.
  Since $\ia(\alpha)\cap\RN(R)=\emptyset$ and $\RN(P)\subseteq\RN(R)$, also $\ia(\alpha)\cap\RN(P)\mathbin=\emptyset$.
  Using that $P$ is clash-free, by induction $P'$ is clash-free and $\RN(P')\subseteq \RN(P)$.
  Moreover, $\bn(\alpha) \cap \RN(P')=\emptyset$.
  Now $\RN(R') = \RN(P') \cup \{y\} \subseteq \RN(P) \cup \{y\} = \RN(R)$.\vspace{1ex}

  As $\bn(\alpha) \cap \RN(P')\mathbin=\emptyset$ and $y \notin\n(\alpha)$,
  $\bn(\alpha) \cap \RN(R')\mathbin=\emptyset$.\vspace{1ex}

  To show that $R'$ is clash-free, using that $P'$ is clash-free, it suffices to show
  (\ref{i}) $\Fn(R')\cap \RN(R')=\emptyset$ and (\ref{iii}) $y \notin \RN(P')$.
  The latter follows since $y \notin \RN(P)$ by the clash-freedom of $R$.
  The former follows exactly as in the case of \hyperlink{pi-relabelling}{\textsc{\textbf{\small par}}}.

\item Suppose $R \goto{\alpha}_\bullet R'$ is derived by rule \hyperlink{ES}{\textsc{\textbf{\small symb-open}}}.
  Then $R\mathbin=(\nu y)P$, $\alpha\mathbin=M\bar x(y)$, $P \goto{M \bar x y}_\bullet R'$,
  $y\mathbin{\neq} x$ and $y \mathbin{\notin} \n(M)$.
  Using that $P$ is clash-free, by induction $R'$ is clash-free and $\RN(R')\subseteq \RN(P)\subseteq \RN(R)$.
  Since $R$ is clash-free, $y \notin \RN(P)\supseteq \RN(R')$, so $\bn(\alpha)\cap \RN(R')=\emptyset$.

\item Suppose $R \goto{\alpha}_\bullet R'$ is derived by rule \hyperlink{pi-relabelling}{\textsc{\textbf{\small relabelling}}}.
  Then $R\mathbin=P[\sigma]$, $P \goto\beta_\bullet P'$, $\beta\as\sigma \mathbin= \alpha$, $\bn(\alpha) \cap \Fn(R)\mathbin=\emptyset$ and 
  $R'=P'[\sigma]$.
  Since $\ia(\alpha)\cap\RN(R)=\emptyset$, also $\ia(\alpha)\cap\RN(P)\mathbin=\emptyset$.
  Using that $P$ is clash-free, by induction $P'$ is clash-free, $\RN(P')\subseteq \RN(P)$ and $\bn(\alpha) \cap \RN(P')=\emptyset$.
  It follows that $\RN(R')\subseteq \RN(R)$.\vspace{1ex}

  Suppose $z \in \bn(\alpha) \cap \RN(R')$. Then there are $v\in \bn(\beta)$ and $w\mathbin\in\RN(P')$ with $v\as\sigma \mathbin= w\as\sigma$.
  As $\bn(\alpha) \cap \RN(P')\mathbin=\emptyset$ one has $v\neq w$,.
  By \lem{restriction bound names} $v \in \RN(P)$. Moreover, $w \in \RN(P)$. This contradicts Condition (\ref{iv})
  of the clash-freedom of $R$.\vspace{1ex}

  To show that $R'$ is clash-free, using that $P'$ is clash-free, it suffices to show
  (\ref{i}) $\Fn(R')\cap \RN(R')=\emptyset$ and (\ref{iv}) $y\as\sigma \neq z\as\sigma$ for different $y,z\in\RN(P')$.
  The latter follows from the same condition for $\RN(P)$. The former follows exactly as in the case of \hyperlink{pi-relabelling}{\textsc{\textbf{\small par}}}.
  \qed
\end{itemize}
\end{proof}
\vfill

Let $\fT_\bullet$ be the identity translation from the clash-free processes in $\piRp$ equipped with
the standard transition relation $\goto\alpha$ to clash-free processes in $\piRp$ equipped with the
alternative transition relation $\goto\alpha_\bullet$. The results above imply that $\fT_\bullet$ is
valid up to strong barbed bisimilarity.

\begin{theorem}\label{thm:eliminating alpha-conversion}
  $\fT_\bullet(P) \mathbin{\sbb} P$ for any clash-free $P \in \T_{\piRp}$.
\end{theorem}
\begin{proof}
  It suffices to show that the symmetric closure of 
  \[{\R}:=\{(U,R) \mid U,R \in \T_{\piRp} \wedge U \eqa R \wedge R \mbox{~clash-free}\}\]
  is a strong barbed bisimulation.  Take $(R,U)\in{\R}$.

  Let $U\goto\tau U'$.
  By \lem{eqa matches} $R\goto\tau R^\dagger$ for some $R^\dagger \eqa U'$.
  So $R\goto\tau_\bullet R'$ for some $R' \eqa R^\dagger$ by \lem{alpha eliminate}. 
  By \lem{preservation of clash-freedom} $R'$ is clash-free. Hence $(U',R')\in{\R}$.

  Let $R\goto\tau_\bullet R'$.
  Then $R\goto\tau R^\dagger$ for some $R^\dagger \eqa R'$ by \lem{alpha introduce}.
  By \lem{eqa matches} $U\goto\tau U'$ for some $U' \eqa R^\dagger$.
  By \lem{preservation of clash-freedom} $R'$ is clash-free. Hence $(U',R')\in{\R}$.

  Let $U{\downarrow_b}$ with $b\inp \nN\cup\overline\nN$.
  Then $U\goto{by} U'$ or $U\goto{b(y)} U'$ for some $y$ and $U'$, using the definition of $O$ in \sect{barbed}.
  By Lemmas~\ref{lem:input universality confluence} and~\ref{lem:bound output universality relabelling} I may assume,
  without loss of generality, that $y \notin \RN(R)$ in case of an input action $by$, and $y \notin \Fn(R)$ in case of a bound output action $b(y)$.
  Hence $R\goto{by}$ or $R\goto{b(y)}$ by \lem{eqa matches} and
  $R\goto{by}_\bullet$ or $R\goto{b(y)}_\bullet$ by \lem{alpha eliminate}.
  Thus $R{\downarrow_b}$.

  The implication $R{\downarrow_b} \Rightarrow U{\downarrow_b}$ proceeds likewise.
\end{proof}

\subsection{Eliminating restriction}
\label{sec:pir}\hypertarget{pir}{}

Let $\pir$ be the variant of $\piRp$ without restriction operators.
Hence there is no need for rules  \hyperlink{ES}{\textsc{\textbf{\small res}} and \textsc{\textbf{\small symb-open}}}.
Since the resulting semantics cannot generate transitions labelled $M\bar x(z)$, rule
\hyperlink{pi-relabelling}{\textsc{\textbf{\small e-s-close}}} can be dropped as well.
Moreover, $\bn(\alpha)\mathbin=\emptyset$ for all transition labels $\alpha$.
So the side condition of rules \hyperlink{pi-relabelling}{\textsc{\textbf{\small par}} and \textsc{\textbf{\small relabelling}}} can be dropped too.
Hence one is left with \tab{pi-relabelling} without the orange part.

In $\pir$ I also drop the restriction that $\fn(P)\subseteq\{x_1,\ldots,x_n\}$ in defining equations
\plat{$A(\vec x) \stackrel{{\rm def}}{=} P$}. Thus $A(\vec{x})$ can just as well be denoted $A$.%
\footnote{Here I assume that all sets $\K_n$ are disjoint,
  i.e., the same $\pi$-calculus agent identifier $A$ does occur with multiple arities. When this
  assumption is not met, an arity-index at the $\pir$ identifier $A$ is needed.}
On $\pir$ I don't use $\Fn$.

The fifth step $\fT_\nu$ of my translation goes from $\piRp$ to $\pir$.
It simply drops all restriction operators.
It is defined compositionally by $\fT_\nu((\nu y)P) = \fT_\nu(P)$
and $\fT_\nu(A(\vec{x}))=A_\nu(\vec{x})$, where $A_\nu$ is a fresh agent identifier with defining
equation \plat{$A_\nu(\vec x) \stackrel{{\rm def}}{=} \fT_\nu(P)$} when
\plat{$A(\vec x) \stackrel{{\rm def}}{=} P$} was the defining equation of $A$;
the translation $\fT_\nu$ acts homomorphically on all other constructs.

Although this translation in not valid in general,
I proceed to prove its validity for clash-free processes.
By \thm{eliminating alpha-conversion} I may use the transition relation $\goto\alpha_\bullet$ on $\piRp$.

For $\alpha$ an action in $\piRp$
one defines the \emph{debinding} of $\alpha$ by $\debind{\alpha}\mathbin{:=}\alpha$ if $\alpha$ has the form
$M\tau$, $Mxz$ or $M\bar x y$, and $\debind{M\bar x(y)}:=M\bar xy$.

\begin{lemma}\rm\label{lem:nu-out new}
If $R \goto{\alpha}_\bullet R'$,
then $\fT_\nu(R) \goto{\debind\alpha} \fT_\nu(R')$.
\end{lemma}
\begin{proof}
By induction on the derivation of $R \goto{\alpha}_\bullet R'$.
\begin{itemize}
\item Suppose $R \goto{\alpha}_\bullet R'$ is derived by rule \hyperlink{pi-relabelling}{\textsc{\textbf{\small early-input}}}.
  Then $R\mathbin=x(y).P$, $\alpha\mathbin=xz$, $R'=P[\renb{z}{y}^\aN]$ and $\fT_\nu(R)\mathbin=x(y).\fT_\nu(P)$.
  Moreover, $\fT_\nu(R)\goto{\debind\alpha} \fT_\nu(P)[\renb{z}{y}^\aN] = \fT_\nu(P[\renb{z}{y}^\aN])$.
\item The cases that  $R \goto{\alpha}_\bullet R'$ is derived by 
  \hyperlink{pi-relabelling}{\textsc{\textbf{\small tau}} or \textsc{\textbf{\small output}}} are even more trivial.
\item Suppose $R \goto{\alpha}_\bullet R'$ is derived by \hyperlink{pi-relabelling}{\textsc{\textbf{\small ide}}}.
  Then $R=A(\vec{x})$, say with \plat{$A(\vec{x}) \stackrel{{\rm def}}{=} P$},
  and therefore $P\goto{\alpha}_\bullet R'$.
  Moreover $\fT_\nu(R)=A_\nu(\vec{x})$, with $A_\nu$ defined by \plat{$A_\nu(\vec{x}) \stackrel{{\rm def}}{=} \fT_\nu(P)$}.
  Now $\fT_\nu(P)\goto{\debind\alpha} \fT_\nu(R')$ by induction.
  With rule \hyperlink{pi-relabelling}{\textsc{\textbf{\small ide}}} one infers
  \plat{$\fT_\nu(R) \goto{\debind\alpha} \fT_\nu(R')$}.
\item The cases that  $R \goto{\alpha}_\bullet R'$ is derived by 
  \hyperlink{pi-relabelling}{\textsc{\textbf{\small sum}}, \textsc{\textbf{\small symb-match}},
  \textsc{\textbf{\small par}} or \textsc{\textbf{\small e-s-com}}} are trivial.
\item Suppose $R \goto{\alpha}_\bullet R'$ is derived by \hyperlink{pi-relabelling}{\textsc{\textbf{\small e-s-close}}}.
  Then $R=P|Q$, $\alpha=\match{x}{v}MN\tau$, $P\goto{M\bar x(z)}_\bullet P'$, $Q\goto{Nvz}_\bullet Q'$,
  $R'\mathbin=(\nu z)(P'|Q')$ and $\fT_\nu(R)\mathbin=\fT_\nu(P)|\fT_\nu(Q)$.
  Now $\fT_\nu(P)\goto{M\bar xz} \fT_\nu(P')$ and $\fT_\nu(Q)\goto{Nvz} \fT_\nu(Q')$ by induction.
  Thus  $\fT_\nu(R) \goto{\debind\alpha} \fT_\nu(P')|\fT_\nu(Q')=\fT_\nu(R')$
  by application of rule \hyperlink{pi-relabelling}{\textsc{\textbf{\small e-s-com}}}.
\item Suppose $R \goto{\alpha}_\bullet R'$ is derived by \hyperlink{ES}{\textsc{\textbf{\small res}}}.
  Then $R=(\nu y)P$, $R'=(\nu y)P'$, $P\goto{\alpha}_\bullet P'$. 
  Now $\fT_\nu(P)\goto{\debind\alpha} \fT_\nu(P')$ by induction.
  So $\fT_\nu(R) = \fT_\nu(P) \goto{\debind\alpha} \fT_\nu(P') = \fT_\nu(R')$.
\item Suppose $R \goto{\alpha}_\bullet R'$ is derived by \hyperlink{ES}{\textsc{\textbf{\small symb-open}}}.
  Then $R=(\nu y)P$, $\alpha=M\bar x (y)$ and $P\goto{M\bar x y}_\bullet R'$.
  By induction $\fT_\nu(P)\goto{M\bar x y} \fT_\nu(R')$.
  So $\fT_\nu(R) \mathbin= \fT_\nu(P) \goto{\debind\alpha} \fT_\nu(R')$.
\item Suppose $R \goto{\alpha}_\bullet R'$ is derived by \hyperlink{pi-relabelling}{\textsc{\textbf{\small relabelling}}}.
  Then $R=P[\sigma]$, $P\goto{\beta}_\bullet P'$, $\beta\as\sigma = \alpha$ and $R'=P'[\sigma]$.
  By induction $\fT_\nu(P)\goto{\debind\beta} \fT_\nu(P')$.
  Hence $\fT_\nu(R) = \fT_\nu(P)[\sigma]\goto{\debind\alpha} \fT_\nu(P')[\sigma]=\fT_\nu(R')$.
\qed
\end{itemize}
\end{proof}
The next lemma makes use of the set $\no(\beta)$ of \emph{non-output} names of an action $\beta$ in
\plat{$\pir$}.
Here $\no(\beta):=\n(\beta)$ if $\beta$ has the form $M\tau$ or $Mxz$, whereas
$\no(M\bar x y):=\n(M)\cup\{x\}$.

\begin{lemma}\rm\label{lem:nu-in new}
  If $R$ is clash-free and $\fT_\nu(R) \goto{\beta} U$,
  where $\no(\beta)\cap\RN(R)=\emptyset$,
  then \plat{$R \goto{\alpha}_\bullet R'$} for some $\alpha$ and
  $R'$ with $\debind\alpha=\beta$ and $\fT_\nu(R') = U$.
\end{lemma}
\begin{proof}
By induction on the derivation of $\fT_\nu(R) \goto{\beta} U$, and a nested structural
induction on $R$.
\begin{itemize}
\item The cases $R=\tau.P$ and $R=\bar x y.P$ are trivial.
\item Let $R=x(y).P$.
  Then $\fT_\nu(R)=x(y).\fT_\nu(P)$.
  Hence $\beta\mathbin=xz$ and $U\mathbin=\fT_\nu(P)[\renb{z}{y}^\aN]$.
  Furthermore, $R\mathbin{\goto{xz}_\bullet} P[\renb{z}{y}^\aN]$ and
  $\fT_\nu(P[\renb{z}{y}]^\aN)\mathbin=\fT_\nu(P)[\renb{z}{y}^\aN] = U$.
\item 
  Let $R=A(\vec{x})$ with \plat{$A(\vec{x}) \stackrel{{\rm def}}{=} P$}.
  Then $\fT_\nu(R)=A_\nu(\vec{x})$ with  \plat{$A_\nu(\vec{x}) \stackrel{{\rm def}}{=} \fT_\nu(P)$}.
  So $\fT_\nu(P)\goto{\beta} U$. 
  Process $P$ is clash-free by \df{clash}(\ref{vii}), and $\RN(P)\mathbin=\RN(R)$ by \df{RN}.
  So $\no(\beta)\cap\RN(P[\renbt{y}{z}^\aN])=\emptyset$.
  Thus, by induction, $P\goto{\alpha}_\bullet R'$ for some $\alpha$ and
  $R'$ with $\debind\alpha\mathbin=\beta$ and $\fT_\nu(R') \mathbin= U$.
  Now \plat{$R\goto{\alpha}_\bullet R'$} by \hyperlink{pi-relabelling}{\textsc{\textbf{\small ide}}}.
\item
  The cases that $R$ is $\nil$, $P+Q$ or $\Match{x}{y}P$ are trivial.
\item Let $R \mathbin= P[\sigma]$. Then $\fT_\nu(R)\mathbin=\fT(P)[\sigma]$ and $\fT(P)\mathbin{\goto{\delta}}V$
  for some $\delta$ with $\delta[\sigma]=\beta$, and some $V$ with $V[\sigma]=U$.
  By definition $P$ is clash-free.
  Considering that $\RN(R)\supseteq\{y\as\sigma \mid y \inp \RN(P)\}$, one has $\no(\delta)\cap\RN(P)=\emptyset$.
  So by induction \plat{$P\goto{\gamma}_\bullet P'$} for some $\gamma$ and
  $P'$ with $\debind\gamma=\delta$ and $\fT_\nu(P') = V$.
  Let $\alpha := \gamma\as{\sigma}$. Then
  $\debind\alpha = \debind{\gamma\as{\sigma}} = \debind\gamma\as{\sigma} =\delta\as{\sigma}=\beta$.
  In case $z\in\bn(\alpha)$ then $z=w\as\sigma$ with $w \in\fn(\gamma)$,
  so $w \in \RN(P)$ by \lem{restriction bound names} and hence $z\mathbin\in\RN(R)$, so
  the clash-freedom of $R$ implies $z \mathbin{\notin}\Fn(R)$.
  Therefore, by \hyperlink{pi-relabelling}{\textsc{\textbf{\small relabelling}}}, $R = P[\sigma]\goto{\alpha}_\bullet P'[\sigma]$,
  and $\fT_\nu(P'[\sigma]) =\fT_\nu(P')[\sigma] = V[\sigma] =U$.
\item
  Let $R=P|Q$. Then $\fT_\nu(R)=\fT_\nu(P)|\fT_\nu(Q)$. Suppose $\fT_\nu(R) \goto{\beta} U$
  is derived by \hyperlink{pi-relabelling}{\textsc{\textbf{\small par}}}.
  Then $\fT_\nu(P) \goto{\beta} V$ and $U=V|\fT_\nu(Q)$. By definition $P$ is clash-free.
Since $\RN(P)\subseteq\RN(P|Q)$, $\no(\beta)\cap\RN(P)=\emptyset$.
So by induction \plat{$P \goto{\alpha}_\bullet P'$} for some $\alpha$ and $P'$ with
$\debind\alpha\mathbin=\beta$ and $\fT_\nu(P') = V$.
By \lem{restriction bound names} $\bn(\alpha)\subseteq\RN(P)\subseteq\RN(R)$, so
$\bn(\alpha)\cap\Fn(R)=\emptyset$ by the clash-freedom of $R$.
Thus $R \goto{\alpha}_\bullet P'|Q$ by \hyperlink{pi-relabelling}{\textsc{\textbf{\small par}}},
and $\fT_\nu(P'|Q) = V|\fT_\nu(Q)\mathbin=U$.\vspace{3pt}

  Now suppose $\fT_\nu(R) \goto{\beta} U$ is derived by \hyperlink{pi-relabelling}{\textsc{\textbf{\small e-s-com}}}.
  Then $\beta=\match{x}{v}MN\tau$, $\fT_\nu(P)\goto{M\bar xy} V$, $\fT_\nu(Q)\goto{Nvy} W$ and $U\mathbin=V|W$.
  By definition $P$ and $Q$ are clash-free.
  Since $\RN(P)\subseteq\RN(P|Q)$, $\no(M\bar x y)\cap\RN(P)=\emptyset$.
  So by induction either \plat{$P \goto{M\bar x y}_\bullet P'$}
  or \plat{$P \goto{M\bar x (y)}_\bullet P'$} for some $P'$ with $\fT_\nu(P') = V$.

  In the first case $y\in\Fn(P)\subseteq\Fn(R)$ by \lem{free names actions relabelling}, so $y\notin\RN(R)\supseteq\RN(Q)$
  by the clash-freedom of $R$. Hence also $\no(Nv y)\cap\RN(Q)=\emptyset$.
  By induction \plat{$Q \goto{Nv y}_\bullet Q'$} for some $Q'$ with $\fT_\nu(Q')=W$.
  So $R \goto{\alpha}_\bullet P'|Q'$ by \hyperlink{pi-relabelling}{\textsc{\textbf{\small e-s-com}}},
  and $\fT_\nu(P'|Q') = V|W\mathbin=U$.

  In the second case $y\in\RN(P)\subseteq\RN(R)$ by \lem{restriction bound names},
  so $y\notin\Fn(R)\cup\RN(Q)$ by the clash-freedom of $R$.
  Hence $\no(Nv y)\cap\RN(Q)=\emptyset$.
  By induction \plat{$Q \goto{Nv y}_\bullet Q'$} for some $Q'$ with $\fT_\nu(Q')=W$.
  So $R \goto{\alpha}_\bullet (\nu y)(P'|Q')$ by \hyperlink{pi-relabelling}{\textsc{\textbf{\small e-s-close}}},
  and $\fT_\nu((\nu y)(P'|Q'))=\fT_\nu(P'|Q') = V|W\mathbin=U$.
\item Finally, let $R=(\nu y)P$. Then $\fT_\nu(R)=\fT_\nu(P)$.
  Moreover, $\no(\beta)\cap\RN(P)=\emptyset$.
  By induction, \plat{$P \goto{\alpha}_\bullet P'$} for some $\alpha$ and
  $P'$ with $\debind\alpha=\beta$ and $\fT_\nu(P') = U$.

  In case $y\notin\n(\alpha)$, $R \goto{\alpha}_\bullet (\nu y)P'$ by \hyperlink{ES}{\textsc{\textbf{\small res}}}.
  Moreover, $\fT_\nu((\nu y)P')=\fT_\nu(P') = U$.

  In case $y\in\n(\alpha)=\n(\beta)$, using that $\RN(R)\mathbin\ni y \mathbin{\notin} \no(\beta)$,\linebreak[3]
  $\beta$ must have the form $M\bar x y$ with $y\neq x$ and $y\notin\n(M)$.
  So $\alpha$ is either $M\bar x (y)$ or $M\bar x y$. If $\alpha=M\bar x (y)$ then $y\in\RN(P)$ by \lem{restriction bound names},
  contradicting the clash-freedom of $R$. So $\alpha=M\bar x y$.
  Now $R\goto{M\bar x (y)}_\bullet P'$ by \hyperlink{ES}{\textsc{\textbf{\small symb-open}}}.
\qed
\end{itemize}
\end{proof}

\begin{theorem}\label{thm:step 6}
If $P$ is clash-free then $\fT_\nu(P)\sbb P$.
\end{theorem}
\begin{proof}
  It suffices to show that the symmetric closure of 
  \[{\R}:=\left\{(R,\fT_\nu(R)) \mid R \mbox{~in~} \piRp \mbox{~is clash-free}\right\}.\]
  is a strong barbed bisimulation.
  So suppose $R$ is clash-free.

  Let $R\goto\tau_\bullet R'$.
  Then $\fT_\nu(R)\goto{\tau} \fT_\nu(R')$ by \lem{nu-out new}.
  Moreover, $R'$ is clash-free by \lem{preservation of clash-freedom}, so $R' \R \fT_\nu(R')$.

  Let $\fT_\nu(R) \goto{\tau} U'$.
  Then, by \lem{nu-in new}, \plat{$R \goto{\tau}_\bullet R'$} for some $R'$ with $\fT_\nu(R')= U'$.
  Moreover, $R'$ is clash-free by \lem{preservation of clash-freedom}, so $R' \R U'$.

  Now let $R{\downarrow_b}$ with $b\in \nN\cup\overline\nN$.
  Then $R\goto{by}_\bullet$ or $R\goto{b(y)}_\bullet$ for some $y$, using the definition of $O$ in \sect{barbed}.
  So $\fT_\nu(R)\goto{by}$ by \lem{nu-out new}.  Thus $\fT_\nu(R){\downarrow_b}$.

  Finally, let $\fT_\nu(R){\downarrow_b}$ with $b=x\in\zN$ or $b=\bar x$ with $x\in\overline\nN$.
  Then $\fT_\nu(R)\goto{by} U'$ for some $y$ and $U'$.
  Since $R$ is clash-free, $\RN(R)\mathbin\subseteq \pN$, so $x\mathbin{\notin}\RN(P)$.
  By \lem{input universality confluence} I may assume that if $b\mathbin=x$ then $y\mathbin{\notin}\RN(R)$.
  Hence $\no(by)\cap\RN(R)=\emptyset$.
  Thus, by \lem{nu-in new}, $R\goto{by}_\bullet R'$ or $R\goto{b(y)}_\bullet R'$ for some $R'$.
  Hence $R{\downarrow_b}$.
\end{proof}

\subsection{The last step}
\label{sec:ccstrig}

The language $\pir$ can almost be recognised as an instance of {\CCST}.
Let $Act$ be the set of all actions $M\tau$, $M\bar xy$ and $Mxy$ with names from $\M$.
As parameters of {\CCST} I take $\K$ to be the disjoint union of all the sets $\K_n$ for $n \in \IN$,
of $n$-ary agent identifiers from the chosen instance of the $\pi$-calculus, and
$\A:=Act{\setminus}\{\tau\}$. The set $\Sy\subseteq\A$ of synchronisations consists of all actions $M\tau$ with $M\neq\varepsilon$.
The communication function  $\gamma:(\A{\setminus}\Sy)^2\mathbin\rightharpoonup \Sy\cup\{\tau\}$ is given by
$\gamma(M\bar xy,Nvy)=\match{x}{v}MN\tau$, and its commutative variant.
Now the parallel composition of \plat{$\pir$} turns out to be the same as for this instance of {\CCST}.
Likewise, the silent and output prefixes are instances of {\CCST} prefixing,
and the agent identifiers of $\pir$ are no different from {\CCST} agent identifiers.
However, the input prefix of \plat{$\pir$} does not occur in {\CCST}.
Yet, one can identify $Mx(y).P$ with $\sum_{z\in\M} Mxz.(P[\renb{z}{y}^\aN])$,
for both processes have the very same outgoing transitions.
The $\pir$ matching operator is
no different from the triggering operator of {\sc Meije} or {\CCST} (see \sect{triggering}): both rename only
the first actions their argument process can perform, namely by adding a single match $\match{x}{y}$
in front of each of them---this match is suppressed when $x{=}y$.

This yields to the following translation from $\pir$ to {\CCST}:
\[\begin{array}{@{}l@{~:=~}ll@{}}
\fT_\gamma(\nil) & \nil \\
\fT_\gamma(\textcolor{DarkBlue}{M}\tau.P) & \textcolor{DarkBlue}{M}\tau.\fT_\gamma(P) \\
\fT_\gamma(\textcolor{DarkBlue}{M}\bar xy.P) & \textcolor{DarkBlue}{M}\bar xy.\fT_\gamma(P) \\
\fT_\gamma(\textcolor{DarkBlue}{M}x(y).P) & \sum_{z\in\M} \textcolor{DarkBlue}{M}xz.\big(\fT_\gamma(P)[\renb{z}{y}]\big) \\
\color{purple}\fT_\gamma(\Match{x}{y}P) & \color{purple}\match{x}{y}{\Rightarrow}\fT_\gamma(P)\\
\fT_\gamma(P\mid Q) & \fT_\gamma(P)\|\fT_\gamma(Q) \\
\fT_\gamma(P+Q) & \fT_\gamma(P)+\fT_\gamma(Q) \\
\fT_\gamma (A) & A \\
\end{array}\]
where the {\CCST} defining equations $A=\fT_\gamma(P)$ of agent identifiers $A$ are inherited
verbatim from $\pir$.

Here the use of the triggering operator can be avoided by restricting attention to the
$\pi$-calculus with implicit matching. For that language the clause for $\fT_\gamma(\Match{x}{y}P)$ can
be dropped, at the expense of the addition of the blue $\textcolor{DarkBlue}{M}$s above, which are absent when dealing with the
full $\pi$-calculus.

\begin{theorem}
$\fT_\gamma(P)\bis{} P$ for each $\pir$ process $P$.
\end{theorem}
\begin{proof}
Trivial.
\end{proof}

Putting all steps of my translation from $\pi_L(\N)$ to {\CCST} together, I obtain
\[\begin{array}{@{}l@{~:=~}ll@{}}
\fT(\nil) & \nil \\
\fT(\textcolor{DarkBlue}{M}\tau.P) & \textcolor{DarkBlue}{M}\tau.\fT(P) \\
\fT(\textcolor{DarkBlue}{M}\bar xy.P) & \textcolor{DarkBlue}{M}\bar xy.\fT(P) \\
\fT(\textcolor{DarkBlue}{M}x(y).P) & \sum_{z\in\M}\textcolor{DarkBlue}{M}xz.\big(\fT(P)[\renb{z}{y}^\aN]\big) \\
\fT((\nu y)P) & \fT(P)[p_y]\\
\color{purple}\fT(\Match{x}{y}P) & \color{purple}\Match{x}{y}{\Rightarrow}\fT(P) \\
\fT(P\mid Q) & \fT(P)[\ell]\mathbin{\|}\fT(Q)[r] \\
\fT(P+Q) & \fT(P)+\fT(Q) \\
\fT(A(\vec y)) & A[\renb{\vec{y}}{\vec{x}}^\aN]  \\
\end{array}\]
where the {\CCP} agent identifier $A$ has the defining equation $A=\fT(P)$ when
\plat{$A(\vec{x}) \stackrel{{\rm def}}{=} P$} was the defining equation of the $\pi_L(\N)$ agent identifier $A$.
Abbreviating $[\renb{z}{y}^\aN]$ by $\rensq{z}{y}$ and $[\renbt{y}{x}^\aN]$ by $[\rename{\vec{y}}{\vec{x}}]$,
this is the translation presented in \sect{the encoding}.

\end{document}